\documentclass[
aps,
twocolumn,
prx,
notitlepage,
superscriptaddress,
floatfix,
longbibliography,
nofootinbib
]{revtex4-2}
\usepackage{ragged2e}
\usepackage{array}
\usepackage{comment}

\usepackage{etoolbox}
\patchcmd{\section}
  {\centering}
  {\raggedright}
  {}
  {}
\patchcmd{\subsection}
  {\centering}
  {\raggedright}
  {}
  {}
  \patchcmd{\subsubsection}
  {\centering}
  {\raggedright}
  {}
  {}

\usepackage{graphicx}
\usepackage[toc,page]{appendix}
\usepackage{bm,bbm}
\usepackage{physics}
\usepackage{amssymb}
\usepackage{amsmath}
\usepackage{amsthm}
\usepackage{dsfont}
\usepackage[dvipsnames]{xcolor}
\usepackage{soul}
\usepackage{multirow}
\usepackage{mathtools}
\usepackage{textgreek}
\usepackage{tikz}
\usetikzlibrary{cd}
\usetikzlibrary{quantikz2}
\usepackage[colorlinks, linkcolor = blue, citecolor = blue, urlcolor=blue, pdfborder={0 0 0 [0 0]},bookmarks=false]{hyperref}
\usepackage[normalem]{ulem}
\usepackage{soul}
\usepackage{algorithm}
\usepackage{algpseudocode}
\newtheorem{theorem}{Theorem}
\newtheorem{lemma}{Lemma}

\newtheorem{cor}{Corollary}
\newtheorem{definition}{Definition}

\newcommand{\rc}[1]{{\color{red}{#1}}}
\newcommand{\bc}[1]{{\color{blue}{#1}}}
\newcommand{\gc}[1]{{\color{ForestGreen}{#1}}}
\definecolor{amber}{rgb}{0.8, 0.6, 0.0}

\newcommand{\be}{\begin{equation}}
\newcommand{\ee}{\end{equation}}

\newcommand{\bx}{\overline{\mathcal{B}}}
\newcommand{\nx}{\mathcal{N}}
\newcommand{\zz}{\mathbb{Z}}

\DeclareMathOperator{\Aut}{Aut}

\makeatletter
\def\l@subsection#1#2{}
\def\l@subsubsection#1#2{}
\makeatother

\begin{document}

\title{Universal fault tolerant quantum computation in 2D without getting tied in knots}
\author{Margarita Davydova}
\affiliation{\footnotesize Walter Burke Institute for Theoretical Physics and Institute for Quantum Information and Matter,
 California Institute of Technology, Pasadena, CA 91125}
 \affiliation{\footnotesize IBM Quantum, T. J. Watson Research Center, Yorktown Heights, New York 10598, USA}
 \affiliation{\footnotesize Department of Physics, Massachusetts Institute of Technology, Cambridge, MA 02139, USA}
 
\author{Andreas Bauer}
\affiliation{\footnotesize Department of Mechanical Engineering, Massachusetts Institute of Technology, Cambridge, MA 02139, USA}
\author{Julio C. Magdalena de la Fuente}
\affiliation{\footnotesize Dahlem Center for Complex Quantum Systems, Freie Universit\"at Berlin, 14195 Berlin, Germany}
\author{Mark Webster}
\affiliation{\footnotesize Department of Physics \& Astronomy, University College London, London, WC1E 6BT, United Kingdom}
\affiliation{\footnotesize School of Physics, The University of Sydney, NSW 2006, Australia}
\author{Dominic J. Williamson}
\affiliation{\footnotesize School of Physics, The University of Sydney, NSW 2006, Australia}
\affiliation{\footnotesize IBM Quantum, IBM Almaden Research Center, San Jose, CA 95120, USA}
\author{Benjamin J.~Brown}
\affiliation{\footnotesize IBM Quantum, T. J. Watson Research Center, Yorktown Heights, New York 10598, USA}
\affiliation{\footnotesize IBM Denmark, Sundkrogsgade 11, 2100 Copenhagen, Denmark}

\begin{abstract}
We show how to perform scalable fault-tolerant non-Clifford gates in two dimensions by introducing domain walls between the surface code and a non-Abelian topological code whose codespace is stabilized by Clifford operators. We formulate a path integral framework which provides both a macroscopic picture for different logical gates as well as a way to derive the associated microscopic circuits. We also show an equivalence between our approach and prior proposals where a 2D array of qubits reproduces the action of a transversal gate in a 3D stabilizer code over time, thus, establishing a new connection between 3D codes and 2D non-Abelian topological phases. We prove a threshold theorem for our protocols under local stochastic circuit noise using a just-in-time decoder to correct the non-Abelian code.
\end{abstract}

\maketitle
{
\small \tableofcontents
}

\section{Introduction}

A scalable quantum computer must have a universal set of fault-tolerant quantum logic gates. These gates require a practical scheme for their execution with quantum error correcting codes that have low-depth circuits for syndrome readout as well as an efficient decoding strategy~\cite{Terhal2015}. Topological codes~\cite{Kitaev1997, Dennis2002, Levin2004} are promising in this regard thanks to their local connectivity and efficient high-threshold decoding algorithms. Furthermore, their rich underlying physics makes them an excellent playground for designing new fault-tolerant logic gates~\cite{Kitaev1997, nayak2008non, Barkeshli2019symmetry}.

In this work, we show how to perform non-Clifford logic gates by interfacing surface codes with a topological code stabilized by Clifford operators.
The codespace of the latter code is the ground state space 
of the type-III twisted quantum double model~\cite{Hu_2013, Yoshida2016topological} which represents a \textit{non-Abelian} topological phase~\cite{Dijkgraaf_1990,Propitius_1995}.
We derive circuits for syndrome readout for this non-Abelian code, together with an efficient decoder that demonstrates a threshold. The logical information at the beginning and end of the protocol is encoded in copies of the surface code or an equivalent code in the same topological phase. 
A universal set of logic gates for the surface codes can therefore be completed using standard methods to perform Clifford gates in two dimensions~\cite{Dennis2002, Fowler2012, Horsman_2012,Litinski2019gameofsurfacecodes,Raussendorf2007,Bombin_2009,Brown2017poking,Bombin2006,yoshida2015gapped,Kubica2015,Potter17,Barkeshli23}.

Our results can be interpreted using a number of different perspectives, each with their own advantages. First, we can view certain logic gates as a code deformation between the surface code and the non-Abelian code. This mechanism can also be viewed as an instance of a gauging logical measurement. We discuss the code deformation and the gauging perspectives in Section~\ref{sec:CZ_minimal}. 
More generally, we can design a multitude of non-Clifford logic gates by changing boundary and domain wall configurations in spacetime using the path integral formalism. We present this picture in Section~\ref{section:loopsum}. Examples include logical $CCZ$ and $T$ gates as well as magic $CZ$- and $T$-state preparation.   In Section~\ref{section:examples}, we use the path integral approach to systematically derive microscopic circuits for the logical protocols.
Finally, we show a close relationship between the logic gates presented here and prior work on non-Clifford gates in two-dimensional codes~\cite{bombin2018,Brown2020universal}. There, the transversal non-Clifford gates of a three-dimensional code are turned into linear-depth protocols on a two-dimensional array of qubits. In Section~\ref{sec:3d-2d}, we show that this approach is a special instance that follows from our general results. Thus, both of the examples in Refs.~\cite{bombin2018,Brown2020universal} realize the same non-Abelian code during their intermediate steps, which reveals a new connection between 3D codes with non-Clifford transversal gates and 2D non-Abelian topological phases.

Unlike conventional stabilizer codes, non-Abelian codes require more sophisticated decoding techniques. 
We prove that our logic operations exhibit a threshold when the circuits are realized using a just-in-time decoder~\cite{bombin2018, Brown2020universal, Scruby_2022, scruby2025faulttolerantquantumcomputationdistillation} in Section~\ref{Sec:JustInTime}. Although the decoders we use are functionally similar to those presented in prior work, our results use just-in-time decoding in the context of non-Abelian codes. We also argue for fault tolerance in the end-to-end implementation of the full circuits that implement our logic gates.

Our results offer an alternative to magic state distillation~\cite{knill2004faulttolerantpostselectedquantumcomputation,PhysRevA.71.022316,meier2012magicstatedistillationfourqubitcode,Campbell_2017}. Non-Clifford gates and magic states are widely regarded to be the most resource-intensive components in two-dimensional quantum computing architectures at extremely low logical error rates, as well as in a number of other settings~\cite{Campbell_2017,BravyiKoenig,PhysRevA.91.012305,Webster2022}. Further development of our proposal might offer a reduction in the resource cost of a scalable quantum computer.

It is instructive to contrast our results with existing schemes for topological quantum computation by braiding anyons or holes of a non-Abelian phase~\cite{Kitaev1997, Freedman2002, nayak2008non, Laubscher2019}. 
First, the microscopic circuits for braiding-universal phases, such as the Fibonacci model, require deep syndrome extraction circuits~\cite{Bonesteel2012, PhysRevX.12.021012,Bseiso2024,Minev2024}. A round of syndrome extraction in our protocols has comparatively low depth. Furthermore, the circuits can be expressed using CNOT operations, single-qubit Pauli measurements, and non-Clifford $T$, $CS$ or $CCZ$ gates. Second, braiding-universal topological phases require sophisticated strategies for decoding and correction~\cite{Schotte2023, Dauphinais2016}.
The nilpotency of the TQD topological phase used in our protocols allows us to design straight-forward decoding strategies using a just-in-time decoder~\cite{bombin2018, Brown2020universal}.

We can also compare our scheme to proposals that prepare magic states by measuring macroscopic logical operators directly~\cite{Cui2015,Cong2017, Chamberland_2019,gidney2024magicstatecultivationgrowing}. In our work, we infer the value of high-weight logical operators using low-weight checks. This gives us a scalable way to identify and correct errors as the gate is conducted. In this context, our results can be viewed as a generalization of logic gates by lattice surgery~\cite{Horsman_2012,Cohen2022low,Williamson2024Gauging}, where the Clifford logical operators are measured using local operations to complete the universal gate set.

Our results offer a new point of view on universal quantum computation in two dimensions along with a variety of fault-tolerant implementations. We argue that our approach can be extended to other non-Abelian phases and more general qLDPC codes. This provides a new avenue to realize universal fault-tolerant quantum computation under geometric locality restrictions that are inherited from physical architectures.

\section{Sketch of the idea}
\label{sec:description}

The protocols in this paper all use a non-Abelian topological phase to perform non-Clifford gates on copies of the surface or toric code. Here, we give a high-level overview of the main idea with two expository examples.

We use the \emph{type-III twisted quantum double} (TQD)\footnote{That is, the twisted quantum double model associated with the type-III 3-cocycle $[\omega_{III}]\in H^3(\mathbb{Z}_2\times\mathbb{Z}_2\times\mathbb{Z}_2,U(1))$ given by $\omega_{III}((a,b,c),(a',b',c'),(a'',b'',c''))=(-1)^{ab'c''}$.}, which is related to three copies of the toric code phase by a 3-cocycle twist~\cite{Dijkgraaf_1990,Hu_2013}. This model represents the same topological phase as the quantum double of $D_4$~\cite{Propitius_1995, Kitaev1997, PhysRevB.95.035131}.
We show how to use domain walls to transfer logical information encoded in toric or surface codes~\cite{Kitaev1997, Bravyi1998, Dennis2002, Fowler2012} into the TQD and back again to perform non-Clifford gates. While information is encoded in the twisted quantum double model, we take advantage of its twist to introduce a relative phase factor in the encoded state to implement non-Clifford gates.

The first expository example, shown in Fig.~\ref{fig:protocol2}, realizes the non-Clifford unitary $\overline{CCZ}$ gate. The spacetime layout is based on that in Ref.~\cite{Brown2020universal}; however, we derive the action of the gate by considering the twisted quantum double as an intermediate two-dimensional phase. Fig.~\ref{fig:protocol2}(b) shows how the gate can be viewed as ``sweeping'' the third copy of the surface code (labeled green) past the other two (labeled red and blue). At the points in time where all three of the surface code regions overlap, we realize the twisted quantum double, which is responsible for a relative phase over the course of the operation.
Specifically, if and only if all three codes start in the logical-$1$ state, a $-1$ phase is acquired. As a result, the protocol applies the phase $(-1)^{{r} {g} {b}}$, which is the same as applying the logical $ \overline{CCZ}_{rgb}$ gate.

Throughout our work, it is helpful to view our logic operations using a 2+1-dimensional spacetime perspective, which we explore systematically in Sec.~\ref{section:loopsum}. 
In spacetime, the relative phase responsible for the logical action appears as a triple intersection between the membranes associated with the logical states.
The logical states of the toric codes are in correspondence with continuous nontrivial loops in space, which turn into continuous membranes when viewed in spacetime. Likewise, the states in the non-Abelian phase are defined by continuous loops (membranes in spacetime) that we label with their three respective colors. However, unlike in the toric/surface codes phase, every point in spacetime where all three types of membranes intersect (we call this a``triple intersection point'') acquires a $-1$ phase due to the 3-cocycle twist.
The spacetime configuration of topological codes and their domain walls and boundaries shown in Fig.~\ref{fig:protocol2}(a) is chosen such that for the all-$1$ input logical state, there is strictly an odd number of logical membrane intersections, which gives rise to the $\overline{CCZ}$ logic gate.
We explain this gate in detail in Sec.~\ref{section:loopsum}. In App.~\ref{subsec:loopsum-and-stabs} we give more details on the loop-sum picture.

\begin{figure}[!t]
\includegraphics[width=1.0 \columnwidth]{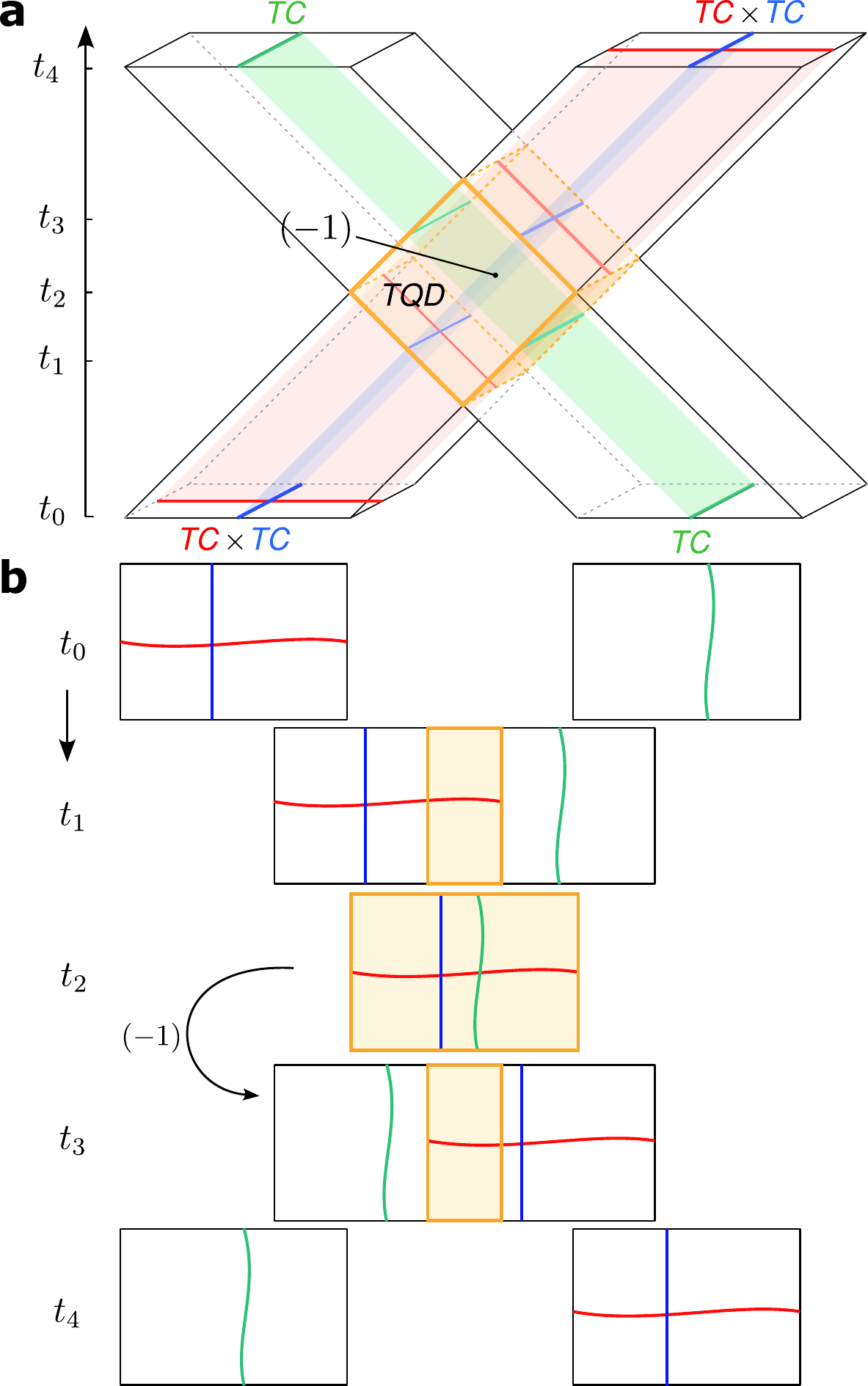}
\caption{(a) The spacetime picture of the protocol that fault-tolerantly implements the non-Clifford $\overline{CCZ}$ gate on three copies of the surface code in $O(d)$ time. 
The spacetime region containing the non-Abelian twisted quantum double phase is shaded in yellow. 
A particular boundary configuration of three surface codes is initialized at time $t_0$. Upon completion of the protocol, the state acquires a $-1$ phase only if all three surface codes are simultaneously in the $\ket{1}$ logical state, corresponding to  $\overline{CCZ}_{\rc{r} \gc{g} \bc{b}}$. 
(b) A step-by-step depiction of the protocol, which can be interpreted as ``sweeping'' one copy of the surface code past two others. The copies are coupled by a twist in the yellow-shaded region, where the twisted quantum double is realized. 
}
\label{fig:protocol2}
\end{figure}

We now turn to the second example, which is explained together with microscopic details in Sec.~\ref{sec:CZ_minimal}. It simply projects two copies of the toric code onto the twisted quantum double model on the torus before reversing the operation, as shown in Fig.~\ref{fig:protocol1}. The resulting operation can be used to prepare a magic state known as the $CZ$-state~\cite{Dennis2001toward,Gupta2024encoding}.

The protocol is defined with two copies of the toric codes as an input, where each code encodes two logical qubits. We label the two copies red and blue, and index the encoded qubits by the color and the orientation of their respective logical Pauli-$X$ string operator, i.e. $\{ \rc{rH},  \bc{bV}, \rc{rV}, \bc{bH}\}$.   Then, the protocol transforms from the toric code phase to the twisted quantum double phase and back, which performs a logical measurement, projecting the input state onto an eigenstate of the logical $\overline{CZ}_{\rc{rH},\bc{ bV}}\overline{CZ}_{\rc{rV} ,\bc{bH}}$ operator.  
Suppose that we start with the input state $|+\rangle_\rc{rH}|+\rangle_\bc{ bV}|0\rangle_\rc{rV}|0\rangle_\bc{ bH}$, then this measurement probabilistically prepares the state $|CZ\rangle_{\rc{rH}, \bc{ bV}}|0\rangle_\rc{rV}| 0\rangle_\bc{ bH}$ where the $CZ$-magic state $|CZ\rangle$ is an equal-weight superposition of the $+1$ eigenstates of the $\overline{CZ}$ operator: $\ket {\overline{CZ}} \propto \ket{00} + \ket{01} + \ket {10}$.

\begin{figure}[t]
\includegraphics[width=1.0 \columnwidth]{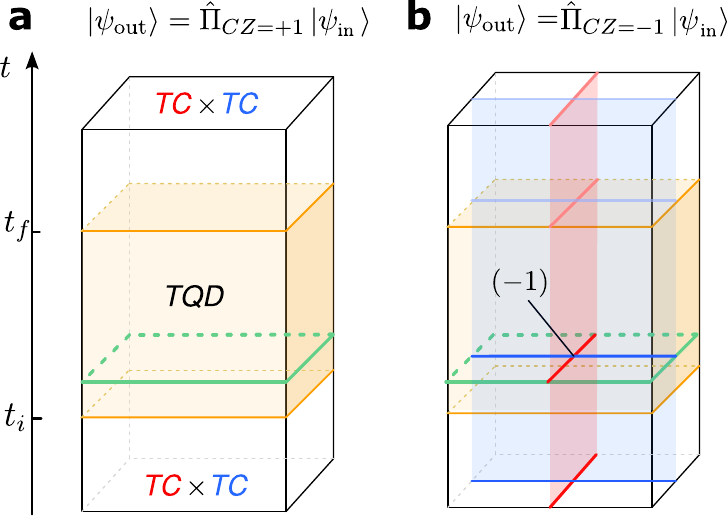}
\caption{Schematic spacetime diagram of the protocol that projects an input logical $4$-qubit state onto an eigenstate of the $\overline{CZ}_{\rc{rH}\bc{ bV}}\overline{CZ}_{\rc{rV} \bc{bH}}$ operator. The protocol transitions from two copies of the toric code to the non-Abelian twisted quantum double phase and back. 
The spacetime region with the twisted quantum double phase is shaded in yellow. 
The projection occurs because, inside the twisted quantum double phase, the product of all vertex stabilizers of the green color (in one time slice) counts the number of red-blue string intersections modulo 2. 
If the toric codes are initialized in the all-$\ket{++00}$ state, the protocol probabilistically prepares the ``$CZ$'' magic state.
}
\label{fig:protocol1}
\end{figure}

As with the first example, we can use the spacetime picture to see the action of the gate.
First of all, we notice that, with our choice of boundary conditions, a green membrane is supported at the $2D$ interface between the toric codes and the twisted quantum double phase, which gives rise to two possible nontrivial triple intersections.  One is between the qubits $\rc{rH}$, $\bc{ bV}$ and the green membrane degree of freedom which we label $\gc{g}$, while the second is between $\rc{rV}$, $\bc{ bH}$ and $\gc{g}$. One such intersection is shown in Fig.~\ref{fig:protocol1}. Once the state is in the twisted phase, the product of all the green vertex stabilizers on any timeslice produces the green membrane. Hence, measuring these stabilizers results in the projective measurement of the total number of triple intersections on a timeslice modulo 2. This number is given by the red-blue membrane intersections in the toric code state that enters the twisted phase, due to the geometry. The green membrane insertion at each step of the twisted quantum double phase can be thought of as initializing the auxiliary degree of freedom in the $|+\rangle_\gc{g}$ state, after which, a nontrivial phase factor $\overline{CCZ}_{\rc{rV},\bc{ bH},\gc{g}} \overline{CCZ}_{\rc{rH},\bc{ bV},\gc{g}}$ counting the parity of triple intersections is implemented, then the extra degree of freedom is measured in the logical Pauli-$X$ basis. The effect of this is the measurement of the logical operator $\overline{CZ}_{\rc{rH},\bc{ bV}}\overline{CZ}_{\rc{rV}, \bc{bH}}$.

Another way to understand this logical action is by noticing that the domain wall occurring at time $t_i$ when the protocol switches from two copies of the toric code to the twisted quantum double corresponds to \emph{gauging} a global symmetry of the red and blue copies of the toric codes by measurement. This symmetry has the logical unitary action of $\overline{CZ}_{\rc{rH}\bc{ bV}}\overline{CZ}_{\rc{rV} \bc{bH}}$. Gauging by measurements results in a projection onto the eigenstate of this symmetry. The second transition at $t_f$ is an ungauging operation, that reverts the system to the two copies of the toric code~\cite{Williamson2024Gauging}.

An important feature of our protocols is that the logical operations implemented by a certain spacetime configuration depend solely on the topological phase of each volume, as well as the phases of spacetime defects such as boundaries, domain walls, and corners, which are sometimes referred to as ``universality classes'', or ``super-selection sectors''.
This can be used to derive various microscopic implementations of the same logic gates, as these phases can have many different lattice realizations.
In this work, we focus on the description of the non-Abelian model as the twisted quantum double, which we use in the continuum closed-membrane picture for computing the logical action implemented by different spacetime configurations.
Various microscopic implementations representing the same topological phase are discussed in Sections~\ref{section:examples} and \ref{sec:3d-2d}.
In App.~\ref{sec:more_gates} we present more microscopic examples derived from gauging Clifford symmetries.

\section{Preparing a magic state on a torus}
\label{sec:CZ_minimal}

In this section, we discuss the protocol that is shown in Fig.~\ref{fig:protocol1}. 
We provide a succinct microscopic explanation that is centered around stabilizer and logical operators of the quantum error-correcting code that is transformed throughout the protocol. 
We then explain how this protocol can be viewed as a gauging logical measurement, which is a type of code deformation. 
In Sections~\ref{section:loopsum} and~\ref{section:examples} we introduce a comprehensive method that generalizes the example presented in this section to construct other logic-gate protocols and the associated circuits based on the topological path integral.  We also show a planar version of this protocol that operates on surface as opposed to toric codes in Sec.~\ref{section:loopsum}. 

\subsection{Stabilizers of the twisted quantum double} 
\label{sec:iiib}

In this example, we use a microscopic lattice realization of the non-Abelian topological phase that is needed for the protocol. This model shares the codespace (the ground state subspace) with a variant~\cite{Yoshida2016topological,Iqbal2024nonAbelian} of the type-III non-Abelian twisted quantum double~\cite{Hu_2013}. 
In the rest of the paper, we refer to this model simply as the ``twisted quantum double" (TQD), as it is the only twisted quantum double model that appears in this work. 
We discuss alternative microscopic realizations of the same non-Abelian phase in Secs.~\ref{section:examples} and~\ref{sec:3d-2d}, which are used for alternative circuits that have the same logical action. 

We start with a single copy of the (untwisted) toric code. 
We place qubits on the edges $e$ of a triangular lattice. 
The stabilizer group is generated by star and plaquette stabilizers associated to vertices $v$ and faces $f$ as follows:
\begin{equation}
 \label{eq:TC_stab_B}
A_v = \prod_{e: v \in \partial e } X_e,\quad B_{f} = \prod_{e \in \partial f} Z_e, 
\end{equation}
where $X_e, Z_e,$ are the standard Pauli matrices acting on a qubit at location $e$, and $\partial $ is the boundary map. 
The codestates of the toric code are the simultaneous $+1$ eigenstates of all stabilizer operators. 
The logical Pauli operators of the toric code are strings of Pauli operators that run along non-contractible cycles of the manifold, on the primal lattice for $Z$ and on the dual lattice for $X$. 
On a torus, we label them by the horizontal/vertical direction of winding, i.e. $\overline Z_{H/V}$ and $ \overline X_{H/V}$. 
With this definition, the nontrivial commutation relations are as follows, $\overline Z_{H} \overline X_{V} = -  \overline X_{V} \overline Z_{H}$,  $\overline Z_{V} \overline X_{H} = -  \overline X_{H} \overline Z_{V}$.

We now introduce the non-commuting stabilizer model that realizes the TQD topological phase. 
Informally, this model can be understood as three copies of the toric code that are coupled via a ``twist''.
When picturing the qubit configurations that appear in the ground state as closed-loop patterns, the twist can be understood as counting the number of triple intersections, as we discuss in Sec.~\ref{section:loopsum} and App.~\ref{subsec:loopsum-and-stabs}. 
Consider three overlapping triangular sublattices, which we color red ($r$), green ($g$), and blue ($b$), as shown in Fig.~\ref{fig:ToricStabilizers}. 
Accordingly, there is one qubit at each edge of each sublattice, marked by a dot. 
For the sublattice of color $c \in \{ r,g,b\}$, we label the vertices as $(v,c)$, edges as $(e,c)$, and plaquettes as $(f,c)$. 
The plaquette operators are the same as those of three individual toric codes placed each on a sublattice of an associated color. We denote them $\mathcal B_{f,c} = B_{f,c}$. The vertex terms are modified by the twist, and become Clifford stabilizers:
\begin{equation} \label{eq:TQD-star}
\begin{split}
\mathcal A_{v,c} &= \prod_{\partial e \ni (v,c)} X_{e} \prod_{i \in  \langle (v,c) \rangle} CZ_{i, i+1} \\
&\equiv \quad   A_{v,c}   \prod_{i \in  \langle (v,c) \rangle} CZ_{i, i+1}
\end{split}
\end{equation}
where the product of the $CZ$ phase operators is applied to the qubits of two complementary colors belonging to a hexagon surrounding the vertex $(v,c)$.
Both types of stabilizers are depicted in Fig.~\ref{fig:ToricStabilizers} for the blue sublattice, and those of the red and green colors are defined analogously.

\begin{figure}[t]
\includegraphics[width = \columnwidth]{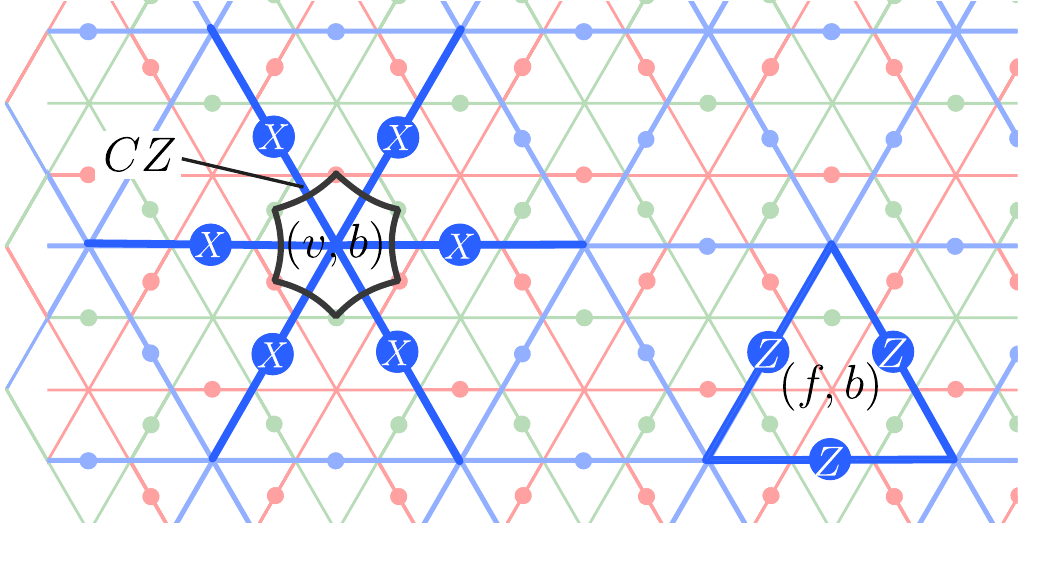}
\caption{Clifford vertex (left) and Pauli plaquette (right) stabilizers on the blue sublattice of the twisted quantum double. 
The stabilizers on the red and green sublattices are defined analogously.    \label{fig:ToricStabilizers}}
\end{figure}

In non-commuting stabilizer codes~\cite{Ni_2015,Magdalena_de_la_Fuente_2021,Webster2022xpstabiliser}, similarly to the usual stabilizer codes, the codestates can be defined as the the $+1$-eigenstates of all the star (vertex) and plaquette stabilizers.
In other words, the corresponding Hamiltonian 
\begin{equation}
    \mathcal{H} = - \sum_{v,c} \mathcal{A}_{v,c} - \sum_{f,c} \mathcal{B}_{f,c}\, ,
\end{equation}
is frustration-free.
However, these ``stabilizers'' do not commute in the full Hilbert space, rather they only commute 
in the $+1$ subspace of $\mathcal{B}_{f,c}$ operators, which we make more precise below.
The algebra of the stabilizers has a different structure to the analogous algebras that occur in commuting stabilizer models. This also applies to relations between the logical operators, which are discussed in the next subsection.

More specifically, the group commutator $[V,W] \coloneqq VWV^{-1}W^{-1}$ between Clifford vertex stabilizers $\mathcal A_{v,c}$ always lies in the subgroup $\mathcal{S}^Z$ of the stabilizer group generated by the $\mathcal{B}_f$ plaquette terms:
\begin{equation}
\left[\mathcal{A}_{v,c}, \mathcal{A}_{v',c'}\right] = \mathcal B_{f, c''} \mathcal B_{f', c''} \quad c'' \neq c, c',
\label{Eqn:StarCommutator}
\end{equation}
if  $c \neq c' \in \{ r,g,b\}$ and the operators $\mathcal{A}_{v,c}$ and $ \mathcal{A}_{v',c'}$ overlap, and is trivial otherwise. 
Here, $(f,c'')$ and $f'$ are the plaquettes of the third color whose centers coincide with $(v,c)$ and $(v',c')$, respectively. 
As we can see, the vertex operators commute in the subspace where all the plaquette stabilizers are $+1$, and anticommute in the presence of nearby excitations (associated with violated plaquettes). 

Finally, violations of stabilizer terms take the state out of the codespace.
Similar to the toric code, we call the point-like excitations associated with violations of the vertex stabilizers charges (which are Abelian) and those associated with the plaquette stabilizers fluxes (which are non-Abelian). 
The three charge (flux) generators are of red, green and blue color in correspondence with the color of the violated stabilizer, and we denote them as $e_{c}$ ($m_c$) where $c = \{ r,g,b\}$. We direct the reader to Secs.~\ref{section:loopsum},~\ref{section:examples} where this discussed in more detail as well as to App.~\ref{sec:excitations} for a brief review of anyons of the TQD model.

\subsection{Logical string operators}
\label{sec:logicals}

We now discuss the logical string operators of the twisted quantum double model, which are the operators that can generate rotations within the codespace. 
The explicit form of these operators, and the algebra they generate, is useful for understanding the logical action of the protocol. 
Like the stabilizers, we find it convenient to write down the logical operators in terms of standard toric code logical operators that are decorated with a twist.
The Pauli $Z$-type operators are the same as in the untwisted case, $\overline{\mathcal{Z}}_{H/V}^c \equiv \overline Z_{H/V}^c$.
Physically, we can think of these operators as transporting electric charges of their corresponding color around non-contractible loops of the torus.

We obtain complementary logical operators for the TQD by supplementing the logical operators $\overline{X}_{V/H}^c$ of the standard toric code with additional $CZ$ terms~\cite{Iqbal2024nonAbelian,Lyons2024}.
An example of one such operator, of color blue, is shown in Fig.~\ref{fig:ribbon}.
In order to specify the $CZ$ interaction terms of this operator we must order the qubits along a path around the torus. We choose an arbitrary starting point (in this case we set it to be the green qubit $g_1$, see Fig.~\ref{fig:ribbon}) and enumerate the green and red vertices in increasing order following the blue path around the torus. This allows us to write
\be
\bc{\overline{\mathcal{X}}^b_{H}} = \bc{\overline{{X}}^b_{H}} \left ( \prod_{\rc{r_i}} \prod_{\gc{g_j}: j\le i} CZ_{\gc{g_j},\rc{r_i} } \right ).
\ee
We find other logical operators by changing the color and orientation of this example.
This operator is implemented by a unitary circuit of depth linear in the length of the string, consisting of $CZ$ gates that are applied from each red qubit to all the green qubits to its left.

The logical operators of the twisted quantum double generate a non-Pauli (and non-Abelian) algebra. 
In fact, the algebra generated by all $\overline{\mathcal X}$ operators does not preserve the codespace of the TQD, i.e. some combinations of these operators map states out of the codespace.
We consider the commutator between an $X$-type logical operator and a vertex stabilizer with overlapping support. 
E.g., for a blue logical operator and a red stabilizer, we find
\begin{equation} \label{eq:str_algebra}
[\bc{\overline{\mathcal{X}}_{H/V}^b}, \rc{\mathcal A_{v,r}}] = \gc{\mathcal{B}_{f,g}},
\end{equation}
where the plaquette $(f,g)$ is located within the star associated with the vertex $(v, r)$. 
For each logical operator representative, there exists one vertex for which this commutation relation is different, namely:
\begin{equation} \label{eq:str_algebra1}
[\bc{\overline{\mathcal{X}}_{H/V}^b}, \rc{\mathcal A_{v,r}}] = \gc{\mathcal{B}_{f,g} \overline{\mathcal{Z}}_{H/V}^{g}},
\end{equation}
where the representative of the logical $\overline{\mathcal{Z}}_{H/V}^{g}$ is supported on the same lattice cells as the corresponding $\overline{\mathcal{X}}_{H/V}^b$.  
The special vertex exists due to the non-translation-invariant  logical operator shown in Fig.~\ref{fig:ribbon}. 
 In our example the special vertex appears at the ``start point'' of the loop. 
For example, for the blue logical operator in Fig.~\ref{fig:ribbon} the special red vertex operator is the one acting on qubits $r_1$ and $r_2$. 
The remainder of the commutation relations are the same under color permutations. 
Before concluding this subsection, let us define the charge parity operator that appears frequently throughout the remainder of this section:
\begin{equation}
\hat{C}_c = \prod_v \mathcal A_{v,c}.
\end{equation}
This is an operator counts the parity of charges of color $c$ on the torus. It follows from Eqn.~(\ref{eq:str_algebra}) that
\begin{equation} \label{eq:str_algebra2}
[\overline{\mathcal{X}}_{H/V}^c, \hat{C}_{c'}] = \overline{\mathcal{Z}}_{H/V}^{c''},
\end{equation}
where $c'' \neq c, c' \in \{ r,g,b\}$.

\begin{figure}[t]
\includegraphics[width = 0.9\columnwidth]{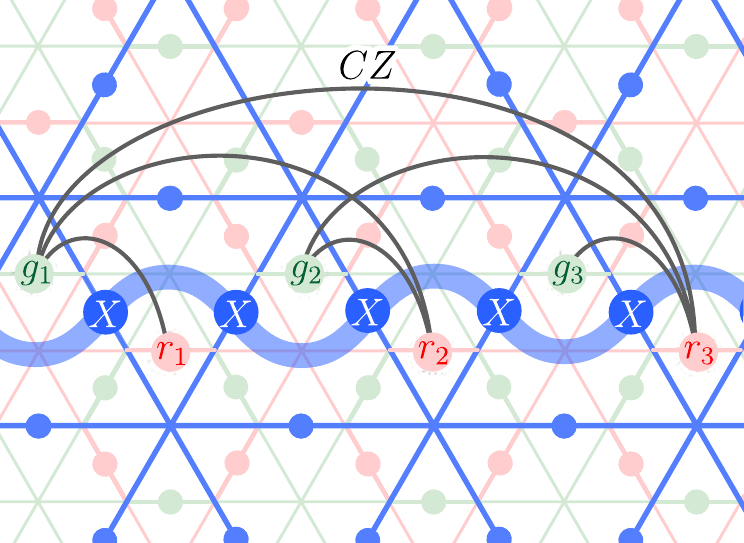}
\caption{A logical operator $\hat{X}_H^{b}$ for the twisted quantum double model. The gray lines depict $CZ$ gates applied between the qubits at their endpoints. The blue string operator is decorated with controlled-phase gates connecting qubits on the red and green sublattices.
\label{fig:ribbon}}
\end{figure}

\subsection{Codestates and charge sectors}
\label{sec:codestates}

The codestates of the TQD model are simultaneous $+1$ eigenstates of all stabilizer operators. 
Because of the non-Pauli algebra formed by the logical operators,
the codespace of the TQD on a torus (and the ground state space of the associated Hamiltonian model) is 22-dimensional, which reflects the non-additive nature of the code. 
A basis for the codespace can be generated by applying certain combinations of the $\overline{\mathcal{X}}$-type logical operators to the all-0 (`vacuum') logical state $\ket{\bm 0}$. The $\ket{\bm 0}$ state is defined to be a simultaneous $+1$ eigenstate of all $\mathcal Z^c_{H/V}$ operators in addition to the stabilizers. Up to normalization it can be written as
\be
\ket{\bm{0}} = \prod_{v,c}\frac{1 + \mathcal{A}_{v,c}}{2}\ket{0}^{\otimes n} ,
\ee
where $\frac{1 + \mathcal{A}_{v,c}}{2}$, $c \in \{ r,g,b\}$ is the projector onto the $+1$-eigenstate of the associated vertex operator. 
Of the 64 combinations of $\overline{\mathcal{X}}$-type string operators of three colors and three orientations, only 22
result in an eigenstate of all the vertex operators. 
To see this, we recast the product of all the vertex operators as\footnote{To show this, we use that $\prod_v \mathcal A_{v,c}$ must (i) be a diagonal unitary (ii) act trivially on the codespace and (iii) produce the commutation relations in Eq.~\eqref{eq:str_algebra}. 
One can verify that the solution \eqref{eq:str_algebra_1} satisfies all three conditions. 
Alternatively, one can show this directly as in Ref.~\cite{Iqbal2024nonAbelian}.} 
\begin{equation} \label{eq:str_algebra_1}
\hat{C}_c = (-1) ^{(1 - \overline{\mathcal Z}^{c'}_{,H})(1 - \overline{\mathcal Z}^{c''}_{V})/4 } (-1) ^{(1 - \overline{\mathcal Z}^{c'}_{V})(1 - \overline{\mathcal Z}^{c''}_{H})/4 } .
\end{equation}
The product of all vertex operators of color $c$ is $+1$ ($-1$) when the total parity of intersections of the $\overline{\mathcal{X}}$-type strings of two complementary colors is even (odd). 
This has a clear physical interpretation, as we recall that $\hat{C}_c$ counts the total parity of charges of color $c$.

To obtain an odd parity of charges on a manifold with periodic boundary conditions we use the fact that an intersection between two  $\overline{\mathcal{X}}$ operators of different colors contains an unpaired charge of the third color~\cite{Iqbal2024nonAbelian}. 
An odd total number of charges is thus possible in the states that are obtained from $\ket{\bm{0}}$ by applying a combination of $\overline{\mathcal{X}}$ operators of different colors and orientations such that the resulting total parity of intersections is odd. 
The right-hand side of Eq.~\eqref{eq:str_algebra_1} counts precisely the parity of these intersections.  

The 22 basis states for the codespace of the TQD on a torus are depicted below.
\be \nonumber
\includegraphics[width = 1\columnwidth]{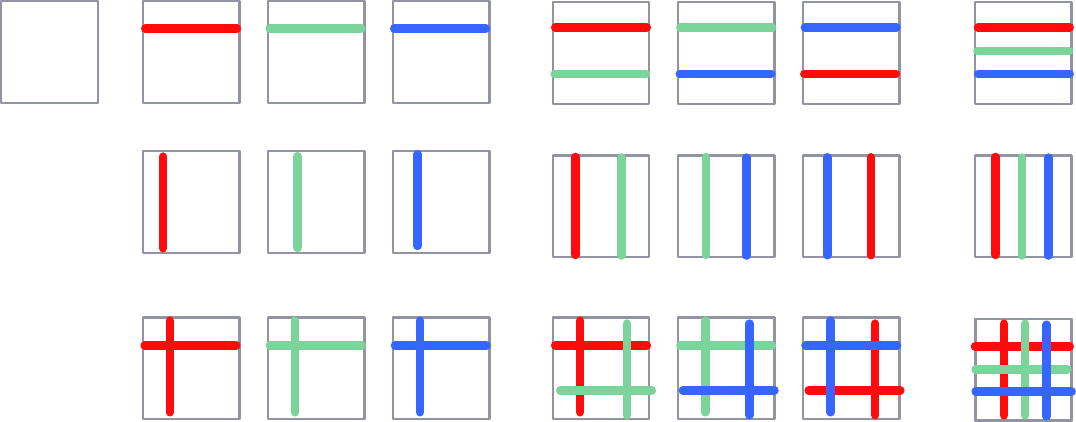}
\ee
Here, each colored line denotes the application of an $\overline{\mathcal X}$ operator of corresponding color and orientation. 
All of the above states contain an even number of intersections between strings of any pair of distinct colors, and there are 22 such combinations.

In what follows, we must also the case when our protocol obtains an odd charge parity subspace when we transform between the toric code model and the TQD. For the specific protocol that we consider in this section, odd parity of green charges can occur.
There are 6 basis states containing a single unpaired green charge $e_g$, depicted below.
\be \nonumber
\includegraphics[width = 0.75\columnwidth]{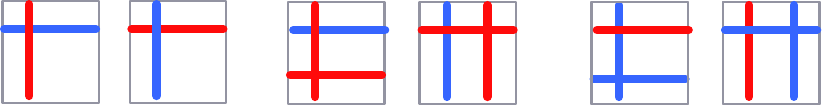}
\ee
These states have an odd total number of intersections between strings of red and blue color. 
This subspace is characterized by the  $-1$ eigenvalue of the $\hat{C}_g$ operator. 
The states with an odd number of $e_r$ or $e_b$ charges are obtained by appropriate color permutations. 

For completeness, we characterize all types of charge-parity sector that occur in the TQD model. 
There are 6 states that simultaneously support unpaired charges of two different colors. 
For example, the subspace with unpaired $e_r$ and $e_g$ is spanned by the states depicted below. 
\be \nonumber
\includegraphics[width = 0.75\columnwidth]{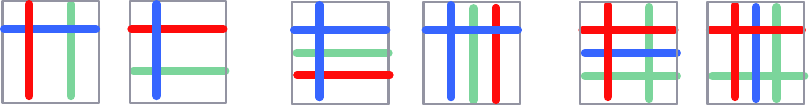}
\ee
Similar states occur for $(e_g,e_b)$ as well as $(e_b, e_r)$.

Finally, there are 6 states in the subspace with simultaneously unpaired $e_g, e_r$ and $e_b$ charges, depicted below. 
\be \nonumber
\includegraphics[width = 0.75\columnwidth]{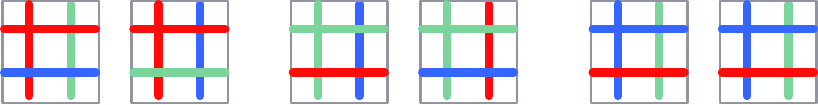}
\ee
In total, the classes above result in 
$$\underset{1}{ \underbrace{22}}  + \underset{ e_c}{ \underbrace{3 \times 6}} + \underset{(e_c,e_{c'})}{ \underbrace{3 \times 6}} +\underset{(e_r, e_g, e_b)}{ \underbrace{6}} = 64$$ 
combinations, where the braces indicate if there is an odd parity of certain number of charge types.

\subsection{The $\overline{CZ}$ measurement protocol}

We now describe the implementation of the simple measurement protocol and its logical action. We postpone the derivation of the protocol to Sec.~\ref{section:loopsum}.
For simplicity, we assume noisless operations; fault tolerance is discussed in Sec.~\ref{Sec:JustInTime}. 

The protocol is implemented as follows. 
We start with encoded states of the red and blue copies of the toric code on triangular sublattices of the corresponding color,
as shown in Fig.~\ref{fig:ToricStabilizers}. The qubits on the green sublattice are initialized in the $\ket{0}$ state. The initial stabilizer group on all qubits can be written as
\be
\mathcal{S}_1 = \langle \rc{ A_{v,r}}, \rc{ B_{f,r}},\bc{ A_{v,b}},\bc{B_{f,b}}, \gc{Z_{e,g}}, \gc{B_{f,g}},\gc{\overline{ Z}^g_H},\gc{\overline{Z}^g_V}  \rangle ,
\ee
where we have listed redundant operators to highlight that the state is a $+1$ eigenstate of the ${B^g_f} \equiv {\mathcal{B}^g_f}$ plaquettes, as well as the ${\overline{ Z}^g_{H,V}}$ operators.

With this initialization, we find that the code is stabilized by the red and blue vertex terms of the TQD ${\mathcal A_{v,r}}$ and ${\mathcal A_{v,b}}$, because the $CZ$ operators in these vertex terms always act on at least one green qubit in the $\ket{0}$ state. Thus, we present the stabilizer group $\mathcal{S}_1$ alternatively as follows
\be
\mathcal{S}_1 = \langle \rc{\mathcal A_{v,r}}, \rc{\mathcal B_{f,r}},\bc{\mathcal A_{v,b}},\bc{\mathcal B_{f,b}}, \gc{Z_{e,g}},  \gc{\mathcal B_{f,g}},\gc{\overline{\mathcal Z}^g_H},\gc{\overline{ \mathcal Z}^g_V} \rangle .
\ee
Thus, the only TQD stabilizers that have not been prepared are the ${\mathcal A_{v,g}}$ operators.

Next, we implement a transition to the TQD model by measuring its green vertex operators  ${\mathcal A_{v,g}}$. 
We explain why this measurement prepares the twisted quantum double in greater depth in the next section. 
After this measurement the stabilizer group becomes
\be
\mathcal{S}_2 = \langle \rc{\mathcal A_{v,r}}, \rc{\mathcal B_{f,r}},\bc{\mathcal A_{v,b}},\bc{\mathcal B_{f,b}},  (-1)^{m_v}\gc{\mathcal A_{v,g}} ,  \gc{\mathcal B_{f,g}}, \gc{\overline{\mathcal Z}^g_H},\gc{\overline{\mathcal Z}^g_V}  \rangle.
\ee
Here, we have included the $\pm 1$ measurement outcomes of ${\mathcal A_{v,g}}$ which are captured by $\mathbb{F}_2$-valued variables ${m_v}$.  
These random measurement outcomes specify a distribution of green Abelian charges. 
Upon adding the newly measured operators $ (-1)^{m_v}\gc{\mathcal A_{v,g}}$ to the stabilizer group, we only keep the $Z_g$-type operators from $\mathcal S_1$ that commute with all the newly measured terms, which are green plaquette terms and green logical-$Z$ operators. 
We emphasize that the stabilizer group $\mathcal S_{2}$ we have prepared describes the TQD in a $+1$ subspace of green $\overline{\mathcal Z}$ logical operators. 
Additionally, the product of all green vertex operators is the charge parity operator
\begin{equation} 
\label{eq:main-CZ}
\begin{split}
 \gc{\hat C_{g}} = (-1) ^{\frac{1 - \rc{\overline{\mathcal Z}_{r,H}}}{2}\frac{1 - \bc{\overline{\mathcal Z}_{b,V}}}{2} } \times (-1) ^{\frac{1 - \rc{\overline{\mathcal Z}_{r,V}}}{2}\frac{1 - \bc{\overline{\mathcal Z}_{b,H}}}{2} } ,
\end{split}
\end{equation}
that was introduced in Eq.~\eqref{eq:str_algebra_1}. 

Finally, the transition back to the red and blue toric code copies is achieved by decoupling the green qubits into $Z$-eigenstates by measuring on-site $Z$ operators, giving
\be
\begin{split}
\mathcal{S}_3 = \langle &\rc{\mathcal A_{v,r}}, \rc{\mathcal B_{f,r}},\bc{\mathcal A_{v,b}},\bc{\mathcal B_{f,b}},
\\
&(-1)^{m'_e} \gc{Z_{e,g}}, \gc{\mathcal B_{f,g}}, \gc{\overline{\mathcal Z}^g_H},\gc{\overline{\mathcal Z}^g_V}, \prod_v (-1)^{m_v} \gc{\hat C_g}  \rangle .
\end{split}
\ee
Here, we have added the newly measured operators to the stabilizer group and `kicked out' the stabilizers that do not commute with these measurements. 
The operator $\prod_v (-1)^{m_v} {\hat C_g}$ remains in the stabilizer group as it commutes with the individual Pauli-$Z$ operators (the Pauli-$X$ component  of ${\hat C_g = \prod_v \mathcal A_{v,g} }$ cancels upon multiplication). 

We have included redundant operators ${\mathcal B_{f,g}}, {\overline{\mathcal Z}^g_H},{\overline{\mathcal Z}^g_V,} $ in the stabilizer group $\mathcal{S}_3$ to explicitly show the constraints that measurement outcomes $m'_e$ of single-qubit $Z_{e,g}$ operators on green qubits have to obey. 
In fact, the syndromes must form closed loops that can be corrected by applying products of ${\mathcal {A}_{g,v}}$ operators to fill the regions enclosed by the loops.\footnote{Any correction that removes the closed loops on the green sublattice, but also commutes with the existing vertex and plaquette stabilizers, is valid. 
Products of operators $\mathcal{A}_{v,g}$ meet both of these requirements.}  Once the correction is performed,  the resulting state is stabilized by $\mathcal{A}_{v,r/b}$ and $A_{v, r/b}$ simultaneously, similar to the beginning of the protocol. 
In addition, the state on the green sublattice completely decouples from the state on the red and blue sublattices. 
This recovers the original pair of red and blue toric codes with stabilizer group
\be
\mathcal{S}_3 = \langle \rc{A_{v,r}}, \rc{ B_{f,r}},\bc{A_{v,b}},\bc{ B_{f,b}}, \gc{Z_{e,g}}, \prod_v (-1)^{m_v} \gc{\hat C_g}\rangle ,
\ee
which includes the initial stabilizer group $\mathcal{S}_1$ along with the $\prod_v (-1)^{m_v} {\hat C_g}$ operator.

We now determine the logical action of the protocol. 
It is convenient to work in an eigenbasis of the logical Pauli-$Z$ operators
$\{\rc{\overline Z_V^r}, \bc{\overline Z_H^b},\rc{\overline Z_H^r},\bc{\overline Z_V^b}\}$, as these operators are preserved over the transformation from the Abelian phase onto the TQD and back again. We label the associate states as $\{ \rc{rH}, \bc{bV},  \rc{rV}, \bc{bH}\}$.
On the other hand, upon completing the operation, the operator $ \prod_v (-1)^{m_v}  \hat C_g$ is included in the stabilizer group of the code. The inclusion of this operator in the stabilizer group represents a projection onto the $+1$ eigenvalue eigenspace of $\pm \hat C_g$. Importantly, the expression in  Eq.~\eqref{eq:main-CZ} can be written as the product of logical $\overline{CZ}$ operators
\begin{equation}  \label{eq:Cg-logical}
\begin{split}
 \gc{\hat C_{g}}\equiv  \overline{CZ}_{\rc{rH}, \bc{bV}}\overline{CZ}_{\rc{rV}, \bc{bH}}.
\end{split}
\end{equation}
Thus, the protocol measuring $\hat C_g$ projects the state onto the $\prod_v (-1)^{m_v}$ eigenstate of a product of $\overline{CZ}$ operators.

This logical measurement can be used to prepare magic states if we choose the input logical state to be $\ket{\rc{+_{rH}}\bc{+_{bV}}\rc{+_{rV}}\bc{+_{bH}}}$. 
The protocol prepares a magic state of the following form
{\small
\begin{align}
        &\ket{C_g = +1}_{\text{out}} = \frac{3}{\sqrt{10}}\ket{\overline{CZ}_{\rc{H}\bc{ V}}}\ket{\overline{CZ}_{\rc{V} \bc{H}}} + \frac{1}{\sqrt{10}}\ket{\rc{1}\bc{1}\rc{1}\bc{1}} ,\label{eq:mag1}
        \\
        &\ket{C_g = -1}_{\text{out}} =\frac{1}{\sqrt{2}}\ket{\overline{CZ}_{\rc{H}\bc{ V}}}\ket{\rc{1}\bc{1}} + \frac{1}{\sqrt{2}}\ket{\rc{1}\bc{1}}\ket{\overline{CZ}_{\rc{V}\bc{ H}}} ,\label{eq:mag2} 
\end{align}
}
where 
\be
\ket{\overline{CZ}} = \frac{1}{\sqrt{3}} \left ( \ket{00} + \ket{01} + \ket {10} \right ) ,
\ee
is the $CZ$-magic state \cite{Dennis2001toward, Gupta2024encoding}.  
Both of the output states  of this preparation protocol, Eqs.~\eqref{eq:mag1} and \eqref{eq:mag2}, are magic (non-stabilizer) states.  
In particular, one can use either of them to probabilistically recover the  $CZ$- state in its original form by measuring the pair of qubits ($rH$,$bV$) in the Pauli-$Z$ basis and verifying we obtain the correct outcome such that the remaining state is a $CZ$-state.

We verify the magic state preparation protocol in a code in the \href{https://github.com/m-webster/XPFpackage/blob/main/CCZ2D/CCZ2D.ipynb}{linked} Jupyter notebook by representing the stabilizers of the twisted quantum double phase as $XS$ operators using the embedded code technique from Ref.~\cite{Webster2023diagLO} and simulating measurements of the $XS$ operators using the technique set out in Ref.~\cite{Webster2022xpstabiliser}.

The protocols we present in this paper can be seen as generalized code deformations~\cite{Vuillot_2019, Brown2020universal,Cohen2022low} beyond Pauli stabilizer codes~\cite{Ni2015, Webster2022xpstabiliser}.
Such a code deformation can be formulated as a projective measurement on a stabilizer code $\mathcal{S}$ that is induced by measuring into the stabilizers of a new code $\mathcal{S}'$ whose stabilizers include a logical operator of $\mathcal{S}$. 
In the example above, measuring a Clifford operator prepares a magic state.
This is followed by a reverse code deformation that restores the original code in a magic state, as the measurements that implement this reverse transformation commute with the logical operator that was measured.

\subsection{Relation to gauging}
\label{sec:GaugingMsmnt}

The transitions between the Abelian and the non-Abelian topological codes in the protocol in this section 
correspond to  \emph{gauging} and \emph{ungauging} the $\mathbb Z_2$ symmetry implemented by the operator in Eq.~\eqref{eq:main-CZ} on a pair of toric codes. 
This point of view allows us to explain the protocol purely in terms of universal, emergent anyon data.
Here, we briefly discuss how to understand this in terms of the gauged symmetry group and how to identify the associated logical operator.
This is discussed further in App.~\ref{sec:gauging_CZ} and \ref{sec:more_gates}, including microscopic lattice model details that go beyond the example presented in this section.

Both gauging and ungauging transitions can be described by gapped domain walls between a ``gauged'' phase and the original ``ungauged'' phase~\cite{Williamson2016ungauging,kubica2018ungauging}.
They can be described as \emph{condensation transitions} in which certain non-trivial defects before condensation are identified with the ``vacuum'' after condensation, meaning that after condensation these defects can be freely removed or inserted into the state via local operators~\cite{Bais2009Condensate}.
Gauging a finite global symmetry group $G$ corresponds to the condensation of $G$-domain walls, which are one-dimensional defects separating regions in which the symmetry was applied from regions where it was not applied. 
Ungauging refers to the opposite process in which Rep($G$) bosons, point-like defects at the endpoints of deconfined strings, are condensed~\cite{Barkeshli2019symmetry}.

When $G$ is an Abelian group, gauging $G$ can be implemented via local adaptive quantum circuits~\cite{Williamson2020a,Tantivasadakarn2021LRE}. 
For this reason, the subclass of domain walls that are induced by gauging an Abelian symmetry automatically yield microscopic  circuits for the associated logical gates.
For example, in this section, $G\simeq \zz_2$ is the anyon-permuting symmetry of a pair of toric codes (see App.~\ref{sec:gauging_CZ}) that corresponds to a finite-depth logical unitary $\overline{CZ}_{{rH}, {bV}}\overline{CZ}_{{rV}, {bH}}$ in the code.
Gauging this symmetry implements a projective measurement of the corresponding Clifford logical operator, which we then use to prepare magic states.
More generally, all protocols in this work in which the domain wall between the TQD phase and toric code phases lies perpendicular to the time direction can be regarded as gauging domain walls.

The approach of gauging symmetries of a code to obtain a logical action has the advantage of being readily applicable to general qLDPC codes~\cite{Breuckmann2021} with symmetries. Gauging measurements have been discussed in this context in Ref.~\cite{Williamson2024Gauging}, where they were applied to symmetries of quantum codes to implement measurements of logical Pauli operators.   
However, when gauging a Clifford symmetry, the stabilizers after gauging are Clifford operators and are not guaranteed to commute outside the codespace. 
In this work, we establish fault tolerance for a class of Clifford gauging measurements on topological codes, see Sec.~\ref{Sec:JustInTime}. 
A general analysis of the fault tolerance of deformed qLDPC codes whose description is beyond the Pauli stabilizer formalism is an open problem.

\section{Non-Clifford logical operations from the TQD path integral}
\label{section:loopsum}

In this section, we present a systematic approach to non-Clifford logical operations from interfacing Abelian and non-Abelian topological phases via configurations of topological defects in spacetime, using a path integral description of topological phases.
We first review the TQD model, including its Euclidean spacetime path integral in the presence of defects, and then derive a whole family of protocols that operate similarly to the one introduced in Sec.~\ref{sec:CZ_minimal}. These include $\overline{CZ}$ logical measurement,  $\overline{XS}$ logical measurement (which prepares the $T$-magic state) as well as $\overline{CCZ}$ and $\overline{T}$ unitary gates.

This section aims at deriving the global, logical, properties of the protocols which follow from the universal features of topological phases and their defects. To calculate these global properties we first discuss microscopic continuum descriptions of the relevant phases and their properties. In the following Section~\ref{section:examples}, we derive the microscopic circuit implementations of the protocols discussed in this section.

\subsection{Global topology vs. microscopic implementation}

We start by discussing the difference between two distinct but related aspects of the protocols.
The \emph{global topology} of a protocol corresponds to a configuration in spacetime where each region of spacetime is filled by a topological phase (which can be trivial) and the interfaces between these regions host domain walls between adjacent topological phases (this further extends to interfaces of interfaces and so on, see Refs.~\cite{Aasen2020TDN,williamson2023spacetime}). For example, an interface to the trivial phase corresponds to a gapped boundary. Examples of distinct global topologies have already appeared in the text, in Fig.~\ref{fig:protocol1} and Fig.~\ref{fig:protocol2}.

In this work, we focus on protocols that specifically involve copies of the toric code and TQD phases in 3-dimensional spacetime volumes, gapped boundaries, and domain walls between these phases, as well as line and point defects.
The global topology alone suffices to determine the logical operation performed by a protocol. Given a fixed global topology, the overall logical action is independent of the microscopic details.

On the other hand, we define a \emph{microscopic lattice implementation} to be a specific recipe used to implement a global topology.  
It specifies an implementation of the topological phases, boundaries, domain walls, and their interfaces, by an explicit, geometrically local quantum circuit with global classical communication and feed-forward.
There are two remarks regarding different microscopic implementations of a given global topology. 
First, there are different microscopic implementations of the same emergent topological phase in the continuum. 
For example, the Dijkgraaf-Witten gauge theories based on $D_4$ with no twist, and $\mathbb{Z}_2^3$ with a type-III cocycle twist are equivalent. 
Second, there are different ways to realize a microscopic description of a topological phase on a lattice of qubits (for example, the TQD phase can be implemented either via the model in Ref.~\cite{Hu_2013} or in Ref.~\cite{Yoshida2016topological}).
Such microscopic implementations can be obtained from fixed-point path integral representations of the topological phases, which we discuss
in Section~\ref{subsec:cupproduct-tqd}.
Alternatively, microscopic implementations of protocols that implement logical measurements can be obtained via the gauging logical measurement procedure, see Section~\ref{sec:GaugingMsmnt} and Appendix~\ref{sec:gauging_CZ}.

\subsection{Circuits, path integrals, and the closed-membrane picture in spacetime}

\subsubsection{Fixed-point path integral}

Here, we formulate a spacetime picture of the relevant topological phases. The well-known loop sum picture for the relevant ground states  (see Appendix~\ref{subsec:loopsum-and-stabs}) follows naturally from this point of view, as we discuss below. We then show how to use the path integral approach for the calculation of the logical action implemented by a given global topology.

Any non-chiral topological phase can be described in spacetime through a \emph{fixed-point path integral}~\cite{Dijkgraaf_1990,turaev1992state,Koenig_2010,sahinoglu2016tensor,Bauer_2022,Bauer2023,Bauer2024}. We define the path integral using a cellulation of spacetime where we place a variable on each of its cells. The path integral is a sum over discrete configurations $\vec {c}$ of these variables where each term in the sum is given by a product over local weights $\omega$. Schematically,
\begin{equation}\label{eq:pathintegral_schematic}
Z=\sum_{\substack{\vec c\ }}
\prod_{i} \omega_i(\vec c)\, ,
\end{equation}
where $i$ is a local label in spacetime (usually labeling 3-cells)  and $\omega_i(\vec c)$ only depends on the configuration $\vec{c}$ in the vicinity of location $i$. For the path integrals considered here, the weight $\omega_i(\vec c)$  is either a complex phase or is equal to zero to impose specific local constraints on the configuration (such as closed-membrane constraints).

The value $Z$ has an important meaning when it is evaluated on a lattice with an \emph{input or output} state boundary. The input state boundary is usually placed at the earliest time of the protocol, and the output one is placed at the latest time.
At such boundaries, there is no summation over the variables. For each configuration of variables on the state boundary $\vec c_{\text{b}}$, we obtain a separate amplitude $Z(c_{\text{b}})$. This type of boundary is distinct from a \emph{physical boundary} which is discussed in Sec.~\ref{sec:boundaries-DWs-corners}. 
We use the evaluation of the path integral with an input or output state boundary  $b$ to define a \emph{boundary state} $\ket{\psi_{\mathrm{b}}}$  that is an input or output state for the corresponding boundary type. The computational-basis coefficients are the amplitudes, namely $\bra{\vec c_{\mathrm{b}} }\ket{\psi_{\mathrm{b}}} = Z(\vec c_{\mathrm{b}})$.

To motivate the path integral approach, we briefly explain how the topological fixed-point path integrals are related to circuits consisting of unitary gates and measurements~\cite{Bauer2023,Bauer2024}. The full discussion follows in Sec.~\ref{section:examples}. 
It is relatively straightforward to turn a path integral with input and output state boundaries into a circuit of unitary gates and $+1$-postselected measurements, as the latter are simply projection operators.
For this, we assign a time direction in the spacetime cellulation for the path integral.
We turn the path-integral variables located at the same time, but different space, coordinates into the configurations of qubits at a corresponding time in the circuit.
We then implement the local constraints and phases in the path integral given by weights $\omega_i(\vec c)$ through projectors onto configurations that satisfy these constraints, as well as gates whose matrix elements are equal to the associated complex phases, respectively.
Finally, we can replace the projectors (or the postselected measurements) with actual measurements.
For this, we observe that the circuit with some of the measurement outcomes equal to $-1$ corresponds to a path integral with insertion of line-like defects~\cite{Bauer2023,Bauer2024} at appropriate locations. 

We now describe the fixed-point path integrals that are used in this work.
The first path integral is the one for the toric code phase. A $\mathbb Z_2\simeq \{0,1\}$-valued variable is associated to every edge of a three-dimensional spacetime cellulation.  
The toric code path integral is an equal-weight sum over all configurations with an even number of $1$ variables around each face (that is, there are weights that are nonzero only for configurations that satisfy this constraint).
This constraint is, in fact, the \emph{closed-membrane constraint}: because of it, the Poincar\'e dual faces on each $1$ variable form a closed-membrane pattern on the Poincar\'e dual cellulation.
The path integral is thus
\begin{align}
Z
&= \sum_{  \vec{c}} \prod_{f} \delta\Big({\sum_{\partial f} \vec{c} = 0 \operatorname{mod} 2}\Big)\, 
\\
&= \sum_{\vec c: \text{ closed membranes}} 1 
\label{eq:TC-pathintegral}
\end{align}
where the second line is a schematic representation. 
Any input and output state for this path integral is an equal-weight superposition of all closed-loop patterns, giving rise to the picture presented in Appendix~\ref{subsec:loopsum-and-stabs}.

The second path integral is that of the TQD phase~\cite{Hu_2013}. It is a sum over all closed-membrane configurations of three different colors, $r$, $g$, and $b$, where a $-1$ weight is associated to each triple intersection of three different colored membranes in spacetime, visually depicted below. 
\begin{equation}
\label{eq:triple_membrane_intersection}
\vcenter{\hbox{\includegraphics[width = 0.4\columnwidth]{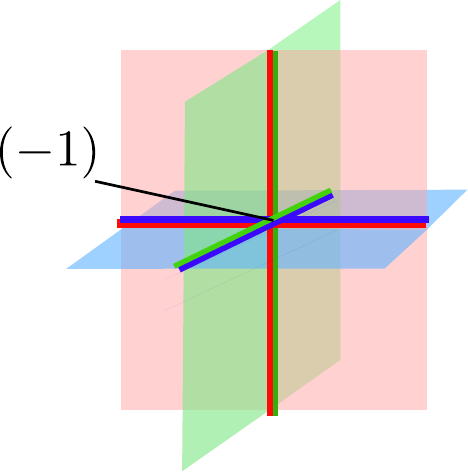}}}
\end{equation}
In the figure above and in the rest of this section, it suffices to use continuous pictures to represent the membranes and lines in spacetime since we are focusing on global topology and global properties of the protocols.
Schematically, the path integral is of the form~\cite{Bauer2024}
\begin{equation}
Z_{TQD}=\sum_{\vec c_r, \vec c_g, \vec c_b: \  \text{ closed membranes}} (-1)^{\text{\# $rgb$ triple intersections}}
\, .
\end{equation}
For the input and output states, one obtains the three-colored closed-loop configurations forming the ground states of the TQD discussed in Appendix~\ref{subsec:loopsum-and-stabs}. The history of these spatial closed-loop patterns evolving through time results in three-colored closed-membrane patterns in spacetime.
Similarly, a spatial triple-crossing event, as shown in Fig.~\ref{fig:TQD_loops}(c), corresponds to a triple-membrane intersection in spacetime, as shown in Eq.~\eqref{eq:triple_membrane_intersection}.

\subsubsection{Charge and flux defects}
\label{subsec:chargeandflux}
In order to relate the path integral picture to circuits, we need to consider the path integrals for relevant phases after the insertion of \emph{defects}.
The defects in the toric code and the TQD phases correspond to either \emph{charge} or \emph{flux} worldlines.

We start by discussing flux defects, in the absence of any charge defects. 
In spacetime, flux defects are worldlines where the membranes of the path integral terminate, and accordingly come in red, green, and blue colors. Below, we show such a worldline of red color as well as a section at a fixed time $t$.
\begin{equation}
\vcenter{\hbox{\includegraphics[width = 0.6\columnwidth]{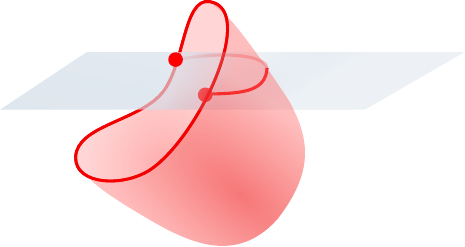}}}
\;
\end{equation}
From the timeslice section, depicted above in gray, we obtain the picture in space. In space, the fluxes are points where the loops of one color are not closed but terminate:
\begin{equation}
\vcenter{\hbox{\includegraphics[width = 0.18\columnwidth]{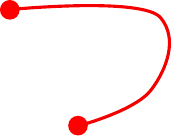}}}
\;
\end{equation}

The path integral with flux defects inserted is modified in comparison to the usual $Z_{TQD}$ by summing over membranes that terminate at the flux worldlines, instead of only the closed ones. 
For example, the path integral with red flux worldline $f$ in the absence of any charge worldlines is:
\begin{equation}
Z_{TQD, \text{ flux }f} = \sum_{\substack{\text{closed $b$,$g$ membr.}\\\text{$r$ membr. terminating at $f$}}} (-1)^{\text{\# $rgb$ triple int.}}\, .
\end{equation}
Microscopically, fluxes are points or lines on the Poincar\'e dual lattice in space as well as in spacetime, that is, they pierce perpendicularly through the faces of the primal lattice.

The insertion of a flux worldline in the TQD phase qualitatively alters the properties of the path integral: after insertion, the number of triple intersection points is no longer invariant under local deformations of the closed-membrane pattern.\footnote{Technically, the resulting action with a defect is not gauge-invariant, i.e., it is not a function of the cohomology class.} 
When a blue-green intersection line is deformed across a red flux worldline, this creates or removes a triple intersection, i.e., a local deformation yields a factor of $-1$.
This implies that the path integral may depend on the geometric shape of the flux worldlines, rather than their topology (more precisely, their cohomology class) alone, as it should normally be for topological phases. 
Hence, the continuously deformable closed-loop or membrane picture is not applicable in the presence of nontrivial flux worldline configurations.

We next discuss what happens in the presence of charge defects with no flux defects. Charges also come in red, green, and blue color.
In spacetime, the charge defects correspond to closed worldlines, where an additional $-1$ weight appears in the path integral for every intersection of a worldline with membranes of the same color. 
\begin{equation}
\vcenter{\hbox{\includegraphics[width = 0.6\columnwidth]{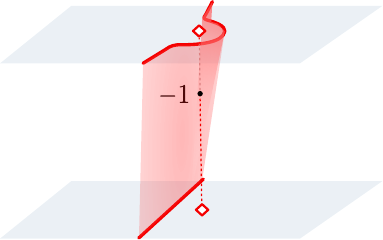}}}
\;
\end{equation}
In the spatial loop-sum picture, charge defects are points (shown as hollow diamonds) such that the states in the superposition have a relative $-1$ for configurations related by a loop crossing this point:
\begin{equation}
\vcenter{\hbox{\includegraphics[width = 0.45\columnwidth]{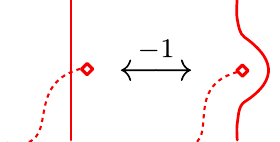}}}
\;
\end{equation}

The only modification to the path integral that is needed in order to add charge defects is to account for the possible $-1$ factor from charge-membrane crossings along every charge worldline. 
For example, adding a red charge worldline $\ell$ leads to the following path integral,
\begin{align}
Z_{TQD, \text{charge }\ell} = \sum_{\text{closed $rgb$ membr.}} (-1)^{\text{\# $rgb$ triple int.}}
\nonumber \\
\times (-1)^{\text{\# red charge $\ell$-membrane int.}}\, .
\end{align}
Charge defects of other colors are described similarly. 
Since the number of charge-membrane intersections is invariant under local deformations (when the fluxes are absent), the entire path integral is still gauge invariant when only charge defects are present.
In fact, the charge worldlines correspond to the Abelian anyons of the TQD. See Appendix~\ref{sec:excitations} for a short review of the anyons of the TQD.

Finally, we discuss what happens if we include both fluxes and charges into the path integral. As mentioned earlier, the presence of flux worldlines ruins the gauge invariance of the path integral (i.e. its invariance under local deformations of the closed-membrane pattern). Gauge invariance, however, ensures that charge worldlines must always be closed (otherwise the path integral evaluates to $0$). As a consequence, if there are flux worldlines, then the charge worldlines \emph{can now terminate on the flux worldlines}. More precisely, a charge worldline of color $c$ can terminate on the flux worldlines of the other two colors $c' \neq c$. In addition, in a circuit realization, any configuration of such charge worldlines that terminate on the associated flux worldlines is equally probable.  This has important consequences for decoding the protocols in the presence of noise, see Sec.~\ref{Sec:JustInTime}.

\subsection{Boundaries, domain walls, and corners of the TQD} \label{sec:boundaries-DWs-corners}

We now briefly discuss the boundaries and the corners of the TQD phase as well as the domain walls between the TQD phase and the toric code. All of these spacetime defects are important for constructing the desired logical actions. 

The most important type of defects in our protocols are \emph{domain walls} that interface the TQD model with copies of the toric code. 
In our fault-tolerant circuits that implement logical protocols, these domain walls can, for example,  appear on a timelike slice such as in Fig.~\ref{fig:protocol1}, on a diagonal plane in spacetime such as in Fig.~\ref{fig:protocol2}, or at a fixed spatial location. 
Several of the global topologies we consider also use domain walls between the TQD model and the vacuum, i.e., \emph{gapped boundaries}. 
A number of the global spacetime topologies we consider include \emph{corners}, which are 1-dimensional interfaces between domain walls or boundaries. 
We present specific realizations of these defects in the closed-loop and closed-membrane picture, below, without explicitly referring to a microscopic lattice. 
Explicit circuits that include corner defects can be derived from this picture by mapping loop configurations onto lattice cocycles.

\subsubsection{Boundaries of the TQD}
\label{sec:pi_boundaries}

The boundaries of non-chiral topological orders in $(2+1)$-dimensions are well understood.
They can be classified into topological superselection sectors, labeled by a finite number of inequivalent Lagrangian algebra objects~\cite{Kitaev2012Models,Kong2014Anyon,Magdalena_de_la_Fuente_2023,beigi2011quantum}.
For example, the toric code has two different boundary phases. 
The rough and smooth boundaries introduced in Ref.~\cite{Bravyi1998} are specific microscopic implementations of the distinct boundary phases. 
For the non-Abelian type-III TQD considered here, there are $11$ distinct topological boundaries~\cite{Bullivant2017,Lan_2015}. 
We briefly describe all $11$ boundaries in the closed-membrane picture without reviewing their full classification. 

Each boundary is determined by the way membranes of different colors are allowed to terminate, and the boundary amplitudes of the path integral. The terminating membranes form a proper subgroup of $\mathbb{Z}_2^3$, and we label the associated boundary by angle brackets containing the generators of this subgroup that are allowed also determine the charge anyon worldlines that are allowed to terminate on each boundary. Due to the non-trivial twist, not all subgroups are allowed.

The 11 distinct boundaries of the TQD are:
\begin{itemize}
\item [(1)] $\langle\varnothing\rangle\equiv \langle\rangle$ boundary, where no membranes are allowed to terminate.  All three generating charge anyon worldlines, and their combinations, are allowed to terminate on this boundary. 
\end{itemize}
This is similar to the rough boundary of each of the three copies of the toric code.

\begin{itemize}
\item [(2,3,4)] $\langle c\rangle$ boundary, where only the membranes of color $c \in \{ r,g,b\}$ are allowed to terminate. 
Only charge anyon worldlines of complementary colors are allowed to terminate on this boundary.\footnote{If we allowed a $c$-colored line to terminate at the boundary, the parity of line-membrane intersections of color $c$ would not be invariant under local deformations.}
\end{itemize}
For three copies of toric codes, this is similar to the boundary that is a smooth for the copy of color $c$ and rough for the copies labeled by other colors.

\begin{itemize}
\item [(5,6,7)] $\langle c_1,c_2\rangle$ boundary, where membranes of colors $c_1 \neq c_2 \in \{ r,g,b \}$ are allowed to terminate.
Only charge anyon worldlines of the third color are allowed to terminate on this boundary.
\end{itemize}
This is similar to a rough-smooth-smooth boundary for three copies of the toric code. In fact, for three copies of the toric code, there is one additional boundary for this subgroup that includes non-trivial boundary weights. 
However, for the TQD, these boundaries with and without the additional weights are phase equivalent.

\begin{itemize}
\item [(8,9,10)] $\langle c_1 c_2\rangle$ boundary, where only the membranes of colors $c_1 \neq c_2\in \{ r,g,b \}$ must terminate simultaneously along the same line (and cannot terminate individually).
Charge worldlines of colors $c_1$ and $c_2$ are allowed to terminate simultaneously on the same point on this boundary, as well as individual anyon worldlines of the third color.
\end{itemize}
For three copies of toric code, this corresponds to a folding boundary between the copies of colors $c_1$ and $c_2$, and a rough boundary for the third copy. 
\begin{itemize}
\item [(11)] $\langle rg,rb \rangle$ boundary, where any pair of membranes of different colors is allowed to terminate simultaneously, but none of them individually. The charge worldlines of all three colors are allowed to terminate simultaneously (but no individual worldlines or pairs are allowed to). 
\end{itemize} 
This boundary requires additional weights, namely a factor of $i$ at each intersection between an $rg$ and an $rb$ termination lines on the boundary.  In space, this weight is associated to a process of exchanging two pair-termination-points on the boundary. This is illustrated and explained in more detail in Appendix~\ref{app:boundaries-DWs-corners}.

\subsubsection{Domain walls between the TQD and toric codes}

Our protocols act on logical states of toric/surface codes, and hence we require interfaces between the toric code and the TQD phases to transfer the logical information into and out of the non-Abelian phase. 
Similar to gapped boundaries, domain walls are described by ways in which membranes can terminate, as well as additional weights on the domain wall.

The domain walls we consider are equivalent to gapped boundaries of a stack consisting of the TQD and toric codes, as the toric codes can be folded onto the same side of the domain wall as the TQD. 
These domain walls are fully classified, however, the complete list of possibilities is too numerous to discuss in full here and so we restrict our attention to domain walls that are relevant for our purposes.

First, we consider the relevant domain walls between a single toric code (which we label with the color purple, denoted by `$p$') and the TQD.
We focus exclusively on domain walls that we can use to transfer all logical information from the toric codes to the TQD. For the types of domain walls used in our protocols, this is ensured by the fact that each type of toric code membrane is forced to terminate simultaneously with some nontrivial TQD membrane.\footnote{
Instead of simultaneous termination, a toric code membrane could be coupled to a TQD membrane via a non-trivial weight on the membrane.
} We note that such domain walls exist between one single toric code, or a pair of toric codes, and the TQD, but not between three toric codes and the TQD.

The termination pattern is given by a subgroup of $\mathbb Z_2^4$, where the first three $\zz_2$ factors\footnote{In this paper, we use $\zz_2$ to refer to the group $ \{ 0,1\}$ with addition modulo two, and therefore, use the notation interchangeably with $\mathbb F_2$ in several places.} correspond to the TQD and the last $\zz_2$ factor corresponds to the toric code.
The domain walls that we use in our protocols are the following:
\begin{itemize}
\item $\langle cp\rangle$ domain wall where $c \in \{ r,g,b \}$. Only the membrane of color $c$ on the TQD side can terminate simultaneously with the toric code membrane $p$ along a common termination line. 
\end{itemize}
\begin{itemize}
\item $\langle c_1p,c_2\rangle$ domain wall with $c_1 \neq c_2 \in \{ r,g,b \}$, where the membrane of color $c_1$ from TQD side can only terminate simultaneously with the $p$ membrane of the toric code, and additionally, the membrane of color $c_2$ can freely terminate. 
\end{itemize}

Our protocols also use the domain walls between the TQD and two copies of the toric code.  In this case, we color the membranes of the copies of toric code purple ($p$) and yellow ($y$). The relevant domain walls are:
\begin{itemize}
\item  $\langle c_1p,c_2y\rangle$  domain wall with $c_1 \neq c_2 \in \{ r,g,b \}$, where TQD membranes of color $c_1$ can terminate simultaneously with the purple toric code membranes, and TQD membranes of color $c_2$ can terminate simultaneously with yellow toric code membranes.
\item  $\langle rgp, rby\rangle$ domain wall, which is derived from the $\langle rg,rb\rangle$ boundary of the TQD by forcing the $rg$ and $rb$ termination lines to coincide with the termination lines of the $p$ and $y$ toric codes, respectively.
\end{itemize}

\subsubsection{Corners of the TQD} \label{subsub:corners}

In some of the logical protocols, the types of corners used play an important role. For this reason, we now review the types of corners that can occur in the TQD phase. 

In a global topology with multiple different boundary conditions, there are also corners, which are lines in spacetime that correspond to interfaces between boundaries.
For each pair of boundaries, the corners can be fully classified into topologically distinct superselection sectors\footnote{
The distinct corners are the irreducible blocks of a *-algebra corresponding to a 2-dimensional path integral obtained by compactifying the bulk into a thin slab with the two boundaries at the top and bottom~\cite{Kitaev2012Models,Bridgeman2018Fusing,Bridgeman2019b,Bridgeman2019a,Magdalena_de_la_Fuente_2023}.
}.
The corners we consider simply interface termination lines on the boundaries on either side.
For example, a corner between the $\langle rg,rb\rangle$ and the $\langle r,g\rangle$ boundary may interface a $rg$ termination line on one side with a separate $r$ and $g$ termination line on the other side (see \eqref{fig:G16} in Appendix~\ref{app:boundaries-DWs-corners} and the surrounding discussion).
In general, there exist different superselection sectors of corners between the same boundary types, which are distinguished by assigning different weights to the points on the corner where termination lines are interfaced.
For the configurations we consider, there are only two different weights that differ by a factor of $-1$, so the choice of corner only changes the protocols by an unimportant single-qubit logical $Z$ operator. 
Specifically, the weights are always $\pm 1$ (and we choose the ``trivial'' $+1$ weight), except for the corners between the $\langle rg,rb \rangle$ and $\langle g,b \rangle$ boundaries, where the weights are $\pm\omega$ with $\omega = e^{-i \pi/4}$ (we choose the $+\omega$ weight).
Note that the $\omega$ weight of this corner is the source of the ``non-Cliffordness'' (magic) in the protocols for the logical $\overline{XS}$ measurement (that can be used for magic state preparation) and the $\overline{T}$ gate. 

Note that the TQD phase also allows for more intricate corners where we sum over additional closed-loop configurations inside the corner worldline, but we do not make use of these. 
We refer the reader to Appendix~\ref{app:boundaries-DWs-corners} for a more complete discussion of the corners. 

Finally, the 0-dimensional defects in spacetime (such as the endpoints or meeting points of the corners) can be identified with the ground states for the spatial configuration obtained by taking the intersection of the spaceitme with a (sufficiently small) 2-sphere surrounding the point.
If the according ground state space has dimension 2 or higher, then this means that part of the logical information is not encoded globally but can be accessed by operators acting locally at this point.
For the global protocol to be fault-tolerant, the ground space dimension thus must be one for every spacetime point, in which case local operators can only gives rise to irrelevant global prefactors.

\subsection{Logical operations}
\label{sec:global_topologies}

In this section, we construct different global topologies that result in specific logical operations. For each global topology, we label each 3-volume by the topological phases that are either the TQD, a single toric code (TC), or two copies of the toric code (TC$\times$TC). We further specify the types of domain walls and boundaries in a given geometry using the notation introduced above. We also designate the state boundaries with input and output states, which we label as logical input ($i$) and output $(o)$ states of the protocol and ensure that they always belong to the toric code phase. We also ensure that the distance between any two points in spacetime, where flux or charge defects termination can affect the logical state, is at least $d$ -- the desired distance of the code. An example of such a pair of points is any pair of points on two distinct boundaries where the same type of flux or charge can terminate.  
For each labeled global topology, we then derive its logical action using the closed-membrane picture for the path integrals.

\subsubsection{Computing the logical action }

First, we explain how to compute the logical operation corresponding to a global topology in general, and then we consider specific examples.
We start by determining the generating cohomology classes of string configurations on the input and output states of the input toric or surface codes, which are in one-to-one correspondence with the logical states in the Pauli-$Z$ basis.
We label the $i$-th input or output state boundary as $I_i$ or $O_i$, respectively. 
The total input and output cohomology classes are then of the form 
\begin{align*}
    \sum_i \alpha_i I_i && \text{ and} && \sum_i \beta_i O_i ,
\end{align*}
 for some $\alpha_i, \beta_i\in \zz_2$, respectively. 
They form a basis for the input and output vector spaces for the logic operation, which we denote by $\ket\alpha$ and $\bra\beta$, respectively. 
Here, we have defined $\zz_2$-valued vectors $\alpha = (\alpha_1, \alpha_2, \dots)$ and $\beta = (\beta_1, \beta_2, \dots)$.
The input and output states that we use in our protocols are placed on either tori (containing the toric code copies), rectangular surface code patches, or surface code patches that are folded into triangles, which are equivalent to the color code~\cite{Kubica2015unfolding}. 

We use the path integral to define an operator $F$ whose matrix elements $\bra{\beta}F\ket{\alpha}$ determine the logical operation applied between input and output states.
To compute these matrix elements, we first use the global topology to determine a set of generating cohomology classes of ``bulk'' membranes in the TQD and toric code phases, which we label $B_j$. These can be thought of as the path along which the logical information can continuously flow between the input and output states.
The bulk cohomology classes are thus $\sum_j \gamma_j B_j$ for some $\gamma_j\in \zz_2$, and we keep track of them via an $\zz_2$-valued vector $\gamma =  (\gamma_1, \gamma_2, \dots)$.  By restricting $B_j$ to the $i$-th input state boundary, we obtain $\mathcal I_{i,j}\in \zz_2$. We define $\mathcal O_{i,j}$ analogously for the restriction of $B_j$ to the $i$-th output state boundary. Both $\mathcal I$ and $\mathcal O$ are thus $\zz_2$-valued matrices.

Finally, for every bulk cohomology class we further determine the weight associated to this closed-membrane configuration by the path integral which we denote $Z( \gamma)$.
We then use these weights to evaluate the desired matrix elements:
\begin{equation}
\bra{\beta}F\ket{\alpha}
= \sum_{\gamma:  \mathcal I \gamma=\alpha, \mathcal O \gamma=\beta} Z( \gamma)\, .
\label{eq:loopsoup_unitary_logic}
\end{equation}

In some cases, the operator $F$ is not unitary, but rather a partial isometry. This occurs when it is not possible to find a correction that maps the corresponding circuit with an arbitrary set of observed measurement outcomes to the $+1$-postselected circuit. In this case, the classical measurement outcomes contain information about the logical input states. 
The resulting logical operation is a quantum channel rather than a unitary.
In our protocols,  this scenario occurs if there exist non-trivial cohomology class generators $C_i$ for charge worldlines that go between input and output state boundaries. 
The different charge cohomology classes are $\sum_i \epsilon_i C_i$ for some $\epsilon_i\in \zz_2$ and $\epsilon = (\epsilon_1, \epsilon_2, \dots)$.
In this case, the path integral weights also depend on $\epsilon$, and we denote them $Z^\epsilon(\gamma)$. 
The resulting logical operation is a measurement defined by a collection of operators $F^\epsilon$, where $\epsilon$ corresponds to the logical measurement outcome. 
The matrix elements of these operators are 
\begin{equation}
\label{eq:loopsoup_general_logic}
\bra{\beta}F^\epsilon\ket{\alpha}
= \sum_{\gamma: \mathcal I\gamma=\alpha, \mathcal O\gamma=\beta} Z^\epsilon(\gamma)\, .
\end{equation}

\subsubsection{$\overline{CCZ}$ unitary gate}
For our first example, we consider the protocol realizing the unitary $\overline{CCZ}$ gate that was briefly described in Sec.~\ref{sec:description} (and Appendix~\ref{sec:more_gates}). The corresponding global topology is shown below. 
\begin{equation} 
\label{eq:loopsoup_ccz_protocol}
\vcenter{\hbox{\includegraphics[width = 0.8\columnwidth]{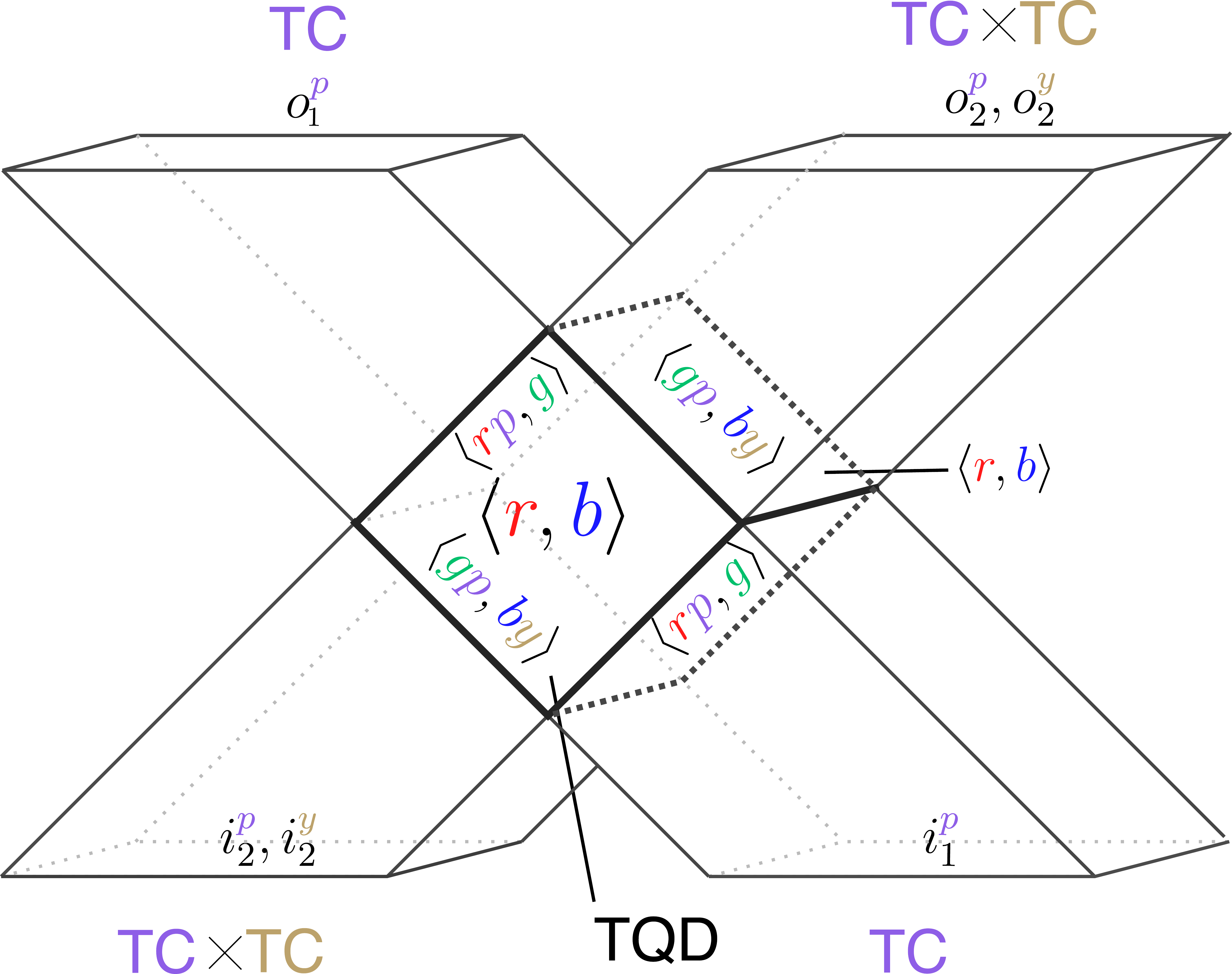}}}
\end{equation}
Here, and all the protocols in this section, we orient the figures such that the time direction used to derive a circuit flows up the page.  
The global topology above is built from a cube containing the TQD phase and regions containing either one or two copies of the toric code phase.
Two of the opposing faces of the TQD cube are interfaces with cubes that contain a single copy of the toric code phase.
Another pair of opposing faces of the TQD cube are interfaces with cubes containing two copies of the toric code. 
The other two faces of the TQD cube correspond to gapped boundaries. 

Each region containing (copies of) the toric code has a face that is an input or an output state boundary labeled by $i$ or $o$, correspondingly. 
There is a subscript on the input or output label enumerating the inputs and outputs, and a superscript corresponding to the copy of the toric code phase  that is labeled by the color $p$ or $y$. 
We omit the boundary labels on the faces of the regions containing (copies of) the toric code phase. 
These are either $\langle \rangle$, $\langle p\rangle$, or $\langle y\rangle$, corresponding to rough and smooth surface code boundary conditions in the associated copy. 
The boundary labels are chosen such that all spatial slices define surface codes with rough and smooth boundaries that produce logical operators consistent with the figures below.

There are three generating input classes and three generating output cohomology classes, corresponding to logical states of the individual surface code patches. 
These are labeled $I_2^p$, $I_2^y$, $I_1^p$, and $O_2^p$, $O_2^y$, $O_1^p$, following the labeling in Eq.~\eqref{eq:loopsoup_ccz_protocol}.
There are three bulk cohomology class generators%
\begin{equation}
\begin{gathered}
B_g= \raisebox{-0.4\height}{ \includegraphics[width = 0.15\columnwidth]{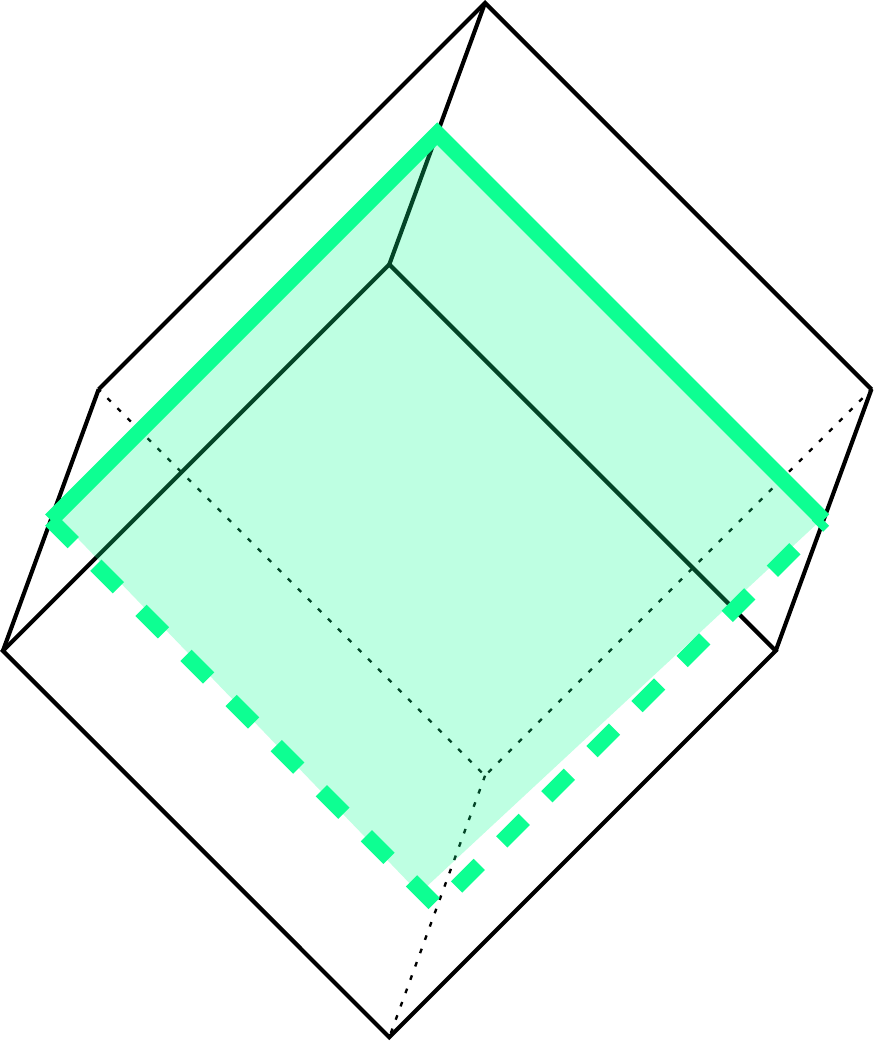}}\, , \, B_b=\raisebox{-0.4\height}{\includegraphics[width = 0.15\columnwidth]{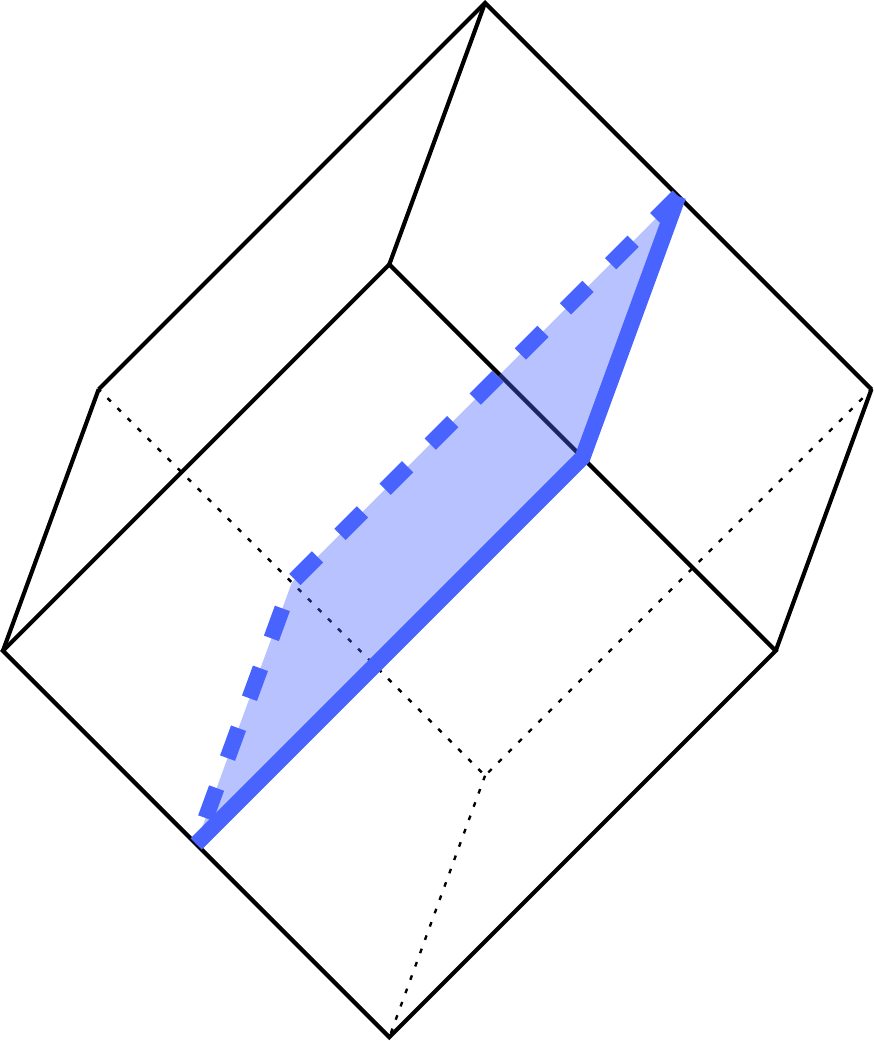}}\, ,\, B_r=\raisebox{-0.4\height}{\includegraphics[width = 0.15\columnwidth]{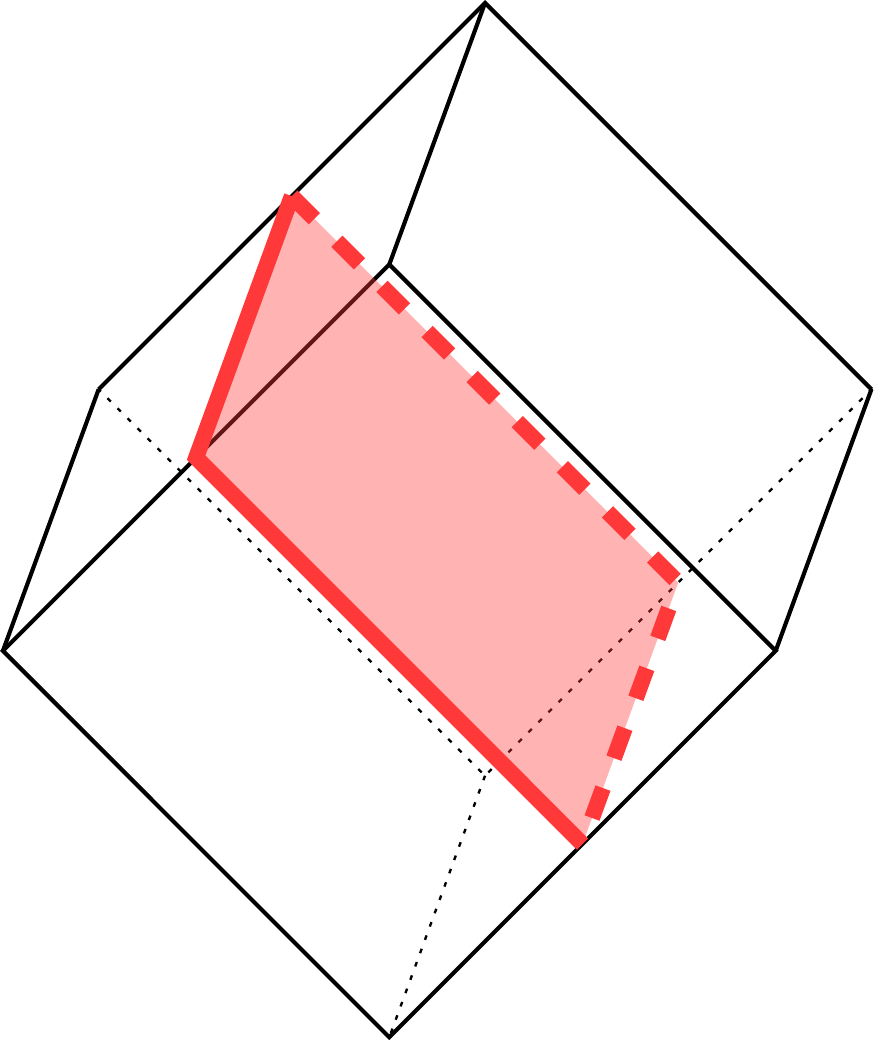}}\, .\\
\end{gathered}
\end{equation}
In this example, there are no nontrivial cohomology classes of charge anyon worldlines. 
Hence, the resulting logical operation is unitary.
From the above figures, we find the following compatibility matrices between bulk and boundary cohomology classes to be 
\begin{equation}
\label{eq:loopsoup_ccz_restriction}
\mathcal I =
\begin{pmatrix}
& I_1^p & I_2^p & I_2^y\\
B_g & 0 & 1 & 0\\
B_b & 0 & 0 & 1\\
B_r & 1 & 0 & 0
\end{pmatrix}
,\
\mathcal O=
\begin{pmatrix}
& O_1^p & O_2^p & O_2^y\\
B_g & 0 & 1 & 0\\
B_b & 0 & 0 & 1\\
B_r & 1 & 0 & 0
\end{pmatrix}
.
\end{equation}
The weight obtained from the path integral as a function of bulk cohomology classes is
\begin{equation}
Z(\gamma)=(-1)^{\gamma_g\cdot \gamma_b\cdot \gamma_r}\, .
\end{equation}
The nontrivial weight is depicted below
\begin{equation}
Z((1,1,1))
=
Z\Big(
\raisebox{-0.4\height}{
\includegraphics[width = 0.2\columnwidth]{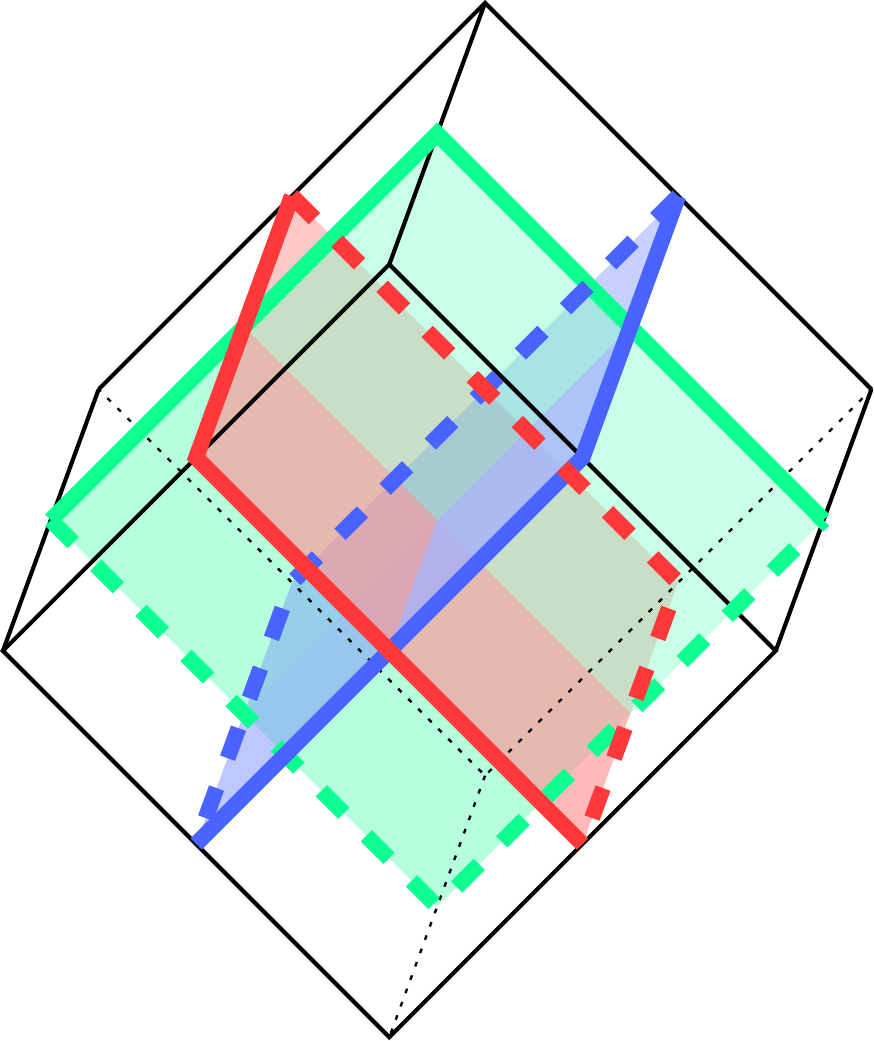}
}
\Big)
=
-1\, ,
\end{equation}
where membranes of all three colors are present with different orientations, which results in an odd number of triple intersection points. 
At this point we have described the necessary data. 
We now evaluate the sum in Eq.~\eqref{eq:loopsoup_unitary_logic}. 
For $\alpha \neq \beta$, the sum is over the empty set and yields zero, and if $\alpha=\beta$, the sum is over a single bulk cohomology class determined by $\gamma=\mathcal I^{-1}\alpha=\mathcal O^{-1}\beta$. This gives the expression:
\begin{equation}
\bra{\beta}F\ket{\alpha} = \delta_{\alpha_1^p,\beta_1^p}\delta_{\alpha_2^p,\beta_2^p}\delta_{\alpha_2^y,\beta_2^y} (-1)^{\alpha_1^p\cdot \alpha_2^p\cdot \alpha_2^y}\, .
\end{equation}
Which coincides with the $CCZ$ logic gate.

Below we depict a deformation of this protocol:%
\begin{equation}
\vcenter{\hbox{\includegraphics[width = 0.55\columnwidth]{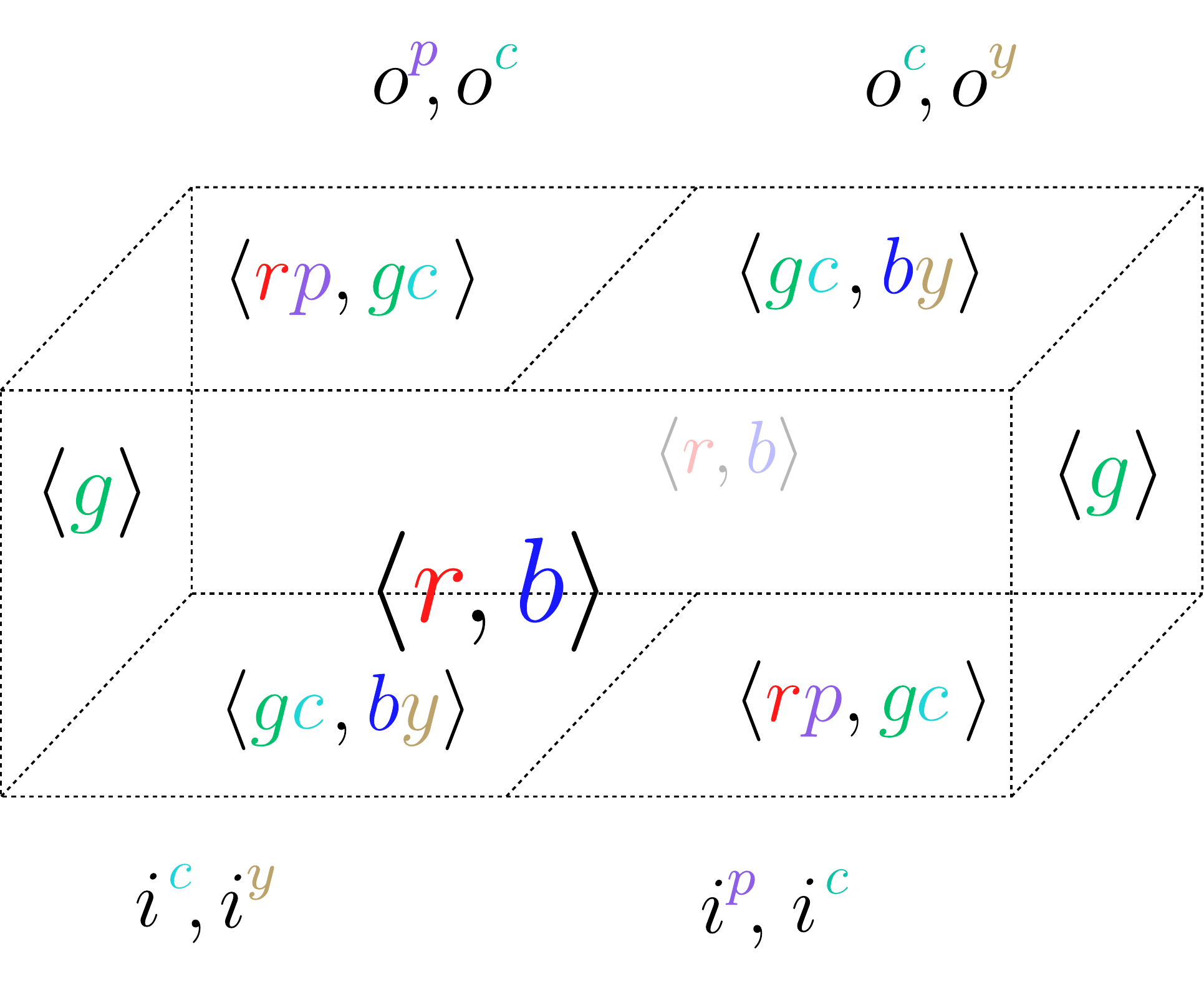}}}
\label{fig:CCZ-alt}
\end{equation}
where we did not show the regions containing the toric code copies (we only show the labels of the corresponding input and output state boundaries; the regions can be added by analogy with Fig.~\eqref{eq:loopsoup_ccz_protocol}). We added a $`c'$-colored copy of the toric code, where $`c'$ stands for ``cyan''. 
Compared to Fig.~\eqref{eq:loopsoup_ccz_protocol}, we introduced additional $\langle g\rangle$ boundaries such that the domain walls between the TQD and toric code are purely spatial.
These additional boundary sections do not change the logical action of the protocol, and the analysis is the same as in Fig.~\eqref{eq:loopsoup_ccz_protocol}.
For example, the bulk cohomology class $\gamma_g=\gamma_b=\gamma_r=1$ with $Z(\gamma)=-1$ is shown below.
\begin{equation}
\vcenter{\hbox{\includegraphics[width = 0.5\columnwidth]{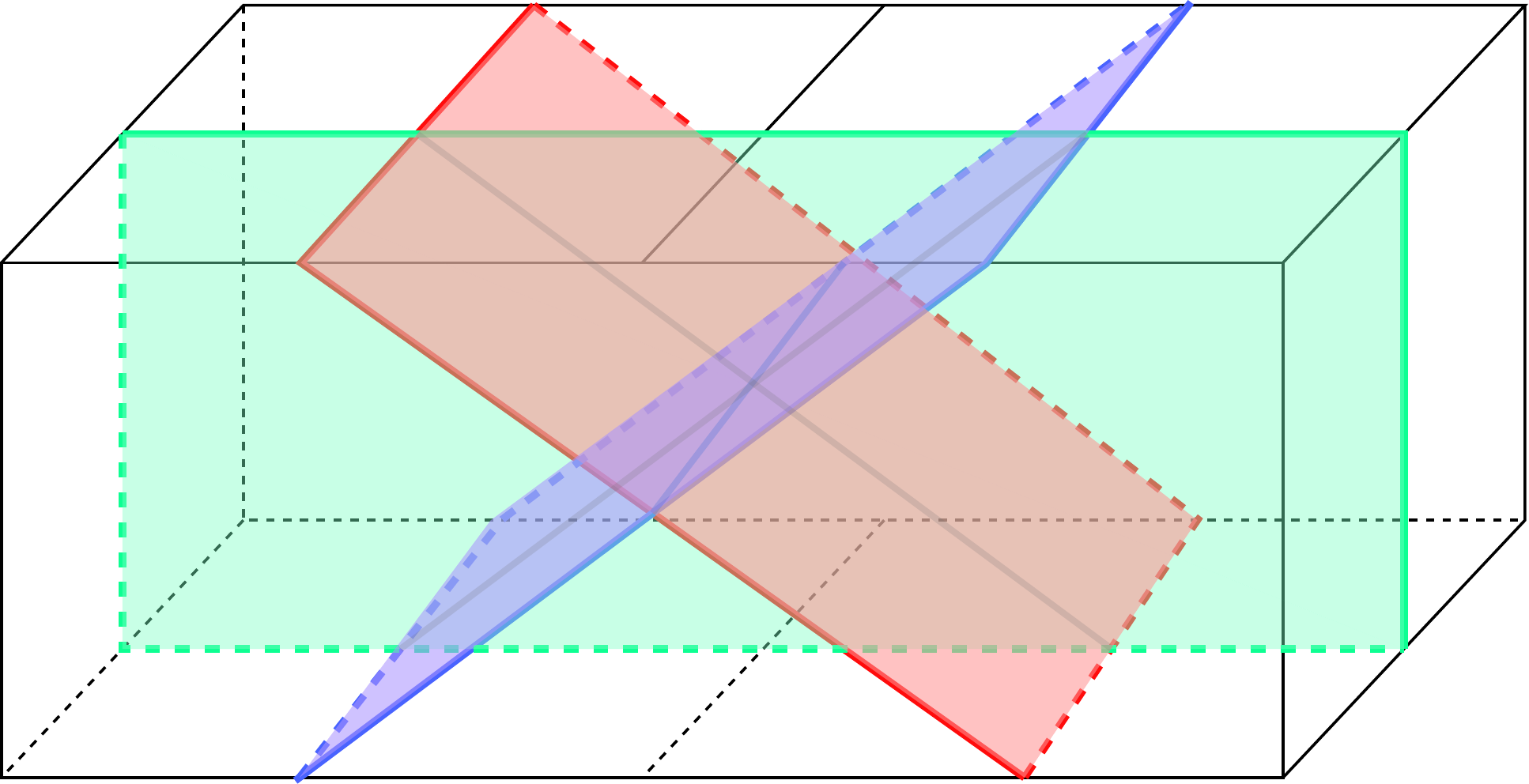}}}\;
\end{equation}

\subsubsection{Logical $\overline{CZ}$ measurement} \label{sec:referenced-example}
The next example we consider corresponds to a global topology that yields a logical $CZ$ measurement. 
This is a variant of the example in Secs.~\ref{sec:description}, and~\ref{sec:CZ_minimal}, with open spatial  boundary conditions. See also the example in Appendix~\ref{sec:more_gates}. 

\begin{equation}
\label{eq:czmeasure_global_protocol}
\vcenter{\hbox{\includegraphics[width = 0.5\columnwidth]{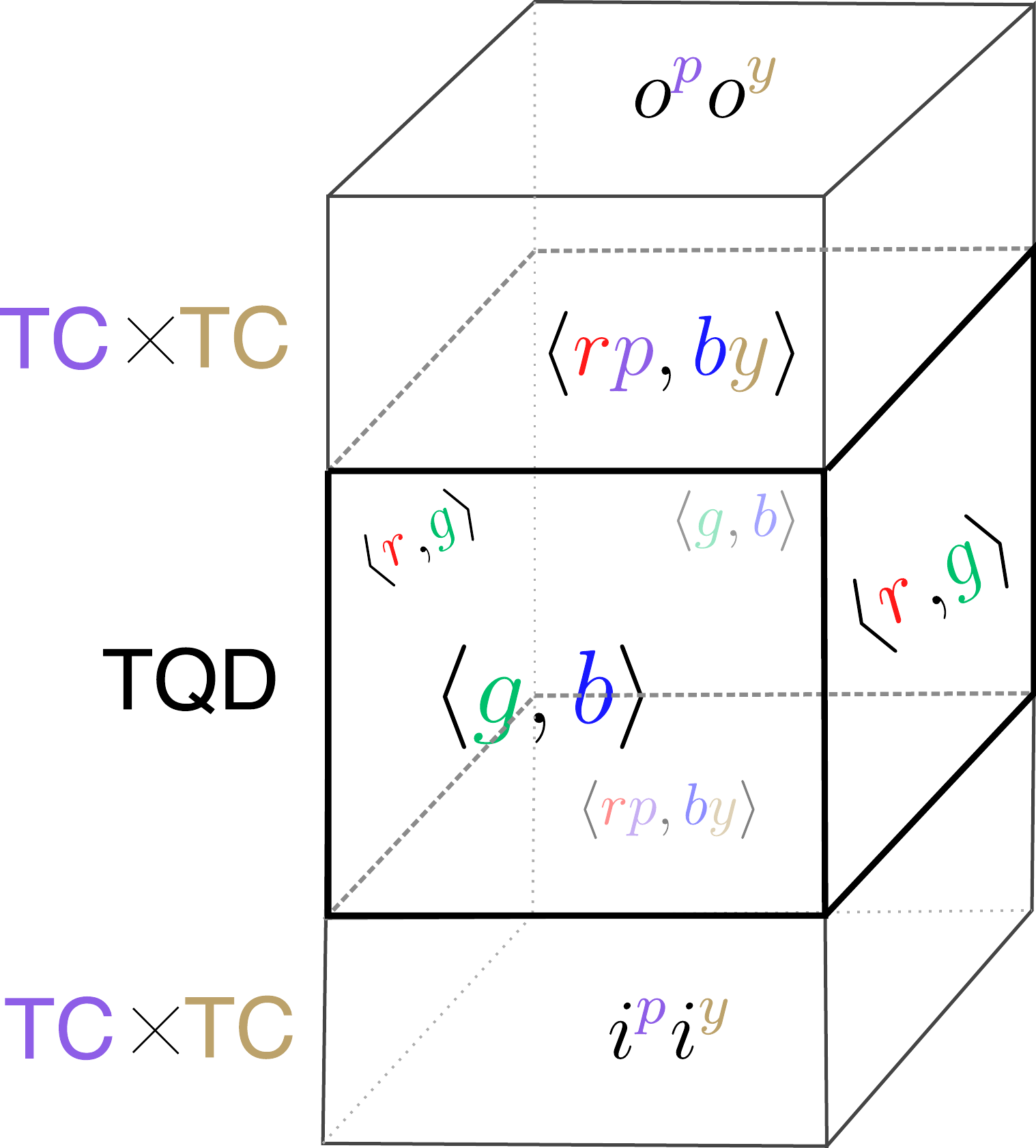}}}\;
\end{equation}
The cube in the middle of this configuration hosts the TQD phase which interfaces two copies of the toric code above and below it.  
The spatial sections of the regions hosting the toric code phase can be chosen to either be tori (as in Sec.~\ref{sec:CZ_minimal}) or rectangles (in which case the initial and final codes are the surface codes). In the latter case, we choose the spatial boundary conditions such that they are consistent with the figures below. We only consider the case with input and output surface codes here; the case with input and output toric codes can be calculated analogously. 

There are two generating input and output cohomology classes $I^p$, $I^y$ and $O^p$, $O^y$, which we label as in Eq.~\eqref{eq:czmeasure_global_protocol}.
There are three generating bulk cohomology classes:
\begin{equation}
\begin{gathered}
B_r=\raisebox{-0.42\height}{ \includegraphics[width = 0.15\columnwidth]{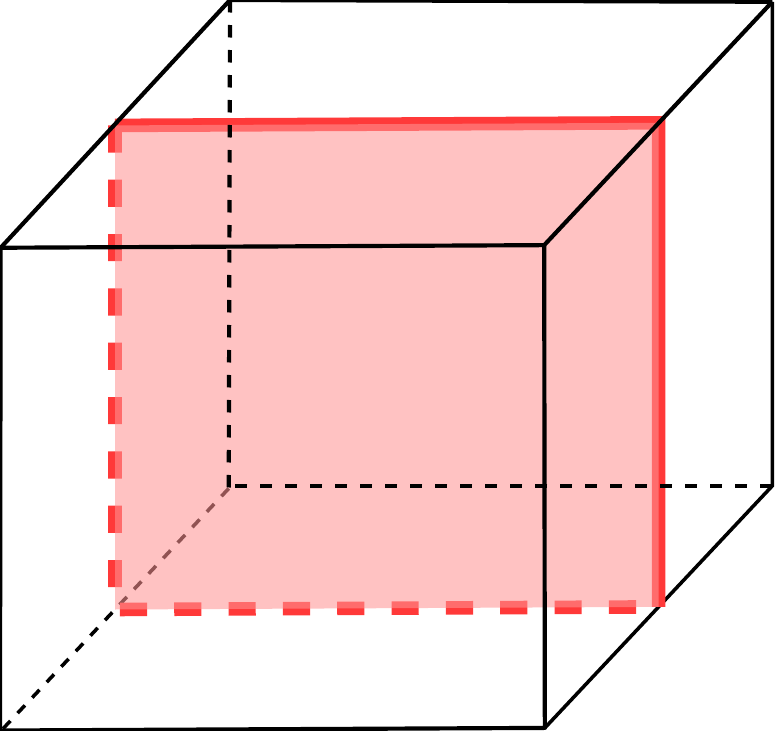}}\,,\, B_b=\raisebox{-0.42\height}{ \includegraphics[width = 0.15\columnwidth]{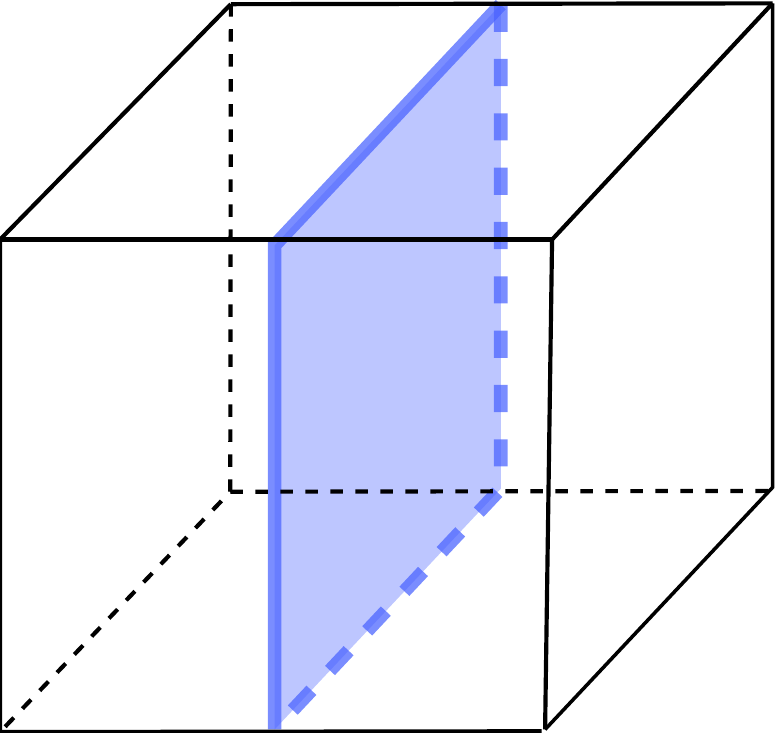}}\,,\, B_g=\raisebox{-0.42\height}{ \includegraphics[width = 0.15\columnwidth]{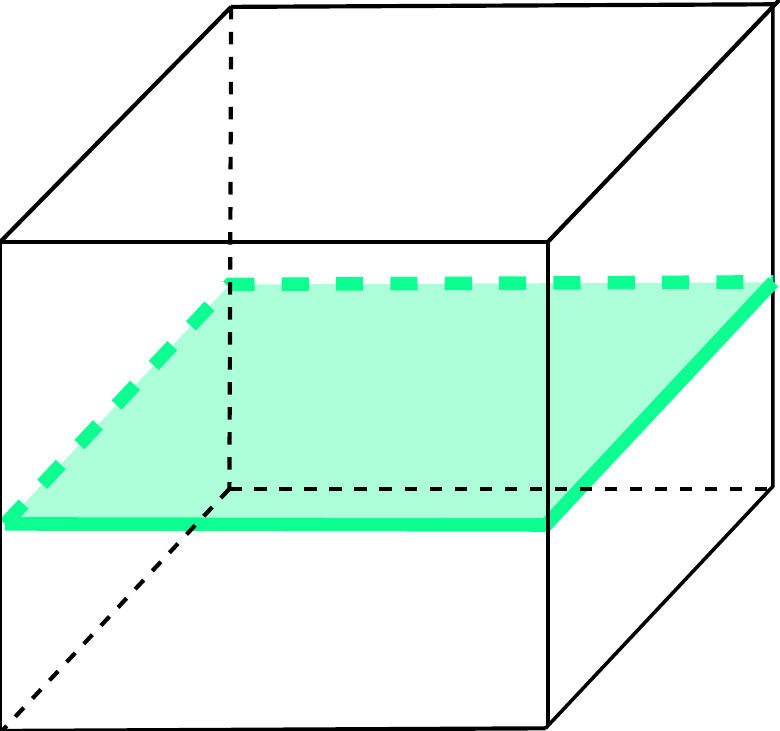}}\, . \\
\end{gathered}
\end{equation}
The restriction of the bulk cohomology classes to the input and output is described by
\begin{equation}
\label{eq:loopsoup_czmeasure_restriction}
\mathcal I =
\begin{pmatrix}
& I^p & I^y\\
B_g & 0 & 0\\
B_b & 0 & 1\\
B_r & 1 & 0
\end{pmatrix}
\, ,\quad
\mathcal O=
\begin{pmatrix}
& O^p & O^y\\
B_g & 0 & 0\\
B_b & 0 & 1\\
B_r & 1 & 0
\end{pmatrix}
\, .
\end{equation}
There is also a single nontrivial bulk charge-line cohomology class generator:
\begin{equation}
C_g = 
\raisebox{-0.4\height}{
\includegraphics[width = 0.18\columnwidth]{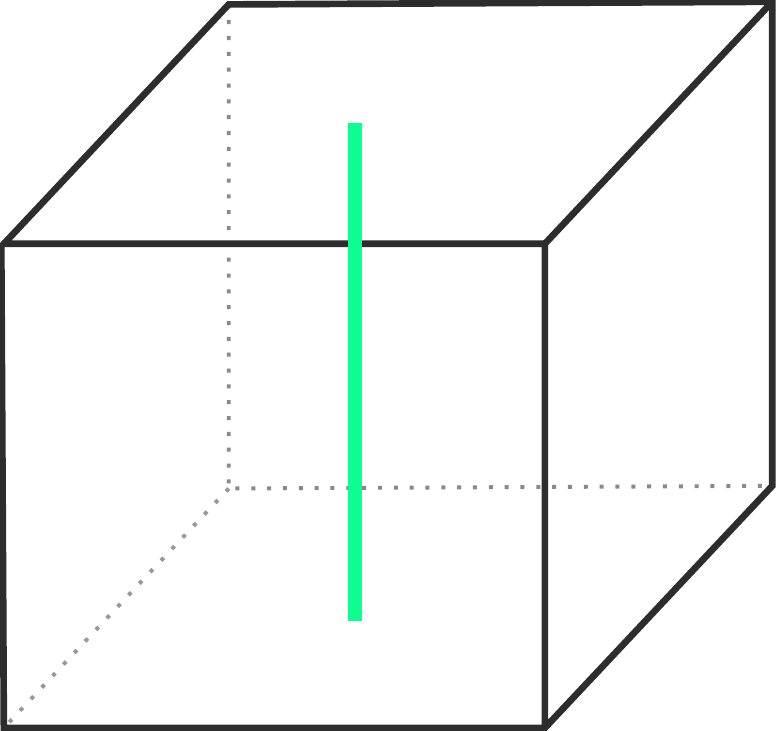}} \, ,
\end{equation}
so the cohomology classes are of the form $\epsilon_g \cdot C_g$ for ${\epsilon_g\in \zz_2}$. The weights of the bulk cohomology classes are
\begin{equation}
Z^{\epsilon_g}(\gamma)= (-1)^{\gamma_r\cdot \gamma_b\cdot \gamma_g + \epsilon_g \cdot \gamma_g}\, .
\end{equation}
We now evaluate the sum in Eq.~\eqref{eq:loopsoup_general_logic}.
For $\alpha\neq \beta$, there is no $\gamma$ such that $\mathcal I\gamma=\alpha$ and $\mathcal O\gamma=\beta$.
So the sum is over the empty set and yields zero.
On the other hand, if $\alpha=\beta$, the sum is over two cohomology classes, namely $\gamma_r=\alpha^p=\beta^p$, $\gamma_b=\alpha^y=\beta^y$, and $\gamma_g\in \{0,1\}$.
For $\epsilon_g=0$ and $\alpha^p=\alpha^y=1$, the two summands cancel.
The same happens for $\epsilon_g=1$ and $\alpha^p=0$, or $\alpha^y=0$.
In all other cases, both summands are $+1$.
Finally, we obtain
\begin{equation}
\bra{\beta}F^{\epsilon_g}\ket{\alpha} = \delta_{\alpha^p,\beta^p}\delta_{\alpha^y,\beta^y} \delta_{\epsilon, \alpha^p\cdot \alpha^y}\, .
\end{equation}
For $\epsilon_g = 0$ and $\epsilon_g = 1$, these are the matrix elements of the two projection operators associated to a projective measurement of the observable $\overline{CZ}$ with measurement outcome $(-1)^{\epsilon_g}$.

\subsubsection{$\overline{T}$-state perparation}
\label{sec:T-state-macroscopic}

The global topology that we use to implement a non-unitary logical ``$XS$ measurement'' operation is shown below. See also the example in Appendix~\ref{sec:more_gates}. 
\begin{equation}
\vcenter{\hbox{\includegraphics[width = 0.45\columnwidth]{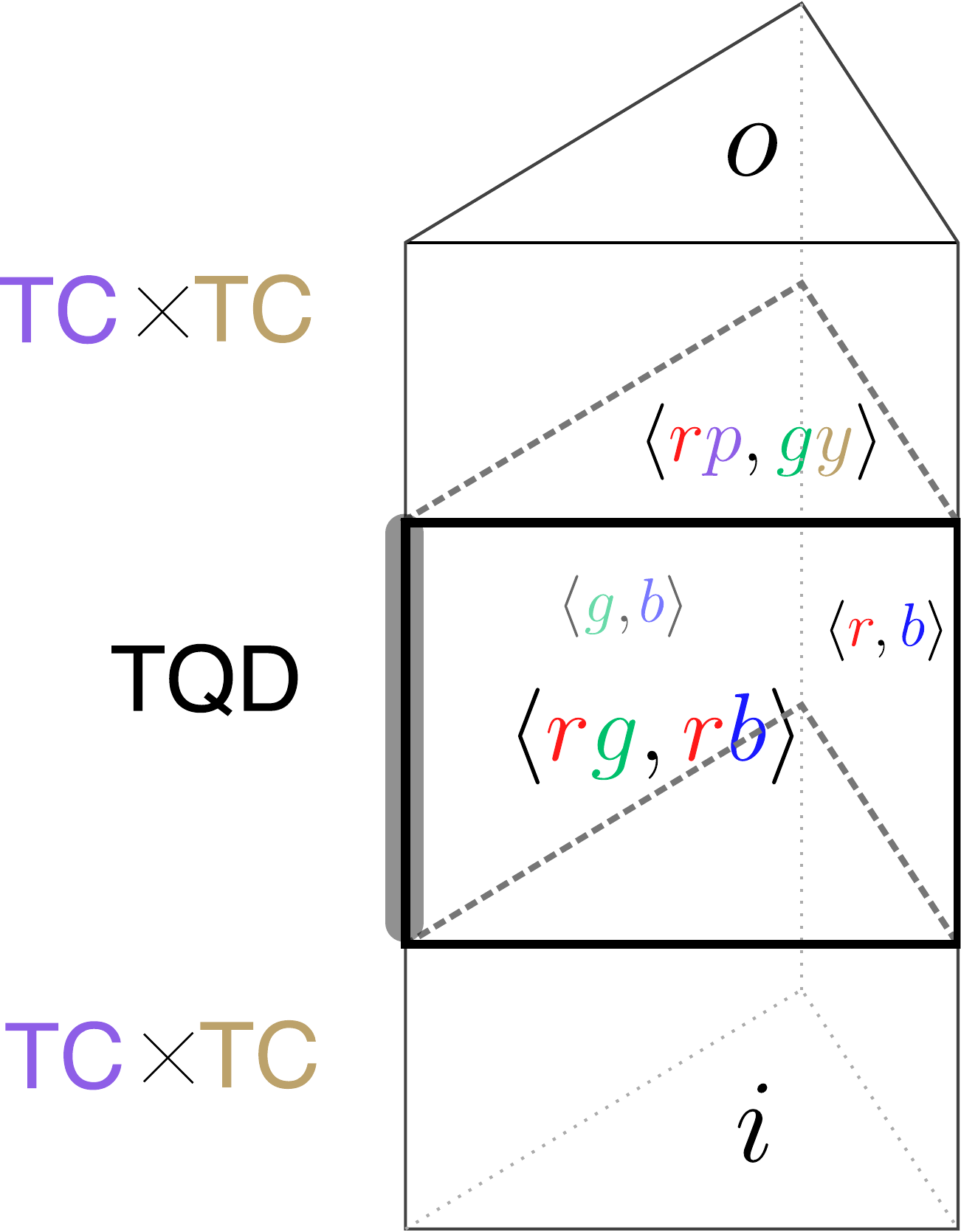}}}
\end{equation}
Here, the triangular prism in the middle supports TQD, with the regions below and above containing two copies of the toric code phase each.  The ``$+ \omega$'' corner between $\langle rg,rb \rangle$ and $\langle g,b \rangle$ boundaries is shown highlighted in dark grey. 
The surface codes are placed on a triangle with the ``fold'' boundary conditions (the ``fold'' corresponds to the $\langle py\rangle$ boundary placed on the front face), equivalent to a single copy of a surface code folded into a triangle. This is analogous to the color code triangle~\cite{Kubica2015unfolding}. The boundary conditions are otherwise chosen to be consistent with both the membranes that can pass into the TQD and condense on the boundaries there, and the figures below.

There is a single generating input, and output, cohomology class, which we label $I$, and $O$, respectively:
\begin{equation}
\begin{gathered}
I=\raisebox{-0.3\height}{\includegraphics[width = 0.17\columnwidth]{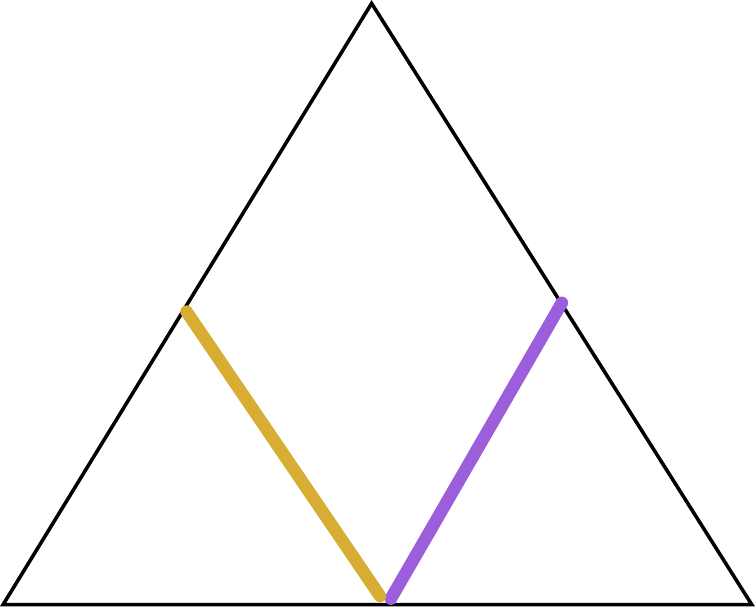}}\,,\quad O=\raisebox{-0.3\height}{ \includegraphics[width = 0.17\columnwidth]{fig/loopsoup_xsmeasure_input}}\, .\\
\end{gathered}
\end{equation}
There are two generating bulk cohomology classes
\begin{equation}
\begin{gathered}
B_0=\raisebox{-0.3\height}{\includegraphics[width = 0.15\columnwidth]{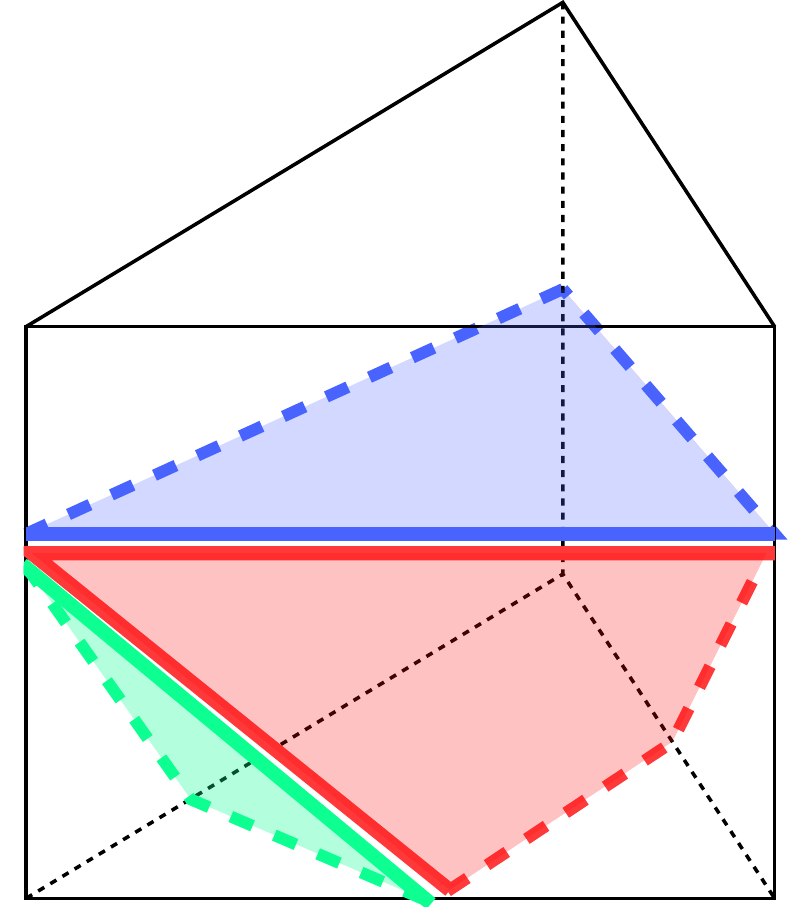}}\,,\quad B_1=\raisebox{-0.3\height}{ \includegraphics[width = 0.15\columnwidth]{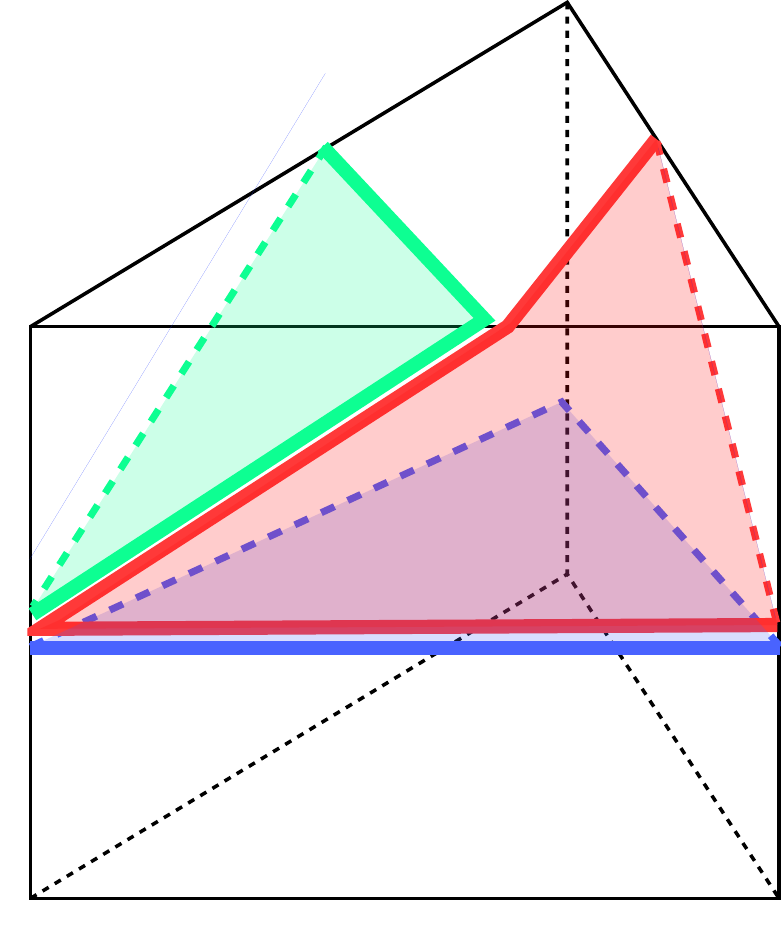}}\, .\\
\end{gathered}
\end{equation}
The restriction maps from the bulk cohomology classes to the input and output state cohomology classes are
\begin{equation}
\label{eq:loopsoup_xsmeasure_restriction}
\mathcal I =
\begin{pmatrix}
& I\\
B_0 & 1\\
B_1 & 0
\end{pmatrix}
\, ,\quad
\mathcal O =
\begin{pmatrix}
& O\\
B_0 & 0\\
B_1 & 1
\end{pmatrix}
\, .
\end{equation}
There is a single nontrivial charge worldline cohomology class generator in the bulk 
\begin{equation}
C_b = \raisebox{-0.3\height}{\includegraphics[width = 0.15\columnwidth]{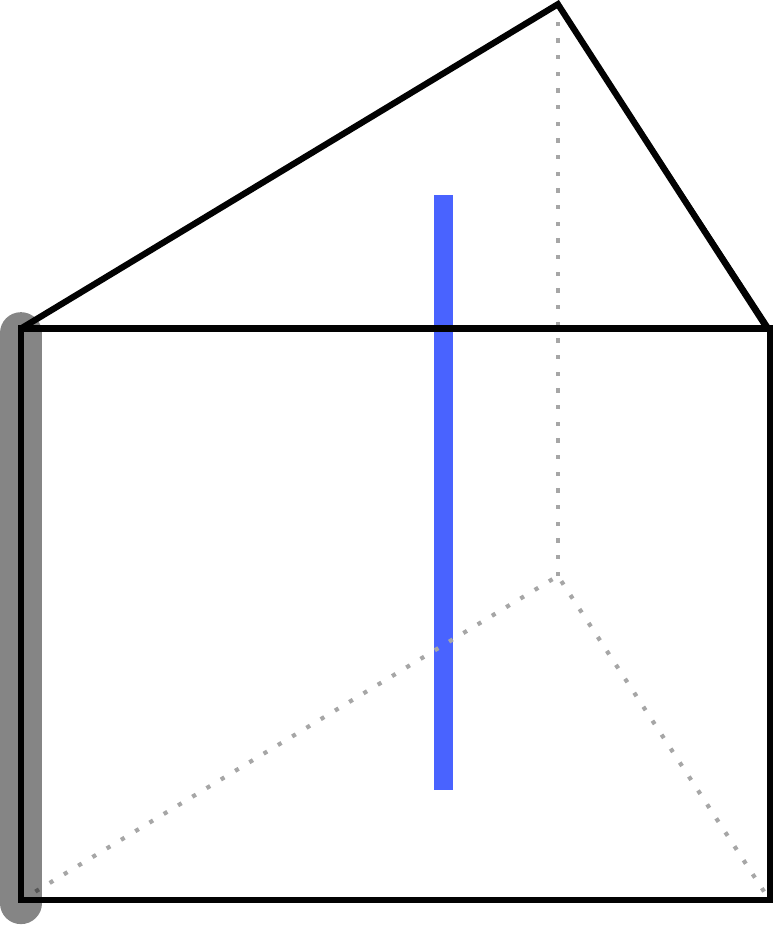}} \, .
\end{equation}
The weights of the bulk cohomology classes are
\begin{equation}
Z^{\epsilon_b}(\gamma)= \omega^{ \gamma_0} \overline{\omega} ^{\gamma_1}(-1)^{\epsilon_b\cdot (\gamma_0+\gamma_1)}\, .
\end{equation}
The factor $\omega^{ \gamma_0}$ ($\overline{\omega} ^{\gamma_1}$)  appears due to the termination of the membranes corresponding to $B_0$ ($B_1$) at the `$+ \omega$' corner between $\langle rg, rb \rangle$ and $\langle g,b \rangle$ boundaries (which is highlighted in dark grey in the figure), following the discussion in Subsec.~\ref{subsub:corners}. We now evaluate the sum in Eq.~\eqref{eq:loopsoup_general_logic}. 
For every $\alpha$ and $\beta$, we sum over only one bulk cohomology class, $\gamma_0=\alpha, \gamma_1=\beta$.
Hence, we obtain the matrix elements 
\begin{equation}
\bra{\beta}F^{\epsilon_b}\ket{\alpha} =  \omega^{\alpha}\overline\omega^{\beta} (-1)^{\epsilon_b\cdot (\alpha+\beta)}\;,
\end{equation}
or in matrix notation
\begin{equation}
F^{0} = \begin{pmatrix}1&\omega\\\overline\omega&1\end{pmatrix}\;,\qquad
F^{1} = \begin{pmatrix}1&-\omega\\-\overline\omega&1\end{pmatrix}
\end{equation}
For $\epsilon_b = 0,1$ these are precisely the projection operators onto the eigenstates of the $\overline{XS}$ operator, which are the $T^\dagger$-magic states with additional relative phase $(-1)^\epsilon_b$. Therefore, we colloquially call this action the ``logical  $\overline{XS}$ measurement''.

\subsubsection{$\overline{T}$ unitary gate}
The final example we present in this section corresponds to a global topology that implements the $\overline{T}$ gate.  This protocol is closely related to the $2+1$-dimensional non-Clifford gate in Ref.~\cite{bombin2018}.
\begin{equation}
\label{eq:loopsoup_t_protocol}
\vcenter{\hbox{\includegraphics[width = 0.4\columnwidth]{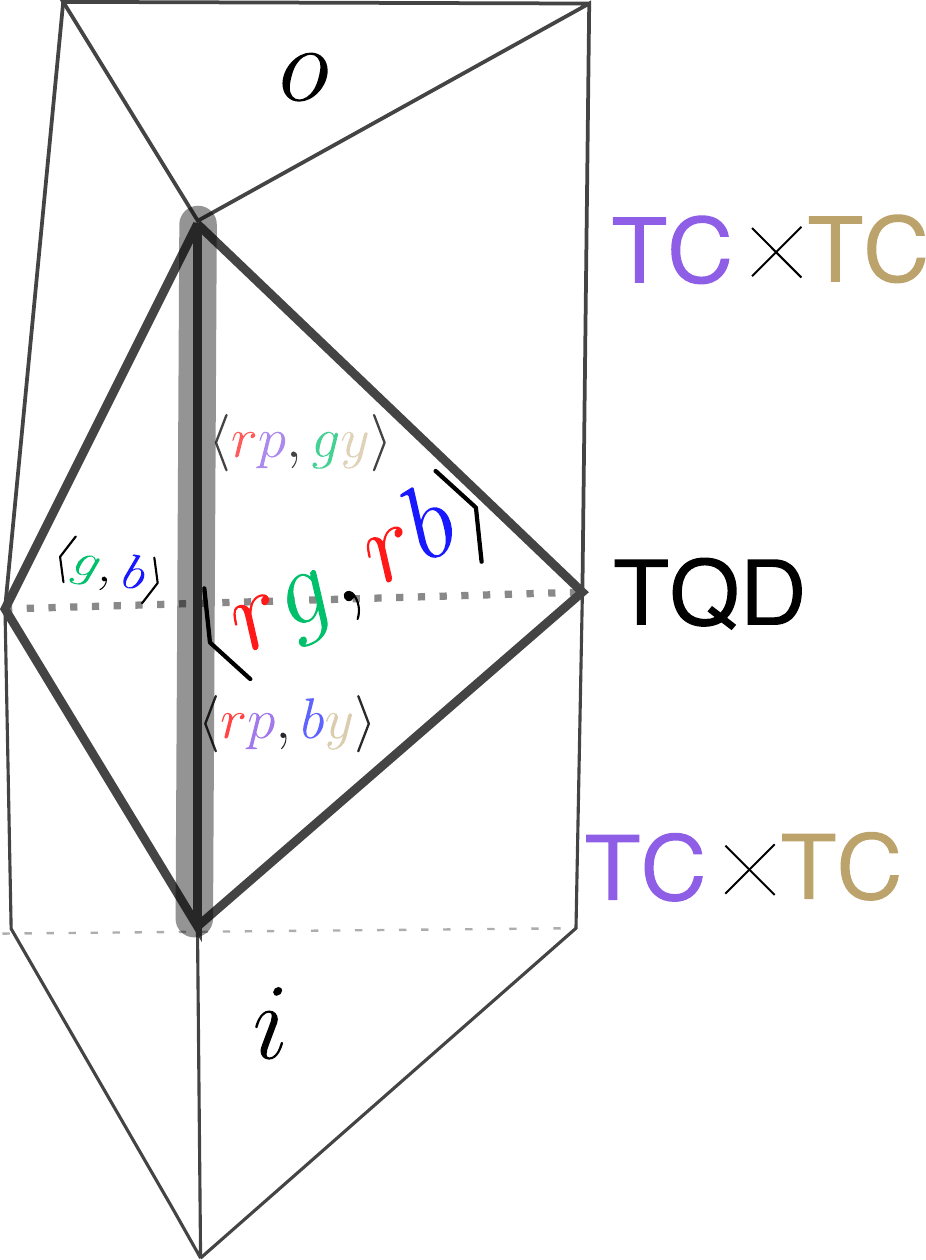}}}
\end{equation}
The tetrahedron in the middle contains the TQD phase, which is connected to two double-toric code phases across two of its triangular faces.  The `$+ \omega$' corner between $\langle rg,rb \rangle$ and $\langle g,b \rangle$ boundaries is shown highlighted in dark color. 
The toric codes have color-code-like boundary conditions, consistent with the membranes that pass into the TQD and condense on boundaries there. 
There is a single generating input, and output, cohomology class which we denote by $I$, and $O$, respectively.
There is also a single generating bulk cohomology class
\begin{equation}
B = \raisebox{-0.35\height}{\includegraphics[width = 0.23\columnwidth]{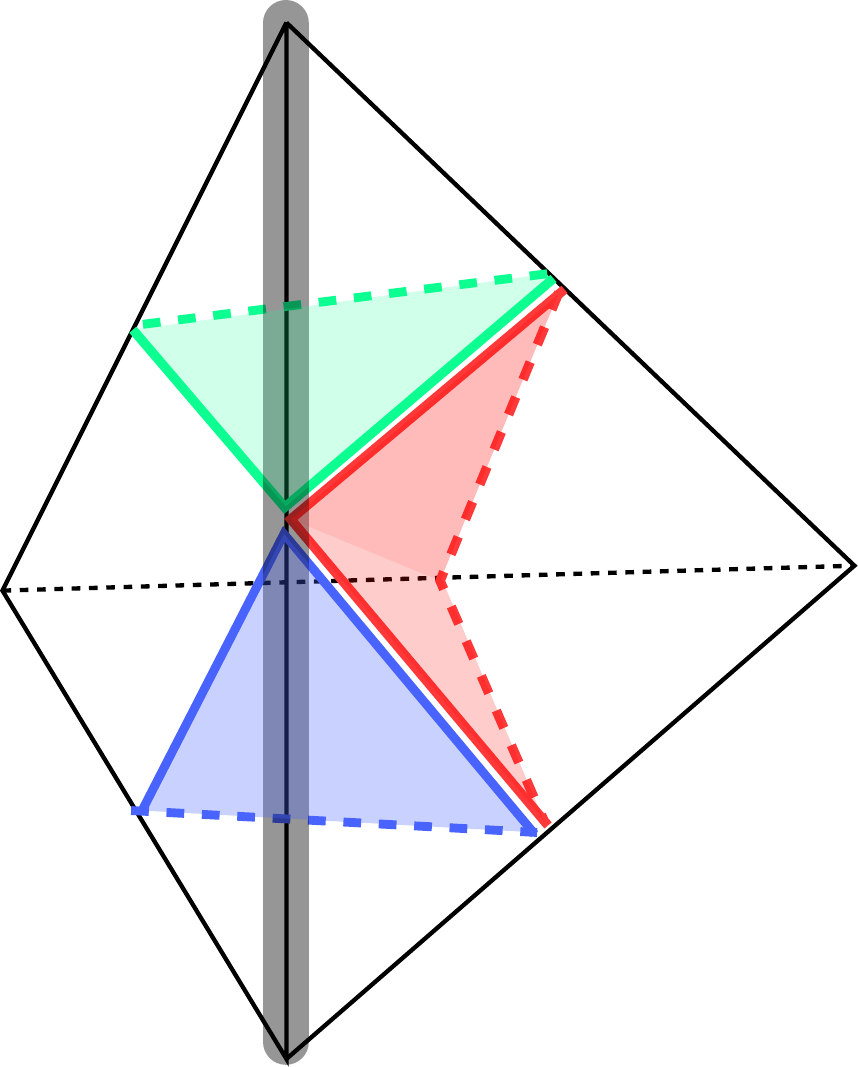}} \, .
\end{equation}
The restriction of the bulk cohomology class to the state boundary is simply
\begin{align}
\label{eq:loopsoup_t_restriction}
\mathcal I =
\begin{pmatrix}
 & I\\
B & 1
\end{pmatrix}
\, ,&&
\mathcal O =
\begin{pmatrix}
 & O\\
B & 1
\end{pmatrix}
\, .
\end{align}
There are no cohomologically nontrivial charge-worldlines. 
The weights of the bulk cohomology classes are determined by whether there is a membrane termination at the `$+ \omega$' corner between $\langle rg, rb \rangle$ and $\langle g,b \rangle$ boundaries (highlighted in dark grey in the figure):
\begin{equation}
Z(\gamma)= \omega^\gamma\, .
\end{equation}
We now evaluate the sum in Eq.~\eqref{eq:loopsoup_general_logic}.
If $\alpha\neq\beta$, due to Eq.~\eqref{eq:loopsoup_t_restriction} the sum is over the empty set and thus is zero.
If $\alpha=\beta$, the sum is over a single bulk cohomology class, $\gamma=\alpha=\beta$.
Finally, we obtain the matrix elements
\begin{equation}
\bra{\beta}F\ket{\alpha} = \delta_{\alpha,\beta}\omega^\alpha\;,
\end{equation}
which can be represented as the matrix
\begin{equation}
F = \begin{pmatrix}1&0\\0&\omega\end{pmatrix}\, .
\end{equation}
Hence, the logical operation is the unitary $\overline{T}$ gate.

A simpler variation of this configuration with the same topology is shown below. 
\begin{equation}
\label{eq:TGateSplit}
\vcenter{\hbox{\includegraphics[width = 0.25\columnwidth]{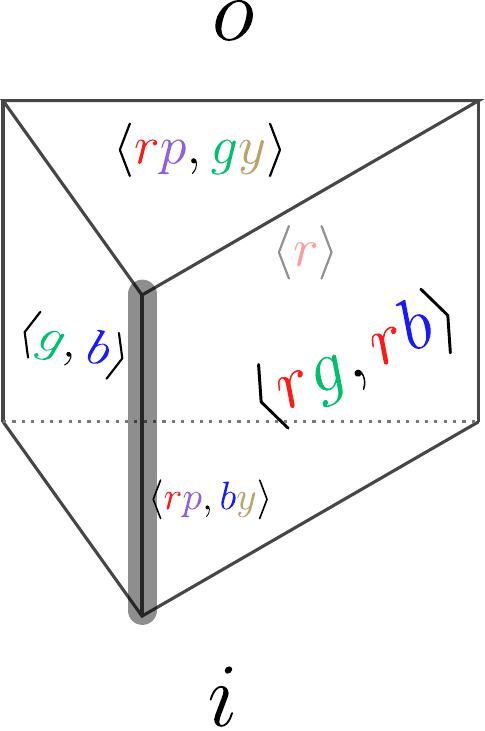}}}
\end{equation}
Here, we do not explicitly show the regions hosting copies of the toric code phase.  
This geometry is obtained by deforming the geometry shown in Eq.~\eqref{eq:loopsoup_t_protocol} and splitting the back edge into a pair of edges separated by the additional $\langle r \rangle$ boundary in the back, followed by shifting one to the bottom and one to the top.
The single generating bulk cohomology class $\gamma=1$ is shown below. 
\begin{equation}
B = \raisebox{-0.55\height}{\includegraphics[width = 0.22\columnwidth]{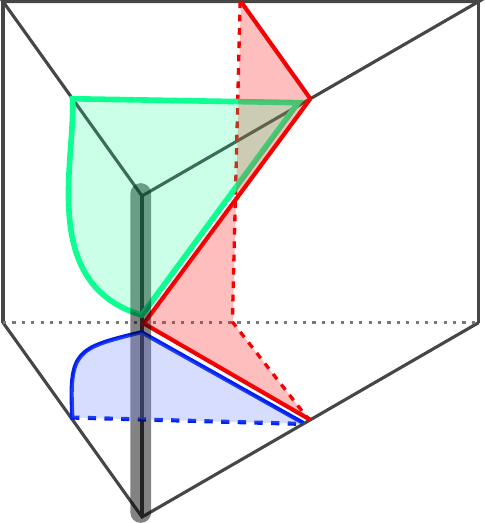}} 
\end{equation}
Following the analysis above, this results in a logical $\overline{T}$ gate. 

\section{Circuit implementations from the path integral}
\label{section:examples}

In this section, we discuss microscopic circuits for spacetime geometries that result in non-Clifford logical operations. We discuss a systematic approach to derive circuits realizing logical protocols on arbitrary lattices (or spacetime cellulations) that can be used to obtain various microscopic implementations of the same protocol. The approach naturally gives rise to specific circuits with a given ordering of unitary and measurement operations. 
The circuits we describe in this section can be brought into a form where they consist of repeated measurements of the Clifford and Pauli stabilizers (for example, such as the circuit considered in Section~\ref{sec:CZ_minimal}) by applying local circuit equivalences and grouping operations.
However, the same method can also be used to derive Floquet-code or MBQC-like circuits as opposed to ones that repeatedly measure the same stabilizers. For example, one can obtain a circuit for the TQD phase that looks like three copies of the honeycomb Floquet code coupled with additional $CCZ$ gates, see Ref.~\cite{Bauer2024}.

We first explain how to derive circuits from the path integral formalism using a specific example, which we refer to as the ``cup product'' implementation. We then briefly discuss other examples (for example, making the connection to the model in Sec.~\ref{sec:CZ_minimal}) and alternative ways to derive and understand these circuits. In addition, Appendix~\ref{sec:more_gates} shows more microscopic examples which are discussed in a language similar to Sec.~\ref{sec:CZ_minimal}.

\subsection{``Cup product'' microscopic lattice realization of the TQD} \label{subsec:cupproduct-tqd}

We first explain how to derive the circuit from the path integral description for the simplest protocol that has trivial logical action, namely the one that simply realizes the TQD phase without boundaries or domain walls (i.e. on a torus).  After that, we explain how to derive the circuits for nontrivial logical protocols corresponding to specific spacetime geometries.

\subsubsection{Preliminaries: circuits for the toric code phase}
 \label{app:pathint}

We start by briefly reviewing how to obtain the circuit for the toric code from the path integral framework (see Appendix~\ref{app:pathint} as well as Refs.~\cite{Bauer2023,Bauer2024a} for a more comprehensive discussion). 
Recall that the path integral for the toric code can be written as
\begin{align}
Z&=\sum_{ \text{confs. } \vec{c}}\ \ \prod_{\text{faces $f$}} \omega_f(\vec c)
\nonumber \\
&= \sum_{  \vec{c}} \prod_{f} \delta\Big({\sum_{\partial f} \vec{c} = 0 \operatorname{mod} 2}\Big)\, .
\label{eq:toric_code_path_integral}
\end{align}
That is, the weights of the path integral 
$w_f(\vec {c} ) = \delta\Big({\sum_{\partial f} \vec{c} = 0 \operatorname{mod} 2}\Big)$, 
enforce the constraint that the variables on the edges $\partial f$ around the boundary of each plaquette $f$ must have even total parity. 
Here, we describe the circuit on a square lattice obtained from the path integral on a cubic lattice.  The approach, however, can be used for arbitrary spacetime cellulations and for other topological phases~\cite{Bauer2024}.

We first turn a path integral with input and output state boundaries into a circuit of unitaries and projector operators (i.e. $+1$-postselected measurements).  Define the time direction to be along one of the axes of the cubic lattice. We call the edges that are parallel to the time direction \emph{timelike} and the other edges \emph{spacelike}. 
Then, applying the projectors $\frac12(1+XXXX)$ for every vertex $v$ and $\frac12(1+ZZZZ)$ for every horizontal face $f$ at each timestep realizes the same action between input and output states as the expression for the path integral, upon appropriate identification between $\zz_2$ variables and qubits (shown in Appendix~\ref{appendix:pathint}). 

Next, we replace the projectors with measurements. The $-1$ measurement outcome of an $XXXX$-operator at vertex $v$ at time $t$ corresponds to a factor of $(-1)^{\text{variable on a $t$-edge}}$ in the path integral, which is the same as inserting the charge worldline at the timelike edge ($t$-edge). The $-1$ measurement outcome of the $ZZZZ$-plaquette operator replaces the parity-even constraint with the parity-odd constraint on that plaquette in the path integral. This means that there is a flux defect worldline going through this plaquette in the path integral. 
In practice, for every vertex (plaquette) measurement we add an ancilla qubit at the associated vertex (plaquette).
At each time step, this ancilla is initialized in the $\ket+$ ($\ket0$) state, then coupled to the qubits participating in the measurements via CNOT gates, and finally measured in the Pauli-$X$ (Pauli-$Z$) basis. The $-1$ outcomes then record the locations of the worldlines of respective defects in spacetime.

\begin{figure}[t]
\includegraphics[width = 0.9\columnwidth]{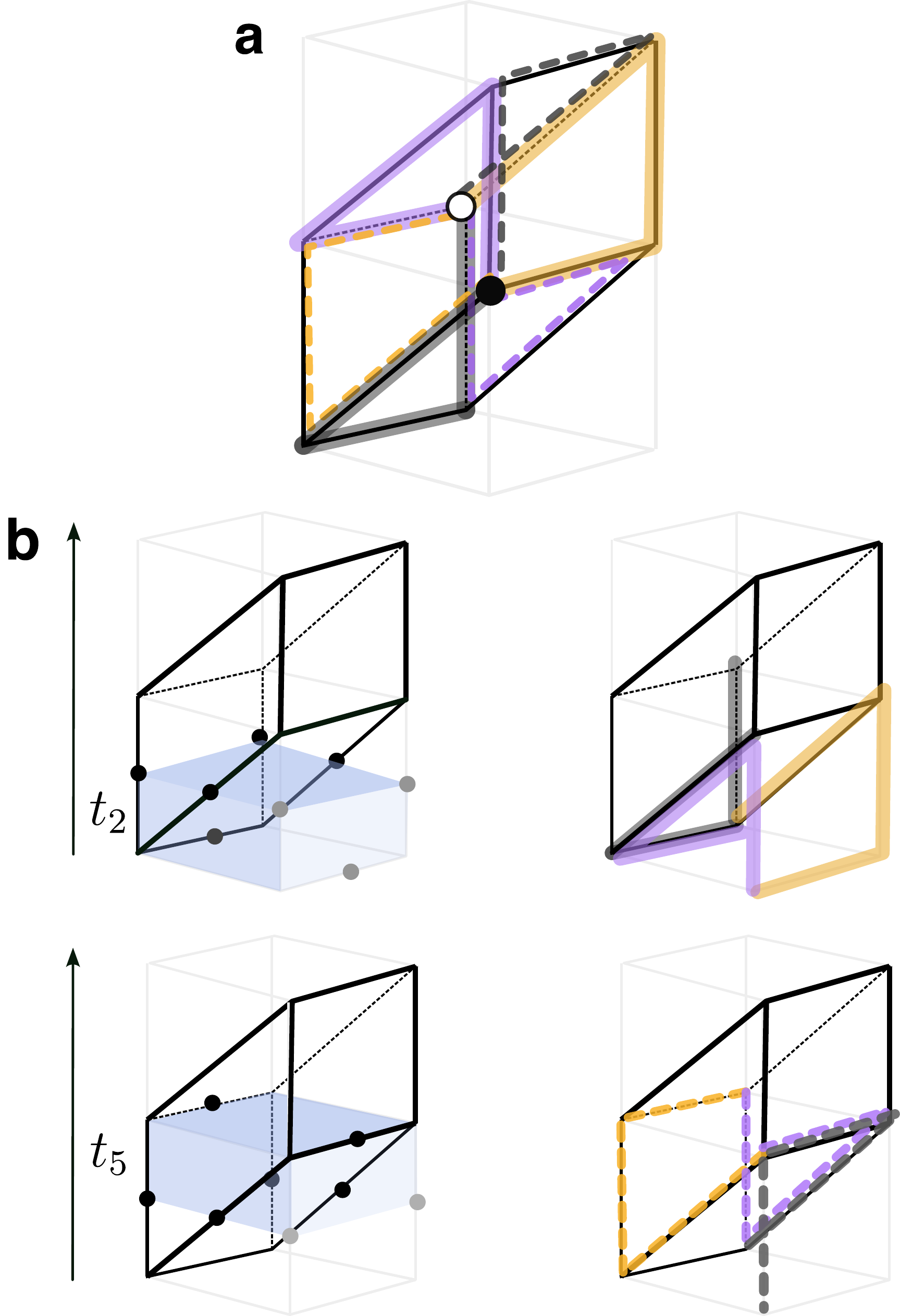}
\caption{ (a) Figure used to define the cup product on a slanted cube in spacetime that appears in the path integral of the TQD phase. The time arrow points upwards. (b) Two timesteps labeled $t_2$ and $t_5$
of the circuit during which we implement the cup-product weights of the TQD path integral. The spacetime regions on the left that are shaded in blue denote the degrees of freedom that correspond to the qubits that are ``live'' in the circuit at the corresponding timestep. On the right, the paths that are supported only on the ``live'' degrees of freedom at a given time step are shown. Some of the paths belong to neighboring slanted cubes. 
\label{fig:derivation-cupproduct}}
\end{figure}

Finally, we can also derive the circuit for the slanted version of the cubic cellulation (shown in Fig.~\ref{fig:derivation-cupproduct}). This circuit is obtained by splitting the timestep $t$ into an appropriate number of substeps. At each substep, the gates in the circuit described above can only implement those weights (here, the parity constraints) of the path integral that involve variables that are currently ``active'', that is, represented by qubits.
Finally, the ancilla qubits are measured, which completes the full period.  The circuit is periodic in time, and each period is subdivided into 5 steps. The step $t_0$ corresponds to single-qubit $m_{X}$ ($m_Z$) ancilla measurements followed by $\ket+$ ($\ket0$) initialization. The circuit is equivalent to the one shown in Fig.~\ref{fig:lattice-cupproduct} if one skips steps $t_2$ and $t_5$, assuming that there is one data qubit per edge and an ancilla qubit per vertex and per plaquette shown as white dots.

\subsubsection{Circuits implementing the TQD phase}
\label{square-TQD-sec}
A very natural microscopic representation of the type-III TQD path integral can be obtained from the action of the according Dijkgraaf-Witten model~\cite{Dijkgraaf_1990} in terms of \emph{cup products}. Informally speaking, the cup product expression precisely counts the number (modulo 2) of triple intersections of membranes of three different types by evaluating this number in each 3-cell and summing over the entire spacetime.

This microscopic representation is defined on a 3-cellulation of spacetime with three $\mathbb Z_2=\{0,1\}$ variables at each edge, which can host three copies of the toric code path integral in Eq.~\eqref{eq:toric_code_path_integral}, which we label in red, green, and blue. 
For the TQD phase, we add a non-trivial weight (or ``twist'') to the path integral using the cup product.
Let $a_r$, $a_g$, and $a_b$ denote the configurations or red, green, and blue toric code variables. For a given configuration, the expression 
\begin{equation}
\label{eq:tqd_pathintegral_weight}
    \prod_{\text{3-cells $v$}} (-1)^{(a_r\cup a_g\cup a_b)(v)}
\end{equation}
counts the total parity of triple intersections between all three colors (see Refs.~\cite{Dijkgraaf_1990,Chen2023higher}, and Appendix~A of Ref.~\cite{Bauer2024}). Then the path integral can be written as: 
\begin{equation} \label{eq:path-integral-cupproduct}
Z = \sum_{\substack{\text{closed membranes} \\ a_r, a_g,  a_b  }} \prod_{\text{$v$}} (-1)^{(a_r\cup a_g\cup a_b)(v)}\, .
\end{equation}
Similar to the toric code above, we again consider a skewed cubic spacetime lattice as shown in Fig.~\ref{fig:derivation-cupproduct}. Note that the skew is not strictly needed, but it makes the resulting circuit more uniform and better parallelized. Then, for a given cube with coordinate $v$, we can evaluate the expression $(a_r\cup a_g\cup a_b)(v)$ following Fig.~\ref{fig:derivation-cupproduct}~(a):
In each cube, there are 6 different paths of edges from the vertex marked in black to the vertex marked in white, shown as black, purple, and yellow, dashed or solid. For each path $p$, we number the edges $(0_p,1_p,2_p)$. Then, we have
\begin{equation} \label{eq:cupproduct-expression}
    (a_r\cup a_g\cup a_b)(v) = \sum_p a_r(0_p) a_g (1_p) a_b(2_p).
\end{equation}
Intuitively, formulas for the cup product can be obtained by shifting three sublattices slightly into a generic position with respect to another, such that the intersection between the membranes becomes unambiguous.
Note that in the example in Section~\ref{sec:CZ_minimal}, the three toric codes are already defined on three separate lattices, such that the intersections are unambiguous.
The vertices marked black and white in Fig.~\ref{fig:derivation-cupproduct} correspond to a choice of a \emph{branching structure} on the cube, that is, a choice of local coordinate system that is needed to define the shift direction.
For further details about microscopic formulas for cup products, we refer the reader to Appendix~A of Ref.~\cite{Bauer2024}, as well as Refs.~\cite{Steenrod1947, Chen2023higher}.

\begin{figure}[t]
\includegraphics[width = \columnwidth]{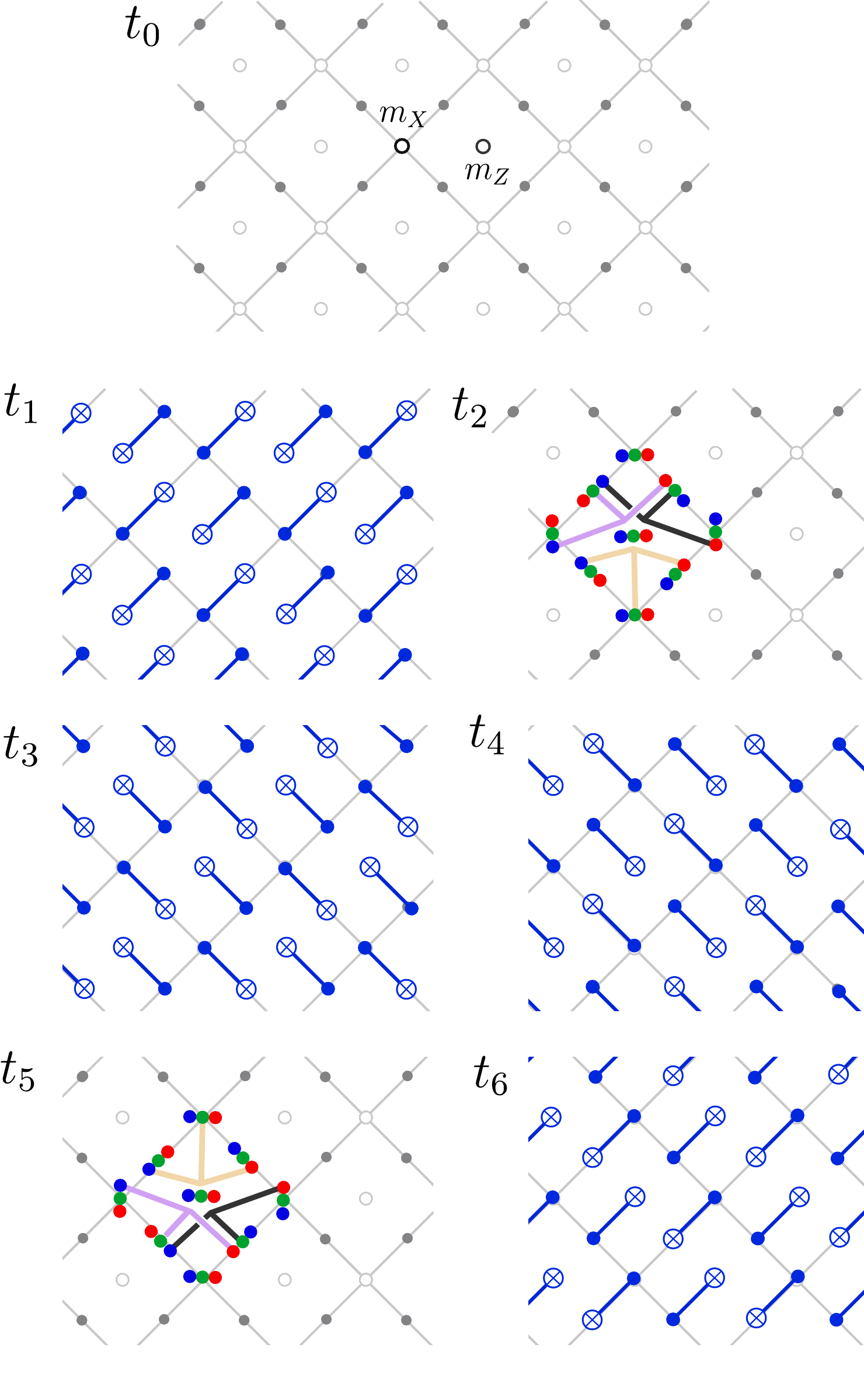}
\caption{ 7-step circuit that implements TQD phase on a torus while detecting the presence of charge and flux excitations. At time $t_0$ all ancilla qubit, shown as hollow circles, are measured in either single-qubit $X$ (labeled as $m_X$) or $Z$ basis (labeled as $m_Z$), and then reset. 
Yellow, purple and black-colored hyperedges at times $t_2$ and $t_5$ show the $CCZ$ gates; each hyperedge ends on a triple of qubits of three different colors on which the physical $CCZ$ gates act. The colors are chosen in correspondence with the paths in Fig.~\ref{fig:derivation-cupproduct}.  This action is repeated on each square plaquette in a translationally-invariant way. CNOT gates are shown at times $t_1, t_3, t_4$ and $t_6$ in blue and are depicted in the standard way. \label{fig:lattice-cupproduct}}
\end{figure}

Finally, we briefly explain the derivation of the circuit for the TQD phase (see Ref.~\cite{Bauer2024} for a more in-depth discussion). The circuit is shown in Fig.~\ref{fig:lattice-cupproduct} and is obtained by considering sections of the spacetime lattice at each timestep, as illustrated in Fig.~\ref{fig:derivation-cupproduct}~(b). The correspondence between the data and ancilla qubits and the path integral variables is analogous to that for the toric code example in the previous subsection. There are three data qubits per edge (labeled by red, green and blue colors) shown as a single dark dot at $t_0$, and three ancilla qubits per vertex and plaquette shown as white dots. If we exclude the steps $t_2$ and $t_5$, the resulting circuit implements precisely three decoupled toric code copies on the square lattice, where the step $t_0$ corresponds to single-qubit $m_X$ or $m_Z$ ancilla measurements. The additional steps $t_2$ and $t_5$ implement the nontrivial phases in the path integral of the twisted phase~\eqref{eq:path-integral-cupproduct}. 

Fig.~\ref{fig:derivation-cupproduct}~(b) illustrates how steps $t_2$ and $t_5$ realize Eq.~\eqref{eq:path-integral-cupproduct}. The left-hand side shows the time step and the 
``active'' path-integral variables that are represented by data and ancilla qubits at this point in time.
The right-hand side shows the paths entering the cup product in Eq.~\eqref{eq:cupproduct-expression} that are supported only on these
active variables.
Note that the paths shown do not all correspond to the marked cube, but some also originate from the cube below.
Each of these paths $p$ corresponds to a multiplicative term of the form $(-1)^{a_r(0_p)a_g(1_p)a_b(2_p) }$ in the path integral. The presence or absence of a membrane at a given edge corresponds to the qubit being in $\ket{0}$ or $\ket{1}$ state, respectively. Therefore, the phase $(-1)^{a_r(0_p)a_g(1_p)a_b(2_p) }$ can be implemented by a $CCZ_{0_p,1_p,2_p}$ gate acting between red, green and blue qubits at specified locations. Finally, each such $CCZ$ gate for a single square plaquette is shown in Fig.~\ref{fig:lattice-cupproduct} at corresponding timesteps. Each gate is colored in correspondence to the path it was obtained from (i.e. as the path in Fig.~\ref{fig:derivation-cupproduct}(b) at the corresponding time). The qubits that each of these gates couples are shown explicitly for clarity. Each square plaquette contains 6 paths for each period of circuit operation, three of which are realized at time step $t_2$ and the other three at time step $t_5$.

\subsection{Full circuits for logical protocols}

Having discussed the derivation of the circuit  realizing the TQD phase, we are now ready to discuss the derivation of circuits for the protocols that implement nontrivial logical operations. These are the circuits corresponding to the global topologies presented in Sec.~\ref{sec:global_topologies}. For this, we need to discuss how to implement the
boundaries of the TQD phase and its domain walls with the toric code.
These elements are then combined to realize a specific spacetime protocol. We explain this by considering an example of the $\overline{CZ}$ measurement protocol, which is an open boundary condition variant of the example in Sec.~\ref{sec:CZ_minimal}. 

We first discuss boundaries. For each copy of the toric code phase, there is only the rough ($\langle \rangle$) and the smooth ($\langle p \rangle$ or $\langle y \rangle$) boundary. 
There is also a``fold'' boundary that joins two copies of the toric code ($\langle py \rangle$). 
A rough boundary is realized by defining a continuous boundary line on the primal lattice and removing the qubits on the ``vacuum'' side, including the boundary itself. 
The circuit with a boundary is obtained by removing the measurement and unitary operations that formerly involved the removed qubits. A smooth boundary is realized similarly, except that the boundary is now defined on the dual lattice. For the $\langle py \rangle$ boundary, we start with a pair of smooth boundaries instead and add a $ZZ$ measurement that `forces' termination of a matching type of membrane on each boundary site whenever a membrane terminates there.  

For the TQD phase, the only  boundaries that appear in the protocols are of $\langle r,g \rangle$  and $\langle rg , rb\rangle$ type (or the analogous boundaries under color permutation). The $\langle r,g \rangle$  boundary is a smooth boundary on the red and green sublattices and rough boundary on the blue sublattice. The $\langle rg , rb\rangle$ boundary is realized by first considering a smooth boundary of each color and then adding a $ZZZ$ measurement on each boundary site that ``glues'' the termination lines of any pair of membranes.
In addition, the path integral includes a complex-phase weight for crossings of termination lines as shown in Eq.~\eqref{eq:eighthroot} of Appendix~\ref{app:boundaries-DWs-corners}.
In the microscopic model, this additional weight can be implemented via a cup product as
\begin{equation} \label{eq:boundary-path-integral}
\prod_{\text{$f$}} i^{( a_r\cup a_b)(f)}\, .
\end{equation}
Here, $f$ runs over all boundary faces, and $a_g$ and $a_b$ are the configurations of green and blue path-integral variables assigned to the boundary edges. The microscopic expression for this cup product
is similar to the one for the bulk in Fig.~\ref{fig:derivation-cupproduct}~(a):
we again consider the two different paths $p_1$, $p_2$ of edges inside the boundary face that go from one vertex to an opposite vertex, as shown in yellow and purple in the following picture for the boundary face shaded gray:
\begin{equation}
\includegraphics[width = 0.4\columnwidth]{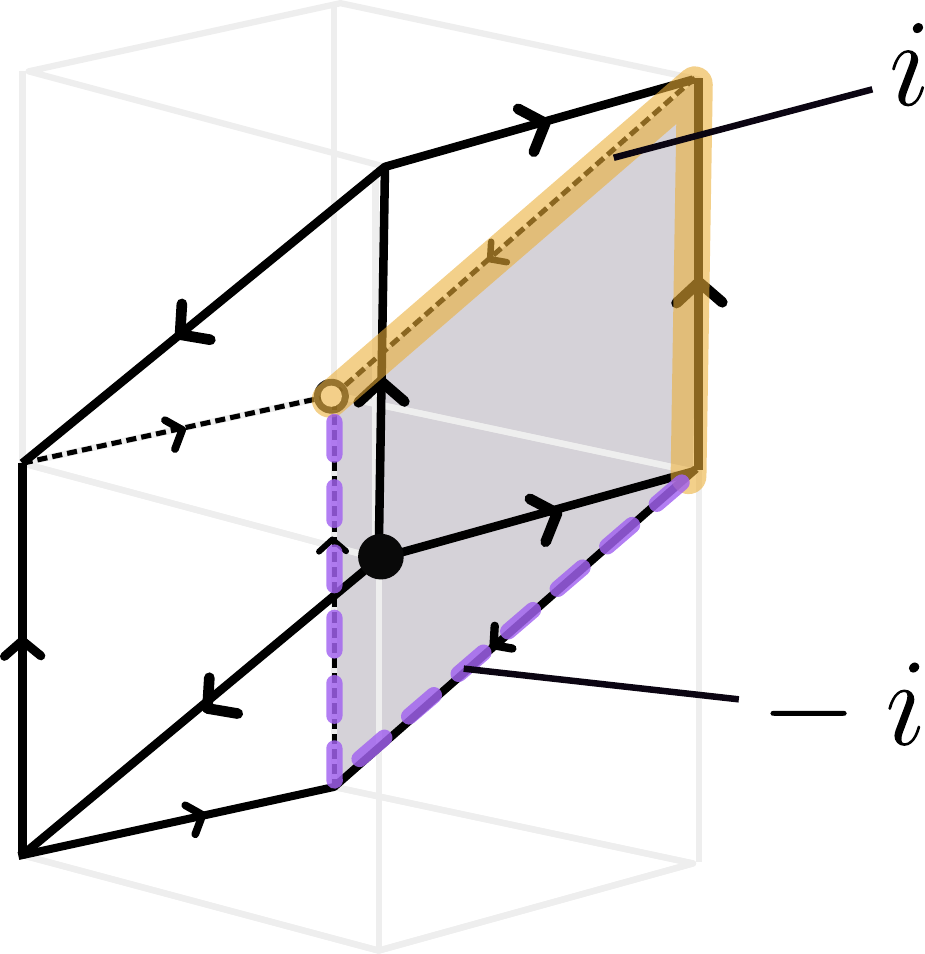}
\end{equation}
The cup product on $f$ is then given by
$$(a_r\cup a_b)(f) = a_r(0_{p1}) a_b(1_{p1}) - a_r(0_{p2}) a_b(1_{p2}),$$ 
that is, the purple path $p_2$ gets an extra factor of $-1$.\footnote{
Note that $-1$ signs are also present in the cup product formula on the cube in Eq.~\eqref{eq:cupproduct-expression}, but they do not matter there since $(-1)^{-a} = (-1)^a$. 
} 
This, the additional path integral weights are of the form $i^{a_r(0_{p1}) a_b(1_{p1})}$, and can be implemented in the circuit by adding physical $CS$ and $CS^\dagger$ gates at steps $t_2$ and $t_5$, acting on the green and blue qubits which represent the corresponding spacetime edges.

Next, we discuss the domain walls between the toric code copies and TQD phases that appear in the protocols. 
These are the $\langle rp, by \rangle$ domain wall and its color permutations.
Two different orientations of this domain wall have appeared in the global topologies above.  
One is timelike, meaning that it occurs at a fixed instant of time which corresponds to globally switching from one topological phase to the other. The second orientation is tilted, which corresponds to a domain wall that moves in space as time progresses. This happens, for example, in the $\overline{CCZ}$ and $\overline{T}$-gate protocols. 
We first address the timelike $\langle rp, by \rangle$ domain wall between a pair of toric codes and the TQD phase. 
In spacetime, this corresponds to the purple and yellow lattices of the toric codes becoming red and blue, which is simply a relabeling of colors. 
For this reason, in many examples we use the color assignment from the TQD phase to label the toric codes as well. 
For the third (green) color, the domain wall is the same as the rough boundary. To realize it in the circuit, we simply prepare the green data and plaquette ancilla qubits in the $\ket{0}$ state and the green vertex qubits in the $\ket {+}$ state at the end of the circuit realizing the toric code phases and continue to run the circuit for the TQD phase. Similarly, to realize such a domain wall in reverse, between the TQD phase and two copies of the toric code phase, we simply measure the green data qubits in the $Z$ basis and perform the usual measurements of the green ancilla qubits. 
Next, we discuss the realization of the tilted domain wall. To move such a domain wall from the toric codes phase to the TQD phase, we initialize a new set of green data qubits of the third color in state $\ket {0}$ in a thin slice, and implement the toric codes circuit on one side of the new domain wall and the TQD circuit on the other side. Similarly, to move such a domain wall from the TQD to the side of the toric codes phase side, we instead measure out green data qubits in $Z$ basis. 

We now discuss the corners. The majority of corners appearing in our protocols are not associated with adding any new weights in the path integral. Microscopically, such corners can be obtained by naively letting two types of microscopic boundary conditions meet at the corner.
An exception is the ``$\omega$'' boundary between the $\langle rg,rb\rangle$ and the $\langle g,b\rangle$ boundaries, which requires the addition of $\omega= e^{-i \pi/4}$ phase weights.
These $\omega$ weights are implemented in the circuit by additional $T$ gates. Apart from these weights, a straightforward blending between the two boundary conditions produces the correct circuit for this type of corner.

\begin{figure}[t]
\includegraphics[width = \columnwidth]{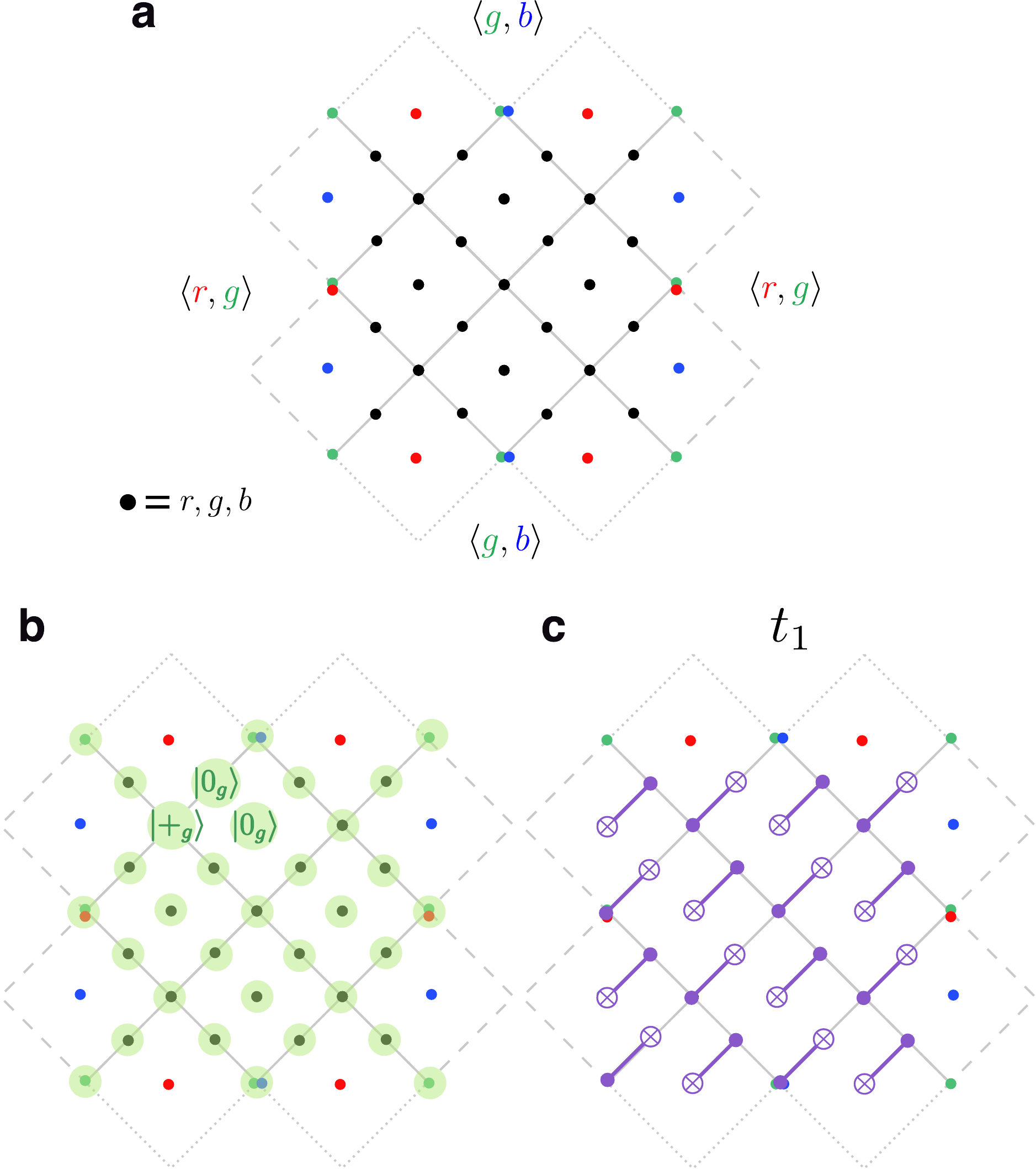}
\caption{ (a) Qubit layout for the logical $\overline{CZ}$ measurement on a pair of surface codes of distance 4, which matches the transient TQD phase and its boundary conditions shown in Fig.~\ref{eq:czmeasure_global_protocol}.
(b) The locations of green qubits that are initialized or measured out to realize the domain wall between the TQD and the toric codes phases. (c) An example demonstrating how the circuit shown in Fig.~\ref{fig:lattice-cupproduct} without boundaries is altered near the boundaries.  \label{fig:protocol-cupproduct}}
\end{figure}

This completes our summary of how to derive circuit implementations for all types of spacetime defects that appear in the global geometries presented in Sec.~\ref{sec:global_topologies}. When appropriately combined together, this produces circuits implementing the associated logical operations by design. The measurement outcomes of vertex ancilla qubits are stored while those of plaquette ancillas are used for the operation of the just-in-time decoder, discussed in Sec.~\ref{Sec:JustInTime}. 

Finally, we present an example of the circuit for the logical $\overline{CZ}$ measurement introduced in Subsec.~\ref{sec:referenced-example}, where the bulk of the TQD phase is implemented as in Fig.~\ref{fig:protocol-cupproduct}. The entire lattice is shown in Fig.~\ref{fig:protocol-cupproduct}~(a), where every black dot represents a triple of a $r$, $g$ and a $b$ qubit, and every colored dot represents an individual qubit of that color. First, the syndrome extraction circuit is run for the $r$ and $b$ copies of the toric codes phase (we have relabeled $p$ to $r$ and $y$ to $b$) with open boundary conditions, which implement rotated surface code patches with distance 4. This circuit is analogous to the one in Fig.~\ref{fig:lattice-cupproduct} with steps $t_2$ and $t_5$ omitted, except that the application of CNOT gates is appropriately truncated at the boundary as shown for one timestep in Fig.~\ref{fig:protocol-cupproduct}~(c). Then, the domain wall from two copies of the toric code to the TQD is realized by initializing the green qubits as shown in Fig.~\ref{fig:protocol-cupproduct}~(b) (all locations involved in the initialization are highlighted in green). Next, we run the circuit realizing the TQD phase with boundary conditions shown in panel (a) according to the appropriately truncated circuit from Fig.~\ref{fig:lattice-cupproduct} for a number of rounds. Finally, we switch from the TQD phase to the pair of surface codes by measuring out the green vertex qubits in the $X$ basis and the green data and plaquette qubits in the $Z$ basis, similar to Fig.~\ref{fig:protocol-cupproduct}~(b). Overall, this performs a logical $\overline{CZ}$ measurement between the logical qubits of the red and blue copy of the surface code.

\subsection{Other microscopic circuit examples}

\subsubsection{Transversal counting of triple intersections}

Defining the three toric code path integrals on different 3-cellulations, it is possible to obtain a circuit where $CCZ$ operations act on triples of qubits at the same location. For example,  three such superimposed 3-cellulations are described in Ref.~\cite{Vasmer2019three}. In this case, while the three cellulations are different, the centers of the edges of each cellulation coincide; in the end, we have three qubits per site that are labeled three different colors.  
In this case, there is one triple intersection weight at every edge center $c$ which takes the simple form
\begin{equation}
(-1)^{a_r(c) a_g(c) a_b(c)}\, .
\end{equation}
Using the approach described in this section, one can turn this path integral into a circuit results in one where the $CCZ$ gates act transversally. The associated stabilizer group appears in Appendix~\ref{sec:3d-2d-toric-codes}. Additionally, the protocol in Ref.~\cite{Brown2020universal} where the original $\overline{CCZ}$ protocol was introduced (where the TQD phase is featured only implicitly) can be viewed as a circuit realizing this path integral. 

\subsubsection{Microscopic lattice from Sec.~\ref{sec:CZ_minimal} and its variations} 

It is possible to also relate the protocol presented in Sec.~\ref{sec:CZ_minimal} to the TQD path integral.\footnote{We remark that a more straightforward to obtain the lattice models shown in Fig.~\ref{fig:ToricStabilizers} and Fig.~\ref{fig:lattice-examples-1} is by gauging an SPT with the type-III 3-cocycle~\cite{Yoshida2016topological} or by directly gauging a global $\mathbb Z_2$ symmetry corresponding to a logical $CZ$ gate of a pair of toric codes, which we describe in Appendix~\ref{sec:gaugingungauging}.} 
Namely, the circuit of alternating $+1$-post-selected plaquette and vertex measurements yields a TQD path integral on three superimposed spacetime lattices.
The plaquette stabilizers implement the parity constraints at the faces of the spacetime lattices, whereas the vertex terms implement the triple-intersection weights at the volumes.
These volumes are diamond-shaped, with a ``bottom'' and ``top'' half that both consist of the triangles surrounding the according vertex in the spatial triangulation.
If we replace projectors with measurements, these measurements in addition detect charge and flux defects. Re-interpreting this model through the lens of the path integral framework is useful, because the path integral approach makes the fault tolerance of the protocol much more apparent  and gives systematic ways to construct full logical circuits, which requires deriving boundaries, domain walls, and corners appearing in them. This can be used to derive circuits for all the logical protocols in Sec.~\ref{section:loopsum} where the TQD phase is implemented via the model from Sec.~\ref{sec:CZ_minimal}.

We now turn to a similar, but simpler, version of the Clifford stabilizer model from Sec.~\ref{sec:CZ_minimal}.  In Fig.~\ref{fig:lattice-examples-1}~(a) we show its stabilizers on the same lattice; the only difference to the lattice in Fig.~\ref{fig:ToricStabilizers} is that the Clifford circuit attached to the vertex stabilizers now consists of two $CZ$ gates only. In Fig.~\ref{fig:lattice-examples-1}~(b) we show why the new vertex stabilizer counts the same loop statistics. For this, we make use of the dual lattice, which is depicted along with the original vertex stabilizer in Fig.~\ref{fig:dual_stab}. 
By slightly shifting the vertices of the red and green dual lattice with respect to one another, we see that there are precisely two locations with possible red-green intersections inside a given blue loop. For each location, a $CZ$ gate detects the presence of such an intersection.
This model has the same properties as the one in Sec.~\ref{sec:CZ_minimal}. However, the stabilizers have weights 3 and 8 as opposed to 3 and 12, respectively, which could be beneficial for practical implementations. 

\begin{figure}[t]
\includegraphics[width = 0.9\columnwidth]{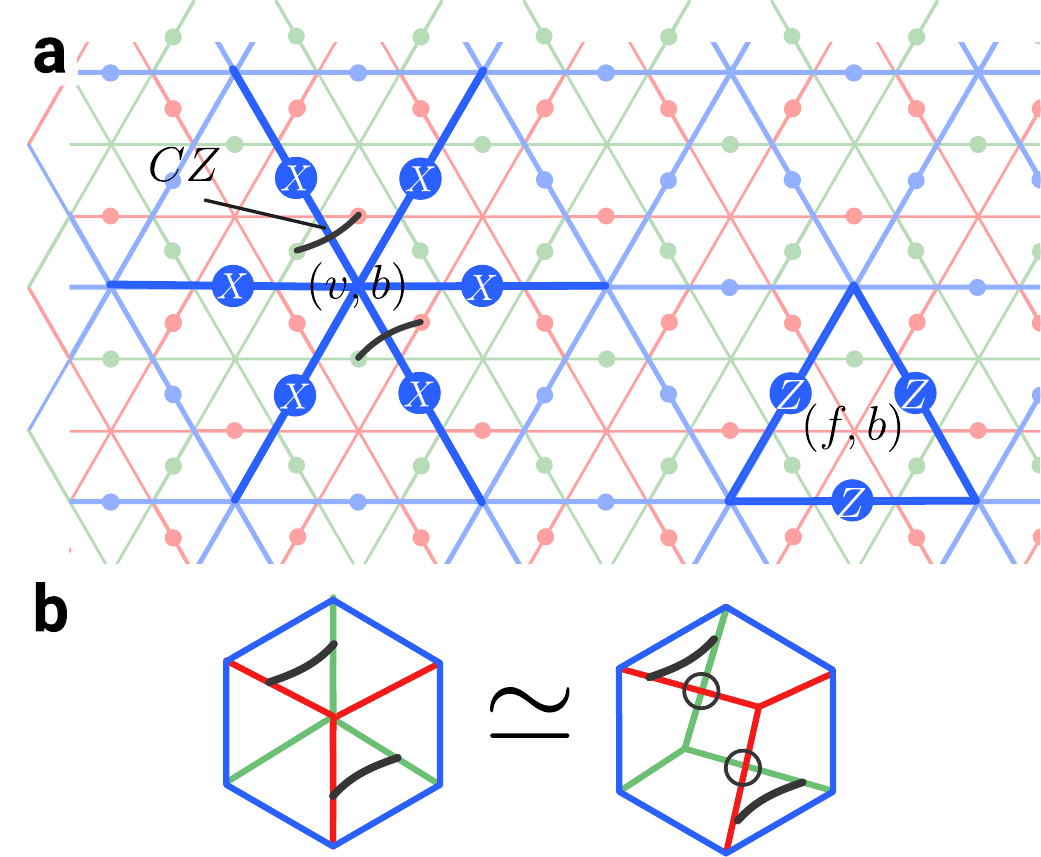}
\caption{ (a) A version of Clifford vertex (left) and Pauli plaquette (right) stabilizers on the blue sublattice of the twisted quantum double on three triangular lattices. 
The stabilizers on the red and green sublattices are defined analogously.  (b) Detection of red-blue loop crossings inside a blue loop on a dual lattice and $CZ$ operators in the associated vertex stabilizer. The possible intersections of red and green loops are resolved on the right and are circled.   \label{fig:lattice-examples-1}
}
\end{figure}

\section{The non-Abelian phase and 3D stabilizer codes}
\label{sec:3d-2d}

In this section, we discuss the connection between the appearance of the two-dimensional non-Abelian TQD model in our protocols and 3D stabilizer codes with finite-depth non-Clifford gates~\cite{Bombin2007topological, Bombin_2015, bombin2018transversal, Vasmer2019three, Brown2024}. 3D codes that admit a fault-tolerant non-Clifford gate can give rise to classes of 2+1D protocols for logical non-Clifford gates~\cite{bombin2018,Brown2020universal}.
In this section, we exemplify this procedure and relation to the TQD phase starting from the 3D color code.\footnote{In fact, the resulting \emph{color path integrals} are also described by a cohomology theory that is equivalent to, but not the same, as cellular (co)homology~\cite{bauer2025planar}.}
The resulting protocols have the appeal that they are based on physical $T$ rotations and local Pauli measurements may be well suited for many hardware architectures.
Specifically, we show how a state in the TQD phase can be prepared starting from a thin slab of 3D code, highlighting the connection between 3D codes and 2D non-Abelian phases.
In Appendix~\ref{sec:3d-2d-toric-codes}, we present an analogous derivation based on three copies of the three-dimensional toric code instead, which results in a different lattice realization of the same TQD model.

\subsection{2+1D circuits from 3+0D measurement-based protocols}
Refs.~\cite{bombin2018,Brown2020universal} introduced two-dimensional protocols that fault-tolerantly realize non-Clifford gates via code switching between a thin 3D code, that admits a transversal non-Clifford gate, and a 2D code. 
To ensure the resulting protocols are fault tolerant, the \textit{dimensional jumping} procedure~\cite{Raussendorf2005, bombin2016dimensionaljumpquantumerror} that transforms between the two codes must be performed repeatedly over $O(d)$ time steps, where $d$ is the spatial code distance of the 2D code. In this section, and in Appendix~\ref{sec:3d-2d-toric-codes}, we show that the non-Abelian TQD model is in fact realized as an intermediate state during both proposals. This points to an equivalence between these examples and the framework we have developed to perform non-Clifford gates using the TQD as an intermediate phase.

Topological stabilizer codes in three dimensions can be prepared in constant time using single-shot error correction~\cite{Bombin2015single, bombin2016dimensionaljumpquantumerror}.
This property can be used to perform single-shot magic-state preparation on the two-dimensional boundaries of the three-dimensional code. 
For example, we summarize below the procedure to prepare a $CCZ$ magic state in constant time using three copies of the three-dimensional toric code:
\begin{enumerate}
    \item Prepare the $\ket{+++}$ logical state of three copies of 3D toric codes in a cube with open boundary conditions such that they admit a transversal $\overline{CCZ}$ gate~\cite{Vasmer2019three}. For this, first initialize the qubits on the lattice that hosts three copies of the 3D toric code in the $\ket{+}^{n}$ state.
    \item   Measure the $Z$ (face) stabilizers of the 3D toric code copies and correct for errors. 
    This is done fault tolerantly with only a single round of measurements using single-shot error correction~\cite{Bombin2015single}. This results in a logical $\ket{+++}$ state.
    \item Apply the constant-depth transversal $CCZ$ gate between the three 3D toric code copies. This prepares a logical magic state $CCZ\ket{+++}$ of the total 3D code.
    \item For each of the copies of the 3D toric code, measure all qubits in the $X$ basis, except those on designated boundaries such that a planar surface code remains. We can use these measurements to reliably correct $Z$ errors in the resulting two-dimensional code such that the residual error is small and local~\cite{bombin2016dimensionaljumpquantumerror}.
\end{enumerate}

The constant depth execution of the protocol means that the example is readily expressed as a post-selected three-dimensional tensor network.
Furthermore, we can interpret the evaluation of this tensor network as summing over all ``qubit configurations'' on internal edges that fulfill the constraints imposed by the local circuit elements and adding weights, depending on the local configuration. It can therefore be interpreted as a \emph{path integral}, in the sense of Eq.~\eqref{eq:pathintegral_schematic}.
The fact that we obtain a scalable family of such networks with macroscopic fault distance renders this tensor network \emph{topological}.\footnote{More specifically, any error that is only supported on a 3D ball-like region in spacetime does not affect the logical action of the post-selected network~\cite{Bauer2024}.}
Using local equivalences rules of the tensors involved, we can interpret the tensor network either as a 3+0D protocol or a 2+1D protocol.
This can be understood as trading one spatial dimension in the 3D code involved in the above protocol for a time direction.
Moreover, the tensor network can be seen as representing a 2+1D topological phase that is equivalent to the phase of the non-Abelian TQD.

In the remainder of this section, we explore the equivalence between 3D codes and the non-Abelian TQD further, invoking a 2D perspective.
We show how a thin 3D code that admits a transversal non-Clifford gate in the bulk can be used to prepare a non-Abelian state on its 2D boundary.
These types of states appear as intermediate states in the protocols from Refs.~\cite{bombin2018, Brown2020universal}, hence explaining that these protocols, in fact, involve a non-Abelian phase.

\subsection{Overview:  the non-Abelian TQD from a 3D Abelian code}

Here, we discuss the relationship between the non-Abelian TQD model in two dimensions and the 3D color code. 
The 3D color code, like three copies of the 3D toric code,\footnote{The 3D color code is equivalent to three copies of the 3D toric code under local unitary operations and addition of disentangled auxiliary qubits~\cite{Kubica2015unfolding}.} can be viewed as realizing a $\zz_2^{\times 3} = \zz_2 \times \zz_2 \times \zz_2$ topological gauge theory.
The logical operators in these codes are related to the nontrivial winding of fluxes and charges of the underlying gauge theory. Symmetries among these topological defects can be implemented by locality-preserving unitaries and correspond to fault-tolerant logical gates of the associated codes. 
The $\zz_2^{\times 3}$ gauge theory in three spatial dimensions admits a global symmetry that implements a logical non-Clifford gate.
For the 3D color code, this symmetry is realized by a transversal $T^{(\dagger)}$ gate and in three copies of toric or surface code it is realized by a transversal $CCZ$ gate acting between the individual codes~\cite{Bombin2007topological,Vasmer2019three}.

To explain the relationship between a three-dimensional slab and a two-dimensional model, we consider a thin slab whose two-dimensional section is a torus, i.e. a manifold with $\mathbb{T}^2\times [0,1]$ topology. 
The slab is shown in Fig.\,\ref{fig:thin3D-TQFT}, and it has two boundaries which we label \textit{top} and \textit{bottom}.
We start with smooth boundary conditions on both boundaries. 
This results in a model that is equivalent to three copies of the 2D toric code. 
Applying the non-Clifford symmetry to the lower half of the state, including the bottom boundary, creates a two-dimensional domain wall defect between the two boundaries.
At the same time, it ``twists'' boundary conditions of the bottom boundary.
As a final step we enforce smooth boundary conditions on the bottom boundary by measuring out the bottom qubits in a suitable single-qubit basis.
This condenses bare fluxes below the domain wall defect.
This creates a model that, when regarded as a two-dimensional model, is equivalent to a non-Abelian TQD.

We now turn the above procedure into an explicit microscopic protocol that prepares a non-Abelian code associated with the TQD. The same procedure occurs in the bulk during the 2+1D protocol obtained from a 3+0D circuit, as we discussed above. The example presented in this section is based on the 3D color code. We present another example based on three copies of the 3D toric code in Appendix~\ref{sec:3d-2d-toric-codes}.

\begin{figure}[t]
    \centering
    \includegraphics[width=\columnwidth]{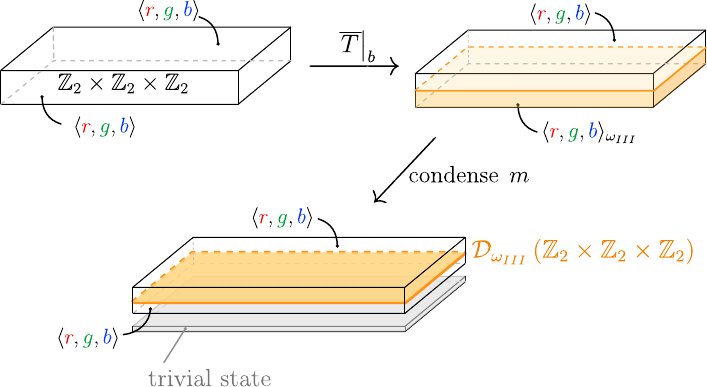}
    \caption{
    The topological quantum field theory picture for the initialization of a non-Abelian TQD from a 3+1-dimensional topological $\zz_2^{\times 3}$ gauge theory defined on a slab with appropriate boundary conditions.
    The outcome is shown at the bottom, displaying a thin slab with the defect configuration that leads to a quasi-2D state equivalent to the TQD model. The top and bottom boundaries are smooth, i.e. flux membranes of all three colors can condense on them (labeled $\langle r,g,b\rangle$), and there is a domain wall defect associated to a non-Clifford symmetry of the bulk indicated in yellow.
    In the limit where the thickness of the slab becomes infinitesimally small, this becomes equivalent to a two-dimensional state of the non-Abelian TQD.
    Proceeding from the top left to the bottom we show how to create this configuration by transversal gates and measurements:
    we start with a thin slab of a 3D code with two opposing smooth boundaries.
    Applying the non-Clifford symmetry $\overline{T}\big|_b$ of the model to the bottom half of the slab creates a domain wall and also ``twists'' the bottom boundary condition, which now becomes $\langle r,g,b\rangle_{\omega_{III}}$.
    Finally, we enforce smooth boundary condition again on the bottom boundary.
    In the stabilizer models that realize this topological theory, this can be achieved by local measurements and adaptive corrections.
    }
    \label{fig:thin3D-TQFT}
\end{figure}

The protocol is summarized in the following steps:
\begin{enumerate}
    \item Initialize a thin slab of 3D code that admits a transversal non-Clifford gate in the bulk, with appropriate boundary conditions, ``smooth'' on both top and bottom boundary.
    \item Apply the non-Clifford gate to the lower half of the slab including the bottom boundary.
    \item Measure out the qubits along the bottom boundary in a suitable single-qubit basis, enforcing smooth boundary conditions again.
\end{enumerate}
These steps are depicted in Fig.~\ref{fig:thin3D-TQFT}.

After these steps, a thin layer of qubits around the top boundary realizes a two-dimensional code that, as we show, is equivalent to the non-Abelian TQD considered in this paper. 

If step 2 (applying the non-Clifford gate in half of the slab) is skipped, we obtain an untwisted version of the 2D model which corresponds to three copies of the 2D toric code. We start by considering this example, as it helps to determine the topological phase of the twisted model below.

\subsection{Non-Abelian TQD preparation from the 3D color code}\label{sec:colorcode}

We define the three-dimensional color code~\cite{Bombin2007topological} with qubits on the vertices of a cubic lattice. 
The color code has two types of Pauli stabilizers associated to the four-colorable brickwork cells shown in Fig.~\ref{Fig:ColorCode}; volume stabilizers $S^X_c$ and face operators $S^Z_f$.
Volume stabilizers are the product of Pauli-$X$ operators acting on the vertices of cell $c$, i.e., $S^X_c = \prod_{v\in  c} X_v$. Faces $f$ interface pairs of adjacent cells, $f = c' \cap c $, and are associated with pairs of colors. We define  $S^Z_f = \prod_{v\in f} Z_v$\footnote{We abuse notation to use labels $c$ and $f$ to index cells and faces, respectively, and also to label the set of qubits associated to cell $c$ and face $f$.}.

\begin{figure}
\includegraphics{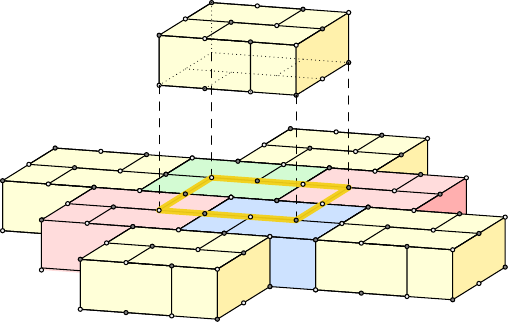}
\caption{Four-colorable 3-cells of the color code brickwork lattice organized with qubits on the vertices of a cubic lattice. \label{Fig:ColorCode}}
\end{figure}

The 3D color code admits a transversal non-Clifford gate of the form
\begin{equation}
\overline{T} = \prod_{\textrm{even } v}T_v \prod_{\textrm{odd } v}T_v^\dagger,
\label{Eqn:TransversalT}
\end{equation}
where $T_v = \exp(i \pi Z_v / 4)$ and $v = (x,y,z)$ is even (odd), if their canonical coordinates on the cubic lattice sum to an even (odd) number, i.e., $x+y+z $ is even (odd)~\cite{Bombin_2015, Kubica2015}.
As a consequence, the code also admits Clifford stabilizers obtained from conjugating the Pauli stabilizers with the non-Clifford logical. First, we define the operator
\begin{equation}
    W_v = \overline{T} X_v \overline{T}^\dagger.
\end{equation}
which is equal to  $e^{i \pi/4}X_v S^\dagger$ if $v$ is on the even sublattice and $e^{-i \pi/4} X_v S$ for the odd sublattice. Then we define
\begin{equation}
S^{W}_c = \overline{T} S^X_c \overline{T}^\dagger.
\end{equation}
They obey the group commutation relation
\begin{equation}
 [S^{W}_t, S^X_b] = S^Z_{f=t\cap b},
\end{equation}
showing that they commute on the codespace.
Outside of the ``flux-free'' subspace, defined by $S^Z_f = +1$, the (Clifford) stabilizers do not necessarily commute.

\subsubsection{Initializing an Abelian 2D state from 3D color code}
\label{sec:2d-abelian-colorcode}
We start with a 3D slab of thickness 1 with all-smooth boundary conditions (i.e.~flux membranes of all three colors can terminate on the top and bottom boundaries). A slab with such boundary conditions is supported on a single layer of bricks\footnote{This can be obtained by considering a 3D color code supported in infinite space and then measuring out the qubits everywhere, apart from one layer of bricks, in the single-qubit $X$ basis.} as shown in Fig.~\ref{Fig:ColorCode} and the qubits on the bottom and top of the layer of bricks are naturally split into the \textit{bottom layer} $b$ and the qubits on the \textit{top layer} $t$.

First, consider what happens if we ``collapse'' the state in the three-dimensional slab onto the top boundary without introducing the non-Clifford domain wall defect. For this, we measure out the bottom layer in the single-qubit Pauli-$X$ basis, which leaves us with a state supported on the top layer only.
The stabilizers for the two-dimensional state that we obtain this way are shown in Fig.~\ref{Fig:ThinCode}. 

The post-measurement state is in the phase of three copies of the 2D toric code.
To see this, we explicitly construct the basis of excitations for this model and verify their statistics. We assign colors to the three copies $r$, $g$ and $b$. The excitations of a single toric code are known as electric charges $e$ and magnetic fluxes $m$. We assign the $e_\mathbf{u}$ excitations a single color that corresponds to their layer $\mathbf{u} = r, g, b$. Here, an excitation of a given color lies at a violated brick stabilizer. In Fig.~\ref{Fig:ThinCode} we show an operator that creates a pair of $e_r$ charges. A yellow electric charge 
$ e_\mathbf{y} = e_r \times e_g \times e_b $
is the fusion product of the electric charges.

\begin{figure}[t]
\includegraphics[width = 0.9\columnwidth]{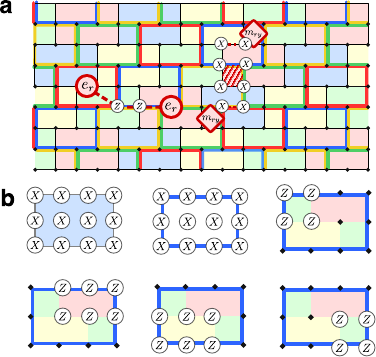}
\caption{
(a) The 4-colorable ($r,g,b$ and $y$) brickwork lattice obtained by ``collapsing'' a thin three-dimensional slab of 3D color code onto its top boundary. Fig.~\ref{Fig:ColorCode} shows two layers of 3-cells in 3D whose intersections with one of the boundaries give two types of bricks that we label ``bottom'' (shaded in solid colors) and ``top'' (shown with bold colored outlines) that come in 4 colors. Qubits are shown as black dots and form a square lattice.  The faces of this lattice are weight-6 and weight-4 rectangles that are colored in pairs of colors corresponding to the intersecting top and bottom bricks. The cell and face stabilizers are shown in panel (b). Additionally, panel (a) shows red charge and flux excitations and their hopping operators. \label{Fig:ThinCode}}
\end{figure}

Magnetic fluxes are assigned a color pair $m_{\mathbf{uv}}$. Electric charges $e_\mathbf{w}$ braid non-trivially with magnetic charges $m_\mathbf{uv} $ if either $\mathbf{w} = \mathbf{u}$ or $\mathbf{w} = \mathbf{v}$. We can associate the magnetic charge $\mathbf{uy}$ to the $\mathbf{u}$-colored toric code since $m_\mathbf{uy}$ braids non-trivially with the electric charge of the corresponding toric code $e_\mathbf{u}$ but trivially with electric charges $e_\mathbf{v}$ with $\mathbf{v}\not=\mathbf{u}, \mathbf{y}$ (recalling that $e_\mathbf{y}$ is the fusion product of the electric charges of the three toric code layers). Magnetic charges have a non-trivial fusion product 
$ m_\mathbf{uv} = m_\mathbf{uw} \times m_\mathbf{vw} $
where all colors $\mathbf{u}$, $\mathbf{v}$ and $\mathbf{w}$ are distinct.

\subsubsection{Initializing a non-Abelian state from 3D color code}

For this particular microscopic model, we are free to choose whether to apply the $\overline T$ symmetry to the top or bottom half in our non-Abelian phase preparation protocol before measuring out the bottom qubits.\footnote{The equivalence can be seen from the fact that the non-Clifford gate is diagonal in the computational basis and hence commutes with the multi-qubit $Z$ measurements.
This allows to commute the transversal gate through the measurement showing that within the full protocol, either bottom or top qubits can be measured out after the application of the non-Clifford gate.
This is in line with the perspective taken in Refs.~\cite{Brown2020universal, bombin2018} where the protocols are interpreted as ``measuring through'' a 3D state to which a transversal gate was applied.} 
We choose to apply it to the bottom half of the slab, denoting it $\eval{\overline{T}}_{b}$. It acts nontrivially on the bottom layer of qubits before we measure them out in the $X$ basis.
The resulting effect on the top qubits is the same as if we instead measured the bottom qubits directly in the basis defined by $W_v$.

Assuming that all measurement outcomes are $+1$ the post-measurement stabilizer group follows directly from the stabilizer group of the thin 3D code.
For other measurement outcomes the stabilizers of the state differ only by their signs.
The $Z$ stabilizers that overlap the bottom qubits are removed from the stabilizer group. 
The two-dimensional state after measurement is stabilized by $S^Z_f$ operators on faces $f$ supported on the top boundary.
To determine the operators that stabilize the state, consider first how the volume operators of a thin slab of the 3D color code are transformed by~$\eval{\overline{T}}_{b}$.
Conjugating the $S^X_c$ and $S^{W}_c$ stabilizers transforms them to ``mixed'' Pauli and Clifford stabilizers
\begin{subequations}
\begin{align}
    S^X_c &\stackrel{\eval{\overline{T}}_b}{\longmapsto} \widetilde{S}^X_c =\prod_{v\in c\cap b} W_v \prod_{v\in c\cap t} X_v ,\\
    S^W_c &\stackrel{\eval{\overline{T}}_b}{\longmapsto} \widetilde{S}^W_c = \prod_{v\in c\cap b} X_v \prod_{v\in c\cap t} W_v.
\end{align}    
\end{subequations}
The measurement in the $X_q$ basis only commutes with $\widetilde{S}^W_c$, and the $X$-type volume stabilizers of the 3D color code that are supported on the top boundary only. 
Therefore, the post-measurement state is supported by the latter as well as the volume stabilizers of type~$\prod_{v\in c\cap t} W_q$.

\begin{figure}[t]
\includegraphics[width = 0.87\columnwidth]{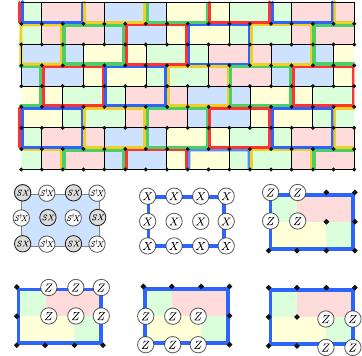}
\caption{ Non-commuting Clifford stabilizers for the color code-based microscopic lattice realization of the TQD phase. Below the figure, we show an example of blue brick stabilizers $S_b^X$ and $S_b^W$, as well as some of the Pauli-$Z$ face stabilizers.  \label{fig:lattice-colorcode}}
\end{figure}

The resulting stabilizers are shown in Fig.~\ref{fig:lattice-colorcode}. 
The lattice shows two types of bricks: ``top-layer'' bricks and ``bottom-layer'' bricks where, in the figure, top-layer bricks are outlined in bold and the bottom-layer bricks are shaded. 
One can see that the lattice of top-layer bricks is a spatial translation from the lattice of bottom-layer bricks on the square lattice of vertices. The top-layer bricks host the $X$-type top boundary cell-stabilizers of the 3D color code, the bottom-layer bricks support the volume stabilizers of the $\prod_{v\in c\cap t} W_q$ type. Each face supports a $Z$-type stabilizer.

\subsubsection{The phase of the non-Abelian code}
We now identify the non-Abelian phase of the Clifford stabilizer model obtained in the previous section by viewing it as a model obtained from gauging a particular anyon-permuting $\zz_2$ symmetry in two copies of the toric code.
This shows that the model hosts non-Abelian anyons precisely equivalent to those of the type-III twisted quantum double of $\zz_2^{\otimes 3}$, see Sec.\,\ref{sec:gauging_CZanyons}.

First, we construct a Pauli model with a suitable $\zz_2$ symmetry.
We start with the two-dimensional Abelian version of the model discussed in Sec.~\ref{sec:2d-abelian-colorcode}, which is equivalent to three copies of the toric code.
Then we condense a single type of boson~\cite{kesselring2022anyon} to obtain a model equivalent to two copies of the toric code.
We choose to condense the $e_{\vb{g}}$ anyons by projecting onto the $+1$ eigenspace of $Z\otimes Z$ operators on green edges -- edges that connect green bricks.
The stabilizer group for the resulting state is then obtained by adding $Z\otimes Z$ hopping operators to it and removing the operators that do not commute with them. The remaining operators are generated by the $Z$ stabilizers of the uncondensed model and all $X$ stabilizers except for the green cells since these anticommute with the enforced hopping operators. 
We illustrate the condensed model in Fig.\,\ref{Fig:NewColorCode}.

\begin{figure}
\includegraphics[width = 0.85\columnwidth]{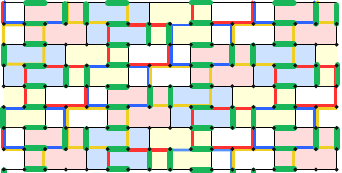}
\caption{The lattice supporting the stabilizer code is locally equivalent to the color-code model. For each red, blue and yellow colored brick, it has the same Pauli-$X$ volume stabilizers as in the uncondensed model shown in Fig.~\ref{Fig:ThinCode}. Similarly, it has all the Pauli-$Z$ face stabilizers of the uncondensed model. The green edges shown in the picture support $Z \otimes Z$ stabilizers. 
A transversal implementation of the logical $S$ gate is achieved by applying $S = i ^{|1\rangle\!\langle 1|}$ to all the qubits on the even sublattice and $S^\dagger$ to all the qubits on the odd sublattice, which is the same as the operator shown in Eq.~\ref{eq:symmetry_green_decomposition}. 
\label{Fig:NewColorCode}}
\end{figure}

The condensed model has a global Clifford symmetry of the form
\begin{align}\label{eq:symmetry_green_decomposition}
    \overline{XS} = \prod_{c\text{ green}} S_c.
\end{align}
where $S_c$ are both the Pauli and Clifford green volume stabilizers of the uncondensed model.
To see that this symmetry preserves the codespace, we recall that if two qubits are stabilized by $Z\otimes Z$ they are also stabilized by any operator of the form $D\otimes D^\dagger$, where $D$ is a diagonal single-qubit unitary.
We use $D = S = i^{\ketbra{1}}$ (or its conjugate on the odd sublattice). In addition, one can explicitly verify that $\overline{XS}$ commutes with the $X$ stabilizers up to a $Z$ stabilizer.
Taken together, we find that $\overline{XS}$ preserves the codespace and acts as a non-trivial logical Clifford operator of order 2 on the condensed code and hence realizes a non-trivial $\zz_2$ symmetry of the condensed code.

Following the general procedure of gauging Abelian on-site symmetries (which we review in more detail in Appendix~\ref{app:onsite-gauging} and \ref{sec:gaugingungauging}), we find that this $\zz_2$ symmetry is gauged by projecting the code space onto the $+1$ eigenspace of each of the (Clifford) stabilizers $\mathcal{S}_c$ associated to the green bricks.
Alternatively, we can measure these operators after ensuring that all $Z$-stabilizers have $+1$ eigenvalues and then applying a suitable correction.
By construction, this gauging procedure produces a state that hosts non-Abelian anyons characterized by the twisted quantum double, as explained in Appendix~\ref{sec:gauging_CZ}.

\subsection{Linking charges in terms of non-Abelian statistics}

Having established this connection between 3D codes and the TQD in spacetime it becomes instructive to investigate the common features that emerge in these two perspectives. In fact, the triple intersection of membranes in spacetime that is instrumental for our protocols (see Sec.~\ref{section:loopsum}) is analogous to the triple intersection resulting in non-Clifford gates in 3D topological codes~\cite{BravyiKoenig,Bombin_2015, bombin2018transversal, zhu2024nonclifford}. Let us explore the analogy more closely at the level of error propagation in 3D codes that occurs upon the application of the non-Clifford gate~\cite{bombin2018transversal}.

Applying a transversal non-Clifford gate to a 3D code with a flux loop excitation gives rise to a superposition of electric charges running along the flux loop~\cite{Yoshida2015topological,bombin2018transversal}.
Because of this the \emph{linking charge} phenomenon can occur in such 3D codes~\cite{bombin2018transversal,Scruby_2022_1}, wherein two non-trivially interlinked flux loops of appropriate color must exchange an odd number of charges upon the application of a transversal gate.

\begin{figure}[t]
 \includegraphics[width= \columnwidth]{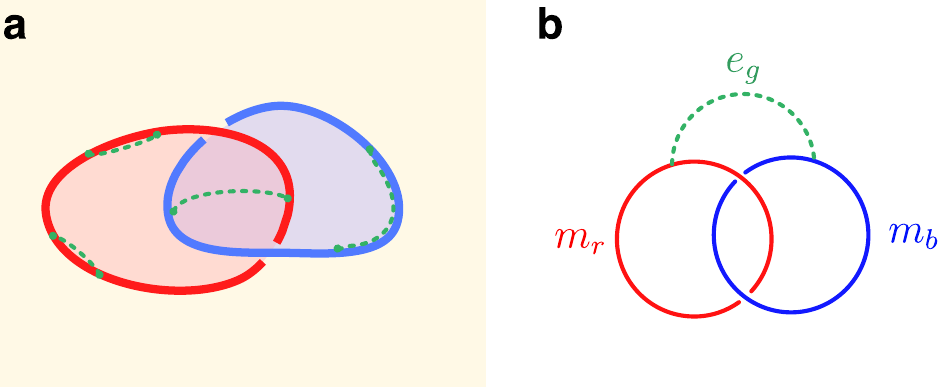}
\caption{ (a) Linking charge phenomenon as defined in Ref.~\cite{bombin2018transversal} in 3D code equivalent to three copies of the toric code or the color code. The red and blue flux loops are linked. When the transversal non-Clifford gate is applied to the three-dimensional volume (depicted by the yellow shading), a superposition of charges of the complementary colors along each flux loop is generated. Furthermore, an odd number of charges of the third, green, color are exchanged between the linked flux loops with certainty. This is analogous to the anyons of the non-Abelian TQD shown in~(b), where braiding red and blue non-Abelian anyons $m_r$ and $m_b$ exchanges a green charge $e_g$.
}
\label{fig:linking-charge}
\end{figure}

There is a close connection between 3D codes and 2+1D phases, and we find that, indeed, the linking charge phenomenon is mirrored in the anyon data of the non-Abelian TQD phase.
The flux worldlines measured in 2+1D can be understood as certain non-Abelian anyon worldlines\footnote{Note that the flux measurement in our protocols is not a measurement of \textit{simple} anyons but still, observing a $-1$ outcome projects onto a sector of the Hilbert space which hosts a combination of non-Abelian anyons.} on which charge worldlines can terminate.
On this basis, the fact that the total charge parity among linked flux anyons is fixed when the non-Clifford gate is applied can be described from the fusion and anyon data of the non-Abelian anyon theory.
From the 2+1D perspective, the charge exchange between linked flux anyon worldlines is related to the braiding data of the TQD anyon model: the $S$-matrix entry for two non-Abelian flux anyons of different types (see Appendix A of Ref.~\cite{Iqbal2024nonAbelian}), corresponding to two linked flux loops that are not connected by any charge strings, is $0$.
As a consequence, if we braid two fluxes, they must exchange an odd number of Abelian charge anyons during the process.
In Fig.~\ref{fig:linking-charge}, we show an explicit example demonstrating the analogy between the charge excitations in the 3D code and a $2+1$-dimensional braiding diagram of the non-Abelian anyons in the TQD.

\section{Fault tolerance and just-in-time decoding}
\label{Sec:JustInTime}

In this section, we discuss how to ensure that the protocols are fault tolerant by introducing just-in-time decoding and proving the existence of a threshold under local stochastic noise.

\subsection{Informal explanation}

The measurements in the circuit that implements the TQD phase are divided into two classes, one of them corresponding to the charge worldlines (which are associated with those of the Abelian charge anyons) and the other corresponding to the flux worldlines. Both charges and fluxes come in red, green, and blue colors.  For simplicity, in our initial discussion we assume that the syndrome-extraction circuit has the property that starting from a noise-free initial state, no $-1$ measurement outcomes would appear in the absence of the noise. In addition, if we start in a state with some distribution of fluxes and run a noiseless circuit, the flux wordlines continue indefinitely in the time direction (unless error correction is performed). All examples considered in Sec.~\ref{section:examples} have this property, but it does not apply to 3+0D, measurement-based quantum computing (MBQC)-like~\cite{Raussendorf2005} and Floquet-like circuits~\cite{Hastings_2021,Bauer2023,Davydova_2024}. While the decoding strategy and the fault tolerance proof remain the same for these other kinds of protocols, they possess an extra element of complexity. We defer a discussion of decoding such protocols to Subsec.~\ref{sec:decoding-MBQC}.

A subtle point is that configurations of fluxes in our circuits do not correspond to specific configurations of non-Abelian excitations. Instead, a state containing some configuration of fluxes contains some superposition of non-Abelian excitations. We focus specifically on correcting flux defects (as opposed to non-Abelian anyons themselves) in our decoding strategy as these describe more naturally what happens in the circuits.

\begin{figure*}[t]
\includegraphics[width=0.9 \textwidth]{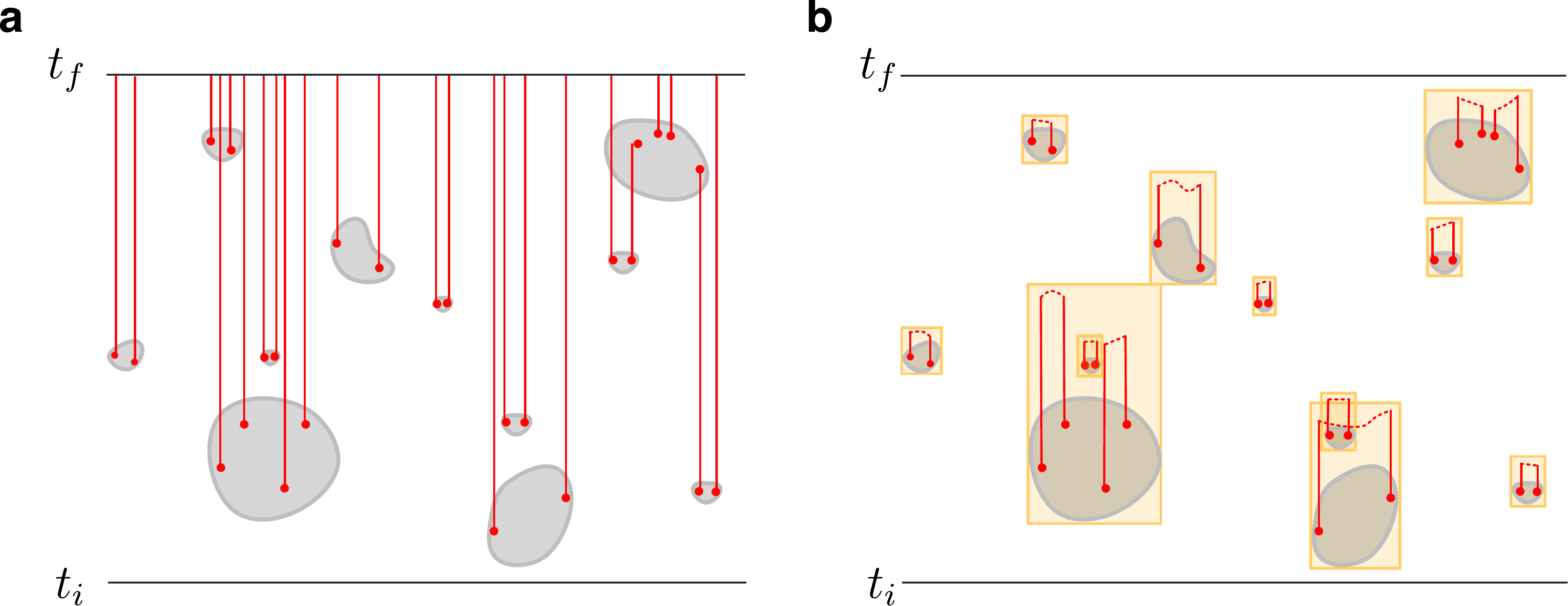}
\caption{
(a) Schematic illustration of the worldlines of flux defects (red lines) created due to noise (contained in gray regions) in the circuit implementation of the non-Abelian TQD, not to scale.  (b) The same spacetime history when operating with the just-in-time decoder, which pairs up the defects as soon as the correction can be deduced reliably. 
Each error cluster together with the associated flux worldlines is contained in some box (in yellow), and these boxes form a sparse distribution.
}
\label{fig:noise-onthefly}
\end{figure*}

The just-in-time decoding in the flux sector is motivated as follows. If we simply collect the measurement data for an extended period of time while in the TQD phase, assuming local stochastic noise as well as measurement errors, a schematic picture of a typical scenario is shown in Fig.~\ref{fig:noise-onthefly}(a). Due to the noise (occurring in gray regions), fluxes (shown in red) appear in the spacetime picture. Their worldlines continue in a timelike direction (shown by vertical red lines) and continue indefinitely in the absence of correction or unless new noise occurs in their path. After a long time (e.g. by $t_f$), 
many flux worldlines accumulate.
For Abelian models, this would not be a problem: by considering the spacetime history of measurements, we determine the most-likely homology class of the flux worldlines in the resulting state and apply the correction at the end; only the resulting homology class matters. This is because in Abelian models all defect worldlines are freely deformable.
However, in the non-Abelian TQD this is not the case. On a flux wordline of a given color, the charge worldlines of two other colors can freely terminate (see Sec.~\ref{subsec:chargeandflux}). Therefore, in the presence of a large density of fluxes (which happens if we simply leave the fluxes be without the corrections), we also end up with an arbitrary configuration of charge wordlines in a large spacetime region, which is no longer correctable.

If we instead apply corrections \emph{just-in-time} (see Fig.~\ref{fig:noise-onthefly}), i.e., we attempt to remove the flux worldlines as we carry out the protocol, it is possible to retain fault tolerance. Such a just-in-time decoding algorithm uses the available measurement history up to the current time and applies the best correction possible based on available information. As a result, the effect of the noise is bounded by a set of spacetime boxes of a certain size, each proportional to the size of the noise cluster.
This is shown by yellow boxes in Fig.~\ref{fig:noise-onthefly}(b). Because of this containment, the state outside of these boxes corresponds to the noiseless TQD phase. The protocol fails if the flux configurations are not contained in sufficiently isolated boxes, which, as we show below, only happens at a large enough noise rate, i.e. above the threshold. 

Thus, the idea behind an appropriate just-in-time decoding scheme is this: we need to close off the flux worldlines that are created due to the noise as soon and as reliably as possible. However, if the decoder applies the correction as soon as the flux is detected, we risk adding large flux worldlines that accidentally connect measurement errors. If the correction is applied too late, the flux worldlines can proliferate as in the example in Fig.~\ref{fig:noise-onthefly}(a). The resolution is to apply the correction as soon as we are confident enough that the detection event (which we also call the noise syndrome) is not a measurement error. As a rule of thumb, a pair of flux worldlines can be closed when the spatial separation between detection events is of the same order as the time these flux worldlines have been observed for. 

The just-in-time decoding scheme discussed here is based on the idea first proposed by Bombin~\cite{bombin2018}, and follows more closely the variant proposed by Brown~\cite{Brown2020universal} (also see Ref.~\cite{Scruby_2022}). Our discussion provides a slightly different explicit proof and makes a more direct relation to topological phases. 

When it comes to error correction of charges (that correspond to Abelian anyons), we can simply store the measurement outcomes of associated stabilizers and process them at the end of the protocol. One difference with the usual decoding of Abelian codes is that detection events of charge measurements can not only be caused by noise but also by the flux worldlines themselves.
Nevertheless, below we show that the resulting distribution must be correctable below the threshold error rate.

Finally, when it comes to the fault tolerance of the whole protocol, we need to take into account additional factors, such as (i) initialization and readout, (ii) dealing with boundaries and domain walls and (iii) extracting logical measurement outcomes from the physical measurement outcomes (for example, in the logical $\overline{CZ}$ measurement protocol). In addition, the spacetime geometry has to be chosen such that the shortest nontrivial and undetectable operator has a size $O(d)$, where $d$ is the spatial distance of the codes used in the computation.

In this work, we assume instantaneous and perfect global classical computation and communication. It would be interesting to explore whether a classically local solution for just-in-time error correction. This would enable classically local universal quantum computation with topological codes in three dimensions or less, which is currently an open problem~\cite{balasubramanian2024localautomaton2dtoric}. 

\subsection{Fault tolerance proof}

\subsubsection{Proof setting}

In this subsection, we assume a circuit whose $+1$-post-selected execution corresponds to a fixed-point path integral for the TQD phase on a torus, such as presented in Secs.~\ref{section:examples}. For such a circuit, the $-1$ measurement outcomes correspond to flux and charge worldlines in the path integral.
As discussed above, we consider the simplified case where the measured flux worldlines can only propagate in the time direction, and discuss the more general case in Subsec.~\ref{sec:decoding-MBQC}.
We assume a noise-free initial state with periodic boundary conditions (i.e.~on a spatial torus of a size $L \times L$). We first show the existence of the threshold in this setting and argue the existence of the threshold for complete logical protocols in the following subsection.

We call the endpoints of charge and flux worldlines in spacetime the noise \emph{syndromes} (there are charge and flux syndromes, correspondingly). This must include the effect of the just-in-time corrections which themselves produce flux worldlines. 
We assume a general noise model (where we label a noise realization as $\widetilde{E}$) that can consist of arbitrary local quantum channels acting on arbitrary qubits at given locations in spacetime as well as faulty gates and measurement outcomes. For the noise realization $\widetilde{E}$, we call $E$ the set of all spacetime points that experience error under a given noise realization. The probability distribution determining the error model are specified later on. Here, we use that for each noise realization and each circuit trajectory, the measurements in the circuit yield a specific distribution of flux and charge worldlines. 

To perform error correction in the flux sector, the circuit is run while simultaneously implementing corrections determined by the just-in-time decoder. The charge-detecting measurement outcomes are stored and processed at the end of the protocol. For each timestep $t$, the decoder determines the set of correction strings $C_t$ (which split into red, green and blue classes for the corresponding sublattices) and passes them to the circuit. 
The circuit then applies string operators that insert flux worldlines along $C_t$, at a fixed time $t$, that close off some of the measured flux worldlines.
For all protocols considered in this paper, these string operators are simply strings of Pauli-$X$ operators.
From the perspective of Section~\ref{sec:CZ_minimal}, as well as circuits in Sec.~\ref{section:examples}, such a Pauli-$X$ string indeed removes the flux defects at its endpoints because it anticommutes with the associated plaquette measurements. 
Furthermore, a Pauli-$X$ string creates a random collection of Abelian charges along its support. 
These charges are detected by the protocol, and can be corrected by the RG decoder at the end of the protocol so long as the clusters containing flux worldlines are sufficiently small (see Sec.~\ref{subsec:chargeandflux}).

\subsubsection{Error model and its properties}

We assume a noise model wherein different error realizations $\widetilde E$ occur with different probabilities $P(\widetilde E)$.
Each noise realization corresponds to adding or replacing some of the quantum channels or measurements in the noise-free circuit with a faulty channel or measurement.
We denote the set of spacetime locations of the faulty channels by $E$.

Further, we assume that the probability of an error configuration is suppressed in its size, or more precisely, that the probability of a fixed set of channels being faulty decays exponentially with the size of this set.
\begin{definition}
    We say that the noise model is $p$-bounded if for any noise realization $\widetilde E$, the probability of a set of spacetime points $A$ entering $E$ is bounded as 
    \begin{equation}
        \mathbb{P}( A \in E) = \sum_{\widetilde E: A\subset E} \mathbb{P}(\widetilde E) \leq p^{|A|}.
    \end{equation}
\end{definition}
In the literature, this is also called local stochastic noise with rate $p$. This noise definition is more general than the local i.i.d. noise and also includes correlated noise.

An important property of this noise model is that there exists a cluster decomposition of the spacetime distribution of noise. We follow the discussion in Ref.~\cite{Bravyi2013quantum} (where the cluster decomposition is called a chunk decomposition); see also Refs.~\cite{gacs1983reliable,Gacs_2001,çapuni2021reliableturingmachine} for the original introduction of the method by G\'acs.  
The cluster decomposition is defined as follows. Suppose that we fix some constant $Q>0$. We call any point in $E$ a level-0 minimal cluster. A subset of $E$ is a level-$n$ minimal cluster if it is a union of two level-$(n-1)$ minimal clusters and its diameter is at most $Q^n/2$; it has precisely $2^n$ points.  We call $E_{\geq n}$ the union of all minimal level-$n$ clusters and define the set of all strictly level-$n$ minimal clusters as $F_{ n} = E_{\geq n} \setminus E_{\geq n+1}$. The minimal cluster decomposition of the error set exists if $E = F_0 \cup F_1 \cup \dots \cup F_m$.
We call an $R$-connected component of a set $E$ a subset that cannot be further partitioned into two subsets separated in distance for at least $R$. We use the following statements that have been proven in Ref.~\cite[Prop. 7]{Bravyi2013quantum}.

\begin{lemma}\textnormal{\textbf{(Existence and properties of the cluster decomposition, Ref.~\cite{Bravyi2013quantum}).}} \label{lemma:clustering}
For a $p$-bounded noise model with $p \leq p_c = 1/(3 Q)^6$ and $Q \geq 6$, the following statements hold:~\footnote{The proof in Ref.~\cite{Bravyi2013quantum} is technically written for the i.i.d. noise; however, it directly generalizes to $p$-bounded noise, see also Refs.~\cite{gacs1983reliable,Gacs_2001,çapuni2021reliableturingmachine}. } 
\begin{enumerate}
    \item The set of spacetime errors can be decomposed into minimal clusters as
    \begin{equation} \label{eq:dec}
        E = F_0 \cup F_1 \cup \dots \cup F_{n_{\mathrm{max}}}
    \end{equation}
    $n_{\mathrm{max}} = \lfloor \log_Q L/2 \rfloor-1$, with probability 
    \begin{equation} \label{eq:decomposition}
    \mathbb{P} \geq 1 - C \left ( \frac{p}{p_c}  \right )^{O(L^\eta)},
    \end{equation} 
    where $C$ is some positive constant and $\eta = 1/\log_2 Q$. 
    \item We refer to the $Q^n$-connected components of $F_n$ as \emph{level-$n$ clusters}.  Any level-$n$ cluster $M$ of $F_n$ has diameter at most $Q^n$ and is separated from $E_n \setminus M$ (but not from $E \setminus M$) by at least $\frac{1}{3} Q^{n+1}$.
\end{enumerate}   
\end{lemma}

\subsubsection{Just-in-time algorithm}
\label{subsec:jit_algorithm}
We now introduce a specific implementation of just-in-time decoding, which we call the just-in-time RG (renormalization group) decoder.  This decoder is not designed to be highly performant in practice, rather it allows us to rigorously prove the existence of a threshold.  We discuss how to perform more efficient practical decoding in Subsec.~\ref{sec:practical-decoding}. For most of this subsection, we focus on the just-in-time decoding in the flux sector and later address the charge sector. To simplify our proof,  we use the usual RG decoder~\cite{Bravyi2013quantum} for the charge sector at the end of the protocol. 
 
We assume that the circuit is associated with spacetime lattices corresponding to possible locations of flux worldlines, charge worldlines and qubits.
We use a rectangular coordinate system on the triple of lattices and the box metric $\ell_\infty$ to determine the distance between points in spacetime: namely, for $x = (x_1,x_2,x_t)$ and $y = (y_1,y_2,y_t)$ call $d(x,y) = \max(|x_1-y_1|,|x_2-y_2|,|x_t-y_t|)$. The flux-detecting measurements performed by the circuit at each time $t$ are stored in an array $m_t$. At each step, the decoder uses this together with the information about all corrections administered up to time $t$ to determine the syndrome at time $t$ that we call $\mathcal M_t$. It includes the measured flux worldline endpoints (excluding the apparent ``termination'' of flux worldlines at the current time $t$) minus the endpoints of the correction strings. 

Below we present the RG just-in-time decoding algorithm.
At each time step $t$, the input to the decoder is the set of flux measurements in the circuit $m_t$. The output of the decoder at each timestep $t$ is the collection of strings $C_{t}$ which the circuit turns into a (Pauli-$X$) correction operator $P_{t}$ that is then immediately applied. In-between two time steps the decoder memorizes some information, namely a set of spacetime points $\Sigma_t$, which is the set of syndromes that are unmatched at time $t$. We define $\partial C_{t}$ to be the endpoints of the string $C_{t}$.

We use the notation `$\oplus$' to denote the symmetric difference between sets. Define  $k_m$ to be the largest integer satisfying $2^{k_m} < L/2$. 

\begin{algorithm}[H]
  \caption{Just-in-time RG decoder}
  \label{alg:greedy}
  At $t=0$, initialize empty $\Sigma_0$. At each timestep $t \geq 0$
   \begin{algorithmic}[1]
   \State Determine $\mathcal M_t$ from  $m_t$ and  $ C_{t-1}$. 
   \State Update $\Sigma_t \rightarrow \Sigma_{t}\oplus \mathcal{M}_t$;
   \State Initialize $C_{t}$ to be empty and the set $s = \Sigma_t$.
   \State  For each $0 \leq k \leq k_m$:
   \begin{itemize}
       \item[3.1:] remove all points $u = (u_1,u_2,u_t)$ in $s$, s.t. $t - u_t < 2^k$; 
       \item [3.2:] for each $u \in s$, if  $u' \in s$ exists s.t. $d(u,u') \leq 2^k$, find any shortest path for the (Pauli-$X$) correction string $C^{(k),i}_{t}$ matching their spatial coordinates.\footnote{We consider different $u$ and $u'$ ``greedily'' in an arbitrary order. The outcome of the decoder depends on the order, but our proof of the threshold for this decoder is order-independent.}  Update $C_{t} \rightarrow C_{t} \cup C^{(k),i}_{t}$ and remove $u, u'$ from $s$ and from $\Sigma_t$;
   \end{itemize}
   \State Pass the string $C_{t}$ to the circuit for correction;
   \end{algorithmic}
\end{algorithm}

We call the index $k$ in the algorithm above the \emph{level} of the operation of the decoder.  Any pair of points $(u,u')$ that is removed from $\Sigma_t$ due to step $3.2$ is said to be ``removed by pairing''. The array of current syndromes $\Sigma_t$ captures exclusively the noise syndromes up to time $t$ that have not been paired yet.

\begin{figure}[!t]
\includegraphics[width= \columnwidth]{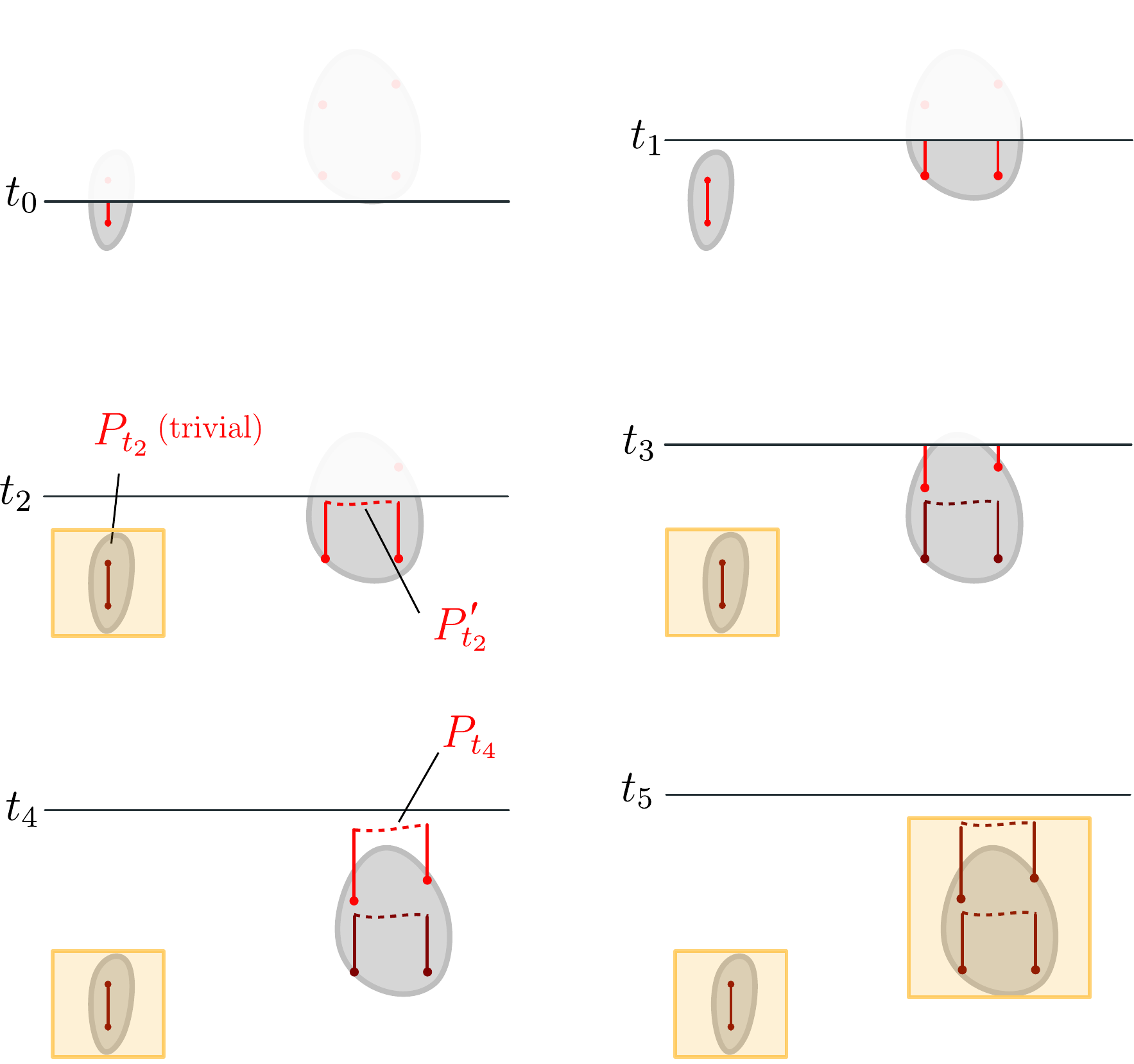}
\caption{ Schematic example of operation of a circuit with the just-in-time RG decoder in the presence of two noise clusters. The cluster on the left produces a measurement error, corresponding to a short timelike open flux worldline. The decoder avoids determining the correction until time $t_2$, when both syndromes have been measured and their distance from the current time slice is large enough (larger than their separation but smaller than the separation from other clusters). The correction $P_{t_2}$ is trivial (no operator is applied). The yellow box shows the smallest box containing both the noise and the correction.
The cluster on the right produces four syndromes which could be due to a pair of string-like errors or a pair of measurement errors. One can directly verify that in either case, the corrections $P'_{t_2}$ and $P_{t_4}$ would close the flux configuration.
}
\label{fig:algorithm-illustration}
\end{figure}

Informally, the algorithm works as follows.  At each point in time, the unpaired syndromes are stored in $\Sigma_t$ (updated at step 1 of the algorithm). Then, for each scale $2^k$, the decoder inspects the history of all unpaired syndromes in spacetime up to the current point in time. The decoder only matches a syndrome if it knows a size-$2^k$ region around it in spacetime. Namely, if a syndrome was detected in less than $2^k$ timesteps in the past, it is set aside and not matched at this timestep. Thanks to this property, the measurement errors (which can temporarily violate the parity of syndromes while the time slice is still going through the error cluster) are not incorrectly matched to errors that belong to a different cluster. Instead, the decoder waits to determine which cluster the syndrome must belong to. For all points that pass the test, the decoder proceeds with greedy pairing, as long as the points are close at a given scale. The operation of the decoder illustrating these points is shown schematically in Fig.~\ref{fig:algorithm-illustration}. We call this scheme a ``RG decoder'' as it operates scale-by-scale similarly to the RG decoder of Bravyi and Haah~\cite{Bravyi2013quantum}, and matches syndromes greedily whenever it can be done reliably.

\subsubsection{Fault tolerance of the just-in-time scheme}

\noindent To prove the existence of a threshold, we assume that we start with a noiseless codestate of the TQD model and then we run the noisy circuit implementation of the TQD phase together with the just-in-time decoder for $T$ timesteps (in the interval of times $[0,T]$), and then run a noiseless version of the same circuit together with the same decoder for another $T_1 \sim O(d)$  timesteps (in the interval of times $[T,T+T_1]$), where $d = L$ is the spatial distance of the code. We explain below how the existence of a threshold extends to full circuits implementing logical protocols.

Here, we state the threshold theorem for the just-in-time scheme in our current idealized setting, which says that under a $p$-bounded noise model, the logical error rate of a circuit (together with the just-in-time decoder) implementing the TQD phase on a $L\times L$ torus decays exponentially in $L^\eta$ (for some constant $\eta > 0$) for noise rate below a threshold.

\begin{theorem} \label{theorem:main}
Consider an error-correcting circuit that implements the TQD model on an $L \times L$ torus together with the just-in-time RG decoder under a $p$-bounded noise model as discussed in this section. Let $\mathcal{C}_{\widetilde{E}}$ be this channel implementing the circuit under a given noise realization $\widetilde{E}$, let $\mathcal{E}$ be the encoding channel into an error-free logical state for the circuit (where the encoded states belong to the 22-dimensional codestate of the TQD model, see Sec.~\ref{sec:codestates}), and $\mathcal{D}$ be the decoding channel consisting of (1)  running the just-in-time flux error correction error-free for the time $O(L)$\footnote{From Lemma~\ref{theorem:onthefly-respects} that we show below, it follows that after running the error-free just-in-time decoding for an extra $O(L)$ time, all the flux errors must have been corrected.} and (2) performing the charge correction error-free using regular RG decoder.

Then, there exists a threshold with a lower bound $p_c = \frac{1}{(3 Q)^6}$ with $Q = 253$, such that for $p< p_c$  we have:
\begin{equation}\label{eq:logical-failure}
\sum_{\widetilde{E}} \mathbb{P}(\widetilde{E})\, \mathcal{D} \circ \mathcal{C}_{\widetilde{E}} \circ \mathcal{E}= (1-p_L)\mathcal{E} +p_L \widetilde{\mathcal E}
\end{equation}
for some (trace-preserving) channel $\widetilde{\mathcal E}$, and the logical error probability is bounded as
\begin{equation}
p_L\leq C \left (\frac{p}{p_c} \right )^{ O(L^\eta)}
\end{equation}
for some constant $C>0$.

\end{theorem}

We prove this theorem at the end of this subsection. The rest of the subsection develops the tools necessary for proving it.

We now look at some fixed error realization  $\widetilde E$.
When $\mathcal{D}  \circ   \mathcal{C}_{\widetilde E} \, \circ  \mathcal{E} = \mathcal{E}$, the outcome of the circuit with decoder is the same as the initial (clean) encoded state, thus, we say that no logical error has occurred.

We assume that for an error $\varepsilon$ contained in a box $\mathcal{B}$ with diameter $b$, its syndromes obtained after measurements in circuits are generally contained in a $\gamma$-neighborhood of $\mathcal{B}$ (including $\mathcal{B}$),  
where $\gamma \geq 0$ is a system size-independent constant which depends on the details of the circuit. We can always rescale the lattice such that $\gamma = 1$, which we assume below.

The following Lemma provides sufficient criteria that ensure that no logical error occurs.

\begin{lemma} \label{lemma:fixed-point}
Consider the channel $\mathcal{C}_{\widetilde{E}}^{\vec f}$ corresponding to a trajectory of the circuit under a given noise realization $\widetilde{E}$ with measurement outcomes for the flux and charge worldlines $W_f$ and $W_c$, as well as the union of all flux correction strings $C_f$ and for the charge correction strings $C_c$, which we collectively denote as $\vec f = (W_f,W_c,C_f,C_c)$.\footnote{One can also think about particular trajectory with fixed $\vec f$ as a trajectory with specific measurement outcomes being post-selected for. }  The following statements hold with probability 1:
\begin{itemize}
    \item [(i)] The union of flux worldlines and their corrections terminate in the vicinity of errors: $\partial (W_f + C_f)\subseteq \nx_1(E)$, where $\partial$ denotes the boundary map. 
    \item [(ii)] The union of charge worldlines and their corrections terminate in the vicinity of errors and fluxes:   $\partial(W_c + C_c)\subseteq \nx_1( E\cup W_f\cup C_f)$.
    \item [(iii)] Connected components of errors are flux-neutral. Namely, consider any connected component  $X$ of $\nx_1(E)$. Then, $\partial(W_f + C_f)\cap X$ contains an even number of points.
    \item[(iv)] A box fully containing one or several connected components of regions that can create charges (i.e. containing noise, flux worldlines, and flux corrections) is charge-neutral\footnote{Alternatively, we define an ``unlinked'' connected component as a connected component whose flux worldlines do not form links with flux worldlines belonging to other connected components.}. Namely, consider any part $Y$ of $\nx_1(E \cap W_f \cap C_f)$ that is contained on the inside of some 2-sphere and disconnected from the rest. Then, $\partial(W_c + C_c)\cap Y$ contains an even number of points.
\end{itemize}
In addition, the following statement is true:
\begin{itemize}
    \item [(v)] There is no logical fault, namely
\begin{equation}
\mathcal{D} \circ  \mathcal{C}_{\widetilde E}^{\vec f} \, \circ \mathcal{E} = \mathcal{E} ,
\end{equation}
if the following two conditions hold:
\begin{itemize}
\item [(a)] the set $\nx_1(E\cup W_f\cup C_f)$ is topologically trivial, i.e. it is contained on the inside of some 2-sphere in the 3D spacetime;
\item [(b)]  for any closure $\widetilde{C}$ of $W_c+C_c$ inside of $\nx_1(E\cup W_f\cup C_f)$, $\widetilde{C}+W_c+C_c$ is homologically trivial.
\end{itemize}
Note that in (b), it is always possible to find $\widetilde{C}$ due to (iii), and the homology of $\widetilde{C}+W_c+C_c$ is independent on the choice of $\widetilde{C}$ if (a) already holds.
\end{itemize}
\end{lemma}

Intuitively, statement (v) says that a logical fault can only occur if the error together with the correction is topologically (or homologically) non-trivial.
The lemma also takes into account the fact that due to the twist in the path integral, charge worldlines can end on flux worldlines, see Sec.~\ref{subsec:chargeandflux}.
These properties follow from the fixed-point properties (namely, zero correlation length) of the path integral that are inherited by the circuit, and of its realization with the flux and charge worldlines, discussed in Section~\ref{subsec:chargeandflux}. We find that they hold for any of the circuits proposed in this paper, but we will not explicitly show this. For explicit proofs of analogous statements, we refer the reader to Refs.~\cite{Bauer2024,Bauer2024a}.

We start by showing an auxiliary result; the decoder that runs up to a high enough level fully corrects an error cluster in complete isolation from any other errors while increasing its diameter by at most a certain fraction.  A subset $A \in E$ of the error set $E$ is called \emph{isolated} from other errors by distance $R$ if $d(A, E \setminus A) \geq R$. We call the box $\mathcal{B}_{\text{is}}$ an isolation region of set $A$ if it includes the $R$-neighborhood of the smallest box $\mathcal{B}$ containing $A$ and does not include any points from $E$.

\begin{lemma} \label{lemma:removal-of-isolated}
    Assume some fixed $0 < n < n_{\mathrm{max}}$, and $Q > 24$ and that we are given a noise realization that contains an error cluster $F$ that fits in a box $\mathcal{B}$ sized $Q^n$ and is isolated from any other error by at least $\frac{1}{3} Q^{n+1}$. Denote its syndromes by $\sigma_F$ and assume that $\mathcal{B}$ starts at time $t_0$. We denote the isolation region of this error in spacetime as $\mathcal{B}_{\mathrm{is}}$.
    
    We continuously run the operation of the circuit just-in-time decoder (Alg.~\ref{alg:greedy}) for all $k$ up to the smallest $k$ such that $2^k \geq  Q^n+2$. We call this value $k_n$.
    \begin{itemize}
        \item[(a)] At each timestep $t \geq t_0$, by the end of operation of level-$k_n$ decoder, the set $s \cap \mathcal{B}_{\mathrm{is}}$ will be empty, and
        \item[(b)] The syndromes $\sigma_F$ will be removed from $\Sigma_t$ by the time $t = t_0 + 3 Q^n + 5$. The union of all corrections determined by level-$k$ operation of the just-in-time decoder $ C^{(k)}$ with $k \leq k_n$ together with $F$ is contained in a $(2 Q^n +5)$-neighborhood of $\mathcal B$.
    \end{itemize}
\end{lemma}

\begin{proof}
    First, we notice that the syndromes $\sigma_F$ must be contained in a $1$-neighborhood of $\mathcal B$, which we call $\mathcal{B}_{\mathrm{syn}}$. In addition, their number must be even (since any isolated error must be neutral, see Lemma~\ref{lemma:fixed-point}(iii)). 

    As we discussed below Alg.~\ref{alg:greedy}, $\Sigma_t$ at each time step only contains unpaired noise syndromes. Consider first $t \in [t_0, t_0 + Q^n]$. At every such time step, a given point $u$ in $\Sigma_t$ must be either left with no action due to step 3.1 (i.e.~removed from $s$ and passed directly to $\Sigma_{t+1}$) or paired with another point in $\sigma_F$ and removed (i.e.~removed from $s$ but not passed to $\Sigma_{t+1}$). Any point $u \in \sigma_F$ cannot be paired with any point $v$ outside of $\sigma_F$ because  $d(u,v) \geq \frac{1}{3}Q^{n+1} - 2 > 2^{k_n}$ (the latter is true for $Q > 24$~\footnote{The way to see this is to use that $k_n$ is the smallest integer satisfying $2^{k^n} \geq Q^n + 2 $. This means that $Q^n + 2  \geq 2^{k_n-1}$ or $2Q^n + 4  \geq 2^{k_n}$. We then require $\frac{1}{3} Q^{n+1} - 2  > 2 Q^n + 4  \geq 2^{k_n} $ for any $n \geq 0$.}). This shows (a) for $t \in [t_0, t_0 + Q^n]$.

    The total correction $C^{k}_{t}$ for each such timestep $t \in [t_0, t_0 + Q^n]$ must be contained in the $1$-neighborhood of the box $\mathcal{B}$. This is because for any $u, u' \in \sigma_F$ that are paired, the shortest path between their spatial coordinates must lie in the spatial shadow of $\mathcal{B}$ by definition of $\mathcal B$; this path placed in spacetime at time $t+1$ for $t \in [t_0, t_0 + Q^n]$ must be in 1-neighborhood of $\mathcal{B}$.

    Consider now $t >  t_0 + Q^n$, i.e. the times later than the latest time coordinate in $\sigma_F$. There must be an even number from points in $\sigma_F$ left in $\Sigma_t$ at the beginning of each time step. At the time $t = t_0 + Q^n + 2^{k_n}$, any point $u \in \sigma_F \cap \Sigma_t$ is closer to any remaining point $v \neq u$, $v \in \sigma_F \cap \Sigma_t$, since $\max_{u,v \in \sigma_F} d(u,v) \leq Q^n + 2 $ and $2^{k_n} \geq Q^n +2 $. In addition, it cannot be removed from $s$ by any level $k < k_n$ of the operation of the decoder. Therefore, the algorithm must remove all the points by pairing at this time at the latest. Because of the shortest matching, the totality of corrections $C^{(k)}$ matching the points in $\sigma_F$ must belong to a $(2^{k}  + 1)$-neigborhood of $\mathcal{B}$. Due to $2 Q^n + 4  \geq 2^{k}$ we obtain the statement of (b). 

    In addition, at each timestep in the paragraph above, any point in $s \cap \mathcal{B}_{\mathrm{is}}$ must be either removed by step 3.1 or paired up by the time $t = t_0 + Q^n + 2^{k_n}$. This completes statement (a) as well. 
\end{proof}

As a direct consequence of the Lemma above, for any noise realization obeying cluster decomposition, in a circuit that is followed by a noiseless just-in-time decoder, all the flux clusters will be corrected by the time $2 Q^{n_{\max}}+5$. Since $Q^{n_{\max} + 1} < L/2$, for large enough $Q$, $2 Q^{n_{\max}}+5 < L/2$. Thus, if we run noiseless decoding for another $T_1=O(L)$ time, the state is guaranteed to return to some clean logical state in the flux sector. 

Now, we show that the just-in-time decoder is able to correct an error realization in spacetime that obeys cluster decomposition with large enough $Q$.

\begin{theorem}\textnormal{\textbf{(The effect of the just-in-time 
RG decoder on noise clusters).}} \label{theorem:onthefly-respects}
     Given a noise trajectory $E$ that obeys cluster decomposition from Lemma~\ref{lemma:clustering} with $Q  > 24$, assume running the circuit with the just-in-time decoder under assumptions of this section. Then, the points in each cluster are paired only with the points within the same cluster (thus, there is a well-defined correction associated with each cluster), and the following holds.
     \begin{itemize}
         \item [(a)] By time $T + T_1$ the syndromes for each cluster $F_{n,i}$ for every $0 \leq n \leq n_{\mathrm{max}}$ are removed by pairing.
         \item [(b)] The correction associated with each cluster $F_{n,i}$ is contained inside a $(2 Q^n +5)$-neighborhood of $F_{n,i}$.
     \end{itemize}
\end{theorem}
\begin{proof}
    We first show that at each timestep $t$, the syndromes of each 
    error cluster $F_{n,i}$,
    which we denote $\sigma_{F_{n,i}}$, are operated upon by the just-in-time decoder at levels $k \leq k_n$ (i.e. the smallest $k$ such that $2^k \geq  Q^n+2$) as if this
    cluster
    was completely isolated from any other errors by $\frac{1}{3} Q^{n+1}$. Once we show this, we can use the statement (b) in Lemma~\ref{lemma:removal-of-isolated} from which the statement of the Theorem follows straightforwardly. 

    We show this by induction. Each
    level-0 cluster
    $F_{0,i}$ is just a single error that is isolated from any other error by at least $\frac 1 3 Q$. In this case, Lemma~\ref{lemma:removal-of-isolated} applies. From  Lemma~\ref{lemma:removal-of-isolated}(a), we additionally have that, at every time step, by the smallest level of decoder operation $k_0: 2^{k_0} \geq 3 $, the syndromes in $\sigma_{F_{0,i}}$ that are contained in $s$ must be removed from it. From the proof of the Lemma, we also have that any in $\sigma_{F_{0,i}}$ can only be paired with another syndrome in this set. This completes the induction base.

    Next, assume that the assumption holds for some $n$. This means that all syndromes of errors in $\cup_{i, m<n} F_{m,i}$ have been treated according to the assumption and, by Lemma~\ref{lemma:removal-of-isolated}(a), their syndromes must be removed independently of the rest of the error clusters. This means that the decoder can be effectively assumed to operate on a decomposition $\cup_{n< s < n_{\mathrm{max}}} F_s$. Consider connected components $F_{n+1,i}$, which we can now assume to be isolated from any other error whose syndromes enter in $s$ by at least $\frac{1}{3} Q^{n+2}$. Applying Lemma~\ref{lemma:removal-of-isolated}, we obtain the statement of the Theorem.
\end{proof}

We denote the $R$-neighborhood of a box $A$ as $\mathcal{N}_{R}(A)$ which we define to be a set that includes the points in the $R$-neighborhood of $A$ as well as $A$ itself. For any given constant $\mu$, we also define $\mathcal{N}_{\mu,n} \equiv \mathcal{N}_{\mu Q^n}$.

We now show that if the decoder ``spreads'' each
cluster
$F_{n,i}$ of the error by at most some factor $\alpha Q^n$ (from Theorem~\ref{theorem:onthefly-respects}, for the just-in-time decoder in this section, we can choose $\alpha = 7$, which is a conservative overestimate\footnote{Setting $\alpha = 7$ is a vast overestimate as for $n > 1$, $\alpha = 2 + \frac{5}{Q}$ is sufficient. One can obtain tighter bounds on $Q$ and the error threshold if we assume the spreading of connected error clusters by $\alpha' Q^n + \mu'$ as opposed to simply setting it to be $\alpha Q^n$. The proofs extend relatively straightforwardly to this case. }, we can consider the set of all ``fattened'' boxes and find a decomposition that is effectively sparse if $Q$ is large enough. 

Calling $\bx_{n,i}$ each $\alpha Q^n$-``fattened'' cluster  $F_{n,i}$, we now have a set of points in spacetime $\bx = \bigcup_{n = 0;i}^{n = n_{\mathrm{max}}}{\bx_{n,i}}$ to which all flux worldlines in the spacetime volume of the circuit are constrained. We also refer to $\bx_{n,i}$ as a ``box'' where it does not cause ambiguity. Some of the fattened clusters can now overlap, see Fig.~\ref{fig:sparsity}. Nevertheless, we now show that $\bx$ admits a sparse decomposition. To this end, we first introduce the following definitions.  

\begin{definition}[\textbf{Linking}] \label{def:linking}
    Consider a union of sets $\bx = \bigcup_{n = 0;i}^{n = n_{\mathrm{max}}}{\bx_{n,i}}$. We say that the box $\bx_{s,i}$ is directly linked to the box $\bx_{m,j}$  if
    \begin{equation}
        d(\bx_{s,i}, \bx_{m,j}) \leq 70 Q^{\min(s,m)}.
    \end{equation}
    (if two boxes overlap, we define the distance between them to always be 0). We refer to the condition above as the existence of an $s$-link.

    We say that the box $\bx_{s,i}$ is (indirectly) linked to the box $\bx_{m,j}$ if there exists a chain of links that starts at the box $\bx_{s,i}$ and ends at the box $\bx_{m,j}$. If such a chain of links does not exist, we say that $\bx_{n,i}$ and $\bx_{m,j}$ are unlinked.
\end{definition}

The parameter $70$ was chosen such that the two following lemmas hold.

\begin{definition}[\textbf{Linked groups}] \label{def:groups}
    We define a linked group $G_{n,i}$ as a union of boxes that consists of at least one box and is obtained as recursively adding every box that is linked to any box that is already in the group, where $n$ is the highest level of a box in the group (we show below that there is a unique highest-level box per group). In other words, $G_{n,i}$ is a connected component of the set $\bx$ under linking.

    For any given $\bx_{n,i}$ define the according $s$-candidate group $C_{n,i}^{s}$ with $s \leq n$ as the union of all $\bx_{m,j}$ with $m \leq s$ that are directly or indirectly linked to $\bx_{n,i}$ via chains of links that do not contain any boxes of level larger than $s$. 
\end{definition}

Examples of a cluster group and a candidate group are shown in Fig.~\ref{fig:sparsity}. By definition, the union of all fattened boxes can be decomposed into cluster groups: $\bx = \bigcup_{n = 0; i}^{n = n_{\max}} G_{n,i}$.
In the following Lemma, we establish some properties of cluster groups that are instrumental to then showing that their union, in fact, provides a sparse cluster decomposition of $\bx$.
 
\begin{figure}[!t]
\includegraphics[width= \columnwidth]{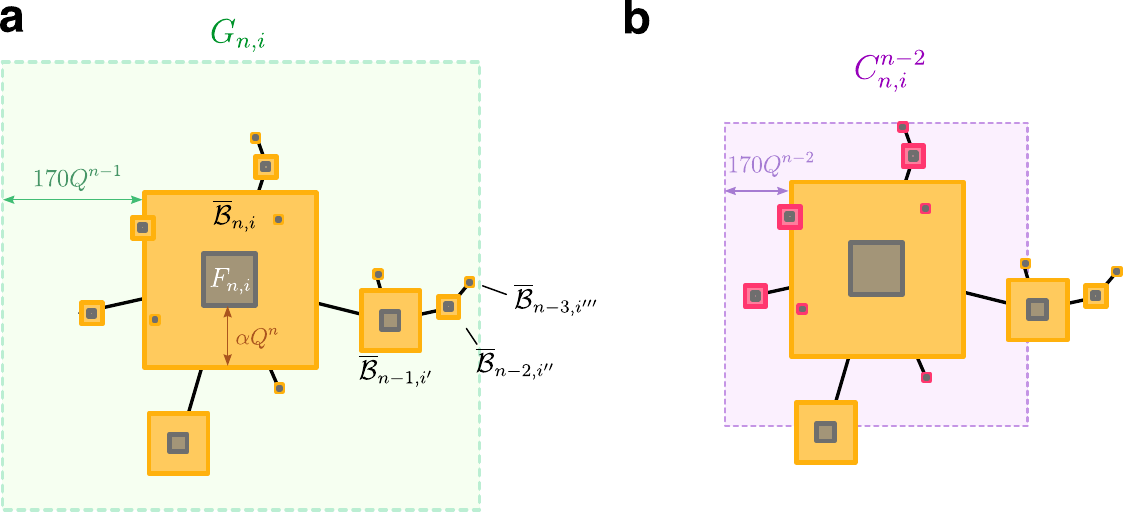}
\caption{ Schematic illustrating linking of ``fattened'' clusters (shown in orange) into (a) a linked group $G_{n,i}$ and (b) $(n-2)$-candidate group $C^{n-2}_{n,i}$ (not to scale). The boxes that are linked are connected by a black line. The group $G_{n,i}$ contains all the boxes linked to $\bx_{n,i}$. The boxes of level $s \leq n-2$ that are linked to $\bx_{n,i}$ via chains that do not contain $m >n-2$ boxes are shown in red in (b); together, these form an  $(n-2)$-candidate group.
}
\label{fig:sparsity}
\end{figure}

\begin{lemma}\textnormal{\textbf{(Linking properties).}} \label{lemma:linking}
Consider a set $E$ obeying the cluster decomposition $E = \bigcup_{n \leq n_{\max}}F_n$ from Lemma~\ref{lemma:clustering} for $Q > 252$.  For each cluster $F_{n,i}$, define the associated ``fattened'' cluster $\bx_{n,i}$ as the smallest box fitting $\alpha Q^n$-neighborhood of $F_{n,i}$.   Now consider the set of points
$\bx = \bigcup_{n = 0;i}^{n = n_{\mathrm{max}}}{\bx_{n,i}}$. Under  definitions~\ref{def:linking},~\ref{def:groups} introduced above, the following properties hold:
\begin{itemize}
    \item [(1)] Consider three boxes $\bx_{n,i}$, $\bx_{m,j}$ and $\bx_{s, \ell}$ with $m \geq n$    
    and $s<m,n$. If the pairs $\bx_{n,i}$, $\bx_{s, \ell}$ and $\bx_{m,j}$, $\bx_{s, \ell}$ are directly linked, then $\bx_{n,i}$ and $\bx_{m,j}$ must be directly linked.
    \item[(2)] Any two boxes same-level boxes $\bx_{n,i}$ and $\bx_{n,j}$ with $i \neq j$ cannot be linked directly. They also cannot be linked indirectly unless there exists a higher-level box $\bx_{m, \ell}$ to which they both are linked. 
\end{itemize}

In addition, the groups of boxes have the following properties:
\begin{itemize}
    \item [(3)] Each linked group $G_{n,i}$ contains a unique highest-level box $\bx_{n,i}$. 
    \item[(4)] For any given box $\bx_{n,i}$ and $s \leq n$, the candidate group $C^s_{n,i}$ fits in the $170 Q^s$-vicinity of the box:
    \begin{equation}
     C_{n,i}^s \subseteq \nx_{170,s} (\bx_{n,i}).
    \end{equation} 
\end{itemize}
\end{lemma}
\begin{proof}
    (1):
    Assume that $\bx_{s,\ell}$ is directly linked to both $\bx_{n,i}$ and $\bx_{m,j}$.
    This means:
    \begin{equation}
        d(\bx_{s,\ell},\bx_{n,i}) \leq 70 Q^s, \quad d(\bx_{s,\ell},\bx_{m,j}) \leq 70 Q^s.
    \end{equation}
    By the triangle inequality this means that $\bx_{n,i}$ and $\bx_{m,j}$ must be not too far from each other:
    \begin{equation}
        d(\bx_{n,i},\bx_{m,j}) \leq 2\times 70 Q^s < 70 Q^{n}.
    \end{equation}
    The second inequality holds for $Q>2$, and shows that $\bx_{n,i}$ and $\bx_{m,j}$ are linked, proving statement (1).

    (2): First of all, we show that a direct link between $\bx_{n,i}$ and $\bx_{n,j}$ is impossible. We establish the following lower bound on the distance between the boxes:
    \begin{equation}
    \begin{split}
           d(\bx_{n,i}, \bx_{n,j}) &\geq  d(F_{n,i}, F_{n,j})  - 2 \alpha Q^n 
           \\
           &\geq \frac{1}{3} Q^{n+1} - 2 \alpha Q^n > 70 Q^n  .
    \end{split}
    \end{equation}
    the last inequality is true for $\alpha = 7$, $Q > 252$. Now consider the case when  $\bx_{n,i}$ and $\bx_{n,j}$ are connected via a chain of direct links, which is irreducible in the sense that only the boxes that neighbor in the chain are linked directly, but no other pair in the chain is. On the contrary to the statement, let us assume that the chain only contains intermediary boxes of the level $s\leq n$. There must exist a box $\bx_{s,\ell}$ with $s < n$ that directly links to two other boxes whose level is larger than $s$, which contradicts (1) by the irreducibility assumption. This shows (2).

    Property (3) follows straightforwardly from (2). By definition, any pair of boxes in $G_{n,i}$  are (indirectly) linked. Any highest-level box $\bx_{n,j}$ in $G_{n,i}$ cannot be linked by (2), to any other box of the same level. Thus, it must be the only box of this level contained in $G_{n,i}$. 

    Finally, we show the last property (4). Consider any $\bx_{n,i}$ and the associated $s$-level candidate group $C^s_{n,i}$. 
    By properties (1) and (2), the linking in the candidate group from $\bx_{n,i}$ to any lower-level box within the group must go monotonically down in levels. Therefore, in the worst-case scenario, a single chain of links from the box $\bx_{n,i}$ goes all the way from $s$ to $0$ while each next box in $\{ \bx_{m,i_{m}}\}_{m = s.. 0}$ is directly linked to the previous one. The length these occupy together with the spacing due to the links can be bounded as
    \begin{equation}
    \begin{split}
        \mathcal{L} &\leq 70 Q^s + \mathrm{diam}(\bx_{s,i_s}) + 70 Q^{s-1} + \mathrm{diam}(\bx_{s-1,i_{s-1}}) + ... 
        \\
        &+ 70 + \mathrm{diam}(\bx_{0,i_0})   .
    \end{split}
    \end{equation}
    Using that $\mathrm{diam}(\bx_{s,i_s}) = (2 \alpha + 1)Q^s = 15 Q^s$, we arrive at 
      \begin{equation}
    \begin{split}
        \mathcal{L} \leq (70+ 15) \frac{Q^{s-1}-1}{Q-1} < 170 Q^s,
    \end{split}
    \end{equation}
    which is true for the chosen range of $Q$. Thus, since all link chains together with respective boxes must be within $\mathcal{L} < 170 Q^s$ distance away from $\bx_{n,i}$, the entire candidate group must be in its appropriate vicinity $C_{n,i}^s \subseteq \nx_{170,s} (\bx_{n,i})$, which gives statement (4). 
\end{proof}
Finally, we are ready to show that the set $\bx$ admits a sparse decomposition (into sets of points that are linked groups of clusters). 

\begin{lemma}\textnormal{\textbf{(Sparse decomposition of fattened clusters).}} \label{lemma:clustering2}
Under the assumptions of Lemma~\ref{lemma:linking}, for $Q > 252$, the decomposition of $\bx$ into linked groups,
\begin{equation}
    \bx = \bigcup_{n = 0; i}^{n = n_{\max}} G_{n,i}
\end{equation}
is sparse. That is, there exists a constant $\beta = 20$ such that
\begin{itemize}
    \item[(a)] each group is bounded in size:
    \begin{equation}
        \mathrm{diam}(G_{n,i}) < \beta Q^n, 
    \end{equation} and
    \item[(b)] the separation between any given group and all the groups that are same level or higher is at least three times the size of this group, namely, for all $G_{m,j}$ with $m\neq n$, $m \geq n$, we have
    \begin{equation} \label{eq:cond-b}
        d\left ( G_{n,i},G_{m, j}\right ) > 3 \beta Q^n.
    \end{equation}
\end{itemize}
\end{lemma}
\begin{proof}
Consider any cluster group $G_{n,i}$ whose highest-level box is $\bx_{n,i}$. By Lemma~\ref{lemma:linking} it follows that $ G_{n,i} \equiv C^{n-1}_{n,i}$. 
From Lemma~\ref{lemma:linking}(4), we have $  G_{n,i} \subseteq \nx_{170, n-1}(\bx_{n,i})$, and thus,
\begin{equation}
\begin{split}
     \mathrm{diam}(G_{n,i}) &\leq \mathrm{diam}(\nx_{170, n-1}(\bx_{n,i})) 
     \\
     & \leq (2 \alpha + 1)Q^n + 2 \times 170 Q^{n-1}  \\
     &= 15 Q^n + 340 Q^{n-1} < \beta Q^n ,
\end{split}
\end{equation}
which is true for $\beta = 20$ and $Q > 252$.  This shows (a). 

We now show (b). Specifically, we show that any box 
$$\bx_{m,j} \in G_{\geq n} \setminus G_{n,i} \bigcup_{\substack{s \geq n; \ell \\ (s,\ell) \neq (n,i)}}G_{s, \ell}$$ 
is further from $G_{n,i}$ than $3 \beta Q^n$. Any such box $\bx_{m,j}$ is either of the same or greater level $m \geq n$  or smaller level  $m<n$ than the considered group $G_{n,i}$, and is not linked to any of the boxes in  $G_{n,i}$ by assumption.  We consider these cases separately and show that the distance between $G_{n,i}$ and any of such boxes is greater than $3 \beta Q^n$, which is equivalent to (b).

In the case $m \geq n$, consider $d(G_{n,i} , \bx_{m \geq n, j})$. By Lemma~\ref{lemma:linking}(4) and since $\bx_{n,i}$ and $\bx_{m \geq n ,j}$ must be unlinked by assumption, we have 
\begin{equation}
\begin{split}
      d(G_{n,i} , \bx_{m \geq n, j}) &\geq d(\bx_{n,i}, \bx_{m \geq n,j}) - 170 Q^{n-1}  
      \\
      &\geq 70 Q^n -  170 Q^{n-1} > 3 \beta Q^n ,
\end{split}
\end{equation}
which is true for $\beta = 20$ and $Q > 252$. This shows the statement (b) for $m \geq n$.

If $m < n$, there must exist some $\bx_{s,\ell} \in G_{\geq n} \setminus G_{n,i}$ with $s \geq n$ that must be linked to it otherwise $\bx_{m,j}$ would be in a lower-level group and would not be contained in $G_{\geq n} \setminus G_{n,i}$. We choose $\bx_{s,\ell}$ to be any lowest-level box satisfying $s \geq n$ that is linked to $\bx_{m,j}$. By assumption that we chose the box of the smallest such level, $\bx_{m,j}$ must be within its associated candidate group $C^{n-1}_{s, \ell}$ \footnote{Namely, if $\bx_{m,j}$ is in any $q$-candidate $C^q_{s,\ell}$ with $q\geq n$, but not in $C^{n-1}_{s, \ell}$, then it must be linked to some other $\bx_{\geq n, \ell'}$, which contradicts the assumption that $\bx_{s,\ell}$ is the smallest-level box of level $\geq n$ that it $\bx_{m,j}$ linked to}. By Lemma~\ref{lemma:linking}(4), $\bx_{m,j} \in C^{n-1}_{s, \ell} \subseteq \nx_{170,n-1}(\bx_{s,\ell})$, and thus
    \begin{equation}
\begin{split}
      d(G_{n,i} , \bx_{m < n, j}) &\geq d(G_{n,i}, \nx_{170,n-1}(\bx_{s \geq n,\ell}))  
      \\ 
      &\geq d(\bx_{n,i}, \nx_{170,n-1}(\bx_{s \geq n,\ell}))  - 170 Q^{n-1}
      \\
      &\geq d(\bx_{n,i}, \bx_{s \geq n,\ell}) - 2 \times 170 Q^{n-1}
      \\
      &\geq 70 Q^n -2 \times 170 Q^{n-1} > 3 \beta Q^n.
\end{split}
\end{equation}
The inequality before last holds because  $\bx_{n,i}$ and $\bx_{s \geq n,\ell}$ must be unlinked, otherwise they would belong to the same group.  The last inequality holds because $\beta = 20$ and $Q > 252$. 
This shows (b) for $m < n$ and completes the proof of the Lemma. 
\end{proof}

The Lemma above gives the following Corollary, which shows that some of the criteria for the absence of logical errors must be satisfied when $Q> 252$:
\begin{cor} \label{lemma:condition_a}
    Condition (a) in Lemma~\ref{lemma:fixed-point} for the absence of logical failure is fulfilled for any error configuration $E$ obeying a cluster decomposition.
\end{cor}
\begin{proof}
By Theorem~\ref{theorem:onthefly-respects} and Lemma~\ref{lemma:clustering2} that $\nx_1(E\cup W_f\cup C_f$ obeys a cluster decomposition as well.
Any sparse cluster decomposition where the size of the largest cluster is smaller than half of the system size is confined to isolated boxes and thus topologically trivial.
\end{proof}

\subsubsection{Fault tolerance in the Abelian charge sector}

For the Abelian charge sector, we simply record the syndromes until the end of the protocol and decode them using the usual RG decoder at the end of the circuit.  We now show that the RG decoder succeeds if the noise obeys a cluster decomposition and the just-in-time decoder fattens the noise clusters such that the outcome is sparse:

\begin{lemma}
\label{lemma:condition_b}
 
      Under assumptions of this section, given any noise realization $\widetilde E$ that obeys cluster decomposition with $Q > 252$, the RG decoder of Ref.~\cite{Bravyi2013quantum} that is run at the end of a circuit with just-in-time decoding, succeeds. Namely, condition (v.b) in Lemma~\ref{lemma:fixed-point} holds (along with condition (v.a) which was shown to hold in the previous subsection).
   
\end{lemma}
\begin{proof}
     We first note that the charge syndromes are restricted to the region $\bx$ due to Theorem~\ref{theorem:onthefly-respects} and Lemma~\ref{lemma:fixed-point}(ii). For any group  $G_{n,i}$ in isolation, the RG decoder admits a correction that is in its $1$-neighborhood. When $Q>252$, from the sparse decomposition of $\bx$ shown in Lemma~\ref{lemma:clustering2}, each group $G_{n,i}$ must have syndromes fitting in a box of size less than $\beta Q^n+2$ (recall that $\beta = 20$) and is separated by distance that is larger than $3 \times \beta Q^n - 2 > 2 \times\mathrm{diam}(G_{n,i}) $ from any group of the same and larger size. Going inductively from smaller to larger levels as in the proof in Ref.~\cite{Bravyi2013quantum}, we find that the decoder must issue an appropriate correction (namely, the one that satisfies $\partial c(n,i) \oplus W_c(n,i) \oplus \sigma^c_{G_{n,i}} = 0$) that is within 1-neighborhood of each cluster.
     Thus, the global RG decoder connects the syndrome associated with every group $G_{n,i}$ inside a 1-neighborhood of $G_{n,i}$.
     Since the boxes together with corrections from the RG decoder are isolated, the homology class of the corrections fits the error configuration, and thus condition (v.b) of Lemma~\ref{lemma:fixed-point} holds.
\end{proof}

Finally, this brings us to the proof of the threshold theorem:

\begin{proof}[Proof of Theorem~\ref{theorem:main}]
Due to Lemmas~\ref{lemma:condition_a} and \ref{lemma:condition_b}, the corrections for both flux and charge sectors fit in a set of disjoint balls.  The conditions (a,b) of Lemma~\ref{lemma:fixed-point} hold for any noise realization satisfying cluster decomposition with $Q = 253> 252$. Thus, Lemma~\ref{lemma:fixed-point} states that there must be no logical failure in the circuit in this case.

The only circumstance when the logical failure can happen is when the $p$-bounded noise does not admit sparse decomposition with the chosen $Q$. By Ref.~\cite{Bravyi2013quantum},  for $p < p_c$ with $p_c = 1/(3 Q)^6$, the probability of this is exponentially suppressed
  \begin{equation} \label{eq:failureprobability}
      \mathbb{P}(\mathrm{fail.}) = p_L \leq C \exp(-O(L^\eta))
  \end{equation}
  with some constants $C, \eta >0$. 
  Summed over all noise realizations, this yields Eq.~\ref{eq:logical-failure} and proves the theorem.
\end{proof}

\subsection{Fault tolerance of the full logical protocol}\label{sec:fault-tolerance}

So far, we have only proven that we can fault-tolerantly store quantum information in the ground-state space of the TQD for an exponentially long time in the absence of domain walls and boundaries as well as spacetime regions with Abelian phases that are generally present in all protocols. 
We now discuss how our arguments can be extended to show fault-tolerance of the full logical protocols proposed in this paper, which must include dealing with spatial boundaries, domain walls, and the regions with Abelian phases at the beginning and the end of the protocol.
In the rest of this section, we present arguments that support our statements and leave rigorous proofs to future work.

\subsubsection{Abelian phase in the protocols}

We first consider the simpler case of protocols where the domain wall between the TQD phase and the Abelian toric codes phase is time-perpendicular (i.e.~a switching occurs in one round in the whole space; note that for every protocol in Sec.~\ref{section:loopsum}, we presented a variation which includes only such domain walls), such as in $\overline{CZ}$ logical measurement protocols. We consider the flux sector only here and discuss the charge sector at the end of this subsection. 

In these protocols, the TQD state is prepared from the toric code state, which is also realized by a noisy circuit. We first present an informal argument showing that the earlier considerations for the absence of logical failure extend to this case, if the just-in-time decoding in the TQD phase is preceded and followed by global RG decoding for the Abelian phase. 
Right before switching to the TQD phase, we perform the usual global RG decoding for the Abelian phase and apply the correction (for the RG decoder to be successful, it might be necessary for the Abelian phase to have existed for $O(d)$ time where $d$ is the spatial distance of the code). Next, the TQD phase is realized by a circuit together with the just-in-time decoder. The residual noise (which will be there in the presence of measurement errors) can be seen by the just-in-time decoder as having started at the beginning of the TQD phase.  After that, the protocol again switches to the Abelian phase, where we can simply continue recording measurement outcomes without applying corrections. At the end, we use the information about the unpaired syndromes left by the just-in-time decoder at the end of the TQD phase together with the syndromes of the new noise in the Abelian phase following the TQD one to perform another round of the usual global RG decoding. With these modifications, the flux and charge worldlines associated with noise and corrections throughout the entire circuit continue being sparsely clustered.

\subsubsection{Domain walls and boundaries}

In other examples, such as unitary $\overline{CCZ}$ and $\overline{T}$ protocols, the regions with the Abelian phase are connected to the TQD region via domain walls that are tilted in spacetime.  In this case, starting from the point where such domain walls appear, we use a (modified) just-in-time decoding in both the Abelian and non-Abelian regions. Namely, the just-in-time protocol now must take the presence of the domain wall into account by appropriately matching the fluxes of specific colors to it. For different flux colors, the decoding must work differently. For an $\langle rp,g\rangle$ domain wall, the decoding for the $r$ fluxes in the TQD and the $p$ fluxes in one of the toric code copies must be performed together (i.e.~viewed as the same type of flux by the decoder) ``ignoring'' the domain wall. Thus, for these sectors the just-in-time decoding works as if the domain wall was absent. On the other hand, the green $g$ fluxes can terminate on the $\langle rp,g\rangle$  domain wall. Therefore, for the green flux sector, the domain wall can be treated as a boundary that is tilted in spacetime. We now explain how the just-in-time decoding must be modified in one flux sector to appropriately account for a boundary where one type of fluxes can terminate.

Consider a TQD circuit with a flux-terminating boundary (domain wall). We call the set of points in spacetime corresponding to the locations on the boundary $\mathcal{Y}$. The spatial coordinates of the boundary at time $t$ are denoted $\mathcal{Y}_t$. For a correction string $C_t$ terminating at the boundary, $\partial C_t\big|_{\mathcal Y}$ is always set to be empty.  We show a modified version of the just-in-time RG decoder which now matches points to the boundary (or a domain wall) when it is appropriate:

\begin{algorithm}[H]
  \caption{Just-in-time RG decoder in the presence of a boundary or a domain wall $\mathcal{Y}$}
  \label{alg:boundary}
  Initialize empty $\Sigma_0$. At each timestep $t\geq 0$
   \begin{algorithmic}[1]
   \State Determine $\mathcal M_t$ from  $m_t$ and  $ C_{t-1}$. 
   \State Update $\Sigma_t \rightarrow \Sigma_{t}\oplus \mathcal{M}_t$;
   \State Initialize $C_{t}$ to be empty and the set $s = \Sigma_t$.
   \State  For each $0 \leq k \leq k_m$:
   \begin{itemize}
       \item[3.1:] remove all points $u = (u_1,u_2,u_t)$ in $s$, s.t. $t - u_t < 2^k$; 
       \item [3.2:] remove each $u \in s$  from $s$ and from $\Sigma_t$,  s.t. $d(u,\mathcal{Y}) \leq 2^k$, if the current coordinate $u$ is outside the spatial domain of the TQD at time $t$;
       \item [3.3:] for any remaining $u \in s$ s.t. $d(u,\mathcal{Y}) \leq 2^k$, find any shortest path for the Pauli correction $C^{(k),i}_{t}$ matching the spatial coordinate of $u$ with the boundary at the current time $\mathcal{Y}_t$  and remove $u$ from $s$ and from $\Sigma_t$;
       \item [3.4:] for each $u \in s$, if  $u' \in s$ exists s.t. $d(u,u') \leq 2^k$, find any shortest path for the Pauli correction $C^{(k),i}_{t}$ matching their spatial coordinates. Update $C_{t} \rightarrow C_{t} \cdot C^{(k),i}_{t}$ and remove $u, u'$ from $s$ and from $\Sigma_t$;
   \end{itemize}
   \State Pass the correction string $C_{t}$ to the circuit;
   \end{algorithmic}
\end{algorithm}
The modification consists of adding steps 3.2 and 3.3. According to 3.2, the noise clusters that are close to the boundary can be ignored if the phase in their spatial support becomes trivial before the correction needs to be applied. Otherwise, according to step 3.3,  the points are preferentially matched to the boundary as long as the cluster is close enough to the boundary in spacetime. However, they are specifically matched to the \emph{current} location of the boundary since the latter may be moving. 

The argument for the existence of the threshold in the presence of a boundary is analogous to the case without boundaries as long as the boundary moves at a slow enough speed. Assume that its speed is $\chi$ per timestep. Consider any cluster $F_{n,i}$ at level $n$. If it is located in the proximity of the boundary in spacetime, in the worst-case scenario the boundary moves away from the cluster by $\chi  (Q^n + 2Q^n + 5)$ in the (upper bounded) time it takes the decoder to make a decision to match.  Therefore, now the padding region can increase by this amount, i.e.~from $2 Q^n + 5$ to $f(Q) \equiv (1+\chi) (2 Q^n + 5) + \chi Q^n $. This effectively changes $\alpha$ in the fault tolerance proof, and thus, the rest of the numerical parameters (which depend on $\alpha$, even though we kept this dependence implicit in our proof) but nothing else. 

Finally, other scenarios involving boundaries include situations when a certain color of fluxes cannot terminate at the boundary (in which case the just-in-time protocol is unchanged and simply does not match to the boundary), or two colors of fluxes can terminate only simultaneously (an example being $\langle rg \rangle $ boundary). In this case,  syndromes corresponding to fluxes of two colors have to be considered simultaneously by the decoder. For syndromes of two different colors $u,v$, we can introduce a boundary-aware distance metric which is computed from the shortest path that connects these two syndromes through the boundaries. If the two such syndromes have to be matched, this is done by a correction string that appropriately changes color at this boundary. 

\subsubsection{Putting it all together}

The fault tolerance of any full protocol comes from combining all the aspects discussed above as well as a few additional points that we address here.

First, we discuss what happens in the charge sector. There, we, as usual, record the measurement outcomes until the end of the protocol. Afterwards, we run the usual RG decoder that must be modified to take into account the possibility of matching to the boundaries and domain walls. Similarly to the case with the trivial protocol, for large enough $Q$, the regions to be corrected form a cluster decomposition such that the RG decoder must succeed.

Another point is that the logical error (see Eq.~\ref{eq:logical-failure}) in actual protocols is suppressed as  $(  
p/p_c)^{O(d^\eta)} $ where $d$ is now the length of the \emph{shortest nontrivial homology class of charge or flux worldlines in spacetime}. This is because we have to compare the size of the largest cluster in the cluster decomposition not only with the spatial dimensions of the manifold but also with the smallest distance between a pair of nontrivial boundaries in spacetime. Specifically, this is the distance between two points that, if they are connected by a cluster of error, affect the logical action of the protocol. Generally, for the protocols in Sec.~\ref{section:loopsum}, all characteristic dimensions in spacetime geometry have to be  $O(L)$, which then gives $d = O(L)$.

Additionally, in some of the protocols (namely, the logical $\overline{CZ}$ and $\overline{XS}$ measurements), the measurement outcomes themselves are used to detect nontrivial homology classes associated with one type of the charge worldlines, such as $C_g$. A protocol experiences a logical failure if we cannot determine the homology class reliably. However, this is the same as the failure of the error correction in the respective charge sector. Thus, the fault tolerance of global measurements (or, in other words, of determining the charge under the global symmetry that is being gauged in the protocol) is the same as fault tolerance in the respective charge sector. 

And finally, since the protocols always end and begin in an Abelian phase, the fault-tolerant initialization and readout is the same as for a standalone Abelian phase, modulo the need to perform decoding and apply corrections immediately before the implementation of the TQD phase in the protocol begins.

\subsection{MBQC-like protocols}
\label{sec:decoding-MBQC}

\begin{figure}[t]
\includegraphics[width=1\columnwidth]{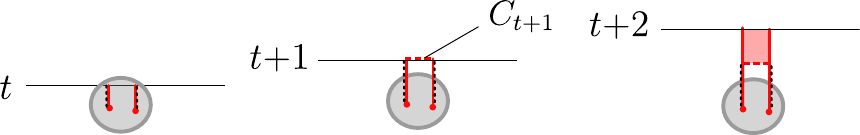}
\caption{
Alternative way of using the output $C_{t+1}$ of the just-in-time decoder (applied to a configuration of two syndrome points caused by an error region in gray) to perform corrections.
Instead of applying a Pauli-$X$ string along $C_{t+1}$, we insert the gauge-fixing membrane, as shaded in red, which is equivalent to changing the circuit such that it is conjugated by the Pauli-$X$ string associated with correction $C_{t+1}$ starting at time $t+1$ and onwards.
After this, we keep measuring $-1$ outcomes at the endpoints of $C_{t+1}$, that is, the flux worldlines continue indefinitely in time direction, however, their properties are ``mitigated'' by the gauge-fixing membrane.
}
\label{fig:MBQC0}
\end{figure}

\begin{figure*}[t]
\includegraphics[width=1\textwidth]{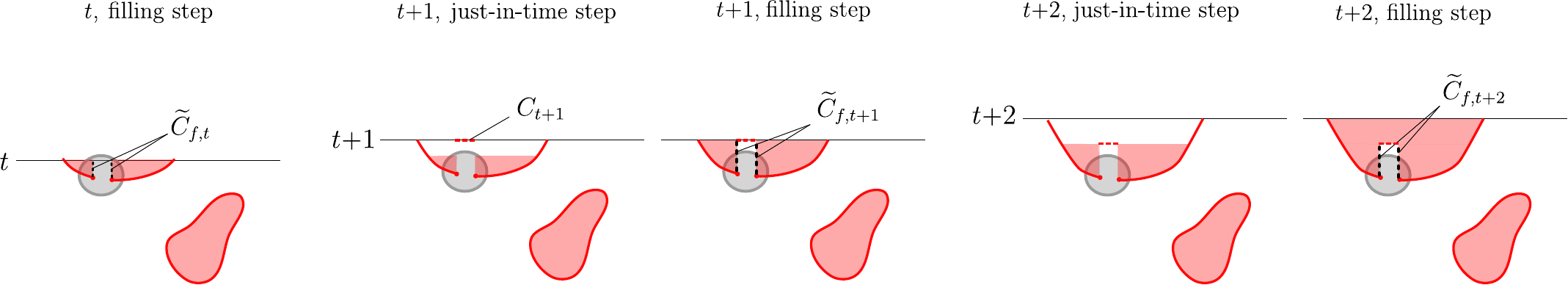}
\caption{
A schematic showing an example of the gauge-fixing membrane filling (shown as red shading) in the presence of noise (shown as a gray region).
At each time step $t$, the decoder finds a string $C_t+\widetilde C_{f,t}$ that closes off the measured (up to time $t$) flux configuration $W_f$, where $C_t$ (shown as red dashed lines) is the output of Algorithm~\ref{alg:greedy}, and $\widetilde C_{f,t}$ (shown as black dashed lines) consists of lines propagating in time direction from any new syndrome point, until it is paired with another syndrome by the decoder and closed off by a $C_t$ string.
Note that, in contrast to discussion earlier in this section, we do not actually apply Pauli-$X$ strings at $C_t$, and we do not measure $-1$ outcomes at $\widetilde C_{f,t}$.
}
\label{fig:MBQC}
\end{figure*}

In the previous subsections, we have restricted ourselves to circuits where the measured flux worldlines can only propagate in the time direction.
We now briefly discuss how to generalize this to other MBQC-like and Floquet-like circuits where flux worldlines can propagate in any direction.
Naively, this constitutes an issue since, even in the absence of noise, we measure a random and dense configuration of flux worldlines.
Since charge worldlines can terminate on the flux worldlines, the fault tolerance proof of the previous subsection does not go through.

Before we show how to solve this problem, we describe an alternative way to perform the corrections for circuits that we discussed up till now, such as in Section~\ref{sec:CZ_minimal} or \ref{square-TQD-sec}, where fluxes propagate \emph{only} in the time direction in the absence of noise.
Instead of applying a (Pauli-$X$) string operator $C_t$ that is determined at step 5 of the just-in-time decoding Algorithm~\ref{alg:greedy}, we instead commute the Pauli-$X$ correction string until the end of the protocol through all following operations in the circuit.
Conjugating the $CCZ$ gates with $X$ yields additional $CZ$ gates in the circuit, and conjugating the $Z$-type measurements at the endpoints of the string $C_t$ will flip all future measurement outcomes at these endpoints to $-1$.
In the path integral, this is the same as inserting a path-integral membrane (by design) as depicted in Fig.~\ref{fig:MBQC0}, which we refer to as inserting the \emph{gauge-fixing membrane}.

While this different way of implementing corrections may seem more complicated, it is precisely the way MBQC-like and Floquet-like protocols must work, such as these derived from considering 3+0D circuits as discussed in Sec.~\ref{sec:3d-2d} and in Refs.~\cite{bombin2018,Brown2020universal}.
In such protocols, if there is no noise, the measured flux configurations  must be all closed loops. In the absence of the noise, all we need to do is to choose any gauge-fixing membrane whose boundary coincides with the measured flux configuration.
In the circuit, the gauge fixing membrane corresponds to conjugating all $CCZ$ gates with Pauli-$X$ operators.

Now, if there is noise in the circuit, then the measured flux worldlines $W_f$ do not have to be closed and so, naively, we do not know where to insert the gauge-fixing membranes. 
In order to determine where the gauge-fixing membrane needs to be inserted, we need to find a loop-closing ``fix'' string, which we obtain in a just-in-time way as a sum of two parts, $\widetilde C_{f,t}+C_t$. The fix string has the same endpoints as $W_f$, $\partial (W_f+\widetilde{C}_{f,t}+C_t)=0$, and we choose the gauge-fixing membrane such that its boundary is given by $W_f+\widetilde{C}_{f,t}+C_t$.

We now describe how $\widetilde C_{f,t}$ and $C_t$ are obtained using the just-in-time decoder.
$C_t$ is given by the correction strings which are the output of the just-in-time decoder in Algorithm~\ref{alg:greedy} (noting that we do not apply a Pauli-$X$ string at $C_t$ here).
$\widetilde C_{f,t}$ consists of timelike flux lines that begin at every flux syndrome point unmatched by the decoder, and propagate straight in the time direction.
When the decoder pairs two syndrome points and outputs an appropriate string $C_t$, this closes off the $\widetilde C_f$ lines from the past timesteps originating from these syndrome points.
This is illustrated in Fig.~\ref{fig:MBQC}.

Finally, we note that in the case where the measured flux worldlines can only propagate in $t$ direction, $\widetilde C_f$ always coincides with the measured flux worldlines $W_f$ originating from the syndrome points, so we only have to start filling in the gauge-fixing membrane after adding the $C_t$ part of the correction, as shown in Figure~\ref{fig:MBQC0}.

\subsection{Comments on decoding logical protocols in practice}
\label{sec:practical-decoding}

So far, we considered one particular way of performing decoding that was chosen specifically to make the proof simple. This leaves vast space for improvement in terms of performance. Here, we propose a more efficient decoding scheme and point out particular areas where improvement is the easiest to achieve.

Below we propose a minimum-weight matching (MWPM) algorithm for the flux sector that is broadly inspired by Bombin's version of the just-in-time algorithm~\cite{bombin2018}. For simplicity, we show its version in the absence of boundaries or domain walls . The modification of the algorithm in the presence of domain walls and boundaries is in line with the discussion in previous subsections. In addition, in the presence of a spatial domain wall, we can use global decoding on the Abelian phase side that appropriately matches defects to the domain wall, followed by the MWPM domain wall-aware decoding on the TQD side from the domain wall.

\begin{algorithm}[H]
  \caption{Just-in-time matching decoder}
  \label{alg:matching}
  Initialize empty $\Sigma_0, \Sigma'_0$. At each timestep $t\geq 0$
   \begin{algorithmic}[1]
   \State Determine $\mathcal M_t$ from  $m_t$ and  $ C_{t-1}$. 
   \State Update $\Sigma_t \rightarrow \Sigma_{t}\oplus \mathcal{M}_t$;
    \State Initialize $C_{t}$  and $ S_t$ to be empty;
   \State Run MWPM on $\Sigma_t$ with smooth boundary at time $t$. Record the spacetime matching outcome in $S_t$;
   \State For every pair $u, u' \in \Sigma_t$ matched in spacetime by $S_t$, run another MWPM round for their spatial coordinates and add the outcome to $C_{t}$; remove $u,u'$ from $\Sigma_t$;
   \State Pass the correction $C_{t}$ to the circuit;
   \end{algorithmic}
\end{algorithm}

Similarly, the charge sector can be decoded by a boundary- and domain-wall aware MWPM at the end of the protocol. In addition, the performance can be improved by explicitly taking into account in the decoding that the charge worldlines can terminate on the flux worldlines, namely treating the measured flux worldlines $W_f$ as the heralded locations where we know the charge strings can terminate.

Apart from the choice of the decoding strategy, the actual performance of any decoder is sensitive to the particular microscopic circuit. It would be interesting to compare different circuits in Sec.~\ref{section:examples} to each other.

We remark that for numerical benchmarking of the non-Clifford protocols, there is no need to perform a full simulation of the circuit involving non-Clifford physical operations.
Given the noise with sufficiently local supports of possible channels, we can compute an effective distribution over flux and charge configurations created by the noise directly relating to measurement outcomes, which are then used by the decoder. 
Given a (sampled) fault configuration, this distribution can be efficiently computed and then sampled from in the simulation.
This computation only has to be performed once for every fault type, yielding a scalable simulation scheme.
In addition, for certain noise models, e.g. i.i.d. single-site $X$ and $Z$ noise, together with probabilistic noise on the measurement outcomes it is possible to obtain an analytic expression for the effective probability distribution of charges and fluxes.\footnote{Specifically, single-qubit Pauli-$X$ errors create a flat distribution over certain flux outcomes in its vicinity. This is in contrast to Clifford circuits where each Pauli error creates a deterministic detection event.}

\section{Dicussion}

In this work, we have shown how to perform non-Clifford gates by introducing domain walls between toric or surface codes and the non-Abelian type-III twisted quantum double phase. This offers an alternative implementation of universal fault-tolerant quantum computation in two dimensions that differs from previous proposals. Given the compatibility of our protocols with the surface code architectures, our logic gates can directly supplement current experimental efforts to realize universal quantum computation with surface codes.

We have presented a number of different specific microscopic implementations of our logical protocols. It will be valuable to further adapt these implementations to run on realistic hardware~\cite{Iqbal2024nonAbelian, Gupta2024encoding}. The path integral framework we have proposed gives us some flexibility to redesign circuits to satisfy hardware constraints. In addition, further work and numerical testing of the just-in-time decoder is essential, as well as careful evaluation of the resource cost of our proposals at a target logical error rate.

An important question that remains open is: what other fault-tolerant logic gates can be achieved by introducing domain walls between different topological phases? 
In this direction, the fault-tolerant non-Clifford gates we have introduced can be in principle generalized to all solvable non-Abelian anyon theories. 
An important related question is whether a version of just-in-time decoding can be used to achieve fault tolerance in these models in the presence of circuit-level noise. 
This presents an opportunity to design fault-tolerant analogs of recent finite-depth adaptive local unitary state preparation results for solvable anyons~\cite{bravyi2022adaptive,Tantivasadakarn2021LRE,Tantivasadakarn2022,Ren2024Efficient,Lyons2024}.

Our results have been derived through a number of different approaches. These include the topological path integral framework, code deformation, and gauging logical operators via measurements. We have also drawn the connection between 2D non-Abelian topological codes and 3D topological codes with transversal non-Clifford gates, thereby illuminating the physical mechanism behind prior proposals to realize non-Clifford logic gates in two dimensions~\cite{bombin2018, Brown2020universal}. It would be interesting if the tools we have developed can be extended to more general quantum error correcting codes, such as different classes of quantum low-density parity-check codes, to find new ways of performing non-Clifford gates on code blocks encoding a large number of logical qubits.
We hope that the present work will serve as a guiding example toward a generalized framework for a universal computation with low-overhead quantum error-correcting codes.

\bigskip
\noindent
\textbf{Note added.} 
Related work~\cite{huang2025generatinglogicalmagicstates} appeared online while we were preparing our manuscript.

\begin{acknowledgements}
We are grateful to S.~Balasubramanian, M.~Barkeshli, P.~Bonderson, J.~Bridgeman, H.~Dreyer, J.~Garre Rubio, A.~Lyons, N.~Schuch, N.~Tantivasadakarn, F.~Verstraete, J.~Wang and G.~Zhu for useful discussions. We would also like to thank M.~Kesselring for fruitful discussions and for suggesting the 3D brickwork geometry used in Section~\ref{sec:3d-2d}.

B.J.B.~thanks the Center for Quantum Devices at the University of Copenhagen for their hospitality.  M.D. was supported by the Walter Burke Institute for Theoretical Physics at Caltech.
Part of this work was done while D.J.W.~was visiting the Simons Institute for the Theory of Computing. 
D.J.W.~was supported in part by the Australian Research Council Discovery Early Career Research Award (DE220100625).
D.J.W.~is currently employed by PsiQuantum. 
A.B.~was supported by the U. S. Army Research Laboratory and the U. S. Army Research Office under contract/grant number W911NF2310255. A.B. and B.J.B. both received support from the U.S. Department of Energy, Office of Science, National Quantum Information Science Research Centers, and the Co-design Center for Quantum Advantage (C2QA) under contract number DE-SC0012704.
J.M.~is supported by the DFG (CRC 183). M.W.~is supported by the Engineering and Physical Sciences Research Council [grant number EP/W032635/1 and EP/S005021/1].
\end{acknowledgements}

\appendix

\section{Loop-sum picture for the TQD model} \label{subsec:loopsum-and-stabs}

Here, we review the \emph{loop sum picture} for the TQD ground states, which is useful for developing a qualitative understanding of our protocols. The loop sum picture presented here can be understood as a continuum version of the lattice model, which nevertheless admits a direct relation to the microscopic realization, where configurations of qubits of any given color that are in the $\ket{1}$ state form closed loops on the lattice (here we refer to states that appear in an expansion of the ground state in the computational basis). 
We start from the well-known description of the codestates of the standard toric code phase in the loop sum picture~\cite{Kitaev1997,Levin2004}.

\begin{figure}[b]
\includegraphics[width=0.9 \columnwidth]{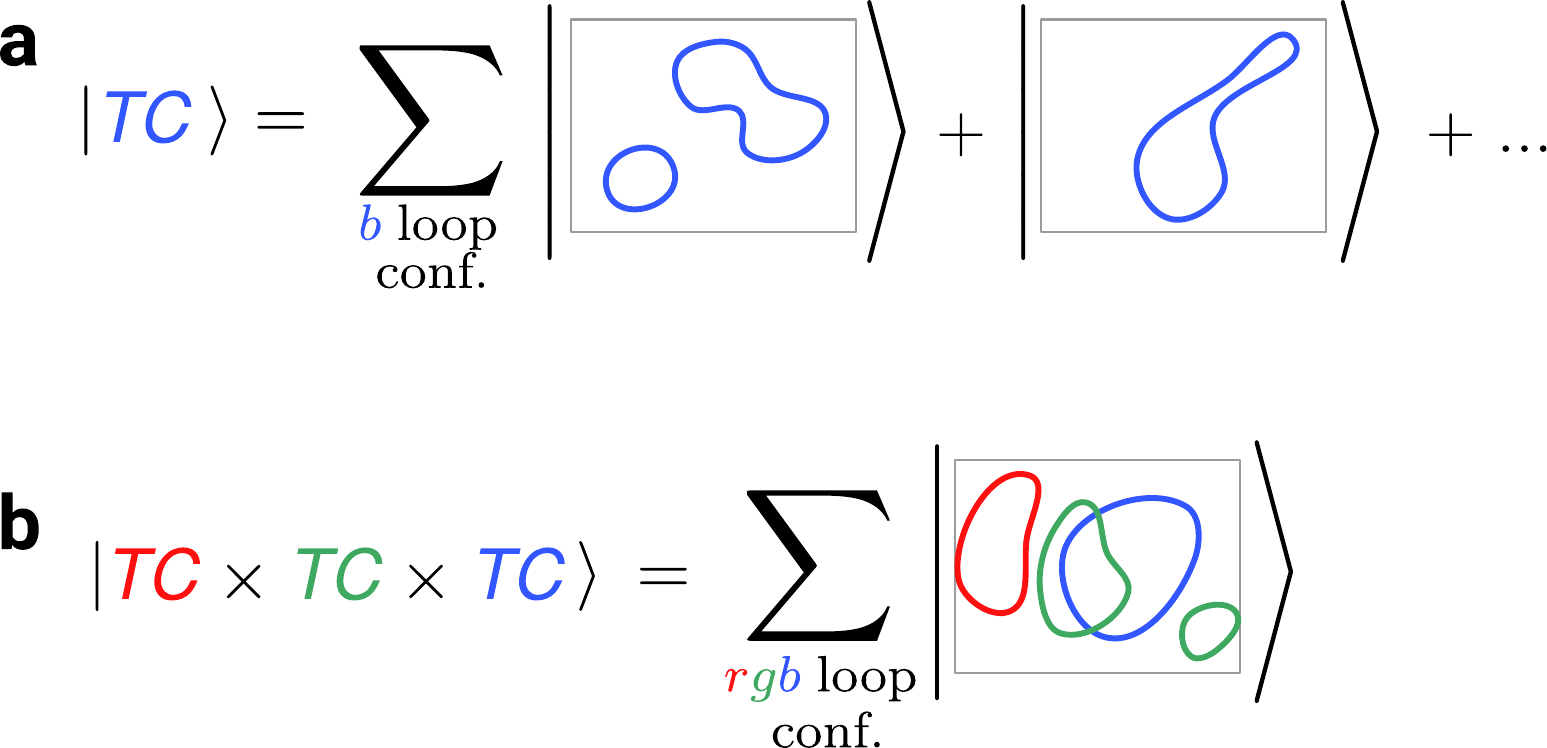}
\caption{Schematic loop sum picture for a ground state of (a) a single toric code and (b) three copies of the toric code labeled red, green, and blue (denoted `TC' in the appropriate color). The three copies of toric code do not have to be supported on the same microscopic lattice, but can be defined on three different superimposed lattices. }
\label{fig:TC_loops}
\end{figure}

The toric code can be defined on an arbitrary two-dimensional cellulation with one qubit on each edge.
A configuration of qubits in the computational basis corresponds to marking all qubits on the Poincar\'e dual cellulation that are in the $\ket{1}$ state (we work with the Poincar\'e dual cellulation unless specified otherwise).
The plaquette stabilizers of the toric code ensure that the codestates have an even number of $\ket{1}$-state qubits around every plaquette, and thus, only closed-loop configurations appear in the expansion of the codestate.  The vertex stabilizers ensure that all possible closed-loop configurations appear with equal amplitudes in the expansion of the codestate. More accurately, all possible closed-loop configurations \emph{within the same cohomology class} appear with equal amplitudes in the expansion of the codestate. 
Thus, the codestates of the toric code in the logical-$Z$ basis are equal-weight superpositions of closed-loop configurations,  as shown in Fig.~\ref{fig:TC_loops}, and otherwise are linear combinations thereof.

The non-Abelian model introduced in Section~\ref{sec:CZ_minimal} is similar to three copies of the toric code defined on three superimposed triangular lattices shown in Fig.~\ref{fig:TC_loops} with the difference that the states do not form an equal-weight superposition (here, by weight we mean the coefficient in front of the state). 
Instead, due to the ``twist", there is a relative phase between any two configurations~\cite{Dijkgraaf_1990,Hu_2013}, as shown in Fig.~\ref{fig:TQD_loops}. If two configurations differ by moving a string of one color over an intersection of strings of the other two colors, their amplitudes differ by a factor of $-1$, as shown in Fig.~\ref{fig:TQD_loops}(b).

\begin{figure}[t]
\includegraphics[width=1.0 \columnwidth]{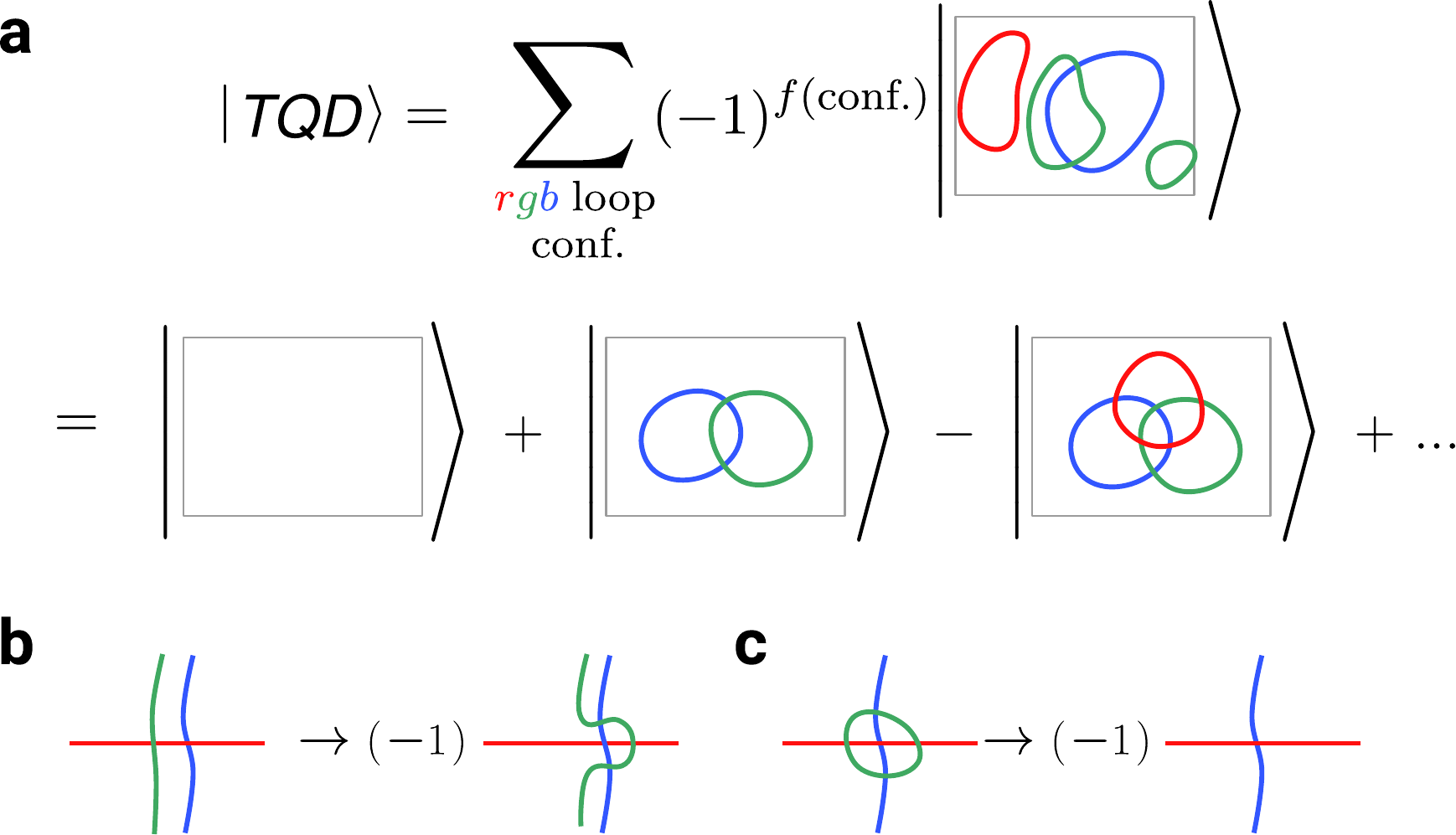}
\caption{(a) Schematic loop sum picture for one of the ground states of the type-III twisted quantum double model (denoted `TQD'). 
It is a superposition of all three-color loop configurations with appropriate complex phases determined by the type-III 3-cocycle function $\omega_{III} (\mathrm{conf.}) = (-1)^{f(\mathrm{conf.})}$. 
The rule that determines the relative phase of the configurations entering the superposition is shown in panel (b). }
\label{fig:TQD_loops}
\end{figure}

We show three superimposed triangular lattices in Fig.~\ref{fig:dual_stab}(a) as well as the Poincar\'e dual of the blue sublattice. In panel (b), we show a closed-loop configuration of three colors.
In Fig.~\ref{fig:dual_stab}(b), we illustrate how the lattice model of Section~\ref{sec:CZ_minimal} corresponds to the closed loop picture.
The stabilizers can be derived from the condition that they stabilize the code state described by the superposition above. 
Similar to the toric code, the plaquette stabilizers ensure that the configurations appearing in the codestate are closed-loop configurations. However, to enforce that the relative amplitudes are changed by the necessary factors, the vertex stabilizers must be modified. 
We now verify that the choice shown in Fig.~\ref{fig:dual_stab}(c) (which are the stabilizers from Sec.~\ref{sec:CZ_minimal} shown on the dual lattice) produces the necessary phase factors. 
Consider the vertex stabilizer in Eq.~\eqref{eq:TQD-star}, where the color $c$ is chosen to be blue without loss of generality. 
When this operator is applied, the Pauli $X$ component of the stabilizer adds a small blue loop to the closed-loop configuration. 
The $CZ$ component of the operator evaluated on a state in the computational basis yields $(-1)^n$ where $n$ is the number of green-red intersections inside the associated blue loop.
An example of this is shown in Fig.~\ref{fig:dual_stab}(d), where adding a small blue loop around a single red-green intersection yields a factor of $(-1)$. This indeed verifies the correct loop statistics.

\begin{figure}[t]
\includegraphics[width = 0.87\columnwidth]{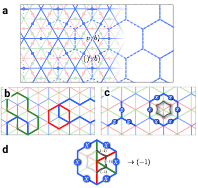}
\caption{ (a) The primal lattice (left) used to define the type-III TQD model and the Poincar\'e dual of the blue sublattice (right).
(b) A three-colored loop configuration shown on the dual lattice.
(c) Stabilizers of the TQD shown on the dual lattice.
(d) A blue vertex stabilizer on the dual lattice showing the $(-1)$ factor due to a red-green loop crossing within its support. \label{fig:dual_stab}}
\end{figure}

\section{Elementary excitations of the TQD model}
\label{sec:excitations}

In this appendix, we provide an informal overview of the excitations of the type-III twisted quantum double which is the non-Abelian model featured in this work.  In the following Appendix, we show how to obtain these excitations from gauging a $\zz_2$ symmetry of a pair of toric codes and discuss more of their properties. For a more in-depth discussion, see Refs.~\cite{dijkgraaf1991quasi,Propitius_1995,Iqbal2024nonAbelian}.

The excitations of the type-III cocycle twisted TQD (which is isomorphic to the $D_4$ quantum double~\cite{Propitius_1995}) can be labeled similarly to three uncoupled copies of the toric code (which we associate with red, green, and blue colors). The generating set is also given by three charges and three fluxes, except that the fluxes are now non-Abelian. There are a total of 22 anyons in this model. 

The electric charges $e_{r,g,b}$ are independent Abelian $\mathbb{Z}_2$ bosons. Together, they generate 8 Abelian superselection sectors, namely $\{1, e_r,e_g,e_b, e_r e_g, e_r e_b, e_g e_b, e_r e_g e_b\}$.  The flux excitations are non-Abelian bosons of quantum dimension $d = 2$, defining 6 sectors $\{ m_r, m_g, m_b, m_{rg}, m_{rb}, m_{gb}\}$. $m_r$ braids nontrivially with $e_r$ (with a phase $-1$) but also with other fluxes.

There are also 6 fermions $\{ f_r, f_g, f_b, f_{rg}, f_{rb}, f_{gb}\}$ that are fusion products as $m_r \times e_r = f_r$ and $m_{rg} \times e_r = m_{rg} \times e_g = f_{rg} $. Finally, there is a semion $s$ (which is contained in the fusion product of three fluxes of different colors) and an anti-semion $\overline{s}$. 

Below are a pair of representative fusion rules for the non-Abelian bosons:
\begin{equation}
    \begin{split}
        &m_r \times e_{g,b} = m_r, \\
        &m_r \times m_r = 1 + e_g + e_b + e_g e_b,
    \end{split}
\end{equation}
and similarly for other colors. The vacuum fusion channel for a pair of, say, $m_r$ fluxes changes to $e_g$ upon braiding with an $m_b$ flux.

\begin{figure}[b]
    \centering
    \includegraphics[width=1\columnwidth]{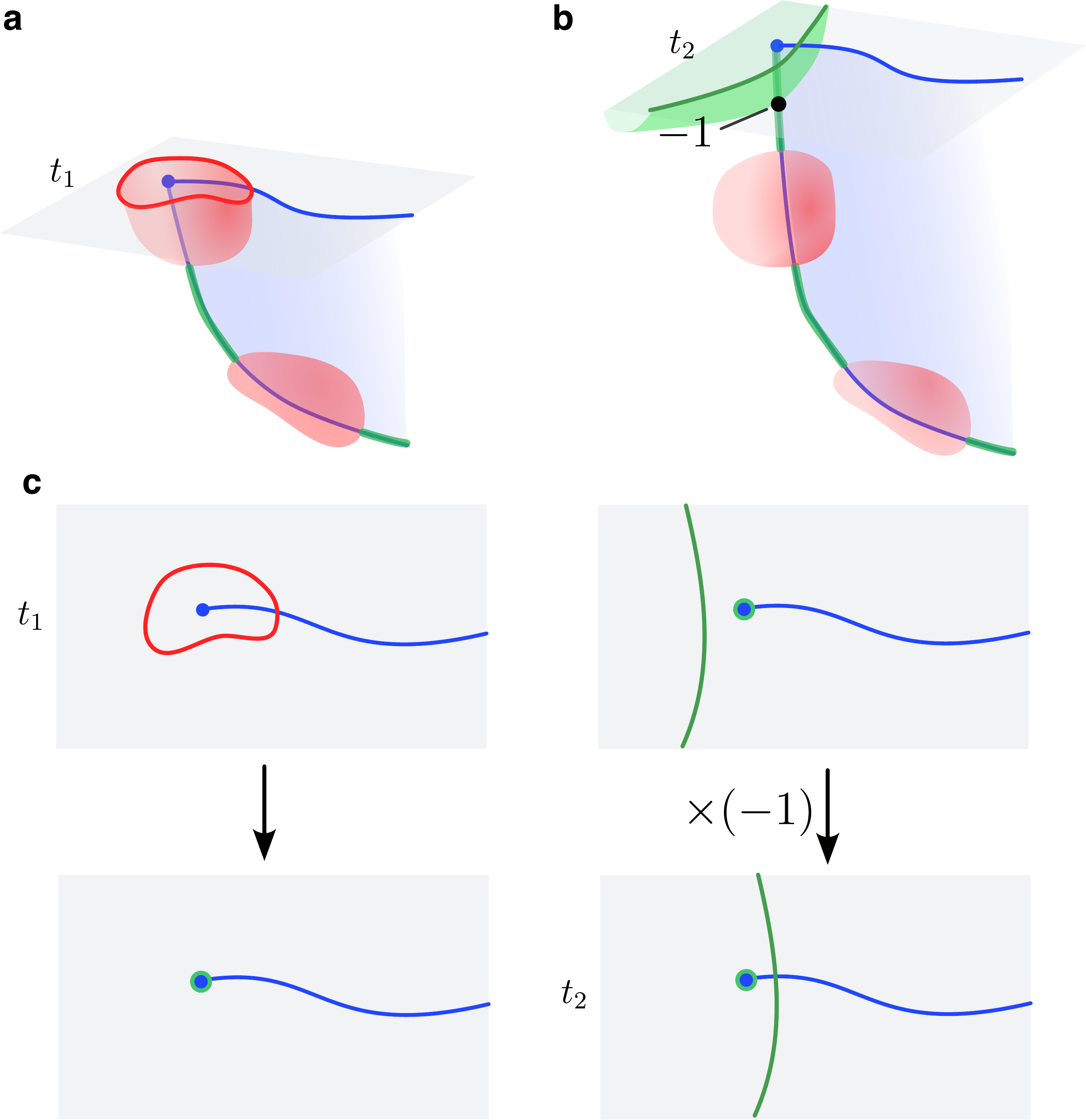}
    \caption{(a) A section taken at time $t_1$ showing a blue $m_b$ anyon worldline in spacetime with an open membrane terminating on it. At intersection points with closed red membranes, additional green strings end that produce $-1$ phases when intersecting with green membranes as shown in (b). The spacetime region in panel (b) displays the same configuration until a later time $t_2 > t_1$. Panel (c) shows 4 time slices between $t_1$ and $t_2$ of the spacetime configuration in (b).  }
    \label{fig:actual-anyons}
\end{figure}

In a circuit realization of topological phases by measurements, if all the measurement outcomes are $+1$, the state is excitation-free. If some of the measurements are violated, an excited state of the model is realized. However, generically, these excitations are not in one-to-one correspondence with specific anyons -- instead, an excitation can be a superposition of several kinds of anyons in the model. In the models from Sec.~\ref{section:examples}, the Clifford (vertex) stabilizers can be associated with $e$ anyons (charges) which are Abelian. However, the states that violate Pauli (usually, $Z$-type)  plaquette stabilizers correspond to superpositions of non-Abelian anyons.  Therefore, we call the violated plaquette measurements  ``flux defects'' (as opposed to ``excitations'' or ``anyons'') which we use in the path integral framework as well as analysis of fault tolerance.

In the closed-loop path integral picture, charge defects are in one-to-one correspondence with Abelian anyons, while flux defects correspond to the loop-shaped termination lines 
of open membranes. However, with some changes we can draw a connection between the flux defects and the non-Abelian anyon excitations. A depiction of a non-Abelian anyon excitation in the path integral picture is shown in Fig.~\ref{fig:actual-anyons}. Panel (a) shows a blue $m_b$ anyon worldline in spacetime that goes up to time $t_1$ and the associated (deformable) blue membrane. There is a green string along the worldline that behaves as a \emph{fluctuating charge} of the third (green) color. This string has the same consequence as the green charge worldline: it introduces a $-1$ weight for each intersection with a green membrane. This is shown in Fig.~\ref{fig:actual-anyons}(b) which depicts a time slice at a later time $t_2 > t_1$. However, the endpoints and the number of such strings are not fixed in the sum over all closed configurations of red membranes in the path integral, which results in different string distributions. The same property holds for $m_b$ excitation worldline and green closed membranes upon exchanging red and green colors. These properties explain the fusion and braiding statistics of the non-Abelian fluxes in the twisted quantum double.

The presence of an open blue membrane can make the path integral non-gauge-invariant. However, for the $m_b$ excitation worldline, this is fixed thanks to the decoration described above. The $-1$ phase due to a possible triple intersection with an open blue membrane is always canceled by the extra $-1$ factor coming from intersections between closed membranes and the strings along the blue anyon worldline, resulting in no overall ambiguity.  

We refer to the Appendix of Ref.~\cite{Iqbal2024nonAbelian} for an interpretation of non-Abelian excitations via gauging one-dimensional SPT states.

\section{Gauging Clifford symmetries}\label{sec:gauging_CZ}

In this section, we first review the algebraic description of gauging an Abelian symmetry of an Abelian anyon theory. 
In general, this can lead to a gauged theory that includes non-Abelian anyons. 
We show explicitly how gauging a $\zz_2$ symmetry on two copies of the toric code can produce $D_4$ anyons, which is equivalent to the anyon model of the type-III non-Abelian TQD. 
We conclude the example by showing how to identify the twisted quantum double model introduced in Sec.\,\ref{sec:CZ_minimal} as a gauged model on the level of the stabilizers that define the code.
This highlights the physics underlying the existing protocols for non-Clifford logical operations on 2D topological codes and enables the design of new procedures.

\subsection{Gauging Abelian symmetries in Abelian anyon models}

The anyon theory realized by a Pauli stabilizer model is always Abelian \cite{Bombin_2012, Haah2021Classification,Ellison2022Pauli}.
The anyon labels themselves form an Abelian group $A$ where the fusion of the anyons is given by the group multiplication.
The braiding is defined by the \textit{topological spin} $\theta: A \to U(1)$ that assigns a complex phase to each anyon type and captures how the wavefunction of a state with an anyon transforms under a $2\pi$ rotation of that anyon.
The trivial anyon, labeled by the identity element $1_A\in A$, has topological spin $\theta_{1_A} = 1$.
For any $a$, if $\theta_a = 1$ we say $a$ is a boson, and if $\theta_a = -1$ we say $a$ is a fermion.

We consider an Abelian symmetry (automorphism) of an Abelian anyon model that acts like some group $G$.
For each group element $g\in G$, there is a map $\rho_g: A\to A$ such that $\rho_g(1_A)=1_A$ and $\rho_g(\rho_h(a))= \rho_{gh}(a)\;\forall {g,h\in G,}a\in A$.
Moreover, an anyon automorphism has to preserve the braiding data, i.e.~${\theta_{\rho_g(a)} = \theta_a } \;\forall g\in G,a\in A$ and be compatible with the fusion, $\rho_g(ab) = \rho_g(a)\rho_g(b)\;\forall g\in G; a,b\in A$.

\begin{figure}
    \centering
    \includegraphics[width=0.8\columnwidth]{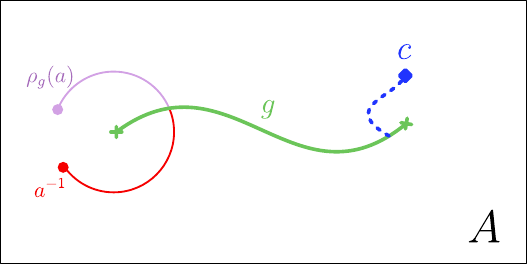}
    \caption{ To any symmetry operation $g\in\Aut(A)$ in an Abelian anyon theory $A$ there is an associated domain wall (green).
    Terminating this domain wall creates twist defects at the endpoints, indicated by crosses. Close to the twist defects, anyons that where previously inequivelant can become identified.
    On the left, we show how the anyon $\rho_g(a)$ can be created together with a $a^{-1}$ anyon close to the defect by applying an operator that crosses the domain wall.
    Without the defect this would not be possible if $\rho_g(a)\neq a$.
    On the right, we show how an individual anyon can be created close to the defect with a local operator if it is the fusion product of $a\times b = c$, where $\rho_g(b)=a^{-1}$.
    This implies that in the vicinity of the twist the total charge of $c$ (possibly fermions) is not conserved, i.e.~they can condense.}
    \label{fig:twist-defects}
\end{figure}

Such a symmetry can be gauged. Before gauging $G$, we need to incorporate the \textit{twist defects} associated to the symmetry $G$ into the theory.
In the literature, this procedure is referred to as constructing the ``$G$-extension of $A$''~\cite{Barkeshli2019symmetry}.
To each group element $g\in G$ we can associate a twist.
In the stabilizer model, twist defects can be introduced at endpoints of an (invertible) domain wall associated with the corresponding group element, see Fig.~\ref{fig:twist-defects}.
The twist defect can \textit{condense} some of the anyons. Namely, for every fermion $c$ of the form $c=a\times b$ with $\rho_g(b) = a^{-1}$ we can create a single $c$ anyon in the vicinity of the twist defect with a local operator, see Fig.~\ref{fig:twist-defects}.
For example, the twist defect in the toric code that exchanges $e$ and $m$ anyons leaves their composite, the fermion $f$, invariant.
In addition, we can apply a local operator close to the twist that creates a single fermion, which is not possible without the twist.

We label the subgroup of anyons that can condense at a $g$-twist by $C_g$.
In the case where all anyons are Abelian the \textit{quantum dimension} of the twist defect $g$ is equal to~\cite{Barkeshli2019symmetry}
\begin{align}\label{eq:quantum_dimension_twist}
    d_g = \sqrt{\abs{C_g}}.
\end{align}
Since $C_g$ is a subgroup of $A$, we can take the quotient $A/C_g$.
This groups the anyons of the original ungauged model $A$ into cosets of $C_g$.
To simplify the notation, we label a coset $aC_g\in A/C_g$ by $a_g$.

Close to the twist defect, anyon $a$ can be transformed into $\rho_g(a)$ by moving it once around the twist.
Clearly, this can be done by applying a local operator and hence, this makes $a$ and $\rho_g(a)$ topologically indistinguishable in the vicinity of the twist.
This decomposes $A/C_g$ into $G$-orbits $[a_g] = \{\rho_x(a)C_g\;|\; x\in G\}$.
These orbits are the (isomorphism classes of) simple objects in the structure that described the symmetric model before gauging.
Similar to anyons in the parent model, before gauging the symmetry, the simple objects can be equipped with a ($G$-graded) \textit{fusion} and \textit{braiding}.
This depends on the details of how the symmetry acts on the anyons of the model.
In the general case, this additional data can lead to an \textit{anomaly}~\cite{Barkeshli2019symmetry} that prevents gauging the symmetry.
However, in the following, we assume that the symmetry is anomaly-free, as it is for the symmetries that we gauge in this paper. 
More specifically, the anomaly-free condition is guaranteed when the symmetry action can be represented by an invertible domain wall within a lattice model realizing the anyon theory, for example by applying a transversal gate to parts of the model.\footnote{Thus, at a microscopic level, this is true for any logical gate that has a (non-anomalous) finite depth local unitary implementation on the associated code. }

Given the data described above, we can construct the labels for the anyons of the gauged model in three steps:
\begin{enumerate}
    \item Decompose $A_g \coloneqq A/C_g$ into $G$-orbits $\{[a_g]\;|\; a_g\in A_g\}$.
    \item For each $G$-orbit $[a_g]$ identify the \textit{stabilizer group} $S_{a_g} = \{h\in G\;|\; \rho_h(a_g)=a_g\}\leq G$ of one of the cosets in the orbit $a_g\in [a_g]$.
    \item Find the irreducible representations of $S_{a_g}$ for each orbit $[a_g]$.
\end{enumerate}
Given these structures, the anyons in the gauged model are labeled by a triple
\begin{align}
    (g,[a_g], \Gamma_a),
\end{align}
where $g\in G$, $[a_g]\subseteq A_g$ a $G$-orbit, and $\Gamma_a$ an irreducible representation of $S_{a_g}$.
Note that since $S_{a_g}$ is Abelian, $\Gamma_a$ can be viewed as group elements of $S_{a_g}$.

To complete the data of the gauged anyon model, we have to also determine the fusion and braiding data.
Both are directly inherited from the data of the symmetry group and associated twist defects.
The quantum dimensions are given by the formula
\begin{align}\label{eq:quantum_dimension_gauged}
    d_{(g,[a_g], \Gamma_a)} = d_g \abs{[a_g]},
\end{align}
where $d_g=\sqrt{\abs{C_g}}$ is the quantum dimension of a $g$-twist in the ungauged theory.
Note that this simple formula only holds for Abelian $G$. For non-Abelian $G$ additional factors appear, related to the conjugacy classes that define the anyon as well as to higher-dimensional irreps of the stabilizer groups $S_{a_g}$.

\begin{figure}
    \centering
    \includegraphics[width=0.65\columnwidth]{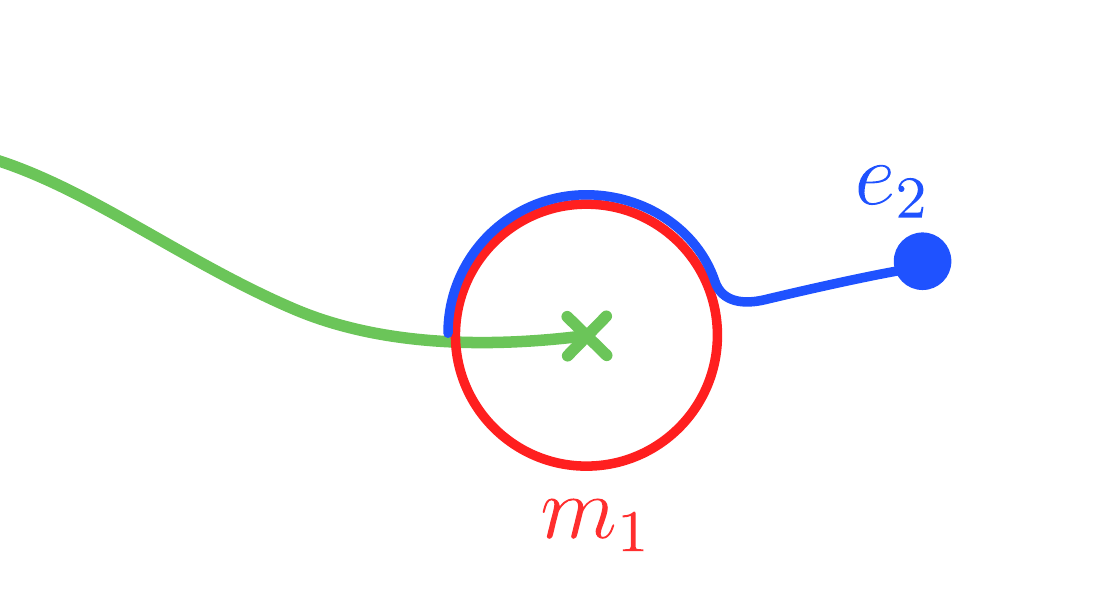}
    \caption{ Abstract depiction of the twist defect associated to the the $\zz_2$ symmetry in two layers of toric code (shown in green).
    We illustrate an anyon from the subgroup of Abelian anyons $C_1=\langle e_1,e_2\rangle$ that can condense at the twist defect.
    There exists a local string operator that creates a single $e_2$ anyon, indicated with the blue dot.
    The operator creating the anyon can be interpreted as creating a pair of $m_1$ anyons (shown in red) and moving one of them over the domain wall attached to the twist defect which acts on as $m_1\mapsto m_1e_2$.
    Fusing it again with the first $m_1$ anyon (closing of the red circle) leaves a single $e_2$ anyon behind.
    }
    \label{fig:Z2-twist-defect-condensation}
\end{figure}

Next, we discuss the topological spins of the anyons in the gauged theory.
They are built from the topological spin of the $g$-twist paired with an anyon in the $[a_g]$ orbit in the ungauged theory\footnote{Note that this quantity depends on the \textit{symmetry fractionalization} \cite{Barkeshli2019symmetry, Tarantino2016} pattern in the ungauged, symmetric, anyon theory. We do not aim to give a complete introduction here but merely show how to, given the complete data associated to the symmetry action $\rho$ within the ungauged theory, obtain the algebraic data describing the gauged anyon theory.}, $\theta_{g, a_g}$ and the irrep $\Gamma_a$.
In particular, we can consider the group character function $\chi_{\Gamma_a}: S_{a_g} \to U(1), g\mapsto \Tr(\Gamma_a(g))$.
The topological spin of the related anyon in the gauged theory is
\begin{align}
    \theta_{(g,[a_g], \Gamma_a)} = \theta_{g, a_g} \chi_{\Gamma_a}(g).
\end{align}
The remaining data of the gauged model, such as the $S$-matrix and the fusion data, are inherited from the data of the $G$-extension of the ungauged model and depend on the details of how $G$ acts on $A$.
In the anomaly-free case this can always be lifted unambiguously to the anyons in the gauged model and the procedure outlined above indeed constructs a proper anyon model.

Upon gauging, the anyons that have quantum dimension larger than 1 are non-Abelian (which can be determined from Eq.\,\eqref{eq:quantum_dimension_gauged}).
For a more in-depth discussion of non-Abelian groups and anyons, we refer the interested reader to Refs.\,\cite{Teo2015gauging, Tarantino2016,Barkeshli2019symmetry}.

\subsection{Non-Abelian TQD from gauging $\zz_2$ symmetry in toric codes}\label{sec:gauging_CZanyons}
Two toric codes are jointly described by the fusion group $A_{TC^{\times 2}} = \langle \{e_i,m_i\;|\; i=1,2\} \rangle\simeq \zz_{2}^{\times 4}$ with topological twist $\theta_{(\vb{e^a}, \vb{m^b})} = (-1)^{a_1b_1+a_2b_2}$, where $(\vb{e^a}, \vb{m^b}) = e_1^{a_1}e_2^{a_2}m_1^{b_1}m_2^{b_2}$ with $a_i,b_i\in \{0,1\}$ labels an anyon in $A_{TC^{\times 2}}$.

We consider a permutation symmetry of $A_{TC^{\times 2}}$ defined by its action on the generators
\begin{align}\label{eq:CZsymmetryTC}
\begin{split}
    e_i \mapsto& e_i, \  i=1,2\\
    m_1 \mapsto& m_1e_2\qcomma  m_2 \mapsto m_2e_1.
\end{split}
\end{align}
It is straightforward to check that it preserves all the topological spins of the model
and squares to the trivial permutation.
As such, it generates a $\zz_2$ symmetry of the anyon model of the two toric codes.

To obtain the gauged theory, we first construct the $G$-extension of $A_{TC^{\times 2}}$.
It has two sectors, ${A_0 = A_{TC^{\times 2}}}$ and $A_1 = A_{TC^{\times 2}}/C_1$, where $C_1$ is the subgroup of bosons that can condense at the non-trivial twist-defect associated to the $\zz_2$ symmetry.
We find that $C_1 = \langle e_1,e_2\rangle$, see Fig.~\ref{fig:Z2-twist-defect-condensation}. 
This leads to $A_1 = \{C_1,m_1C_1,m_2C_1, m_1m_2C_1\}$.
Decomposing both $A_0$ and $A_1$ into $G$-orbits leads to
\begin{subequations}
\begin{align}\label{app:Gorbits-CZsymmetry}
\begin{split}
    A_0 =& \{1\}\sqcup \{e_1\}\sqcup \{e_2\}\sqcup \{e_1e_2\}\sqcup \\
    &\{m_1,m_1e_2\}\sqcup \{m_2,e_1m_2\}\sqcup \{m_1m_2,f_1f_2\}\sqcup\\
    &\{m_1f_2,f_1m_2\}\sqcup \{f_1,f_1e_2\}\sqcup \{f_2,e_1f_2\}\qq{and}
\end{split}\\
    A_1 =& \{C_1\}\sqcup \{m_1C_1\}\sqcup \{m_2C_1\}\sqcup \{m_1m_2C_1\}.
\end{align}    
\end{subequations}
$A_0$ decomposes into 10 orbits and $A_1$ into 4 orbits.
The only possible stabilizer groups are the trivial group $\zz_1\leq \zz_2$ or the full symmetry group $\zz_2$.
The stabilizer groups label the additional ``symmetry charges'' (irreducible representations) carried by the anyons.
The 8 sets containing a single element admit the full stabilizer group and the 6 order-2 sets admit the trivial stabilizer group.
In total, this yields $2\cdot 8 + 6 = 22$ anyons in the gauged theory.
We label an anyon associated to sector $A_{0,1}$, an orbit $\{a\}$ and irrep $j= 0,1$ as $(i,[a], j)$.
The quantum dimensions are obtained by taking the product of the quantum dimension of the associated twist and the cardinality of the orbits.
The Abelian anyons, with quantum dimension 1, are
\begin{align}\label{eq:post-gauge-Abelian}
\begin{split}
    &(0, [1], 0), (0, [1], 1), \\
    &(0, [e_1], 0), (0, [e_1], 1),\\
    &(0, [e_2], 0), (0, [e_2], 1), \\
    &(0, [e_1e_2], 0), (0, [e_1e_2], 1).
\end{split}
\end{align}
They are closed under fusion and form the group $\zz_2^{\times 3}$. In color representation, they can be identified with (for example) $\{1,e_g,e_r,e_r e_g,e_b,e_b e_g,e_r e_b, e_re_be_g \}$, though this is not a unique association, see Sec.~\ref{sec:excitations}. 
All other anyons have quantum dimension 2 and hence, are non-Abelian. This is either because their orbits are of size 2 or because their associated twist defect has quantum dimension 2.

The topological spins of the 22 anyons can also be inferred directly from the data of the anyons and twists that enter an anyon of the gauged model.
We find that all Abelian anyons are bosons.
Moreover, there are three non-Abelian bosons and three non-Abelian fermions associated to anyons in the ungauged toric codes in the trivial sector $A_0$.
Three of the non-Abelian anyons associated to the sector $A_1$ are bosons 
\begin{equation}
    \begin{split}
        &(1,C_1, 0),
        \\
        &(1, m_1C_1, 0),
        \\
        &(1, m_2C_1, 0),
    \end{split}
\end{equation}
which, in the color picture, are translated into $\{m_g, m_r m_g, m_b m_g \}$, see Sec.~\ref{sec:excitations}. Three are fermions 
\begin{equation}
    \begin{split}
        &(1,C_1, 1),
        \\
        &(1, m_1C_1, 1),
        \\
        &(1, m_2C_1, 1),
    \end{split}
\end{equation}
 translated into $\{f_g, m_r f_g, m_b f_g \}$, see Sec.~\ref{sec:excitations}, and two are (anti-)semions, carrying a topological spin of $\pm i$, namely
 \begin{equation}
    \begin{split}
        &(1,m_1m_2C_1, 0),
        \\
        &(1, m_1m_2C_1, 1).
    \end{split}
\end{equation}
In the color picture, these are simply denoted as $s$ and $\overline s$.

Additionally, the anyons inherit their fusion and braiding data from the $G$-extension, i.e. the toric code anyon models together with the twist defects of the $\zz_2$ symmetry.
The fusion and braiding data of the anyons of the gauged theory coincide with the anyons of a quantum double of $\zz_2^{\times 3}$ twisted by a type-III cocycle, respectively the quantum double of the non-Abelian group $D_4$.
We refer to Refs.~\cite{dijkgraaf1991quasi,Propitius_1995,Iqbal2024nonAbelian} for a more in-depth discussion of the braiding and fusion data.

\subsection{Implementation of gauging on-site symmetries}\label{app:onsite-gauging}

We now provide a brief overview of the gauging measurement procedure for a $\mathbb{Z}_2$ on-site symmetry, following Ref.~\cite{Williamson2024Gauging}.

We consider a symmetry action of the form
\begin{align}
U = \prod_{v \in \mathcal{G}} U_v , 
\end{align}
which acts on the vertices of an oriented graph $\mathcal{G}$, equipped with a chosen $\mathbb{Z}_2$-cycle basis with elements labeled by plaquettes~$p$.
The graph should be chosen to admit a sparse cycle basis.
The fault-tolerant gauging measurement procedure is performed as follows:
\begin{enumerate}
    \item Introduce an ancilla qubit initialized in the $\ket{0}$ state on each edge of the graph $\mathcal {G}$. 
    \item Measure local star operators $A_v$, defined below, on every vertex ${v \in \mathcal{G}}$ (this condenses the domain walls associated with the symmetry and results in a deformed intermediate code). 
    \item Repeat $d$ rounds of syndrome extraction for the deformed code. 
    \item Measure out the edge qubits in the $Z$ basis (this ungauges the symmetry and brings us back to the original, un-deformed code). 
    \item Apply an Abelian correction operator $C$, defined below. 
\end{enumerate}
We now describe the deformed code. By construction this code is LDPC even if the symmetry operator that is gauged has total extensive weight.
The deformed code has stabilizer checks $A_v,B_p$, and $\{ \widetilde{s}_j\}$, which we refer to as star, plaquette, and deformed checks, respectively. 
The star checks are
\begin{align}
    A_v = U_v \prod_{e \in \delta v} X_e ,
\end{align}
where $\delta$ denotes a $\mathbb{Z}_2$ coboundary map on $\mathcal{G}$. 
The measured value of the symmetry operator $U$ is given by the product of the measurement outcomes of all vertex checks. 
The new plaquette checks are
\begin{align}
    B_p = \prod_{e \in \partial p} Z_e ,
\end{align}
where $\partial$ denotes a $\mathbb{Z}_2$ boundary map for $\mathcal{G}$. 
To describe the deformed checks $\{ \widetilde{s}_j\}$ we expand the original checks  $\{ s_j\}$ of the (non-deformed) code into a linear combination of operators with definite on-site symmetry charge $\mu$, namely: 
\begin{align}
    s_j = \sum_{\mu} c_j^\mu O_{\mu},
\end{align}
where $c_j^\mu$ are complex coefficients defining the expansion. 
The index $\mu$ runs over the possible charges with respect to each $U_v$ independently.
In fact, since $\{U_v\}_v$ are defined on vertices and generate a $\zz_2$ symmetry, and since $O_\mu$ is a local operators that commutes with the global symmetry $U$, $\mu$ can be considered a $\zz_2$-valued 0-boundary. 
Under this identification we write $\mu=d\phi$ for a $\zz_2$-valued 1-chain $\phi$, where $d$ is the boundary operator of the chain complex defined by $\mathcal{G}$.
The choice of $\phi$ is arbitrary up to adding edges in the support of $B_p$ operators. We choose a $\phi$ with minimal support. 
The deformed check is obtained by replacing each term in the above expansion as follows
\begin{align}
O_{\mu} \mapsto \widetilde{O}_\mu = O_{\mu} \prod_{e \in K} Z_e ,
\end{align}
where $K$ is the support of $\phi$. 
Hence, 
\begin{align}
    \widetilde{s}_j = \sum_{\mu} c_j^\mu \widetilde{O}_{\mu}. 
\end{align}

Upon performing $Z$ measurements on the edges of the graph $\mathcal G$ in step 4, the observed outcomes of these measurements form a $\mathbb{Z}_2$-valued 1-coboundary $h_e$, due to the $B_p$ stabilizer checks of the deformed code. 
This leads to a correction operator
\begin{align}
    C = \prod_{v \in W} X_v ,
\end{align}
where $W$ denotes the support of a $\mathbb{Z}_2$-valued 0-cochain $w$ with coboundary satisfying $(\delta w)_e = h_e$. 

Following the above procedure, gauging a Clifford symmetry of a Pauli stabilizer code can be used to perform logical non-Clifford gates\footnote{Because measurement of a Clifford operator on an appropriate starting state, in general, prepares a magic state; this can then be teleported to produce a non-Clifford unitary gate.} by switching to an intermediate code, that is stabilized by a non-Abelian Clifford stabilizer group, and back. 
This approach was used to derive the procedure in Section~\ref{sec:CZ_minimal}, which can be understood as gauging an anyon-permuting $\mathbb{Z}_2$ symmetry on two copies of the toric code. We explain this for the specific model in the next subsection. 
The gauging procedure also applies to the logical $\overline{XS}$ measurement (i.e.~$\overline{T}$-state preparation) protocol, 
see Appendix~\ref{sec:more_gates}.

\subsection{Gauging a $\zz_2$ symmetry in a microscopic model}
\label{sec:gaugingungauging}

\begin{figure}[t]
\includegraphics[width = \columnwidth]{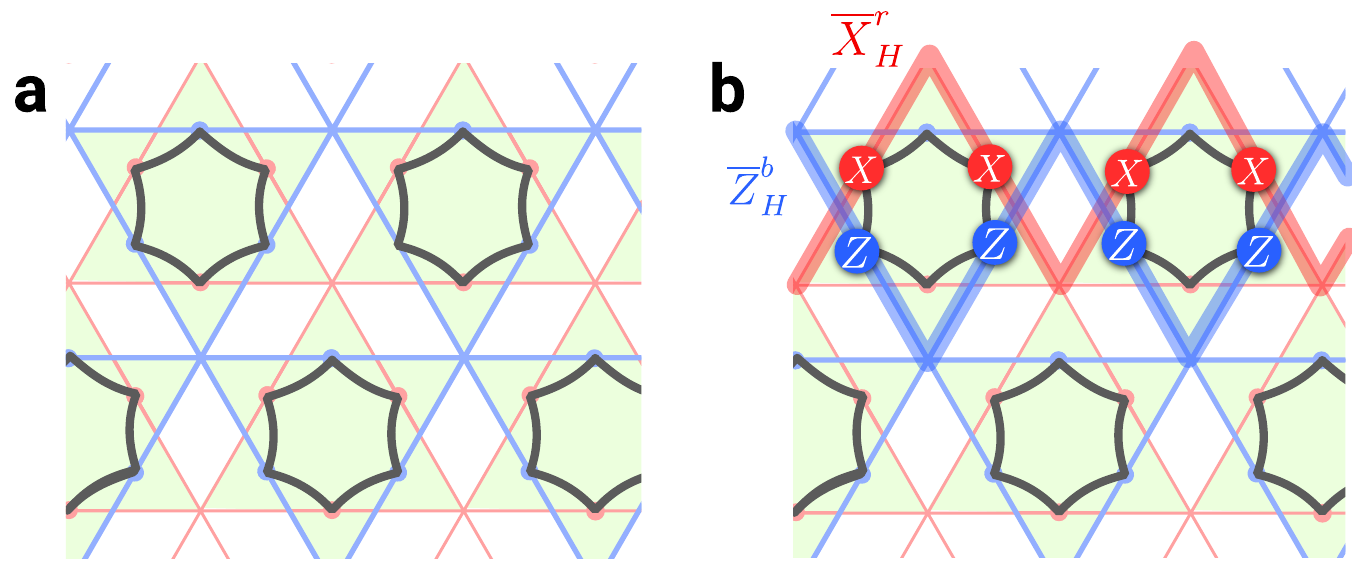}
\caption{\label{Fig:pregauge_model_triangular} (a) The $\mathbb Z_2$ symmetry of red and blue copies of the toric code. Each dark line corresponds to a $CZ$ gate acting between the red and blue qubit at its endpoints. 
(b) shows the action of the corresponding realizations of the symmetry on $\overline X_{H}^r$, giving $\overline X_{H}^r \overline Z_{H}^b$.  }
\end{figure}

In the following, we show how the Clifford stabilizer model introduced in Sec.\,\ref{sec:CZ_minimal} can be thought of as a model obtained from gauging the $CZ$ symmetry (which acts as $\zz_2$ in the way described in \ref{sec:gauging_CZanyons}) in two copies of the toric code on triangular lattices.
In doing so we illustrate how an abstract gauging operation is implemented microscopically by enforcing (projecting or measuring with either postselection or an appropriate correction) local stabilizers whose product is the symmetry operator.
In this method, the global symmetry is enforced via a combination of local symmetries. 

We start with two toric codes on two triangular lattices, shown in red and blue in Fig.\,\ref{Fig:pregauge_model_triangular} that are displaced with respect to each other.
Using the notation from Sec.\,\ref{sec:CZ_minimal}, we start with a code that is stabilized by $A_{v,c}$ and $\mathcal{B}_{f,c}$ for $c=r,b$.
This code has a global $\zz_2$ symmetry generated by products of $CZ$ unitaries acting between qubits on the red and the blue layer.
We depict this symmetry operator and its action on the logical operators in Fig.\,\ref{Fig:pregauge_model_triangular}.

Gauging is performed by adding extra degrees of freedom initialized in the $Z$ basis, e.g. in $\ket{0}$. We choose to place them on the edges of a third triangular lattice, which we label in green, in line with Sec.\,\ref{sec:CZ_minimal}.
The $CZ$ symmetry is gauged by projecting onto the +1 eigenspace $\mathcal{A}_{v,g}$ operators as shown in  Fig.\,\ref{fig:postgauge-cliffordstabs}, which are essentially the domain walls living on the green surrounding domains where the $CZ$ symmetry of the red and blue toric codes has been applied locally. The fact that these operators are stabilizers of the gauged model can be understood as a \textit{proliferation} of domain walls of the $CZ$ symmetry -- i.e. promoting this symmetry to a dynamical gauge field.

The product of all the local operators $\mathcal{A}_{v,g}$ reproduces the global symmetry operator $CZ$ corresponding to the gauged symmetry.
In Fig.\,\ref{Fig:pregauge_model_triangular}, we see that the symmetry operator naturally decomposes into products over groups of 6 qubits.
In Fig.\,\ref{fig:postgauge-cliffordstabs}, we show the local stabilizers that have to be added to the gauged model.
Note that they exactly coincide with the $\mathcal{A}_{v,g}$ Clifford stabilizers of the TQD model in Sec.\,\ref{sec:CZ_minimal}.

\begin{figure}[t]
\includegraphics[width = 1\columnwidth]{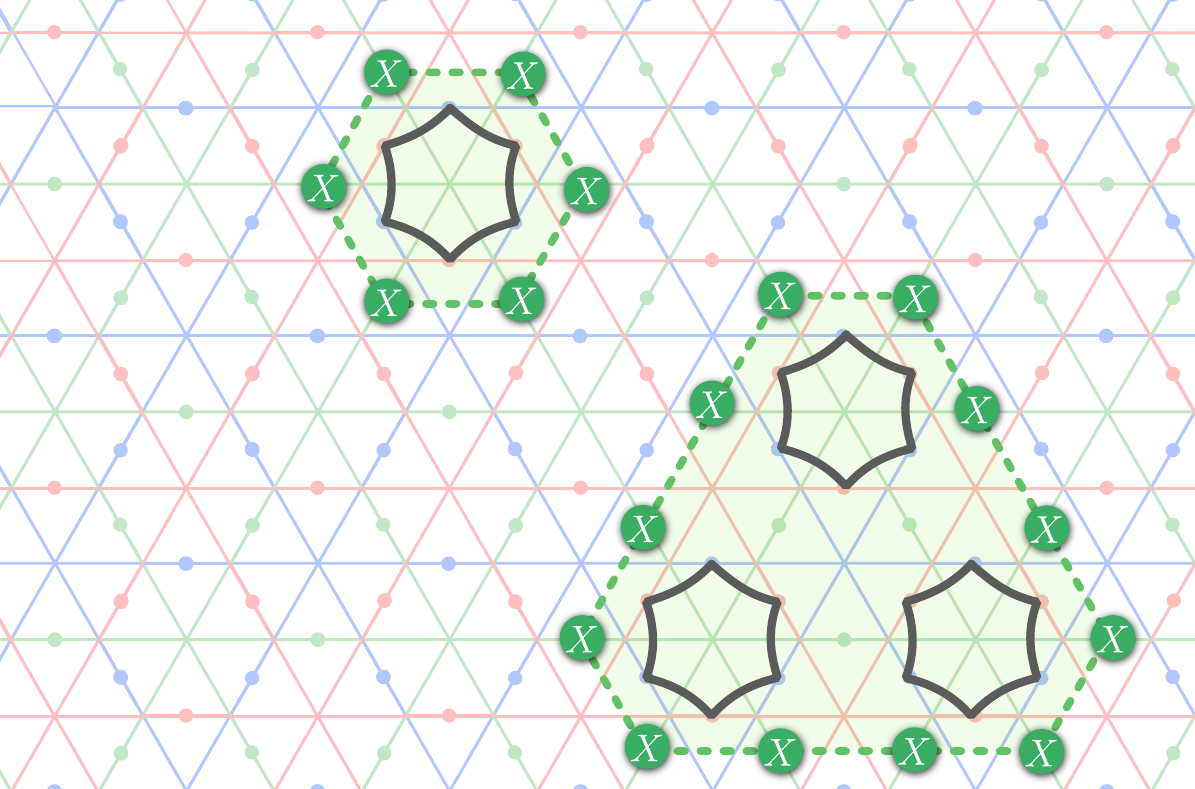}
\caption{\label{fig:postgauge-cliffordstabs} An example of one of the operators that need to be added to the stabilizer group upon gauging two copies of the toric code into the twisted quantum double. These operators consist of the $CZ$-like symmetry shown in Fig.~\ref{Fig:pregauge_model_triangular} applied only in certain domains, surrounded by a green domain wall (seen as a continuous green loop on the dual lattice). }
\end{figure}

Viewed through this lens, we see that the protocol presented in Sec.~\ref{sec:CZ_minimal} is precisely gauging the $CZ$ symmetry where the gauging projection is implemented with quantum measurements, and error correction or post selection.
We also show how the symmetry can be ``ungauged'' using single-qubit Pauli measurements.
This reverts the gauging process while preserving the eigenspace of the gauged symmetry operator.
This way the gauging and ungauging procedure projects onto a definite eigenstate of a logical Clifford operator.

As a side remark, this prescription performs the gauging of symmetry with non-onsite local action (albeit on-site after coarse-graining). For this kind of symmetry action, the usual prescription for gauging that involves the cluster state entangler~\cite{verresen2022efficiently}, would be quite cumbersome. In addition, the usual prescription would prepare the twisted quantum double phase on an entirely new set of qubits, requiring a larger overhead. In contrast, in our prescription, the red and blue qubits are shared between the toric code and the twisted quantum double model.

\section{Additional microscopic examples from gauging}
\label{sec:more_gates}

As we discussed in Section~\ref{sec:GaugingMsmnt} and Appendix~\ref{sec:gauging_CZ}, logical protocols can be implemented via code deformations derived from gauging Clifford symmetries of the input codes. 
We now turn to several additional microscopic examples of logical protocols on codes with open boundary conditions that we derive using this method. We first show a planar preparation of a single $CZ$ state, we then show a version of the $CCZ$ gate, and finally we show a planar model with appropriate boundary conditions to prepare a $T$-state.

\subsection{$CZ$-state preparation with open boundary conditions}

We first discuss a planar version of the $CZ$-magic state preparation scheme. Specifically, we describe how the boundaries are prepared to implement this gate. We also remark that a small implementation of such a $CZ$-state preparation scheme has been demonstrated experimentally in Ref.~\cite{Gupta2024encoding} (there, a slightly different microscopic circuit that is based on a folded color code was used).

We show the lattice geometry in Fig.~\ref{Fig:CZstate}. The figure shows three triangular lattices that each can support a surface code, with qubits on the lattice edges.
Before gauging, the red and the blue copy of the surface code admit a finite-depth logical $\overline{CZ}$ gate, whose expression is 
\begin{equation}
\overline{CZ} \cong \prod_{\textbf{col}(v) = g} \mathcal{A}_{v,g}. 
\end{equation}

We then prepare a gauged code as discussed in Sec.~\ref{sec:CZ_minimal}, except now that the boundary stabilizers are present. For this, we similarly prepare all the green qubits in the $|0\rangle$ state and measure the green star operators, $\mathcal{A}_{v,g}$, as described in Fig~\ref{Fig:CZstate}. 
The product of these measurements returns the value of the parity operator
$
P^g = \prod_v \mathcal{A}_{v,g}
$.
As such, the gauging procedure measures the logical $\overline{CZ}$ operator for the ungauged code. Measurements that project onto the even-charge parity subspace therefore produce a $CZ$-state, assuming we have initialized the red and blue surface code in the logical $\left|\overline{++}\right\rangle$ state.
To initialize, we prepare all the physical qubits on the red and blue lattices in the $|+\rangle$ state and measure the plaquette operators to obtain the logical $|\overline{++}\rangle$ state and the red and blue surface codes. Importantly, before we perform the gauging operation, we have to repeat the plaquette measurements $O(d)$ times to correct the two lattices such that we can reliably recover the $B_p=+1$ subspace.

\begin{figure}[t]
\includegraphics[width = 0.8 \columnwidth]{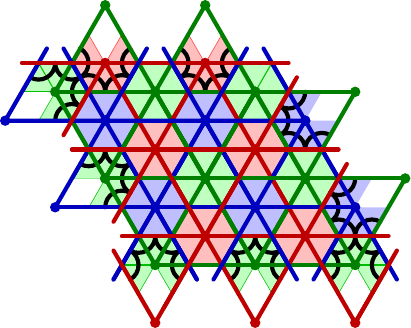}
\caption{ A schematic depiction of the lattice geometry for implementing a logical $\overline{CZ
}$ gauging measurement on two planar surface codes.
We place a qubit on an edge of each color. The surface codes are initialized on red and blue triangular sublattices, which are equipped with appropriate boundary conditions. The stabilizers correspond to 6-qubit Pauli-$X$ vertex terms on red and blue edges and 3-qubit Pauli-$Z$ plaquette terms associated with each triangular plaquette of red or blue color.
The green sub lattice is used as an auxiliary block to implement the gauging.
Upon gauging, we obtain a Clifford stabilizer model with 3-qubit Pauli-$Z$ plaquette stabilizers for each triangle and Clifford stabilizers associated with vertices of all three colors. In the bulk, this model is identical to the one shown in Fig.~\ref{fig:ToricStabilizers}. At the boundary, the $CZ$ parts of stabilizers are shown in black where they differ from the bulk ones. 
}
\label{Fig:CZstate}
\end{figure}

The ungauging procedure is carried out in the same way as described in Sec.~\ref{sec:CZ_minimal}.

\subsection{The $CCZ$ gate}

\begin{figure}[b]
\includegraphics[width = 0.6\columnwidth]{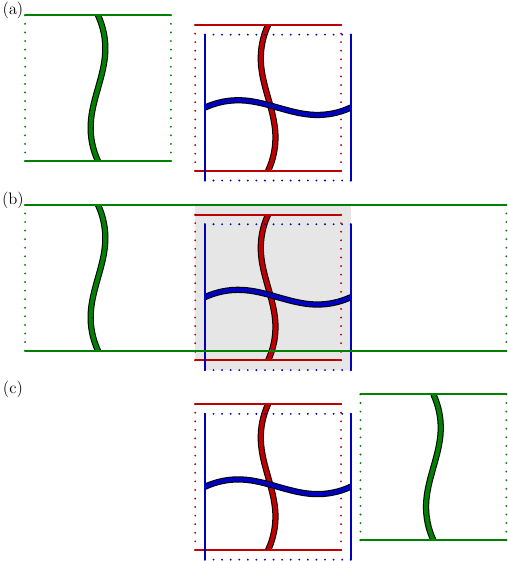}
\caption{\label{Fig:NewCCZ} A constant-depth measurement-based implementation of the logical controlled-controlled phase gate. This accompanies the microscopic details shown in Fig.~\ref{Fig:NewCCZmicroscopic}. Three qubits, colored red, blue, and green, are initialized on surface codes in (a). In step (b), the red and green copies in the middle are gauged, which is done simultaneously to joining them to the green surface code that has been prepared on the left and extending the green surface code to the right. This can be understood as deformation to a new code that consists of an extended region supporting the green surface code and TQD model in the middle.  At the final step, we collapse the left side in the Pauli-$Z$ basis, and ungauge the middle region. This is done by measuring the green qubits in the Pauli-$Z$ basis in the left and middle regions. This completes the $CCZ$ gate with the green logical qubit teleported to the right-hand side, as shown~(c).}
\end{figure}

\begin{figure*}
\includegraphics[width = 0.8 \textwidth]{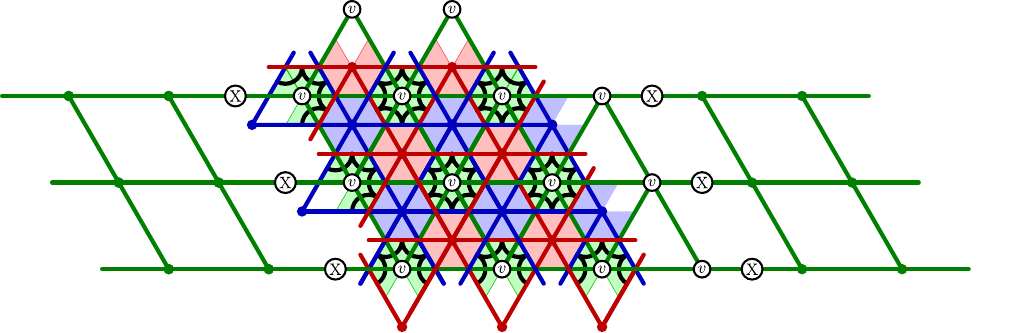}
\includegraphics[width = 0.8 \textwidth]{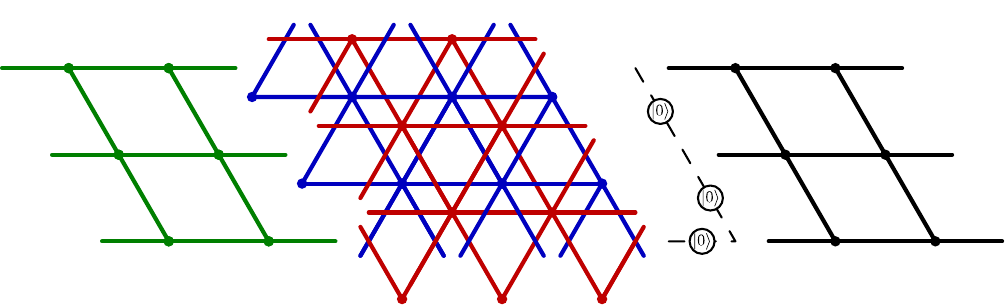}
\caption{The measurement-based CCZ gate by gauging. The system is prepared with one logical qubit on the green surface code to the left of the figure, and two logical qubits are prepared in the red and blue surface codes in the middle of the figure. The auxiliary  (shown in black) surface code on the right is initialized in the logical zero state. The gauging operation merges the central twisted quantum double model with the two adjacent (green and black) surface codes. The product of the marked green vertex operators measure the operator $X_g CZ_{r, b} X_a $. As such, we measure the logical operator needed for the measurement-based $CCZ$ gate, described in the main text, with its inclusion in the stabilizer group of the merged code. \label{Fig:NewCCZmicroscopic}}
\end{figure*}

We now show a version of the protocol for the unitary $CCZ$ gate on three copies of the surface code that we call ``implementation by measurements''.  In Fig.~\ref{Fig:NewCCZ}, we show separate steps of this scheme: the green surface code is deformed over the red and blue surface codes and then brought back. In spacetime this protocol looks like a version shown in Fig.~\ref{fig:protocol2}, except that the spacetime regions are rotated and deformed such that the result matches  In the noise-free case, we can start with the configuration as shown in Fig.~\ref{Fig:NewCCZ} and implement the scheme in just a few steps. In the first step, we prepare a TQD in the middle region with boundary conditions that glue it to an extended region of the green toric code on the left and right from it.
Next, the green code is transported to the other side of the red and blue copy via the ungauging process, where the green qubits that do not maintain the surface code are measured in the Pauli-$Z$ basis.

In Section~\ref{sec:description}, we provided a physical interpretation of the $CCZ$ gate where a green surface code is deformed slowly over a red and blue surface code via measurement. Here, we give an interpretation in terms of measurement-based operations. To show how this gate works, we observe that we can perform a controlled-controlled phase gate on three qubits via measurements together with a single additional qubit prepared in the zero state.  We aim to perform a $CCZ$ operation on the logical qubits indexed $r, g,b$, with the additional qubit indexed $a$, prepared in state $|0\rangle_a$. We begin by measuring $M_1 = \overline{X}_g  \overline{CZ}_{r, b} \overline{X}_a$, and applying a Pauli-$Z$ correction, $\overline{Z}_a$, in case of the $-1$ measurement outcome. Then, we measure the green qubit in the Pauli-$Z$ basis to complete the $CCZ$ operation, while teleporting the green qubit onto the ancilla qubit.
The logical circuit can be depicted via the circuit diagram
\begin{align}
\scalebox{0.8}{ \begin{quantikz}
\lstick{$g$} & \meter[5]{M_1} & & \meter[]{\overline{Z}} \\
\lstick{$r$} & && \rstick{$r$}\\
\lstick{$b$} & && \rstick{$b$}\\
\setwiretype{n} &&\setwiretype{c} & \setwiretype{n} \\
\lstick{$\ket{\overline{0}}$} \setwiretype{q} & & \gate{\overline{Z}}\wire[u][1]{c}  & \rstick{$g$} \\
\end{quantikz}}
\end{align}

In Fig.~\ref{Fig:NewCCZmicroscopic} we show the microscopic details of how the $CCZ$ gate is performed. At the bottom of the figure we show how the system is initialized. The green code is shown to the left and the red and blue codes are shown overlapping in the middle of the figure. An additional surface code region (which we labeled `$a$' above) shown in black is prepared to the right. The qubits of the additional black patch are initialized in the zero state. Next, we perform a gauging operation on the two central code patches, wherein we additionally couple the gauge degrees of freedom to the two adjacent surface codes, i.e., the green logical qubit and the additional logical qubit on the right.
In doing so, we end up measuring the logical operator $ \overline{X}_g  \overline{CZ}_{r, b} \overline{X}_a$ of the initial state. Indeed, this operator is included in the stabilizer group of the new code. Specifically, it is given by the product of all the green vertex (Clifford) stabilizers marked in the top part of Fig.~\ref{Fig:NewCCZmicroscopic}, which are represented by the star operator decoration around each of the green vertices.
As a final step (not shown in Fig.~\ref{Fig:NewCCZmicroscopic}, as it is simply a reflected version of the top panel)  we measure the left code, and the green qubits of the central part in the Pauli-$Z$ basis, to decouple the three surface code copies and to determine the correction we need to apply in the ungauging process. 
This measurement also teleports the green logical qubit onto the additional logical qubit to the right of the system and emulates the motion used in the proposal from Fig.~\ref{fig:protocol2}. As a result, we obtain the $CCZ$ unitary action on the $r,g$ and $b$ qubits up to a Clifford frame $(I \otimes CZ \otimes X)^{m_Z}$ where $m_Z$ is the result of logical $Z$ measurement on the part of the green surface code that was performed during the ungauging procedure.

\subsection{$T$-state preparation}

Finally, we show how we can derive a preparation scheme for a $T$-state, defined as $|T\rangle = T |+\rangle$ where $T = \omega^{|1\rangle\! \langle 1|}$ and $\omega = \exp(i \pi / 4)$ from the gauging approach. The spacetime picture of such scheme is shown in Eq.~\eqref{eq:TGateSplit}; we then derive the miscoscopics from the gauging approach, starting form a color code model on a triangle defined on the lattice shown in Sec~\ref{sec:colorcode}, and obtain the non-Abelian TQD by gauging. We first describe the resulting non-Abelian Clifford stabilizer model. 

In Fig.~\ref{fig:gauged_triangle} we depict the non-Abelian code on a triangle, including the boundary Clifford stabilizers. 
The lattice and stabilizers in the bulk are analogous to Sec.~\ref{sec:colorcode}, except that the plaquettes are now spatially deformed to form hexagons.
The figure also shows how the Clifford stabilizers are defined at each of the boundaries. The boundaries chosen in this geometry are derived from the \emph{color} boundaries of the color code, such that a boundary of color $c$ supports no bricks of color $c$. The left and right boundaries are colored blue and red, respectively, and the bottom boundary is yellow.

\begin{figure}[t]
\includegraphics[width = 0.7 \columnwidth]{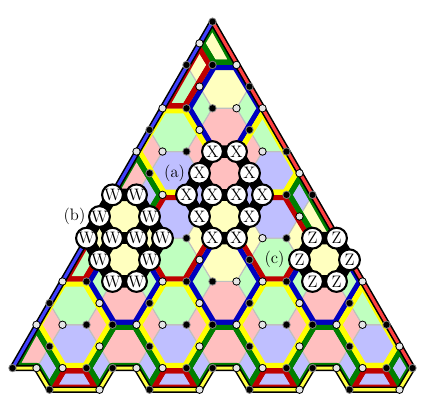}
\caption{Gauged color code on a triangle. Qubits shown in black and white, the stabilizers are associated with bricks and plaquettes. Decorated Clifford stabilizers are associated to the bricks, see, e.g.~(a) and~(b), and Pauli-Z type stabilizers are associated to the ``faces'' which are the intersections of top and bottom bricks,~(c). On the boundary the stabilizers are simply the truncations of the bulk stabilizers. The left and right boundaries are blue and red, respectively, and the bottom boundary is yellow.}
\label{fig:gauged_triangle}
\end{figure}

Following the argument in Sec.~\ref{sec:colorcode} we can show that the gauged model presented in Fig.~\ref{fig:gauged_triangle} is obtained by gauging a symmetry in a model that is equivalent to a 2D color code on a triangle. This color code model is obtained by condensing the green electric charges by measuring all of the green $ZZ$-edge terms of the lattice shown in Fig.~\ref{fig:gauged_triangle}. Specifically, given this color code model the twisted quantum double is obtained by gauging a transversal $\overline{XS}$ gate of such a Pauli model.
On the triangle, the color code encodes a single logical qubit on which the symmetry acts with a logical $XS$ gate.
Gauging this symmetry performs a logical measurement of the associated operator and hence, can be used to prepare a $T$-state, see Sec.~\ref{sec:T-state-macroscopic} and Appendix~\ref{app:onsite-gauging}. See also the appendix of Ref.~\cite{Prem2019gauging} for a similar construction based on gauging the transversal $H$ gate of the 2D color code.

Finally, ungauging this model is performed by measuring the green $Z$ strings similarly to the explanation in Sec.~\ref{sec:3d-2d}. This teleports the logical state back into the ungauged color code.

\section{Toric code phase from the path integral}
 \label{appendix:pathint}

Here, we show how to relate the circuit for the square-lattice toric code to the path integral.  The argument is not specific for a given lattice and can be extended to any spacetime cellulation. We provide only the basic structure and refer the reader to Ref.~\cite{Bauer2023} for a more comprehensive description.  The path integral for the toric code can be written as
\begin{align}
Z&=\sum_{ \text{confs. } \vec{c}}\ \ \prod_{\text{faces $f$}} \omega_f(\vec c)
\nonumber \\
&= \sum_{  \vec{c}} \prod_{f} \delta\Big({\sum_{\partial f} \vec{c} = 0 \operatorname{mod} 2}\Big)\, .
\label{eq:toric_code_path_integral1}
\end{align}
That is, the weight of the path integral 
$$w_f(\vec {c} ) = \delta\Big({\sum_{\partial f} \vec{c} = 0 \operatorname{mod} 2}\Big) ,$$ 
enforces the constraint that the variables on the edges $\partial f$ around the boundary of each plaquette $f$ must have even total parity.

We now explicitly show that applying projectors $\frac12(1+XXXX)$ for every vertex $v$ and $\frac12(1+ZZZZ)$ for every horizontal face $f$ at each timestep realizes the same action between input and output states as the expression for the path integral, upon appropriate identification between $\zz_2$ variables and qubits.

We consider a cubic lattice in spacetime and define the time direction to be along one of the axes of the cubic lattice. We call the edges that are parallel to the time direction \emph{timelike} and the other edges \emph{spacelike}.  Note that the application of $\frac12(1+XXXX)$ at each vertex of the spatial lattice corresponds to a sum over all possible configurations of operators where either $1$ or the $XXXX$ operator is applied at each vertex. We associate the spacelike edges with data qubits for the toric code at corresponding time $t$, and timelike edges with whether the $1$ or the $XXXX$ operator is acting at a given vertex (to which the timelike edge is incident starting at time~$t$).  Here, the timelike edges remain ``virtual'' $\zz_2$ variables, while below, they are treated as ancilla qubits for vertex measurements.%

To compute the global operator resulting from this circuit, we sum over all qubit configuration histories compatible with the actions of the projectors.
Let $g_{t,e}$ denote the value of the qubit at the spacelike square lattice edge $e$ before the time step $t$. Let the $h_{t,v}$ be the variable associated with the timelike edge incident to vertex $v$ that tells us whether the $1 $ or $XXXX$ operator is applied to the edges adjacent to the square-lattice vertex $v$ during time step $t$. We have
$$g_{t+1,e}=g_{t,e}+h_{v_0,t}+h_{v_1,t},$$ where $v_0$ and $v_1$ are the vertices adjacent to the edge $e$, since the qubit configuration at time $t+1$ is obtained from applying $X$ to that at time $t$, according to the variables $h$ at $v_0$ and $v_1$.
The projector $\frac12(1+ZZZZ)$ at a spacelike plaquette $f$ projects onto the space of configurations satisfying
\begin{equation}
g_{t,e_0}+g_{t,e_1}+g_{t,e_2}+g_{t,e_3}=0\;,
\end{equation}
where $e_0$, $e_1$, $e_2$, and $e_3$ denote the edges of $f$.
Hence, the overall operator implemented by the circuit can be written as a sum
\begin{equation}
\label{eq:tc_circuit_path_integral}
\begin{multlined}
\sum_{\vec g, \vec h}  \left (\prod_{\substack{\text{times $t$},\\\text{edges $e$}}} \left (\delta_{g_{t+1,e}=g_{t,e}+h_{v_0,t}+h_{v_1,t}=0} \right ) \right. \times \\
\left.\prod_{\substack{\text{times $t$},\\\text{faces $f$}}} \left (\delta_{g_{t,e_0}+g_{t,e_1}+g_{t,e_2}+g_{t,e_3}=0} \right ) \right)\, .
\end{multlined}
\end{equation}
the second product in the brackets already looks identical to Eq.~\eqref{eq:toric_code_path_integral1} on the spacelike faces. The first product is, in fact, identical to Eq.~\eqref{eq:toric_code_path_integral1} on timelike faces. Thus, the expression that we obtain for the action of the circuit of projectors is identical to  the toric code path integral from Eq.~\eqref{eq:toric_code_path_integral1} with our chosen identification of variables to qubits.
Hence, the evaluation of the $+1$-post-selected toric code circuit is the same as the evaluation of the toric code path integral on a cubic lattice.

When we replace the projectors with measurements, the $-1$ measurement outcome of an $XXXX$-operator at vertex $v$ at time $t$ corresponds to a factor of $(-1)^{h_{v,t}}$ in the path integral at the timelike edge associated with the variable $h_{v,t}$. The $-1$ measurement outcome of the $ZZZZ$-plaquette operator replaces the parity-even constraint with the parity-odd constraint on that plaquette in the path integral. These are associated with flux and charge wordlines, respectively, as we discus in the main text.

\section{Corners and boundaries of the TQD phase}
\label{app:boundaries-DWs-corners}
In this Appendix, we systematically derive the path integral weights that occur at the boundaries and corners of the TQD path integral.
A rigorous approach to this proceeds by defining microscopic path integrals on spacetime triangulations with boundaries, where membranes are represented as cellular cocycles, and the weights are associated to the highest-dimensional simplices.
Then demanding either invariance under local changes in the triangulation (such as Pachner moves) or invariance under gauge transformations that locally change the cellular cocycles.
This yields consistency conditions for the weights that are related to group cohomology that can be solved to determine the weights.
For example, the weight associated to a bulk tetrahedron is a group 3-cocycle, and the weight associated to a boundary triangle is a group 2-cochain that trivializes the bulk 3-cocycle on a subgroup.

Here, we use the simpler continuum closed-membrane picture, where the weights are associated to intersections between lines or membranes.
In this picture, deriving consistency conditions requires us to assume a convention for when membranes or loops are intersecting.
Thus, this method is slightly less rigorous, but it brings us to the correct answer much faster.

We start by discussing the bulk weights associated with the configurations of red, green and blue membranes.
The only special points in the bulk are the triple intersections of the three different membrane types which we associate with a phase $\eta$.
This weight must be a phase to define an path integral associated with a unitary action (i.e.~quantum mechanics is unitary). Changing the orientation of a membrane configuration (i.e.~reflecting it) should correspond to complex conjugating all the weights.
Since the triple intersection is reflection invariant, we have $\eta=\overline\eta$, which has two solutions $\eta=1$ and $\eta=-1$.
$\eta=1$ corresponds to three uncoupled copies of the toric code, whereas $\eta=-1$ is the TQD path integral that we consider here.

Next, we consider the boundaries.
We start by arguing that there exists no topological boundary where membranes of all three colors are allowed to terminate.
Such a boundary would have red, green, and blue termination lines at the boundary.
We would have to associate some weight $\eta_{rg}$ to every intersection of a red and green termination line inside the boundary, and similarly, weights $\eta_{rb}$ and $\eta_{gb}$ to red-blue and green-blue intersections.
Demanding that the weights are invariant under local deformations means that we have the following equation:
\begin{equation}
\begin{gathered}
\raisebox{-0.5\height}{\includegraphics[width = 0.7\columnwidth]{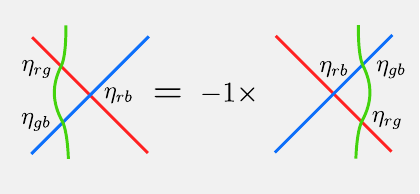}}\\
\eta_{rg} \eta_{rb}\eta_{gb} = -\eta_{rg} \eta_{rb}\eta_{gb} .
\;
\end{gathered}
\end{equation}
Here and below the light gray background represents the boundary in spacetime. 
The two sides show two different patterns of termination lines at the boundary.
To match both sides, the attached bulk membranes (which are not shown) must include a triple intersection on one side, hence the $-1$ factor.
The equation has no non-zero solution, and thus there is no topological boundary where all three membranes terminate.
In the cellular cohomology derivation, the analogous argument is that terminating all membranes corresponds to having the full subgroup of the bulk $\zz_2^{\times 3}$ gauge group on the boundary, and the bulk weight is a cohomologically non-trivial group cocycle.

We next look at the $\langle r,g\rangle$ boundary where the red and green membranes can terminate.
There is a weight $\eta_{rg}$ associated to every intersection of a red and green termination line inside the boundary:
\begin{equation}
\raisebox{-0.5\height}{\includegraphics[width = 0.37\columnwidth]{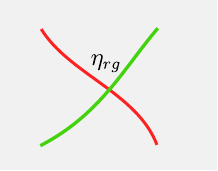}}
\;
\end{equation}
Such red-green intersections are invariant under reflection: in other words, there is no way to distinguish between a ``left-handed'' and a ``right-handed'' intersection.
Thus, unitarity implies that $\bar{\eta}_{rg}=\eta_{rg}$.
So we have two possible solutions, $\eta_{rg}=1$ and $\eta_{rg}=-1$.
These two solutions correspond to two different boundary conditions, however, they correspond to the same topological boundary phase.
This is because one can be obtained from the other by adding a blue membrane along the boundary (``coating'' the boundary with it), which is slightly displaced toward the bulk.
Then, at every red-green boundary intersection, there is a bulk triple-intersection between the attached red and green bulk membranes and the blue coating membrane.
This way, the additional $-1$ factor in $\eta_{rb}$ can be emulated by triple intersections with the help of blue coating membrane.
Since the boundaries are equivalent, we pick the $\eta_{rg}=1$ one.

Next, we consider the $\langle rg,rb\rangle$ boundary.
There can now be both $rg$ and $rb$ termination lines on the boundary.
However, it does not suffice to keep track of where red and green membranes terminate simultaneously, we also need to keep track of whether the red membrane attached to the termination line from the bulk is on the right or on the left from the green membrane attached to this line when looking along the $rg$ ($rb$) termination line.
To see this better, we can imagine that the red and green membranes do not terminate at exactly the same line, but at two $r$ and a $g$ termination lines that are infinitesimally shifted away from each other.
In order to depict different scenarios diagrams, we need to be able to swap the red and the slightly shifted green line, which causes the two attached membranes to cross at a line in the bulk:
\begin{equation}
    \vcenter{\hbox{\includegraphics[width=0.85\linewidth]{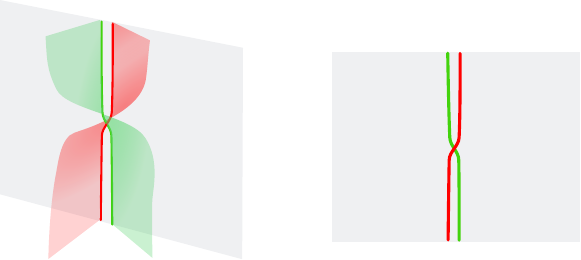}}}
\end{equation}
For the $\langle rg,rb\rangle$ boundary, there are three different types of special points to which we can associate weights. These are the swapping points of $rg$, as well as $rb$ termination lines, and intersections between $rb$ and $rg$ termination lines:
\begin{equation}
    \vcenter{\hbox{\includegraphics[width=0.65\linewidth]{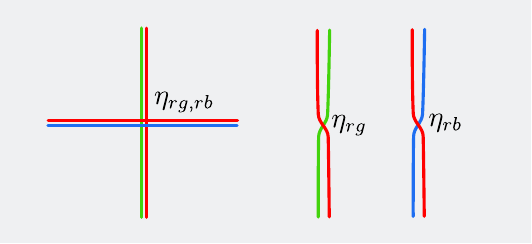}}}
\end{equation}
In space, these weights correspond to the following processes of exchanging two pair termination points on the boundary (which we show by a solid black line in space):
\begin{equation}
    \vcenter{\hbox{\includegraphics[width=0.75\linewidth]{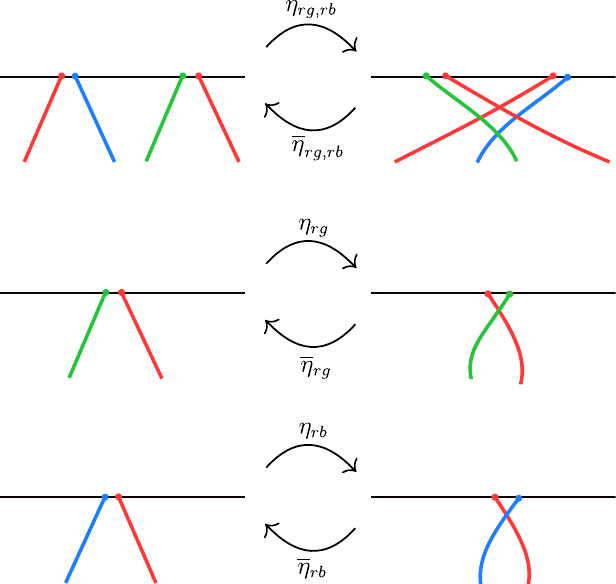}}}
\end{equation}
Note that since $rg$ and $rb$ termination lines consist of infinitesimally separated lines, the above points are not reflection invariant, and can be distinguished from their orientation-reversed versions.
Due to unitarity, the orientation-reversed versions are associated with the complex conjugated weights.
These weights have to fulfill the following consistency condition:
\begin{equation}
    \vcenter{\hbox{\includegraphics[width=0.65\linewidth]{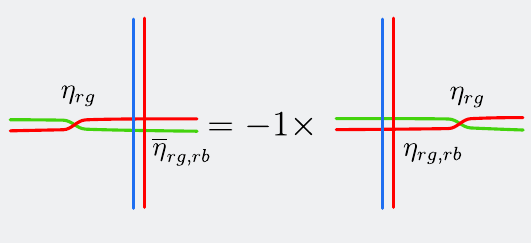}}}
\end{equation}
which gives $  \overline{\eta}_{rg,rb}=-\eta_{rg,rb}$.
The global $-1$ sign occurs because the two sides differ by a triple intersection point between the attached green and blue membranes with the red membrane attached to the vertical red line.
All other consistency conditions are trivial or equivalent to the one above.
The solutions are given by $\eta_{rg,rb}=\pm i$, with $\eta_{rg}$ and $\eta_{rb}$ being arbitrary.
Each of these solutions defines a different boundary condition, but they are all phase equivalent, so we choose the boundary with $\eta_{rg}=\eta_{rb}=1$ and $\eta_{rg,rb}=i$.
To see the equivalence, we first argue that $\eta_{rg}$ can be changed arbitrarily by adding a phase of $\epsilon_{rg}$ to every endpoint of every termination line in every diagram, or $\overline{\epsilon_{rg}}$ if the ordering of red-green is clockwise.
This changes the weight $\eta_{rg}'= \epsilon_{rg}^2 \eta_{rg}$, but it does not affect any of the consistency conditions.
In group cohomology language, the analogous statement is that the two different group 2-cochains are cohomologically equivalent and differ by the boundary of a 1-cochain $\epsilon$.
The same holds for $\eta_{rb}$.
Finally, the $\eta_{rg,rb}=i$ and $\eta_{rg,rb}=-i$ boundaries are equivalent under adding a red membrane that ``coats'' the boundary as
this leads to an additional bulk triple intersection for every boundary $rg$-$rb$ intersection.

Finally, boundaries like $\langle\rangle$ or $\langle r\rangle$ where one or no membrane types are allowed to terminate do not have any non-trivial intersections and therefore there are no non-trivial weights.

Next, we briefly consider the domain walls between the TQD and one or two toric codes.
The domain walls that we use in our protocols simply couple the termination lines of the toric code membranes on one side to the TQD membranes termination lines on the other side of the domain wall.
The weights are the same as for the according boundaries of the TQD.

Finally, we describe the corners that separate different boundary conditions of the TQD.
Since there are $11$ distinct boundaries, there are $11\cdot 10/2 = 55$ different boundary pairs between which we can insert corners. 
For each pair of boundaries, there is a finite number of different superselection sectors of corners.
There exists a standard method for classifying corners in topological fixed-point models by a dimensional reduction or ``folding trick'' to 1+1D via compactifying the bulk into a thin slab with the two boundaries at the top and bottom~\cite{Kitaev2012Models,Bridgeman2018Fusing,Bridgeman2019b,Bridgeman2019a}.
We do not calculate the classification of corners here, but we do derive the path integral weights for several of them from consistency conditions in the closed-membrane picture.

In our fault-tolerant protocols, we only consider corners that are ``simple'' and are fully specified by path-integral weights associated to the points where termination lines cross over the corner from one boundary to another.

As a first simple example, consider corners between the $\langle g,b\rangle$ boundary and $\langle r,b\rangle$ boundary.
Only the blue termination lines can cross over the corner, and the crossing has an associated weight $\eta_c$:
\begin{equation}
    \vcenter{\hbox{\includegraphics[width=0.5\linewidth]{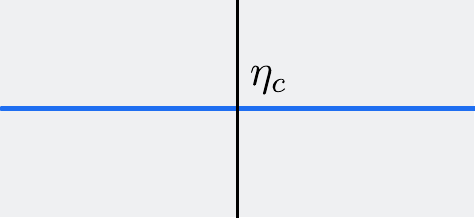}}}
\end{equation}
Here and below, we depict a corner as a black line on a gray background (which denotes both neighboring boundaries). Since the crossing point is invariant under orientation reversal, we have $\eta_c=\overline\eta_c$.
There are no further consistency conditions, so there are two solutions, $\eta_c=\pm 1$.
These two solutions define two distinct corners which are in fact also in different superselection sectors.
However, the $+1$ corner can be turned into a $-1$ corner by ``coating'' it with a blue charge anyon worldline, that is, we put the worldline on the corner and then slightly shift it into the bulk.
The $-1$ factor we get from the blue termination line crossing the $-1$ corner is then emulated by the $-1$ factor we get when the coating anyon worldline crosses the attached blue membrane in the bulk.
Which of the two corners we choose in the global topology of some fault-tolerant protocol changes the resulting logical operation.
However, this change simply corresponds to a Pauli operator, usually a Pauli-$Z$ operation acting on one of the logical qubits.

The corners between all $\langle g,b\rangle$, $\langle r,g\rangle$, and $\langle g,b\rangle$ in our protocols in Section~\ref{sec:global_topologies} were all implicitly $+1$ corners.
If, for example, we instead took the front right corner between the $\langle g,b\rangle$ and the $\langle r,g\rangle$ boundaries in Eq.~\eqref{eq:czmeasure_global_protocol} to be a $-1$ corner, as a result, this would swap the two measurement outcomes of the logical $CZ$ measurement (and not apply a Pauli logical operator).
Note that it is in principle possible to define a corner where the blue termination lines are not allowed to cross.
However, this corner is not irreducible but rather is the direct sum of the $+$ and $-$ corners.
If we used such a corner as a spatial corner in a protocol, it would carry a local logical qubit which is not topologically protected.

We next look at the corners between the $\langle rg,rb\rangle$ boundary and the $\langle r,g\rangle$ boundary. Note that when crossing over the corner, termination lines are allowed to split. For example, a red-green termination line on one side may split into a separate red and green termination line on the other side.
There is a weight $\eta_{c,rg}$ assigned to the points on the corner where an $rg$ termination line on one side splits into separate $r$ and $g$ termination lines on the other side.
\begin{equation}
\vcenter{\hbox{\includegraphics[width = 0.45\columnwidth]{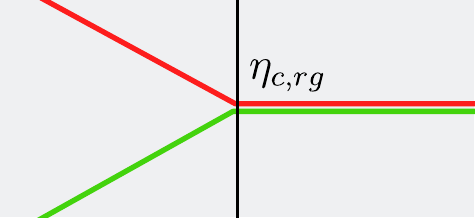}}}
\;
\end{equation}
Due to the infinitesimal shift between the red and the green part of the termination line, the point above is not orientation-reversal invariant.
However, since crossing the red and green termination lines on the left does not produce any non-trivial weights, get a consistency condition
\begin{equation} \label{fig:G16}
    \includegraphics[width=1\linewidth]{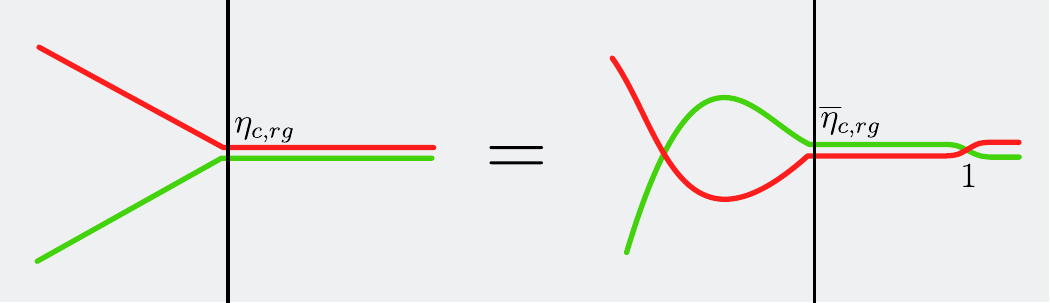}
    \;
\end{equation}
and the two possible solutions are again $\eta_{c,rg}=\pm 1$.
The $-1$ boundary differs from the $+1$ boundary by coating with either a red or by a green charge worldline.
Again, we use the $+1$ corner in our protocols, and using the $-1$ corner instead would only lead to additional logical Pauli operators.

Finally, we consider the corners between the $\langle rg,rb\rangle$ boundary and the $\langle g,b\rangle$ boundary.
There is a weight $\mu$ at every point on the corner where an $rg$ and a $rb$ termination line on one side can terminate at the same corner point as a $g$ and a $b$ termination line on the other side:
\begin{equation}
\vcenter{\hbox{\includegraphics[width = 0.5\columnwidth]{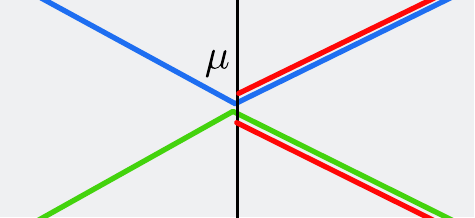}}}
\label{eq:eighthroot}
\end{equation}
In space, this weight corresponds to the process depicted below of moving termination points across the corner (in space, we schematically show the corner as a black dot on the boundary):
\begin{equation}
\vcenter{\hbox{\includegraphics[width = 0.8\columnwidth]{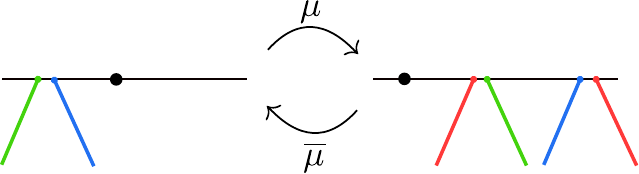}}}
\end{equation}
The weight $\mu$ has to fulfill the following consistency condition:
\begin{equation}
\begin{gathered}
\includegraphics[width = 1\columnwidth]{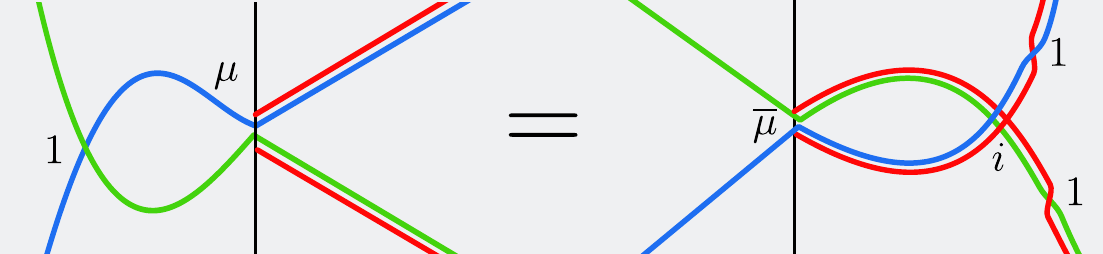}\\
\overline\mu=i\mu\;.
\end{gathered}
\end{equation}
There are two solutions to this equation, namely $\mu=\pm \omega = \pm e^{2\pi i\frac18}$.
The $-\omega$ corner differs from the $+\omega$ corner by coating with either a green or a blue charge anyon worldline.
So again, it does not matter which corner we choose, and we pick the $+\omega$ corner.

In addition to corners between boundaries, our protocols also involve corners that are the termination of a domain wall between the type-III TQD and one or two copies of the toric code. 
For these corners, we can ignore the attached toric code and the path integral weights are the same as for the according corners between two TQD boundaries.

Finally, our protocols also involve 0-dimensional defects in spacetime, which are events where multiple corners meet.
In general, such 0-dimensional defects are classified by vectors in a vector space.
This vector space is the ground state space of the model (i.e.~the logical space of the code) on a 2-dimensional topology, which is obtained by taking the intersection of the global topology with a small sphere that encloses the 0-dimensional point. 
For the protocol to be fault-tolerant, this ground state space must be 1-dimensional.
In fact this is not only true for the 0-dimensional defects but for any point on a corner, boundary, and in the bulk.
So the choices of 0-dimensional defects only differ by a number.
In a fault-tolerant protocol, this number is fixed by the fact that the circuits must be trace-preserving, or in other words, these numbers are only global prefactors which are irrelevant in quantum mechanics.

As an aside, note that there are also more interesting corners that do not just assign weights to crossing termination lines, and which we do not use in our protocols.
For example, consider the following ``non-Abelian'' corner between the $\langle r,g\rangle$ boundary and itself:
(1) the corner is a fixed termination line for the blue membranes;
(2) we sum over all green string configurations inside the corner whose boundary points are given by the points where the red termination lines cross the corner (that is, for a connected corner segment, there are two possible opposite string configurations);
(3) for every green termination line crossing the green string on the corner, we get a factor of $-1$.
\begin{equation}
\vcenter{\hbox{\includegraphics[width = 0.55\columnwidth]{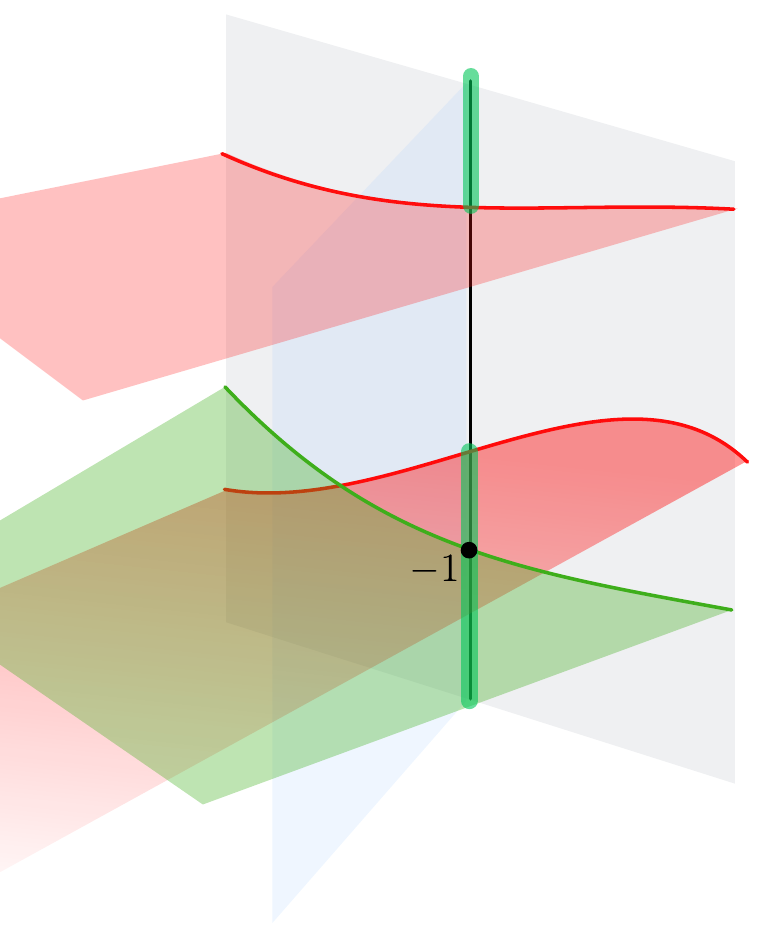}}}
\end{equation}

\section{Non-Abelian TQD model from copies of the 3D toric code}
\label{sec:3d-2d-toric-codes}

Here, to complement the discussion in Sec.~\ref{sec:3d-2d}, we present a similar protocol for alternative microscopic realization of the TQD model from three-dimensional toric codes that support a transversal controlled-controlled-$Z$ gate.
We start with a cubic lattice with three qubits on each edge, indexed by $\{R,B,G\}$, with periodic boundary conditions.
We denote the set of vertices with $V$ and the set of faces with $F$.
Moreover, we bi-color the cubic volumes such that no two cubes of the same color share a face.
We label the set of cubes with the first color, white, with $C_0$ and the cubes with the second color, gray, with $C_1$. Their union, the set of all cubes, is labeled $C$.

The 3D code we start with is defined by a stabilizer group
\begin{align}\label{eq:SABC}
    S^{RBG} = \langle S^R, S^B, S^G \rangle,
\end{align}
where $S^R$ only acts non-trivially on qubits indexed by $R$, $S^B$ on qubits indexed by $B$ and $S^G$ on qubits indexed by $G$.
Each of the three stabilizer groups define a three-dimensional toric code on the qubits $A$, $B$ and $C$ individually.

The code on qubits $R$ is defined via
\begin{equation}\label{eq:SB_bulk}
 \begin{split}
    &S^R = \Big\langle A_{c_0}^R , B_{v}^R \, \big| c_0 \in C_0 , v \in V \Big\rangle, \\ &A_{c_0}^R = \prod_{e\in c_0} X_e^R, \ \ B_{v}^R = \prod_{e\in \delta v \cap C_1} Z_e
 \end{split}
\end{equation}
where, in slight abuse of notation, $e\in c$ denotes the set of edges surrounding the cube $c$.
$S^R$ has two types of generators. The $X$ stabilizer generators are of weight $12$ and act on the qubits surrounding a cube in $C_0$ and the $Z$ stabilizer generators are of weight 3 and act on the qubits in the corner of cubes in $C_1$.
Note that this code is the same as a three-dimensional surface code defined on a triangulation with degree 12 vertices.
In Ref.~\cite{Vasmer2019three} this code was introduced as the `rectified cubic code' and used in Ref.~\cite{Brown2020universal} to realize a logical CCZ gate on three 2D surface codes.

The stabilizer group $S^B$ is defined in the same way as $S^R$ but with the roles of the colors of the cubes reversed, i.e. $C_0\leftrightarrow C_1$, meaning that in $S^B$ the cubes in $C_1$ define the $X$ stabilizers and the cubes in $C_0$ the $Z$ stabilizers.

For the third subset of qubits, $C$, we pick a different microscopic realization of the toric code, defined by the stabilizer group
\begin{equation}\label{eq:SC_bulk}
 \begin{split}
    &S^G = \langle A_v^G , B_f^G  \;|\;
    v\in V, f\in F \rangle, \\
    &A_v^G = \prod_{e \in \delta v} X^G_e,\ \  B_f^G = \prod_{e\in \partial f} Z_e^G 
 \end{split}
\end{equation}
where $\partial f$ denotes the edges that contribute to the boundary of $f\in F$ and $\delta v$ the edges that have $v$ in their boundary.
Note that this simply defines a conventional three-dimensional toric code on a cubic lattice.

We illustrate the stabilizer generators of $S^{RBG}$ in Fig.~\ref{fig:3Dcubic-codes}.
For the rest of this section, we denote the single-qubit Pauli $P$ operator acting at the edge $e$ of the $A$ sublattice as $P_e^R$ (and similarly for $B,C$). 

$S^{RBG}$ defined above has a transversal 3-qubit controlled-controlled-$Z$ gate, that acts on every triple of qubits assigned to the same edge via $CCZ\ket{abc} = (-1)^{abc}\ket{abc}$, where $a,b,c\in\{0,1\}$ label the computational basis states on a system of three qubits.
For the rest of this section we denote the transversally applied $CCZ$ gate with $\overline{CCZ} = \prod_{e} CCZ_e$.
That $\overline{CCZ}$ is a logical gate of the 3D code defined above can be verified by conjugating the Pauli stabilizer group with the transversal gate, $\Tilde{S}^{RBG} =\overline{CCZ} S^{RBG} \overline{CCZ}$.
This maps the generators above to Clifford stabilizers where the $X$ stabilizers acquire additional $CZ$ terms.
By direct calculation we find that any such Clifford stabilizer $\Tilde{A}\in \Tilde{S}^{RBG}$ preserves the Pauli stabilizer group, i.e. $\Tilde{A}^\dagger S^{RBG} \Tilde{A} = S^{RBG}$ and hence is a Clifford stabilizer that acts like the logical identity on the codespace.\footnote{The code defined by $S^{RBG}$ has a macroscopic distance proportional to the linear size of the system, and $\Tilde{S}^{RBG}$ is locally generated.}
Any thin slice of a 3D model obtained by introducing boundaries perpendicular to one of the three spatial directions has the same type of Clifford stabilizers in the bulk (regardless to whether the full code has a logical non-Clifford gate), i.e. for any $X$ stabilizer $S_X$ that is unmodified when introducing the boundaries there is an associated Clifford stabilizer obtained by conjugating $S_X$ with a transversal $CCZ$ gate.

\begin{figure}[b]
    \centering
    \includegraphics[width=\linewidth]{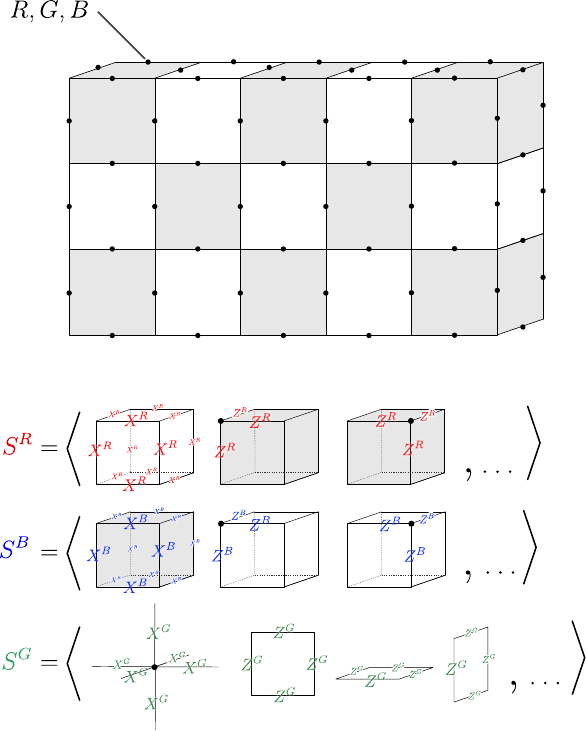}
    \caption{We define three copies of toric codes on a square lattice on which we bi-colore the volumes.
    On each edge, we place is a triple of qubits $R\times B\times G$ and define a stabilizer group $S^{RBG} = \langle S^R, S^B, S^G\rangle$, whose generators are shown on the right.
    The coloring is only necessary to define $S^R$ and $S^B$. $S^G$ defines a usual toric code associating $X$ stabilizers to vertices and $Z$ stabilizers to the faces of the cubic lattice.}
    \label{fig:3Dcubic-codes}
\end{figure}

\subsection{The model in a thin 3D slab}

We now consider the 3D model on a thin slab, namely a manifold with the topology of $\mathbb{T}^2\times [0,1]$ by introducing boundaries perpendicular to the $110$ direction, as depicted in Fig.~\ref{fig:thin3D-stabs}.
Let $Q$ be the qubits within the thin slab, which we assume to have thickness 1 (originating from a single layer of ``rotated'' cubes)
The boundaries are introduced by measuring all qubits in the complement of the slab, $Q^c$ in the single-qubit $X$ basis.
This creates two boundaries of the type labeled by $\langle r,g,b\rangle$ in Fig.~\ref{fig:thin3D-TQFT}.
The stabilizer group of the resulting thin 3D model is generated by $X$-type stabilizers in the bulk, as well as truncated $X$-type stabilizers close to the boundary and all $Z$-type stabilizers that are fully supported in $Q$.
Since we chose the $110$ direction to be the ``thin'' direction, the cubes touch neighboring cubes along a face in one direction and along a single edge in the other direction, see Fig.~\ref{fig:thin3D-stabs}.

The model also admits Clifford stabilizers inherited from the 3D code it was derived from.
For each white cube $c_0$ in the thin slab, the code admits a Clifford stabilizer
\begin{align}
    \widetilde{A}_{c_0}^R = \prod_{e\in c_0} X_e^R CZ^{BG}_e,
\end{align}
for each grey cube $c_1$ in the thin slab, it admits a Clifford stabilizer
\begin{align}
\widetilde{A}_{c_1}^B = \prod_{e\in c_1} X_e^B CZ^{RG}_e,    
\end{align}
and for each six-valent vertex in the thin slab, the code admits the Clifford stabilizer
\begin{align}
    \widetilde{A}^G_v = \prod_{e\in\delta v} X^G_e CZ^{RB}_e.
\end{align}

\begin{figure}[t]
    \centering
    \includegraphics[width=0.8\linewidth]{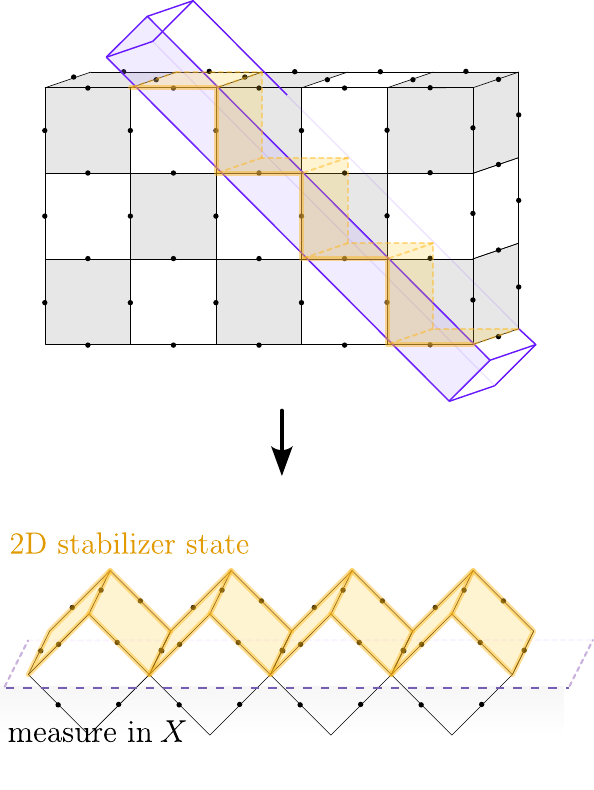}
    \caption{We define a thin 3D model on the qubits in a thin slab that we choose perpendicular to the $110$ direction of the cubic lattice, indicated in blue.
    Measuring all qubits outside of the thin slab in the $X$ direction gives a code defined on a thin slab of qubits $Q$, shown on the right.
    The qubits on the orange edges are part of the top boundary of the thin model, which we label by $Q_t$.
    They form the qubits on which the non-Abelian code is supported after the initialization protocol.}
    \label{fig:thin3D-stabs}
\end{figure}

We now consider two different examples of a 2D code obtained from this 3D code.
In the first example, we simply perform single-qubit Pauli measurements on the qubits of the bulk and bottom boundary, which prepares a codestate on the qubits of the top boundary equivalent to three independent 2D toric codes.
In the second example, before measuring out the bulk and bottom qubits we apply a transversal $CCZ$ gate to the bottom and bulk qubits.
This presents an alternative realization of the 2D code with non-Abelian topological order equivalent to the non-Abelian TQD discussed in the main text.

\subsection{Initializing an Abelian 2D code from 3D toric codes}\label{subsec:init-Ab}

We start with a 3D toric code on a thin slab described in the previous section.
We sort the qubits into two subsets: qubits that are part of the top boundary, $Q_t$, and qubits in the complement, $Q_t^c$.
For the purposes of this section, we chose $Q_t$ as indicated in Fig.~\ref{fig:thin3D-stabs}.\footnote{The bipartition of $Q=Q_t\sqcup Q_t^c$ is not canonical. There is some freedom in the assignment of qubits to the top boundary.
For example, in the context of a repeating protocol that periodically switches between a thin 3D code and a 2D code on its boundary (as in Refs.~\cite{bombin2018, Brown2020universal}) a natural choice for $Q_t$ would only include qubits that are measured out in the next period of switching from the 3D code to the 2D code.
We expect that the protocol works for any choice of $Q_t$ that is not cleanable in the 3D code.}
The qubits in $Q_t$ live on the edges of a 2D square lattice.

Starting with the stabilizer group $S^{RBG}$ of the thin model, we measure the qubits in $Q_t^c$ in the single-qubit Pauli-$X$ basis.
The post-measurement state is stabilized by $S_t^{RBG} = S_t^R\otimes S_t^B\otimes S_t^G$, which factorizes over the qubits $R\cap Q_t$, $B\cap Q_t$ and $G\cap Q_t$, respectively.

Recall that the lattice formed by the boundary qubits $Q_t$ is a square lattice, consisting of horziontal and vertical edges.
We find that $S_t^R$ contains single-qubit $X$-type stabilizers on every other vertical edge.
Since these qubits are in a product state, disentangled from the rest, we can remove these edges from the square lattice. This results in a hexagonal lattice on which $S_t^R$ reduces to the usual toric code stabilizer group on its dual, triangular, lattice.
We illustrate the effective toric code in Fig.~\ref{fig:2D-3D-effective-triangular-TC}. The effective toric code given by $S_t^B$ is found analogously upon switching gray and white cubes and thus lives on a triangular lattice that is shifted by half a lattice site with respect to the lattice defining $S_t^R$.

Lastly, we find that after measuring $Q_t^c$ in the $X$ basis $S_t^G$ is the stabilizer group of a 2D toric code on the square lattice formed by $Q_t$.
Each (4-valent) vertex is associated with an $X$-type stabilizer $A_v^G$ and every (square) plaquette is assocaited with a $Z$-type stabilizer, and $B_f^G$ as defined in Eq.~\eqref{eq:SC_bulk}.

Thus, we identified the 2D code on $Q_t$ to be equivalent to three decoupled two-dimensional toric codes.
This is useful when analyzing the non-Abelian model in the next subsection.

\begin{figure}
    \centering
    \includegraphics[width=0.85\linewidth]{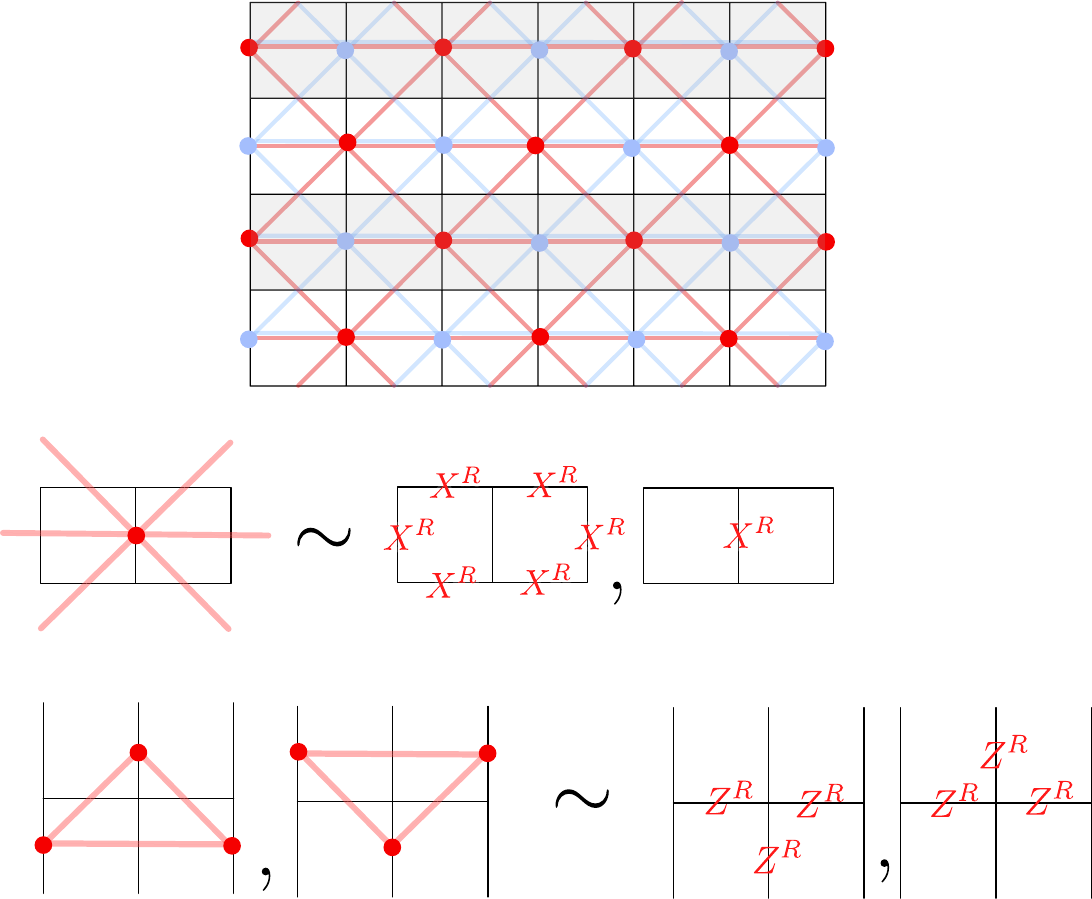}
    \caption{
    Measuring all qubits except the top boundary qubits $Q_t$ in the $X$ basis yields a code defined on a square lattice.
    On both sets of qubits $R$ and $B$ the resulting stabilizer group $S^R$ and $S^B$ is equivalent to toric codes defined on triangular lattices since the stabilizer group contains some single-qubit Pauli operators.
    The lattices are shown in red and green, layed over the square lattice formed by $Q_t$.
    They are shifted horizontally with respect to each other and we can place the qubits associated to the diagonal edges on top of each other.
    The code defined by $S^R\otimes S^B$ admits a non-trivial transversal logical gate $\overline{CZ}$ that acts with a $CZ$ gate on all pairs of qubits on diagonal edges of the red, respectively green, lattice.
    The $G$ qubits are stabilized by a stabilizer group $S^G$ of a toric code defined on the (black) square lattice with plaquette and vertex stabilizers of weight 4.}
    \label{fig:2D-3D-effective-triangular-TC}
\end{figure}

\subsection{Initializing a non-Abelian 2D code from 3D toric codes}\label{subsec:init-nonAb}
To initialize the non-Abelian code, we apply the $\overline{CCZ}$ gate to the qubits in $Q_t^c$ before measuring them in the $X$ basis; the resulting model has Clifford stabilizers.
In Sec.~\ref{subsec:phase-nonAb-TC} we show that this code realizes the non-Abelian topological order corresponding to the type-III twisted TQD by identifying it as a code obtained from gauging an appropriate anyon-permuting $\zz_2$ symmetry in two 2D toric codes (which was discussed in Appendix~\ref{sec:gauging_CZ}).

We start with a thin 3D code in the $+1$-eigenspace of all Pauli-$Z$ stabilizers.
This can be achieved either by post-selection or by error correction in the 3D code.
In this case, the 3D model on the thin slab admits Clifford stabilizers $\widetilde{A}_{c_0}^R, \widetilde{A}_{c_1}^B$ for every white cube $c_0$ and every gray cube $c_1$ that are fully contained in $Q$, as well as $\widetilde{A}_v^G$ for every six-valent bulk-vertex.
Applying $CCZ$ transversally to every triple of qubits on the same edge in $Q_t^c$ maps these Clifford stabilizers to stabilizers that act with Pauli-$X$ on the qubits in $Q_c^t$, namely
\begin{subequations}\label{app:eq:mixed-stabs-thin3D}
\begin{align}
    \widetilde{A}_{c_0}^R &\stackrel{\eval{\overline{CCZ}}_{Q_t^c}}{\longmapsto} \prod_{e\in c_0\cap Q_t^c} X_e^R \prod_{e\in c_0\cap Q_t} X_e^R CZ_e^{BG},\\
    \widetilde{A}_{c_1}^B &\stackrel{\eval{\overline{CCZ}}_{Q_t^c}}{\longmapsto} \prod_{e\in c_1\cap Q_t^c} X_e^B \prod_{e\in c_1\cap Q_t} X_e^B CZ_e^{RG},\\
    \widetilde{A}_{v}^G &\stackrel{\eval{\overline{CCZ}}_{Q_t^c}}{\longmapsto} \prod_{e\in \delta v\cap Q_t^c} X_e^G \prod_{e\in\delta v\cap Q_t} X_e^G CZ_e^{RB}.
\end{align}
\end{subequations}
In contrast, the $X$ stabilizers that have overlap with $Q_t^c$ are mapped to stabilizers that act with a Clifford unitary on these qubits.

Next, we measure the qubits in $Q_t^c$ in the $X$ basis. This preserves the stabilizers shown in Eq.~\eqref{app:eq:mixed-stabs-thin3D} since they commute with the measurements. By discarding the part supported on the qubits in $Q_t^c$ (which can be done because these qubits are measured out in the $X$-basis), we obtain Clifford stabilizers solely supported on $Q_t$.
Together with the $X$- and $Z$-type stabilizers of the thin 3D slab that have already been supported only on $Q_t$, they generate the Clifford stabilizer group of the post-measurement (2D) state.
We show the stabilizer generators of the non-Abelian code in Fig.~\ref{fig:nonAbelianABC}.
In slight abuse of notation we refer to the Clifford stabilizers that only act on $Q_t$ as $\widetilde{A}_{c_0}^R, \widetilde{A}_{c_1}^B$ and $\widetilde{A}_{v}^G$.

\subsection{The phase of the non-Abelian code}\label{subsec:phase-nonAb-TC}
In the following we argue that the Clifford-stabilized code defined above can be obtained by gauging a $\zz_2$ symmetry on two toric codes that is realized by a product of on-site $CZ$ gates, see App.~\ref{sec:gauging_CZ}. 
This shows that the anyons of the code are equivalent to those in the quantum double of $D_4$ and the non-Abelian type-III twisted quantum double of $\zz_2^{\otimes 3}$ considered in this paper.

We start with the 2D Abelian model described in Sec.~\ref{subsec:init-Ab}.
Consider the two decoupled codes defined by the stabilizer group $S_t^R\otimes S_t^B$.
The stabilizer generators can be identified with vertex and plaquette terms of two toric codes defined on triangular lattices, see Fig.~\ref{fig:2D-3D-effective-triangular-TC}.
We refer to these two triangular lattices as $R$- and $B$-sublattices.
Since the lattices are shifted horizontally with respect to each other, the qubits on the diagonal edges of $A$ and $B$ can be identified one-to-one.
Specifically, let $D_A$ be the set of (qubits associated to) diagonal edges of the $A$-lattice and, similarly for the $D_B$ for the $B$-lattice.
For each edge $e\in D_A$ there exists a unique edge in $D_B$ that crosses it. Hence, $e\in D_A$ can be considered as a label for both an element in $D_A$ and $D_B$.

We define the global Clifford unitary
\begin{align}
    \overline{CZ} = \prod_{e\in D_A} CZ_e^{RB}
\end{align}
that acts with a product of $CZ$ gates on the pairs of qubits on the associated diagonal edges.
We find that $\overline{CZ}$ and the $X$-type stabilizers commute up to $Z$ stabilizers.
This shows that $\overline{CZ}$ is a symmetry of the code $S_t^R\otimes S_t^B$.
The action it has on the code space can be inferred from its commutation relation with the logical Pauli operators, which are the anyon string operators along non-trivial (dual) cycles of the $R$- and $B$-lattices.

Let $\overline{L}_m^R$ be a $m$-anyon string operator acting with products of Pauli $X$ operators along a non-trivial cycle of the dual $A$-lattice.
We find that
\begin{align}
    \overline{CZ}\;\overline{L}_m^R \;\overline{CZ} = \overline{L}_m^R \overline{L}_e^B,
\end{align}
where $\overline{L}_e^B$ is an $e$ string operator that acts with Pauli $Z$ operators along a cycle of the $B$-lattice parallel to $\overline{L}_m^R$.
By symmetry the same holds for the string operators along inequivalent cycles and when $A$ and $B$ are interchanged.
This shows that $\overline{CZ}$ acts as the anyon-permuting $\zz_2$ symmetry discussed in Appendix~\ref{sec:gauging_CZ}.

\begin{figure}
    \centering
    \includegraphics[width=\linewidth]{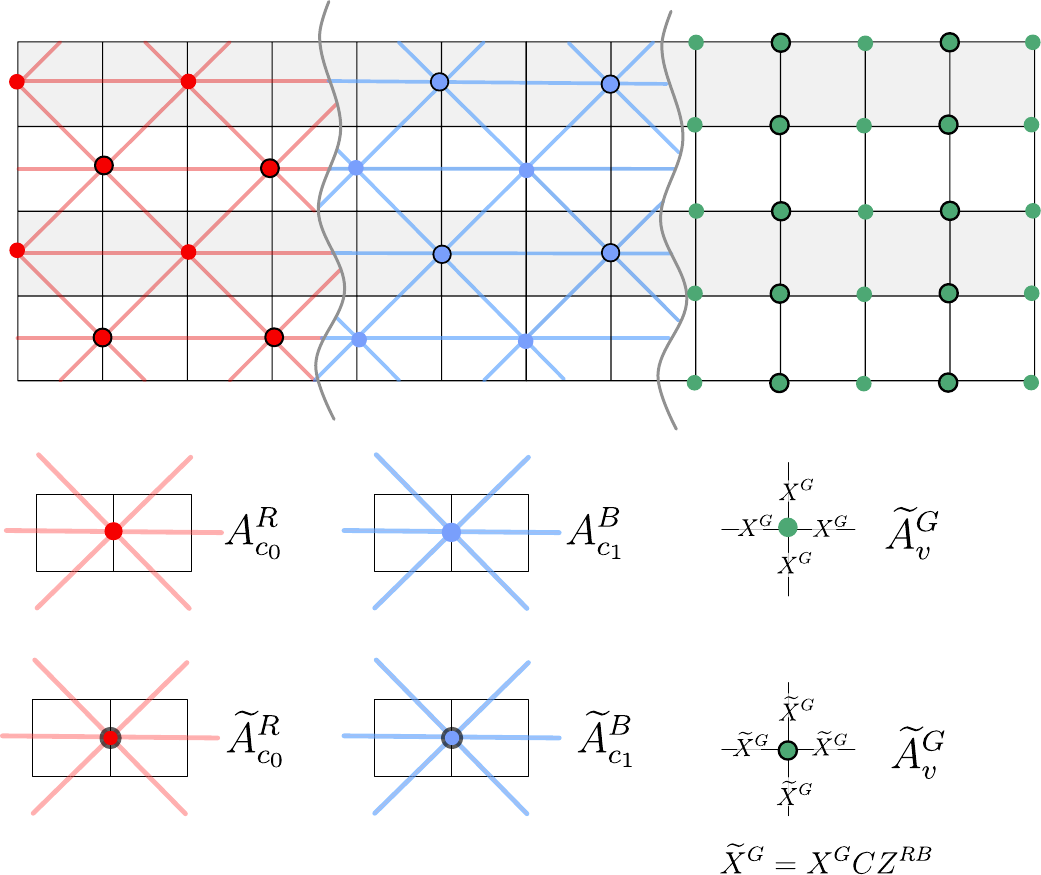}
    \caption{The Clifford stabilizer generators of the non-Abelian code obtained using the transversal $CCZ$ gate in three 3D toric codes. The model also has triangular plaquette Pauli-$Z$ stabilizers for $R,B$ qubits on respective sublattices and square plaquette $Z$ stabilizers on $G$ sublattice. 
    This model can be viewed as a model obtained by gauging a $\zz_2$ symmetry in the two toric codes shown in Fig.~\ref{fig:2D-3D-effective-triangular-TC}.
    The product over all vertex stabilizers $\widetilde{A}_v^G$ inherited from the stabilizers on the $C$ qubits exactly reproduces the $\overline{CZ}$ symmetry of $S^R\otimes S^B$.}
    \label{fig:nonAbelianABC}
\end{figure}

Having identified an on-site symmetry in two triangular toric codes, we turn our attention back to the Clifford-stabilized code from Sec.~\ref{subsec:init-nonAb} and argue that it can be obtained from gauging this symmetry.
Consider the vertex stabilizers of the non-Abelian code associated with the qubits of the $G$ sublattice.
We label them with $\widetilde{A}_{v}^G$.
If $v$ is a bulk vertex of the thin 3D model, then $\widetilde{A}_{v}^G$ is a Clifford stabilizer composed of a product of $X^G CZ^{RB}$. For $v$ a boundary vertex of the thin 3D model, $\widetilde{A}_{v}^G$ corresponds to a Pauli-$X$ stabilizer.
We find that
\begin{align}
    \prod_v \widetilde{A}_{v}^G = \overline{CZ}.
\end{align}
Hence, an eigenstate of all $\widetilde{A}_{v}^G$ stabilizers individually is automatically an eigenstate of $\overline{CZ}$.
We can initialize a given eigenstate of each $\widetilde{A}_{v}^G$ by first fixing the $+1$ eigenstate of all $Z$ stabilizers in $S_t^{RBG}$, by e.g. measuring them and applying suitable $X$ corrections (or via post-selection), and then measuring each of the $\widetilde{A}_{v}^G$ stabilizers.
Given that the measurement was applied to the ``flux-free'' state in the $+1$ eigenspace of all $Z$-type stabilizers, we can apply $Z$-type Pauli corrections to resolve the $-1$ measurement outcomes.
Thus, we obtain the non-Abelian model defined above by gauging the symmetry of the $A$ and $B$ toric codes. At the same time, as we know from Appendix~\ref{sec:gauging_CZ}, this must  produce the state with non-Abelian topological order whose anyons are described by the non-Abelian type-III twisted  quantum double.

\vfill
\bibstyle{plain}
\bibliography{ref}

\end{document}